%% file: main.tex
\documentclass[letterpaper,twocolumn,10pt]{article}
\usepackage{usenix}
\usepackage{amsmath}
\usepackage{amsfonts}
\usepackage{amssymb}
\usepackage{amsthm}
\usepackage{subfig}
\usepackage{graphics}
\usepackage{graphicx}
\usepackage{array,makecell}
\usepackage{soul}
\usepackage[ruled,vlined,linesnumbered]{algorithm2e}
\usepackage{bbold}
\usepackage{hyperref}
\usepackage{url}
\usepackage{booktabs}
\usepackage{enumitem}
\usepackage{microtype}
\usepackage{xcolor}
\usepackage{float}  
\usepackage{tabularx}
\usepackage{etoolbox}  

\providecommand{\citep}[1]{\cite{#1}}
\providecommand{\citet}[1]{\cite{#1}}
\providecommand{\footref}[1]{\textsuperscript{\ref{#1}}}

\newif\iffullversion
\fullversiontrue
\newif\ifhighlightchanges
\highlightchangestrue
\definecolor{ChangeMark}{RGB}{0,70,200}

\newtheorem{theorem}{Theorem}
\newtheorem{proposition}[theorem]{Proposition}
\newtheorem{lemma}[theorem]{Lemma}

\newtheorem{remark}{Remark}
\newtheorem*{lemma*}{Lemma}

\newtheorem*{Proposition*}{Proposition}

\usepackage{nicefrac}

\makeatletter
\renewcommand{\Indentp}[1]{%
  \advance\leftskip by #1
  \advance\skiptext by -#1
  \advance\skiprule by #1}%
\renewcommand{\Indp}{\algocf@adjustskipindent\Indentp{\algoskipindent}}
\renewcommand{\Indm}{\algocf@adjustskipindent\Indentp{-\algoskipindent}}
\makeatother

\AtBeginEnvironment{algorithm}{%
  \setlength{\parskip}{0pt}%
  \setlength{\parindent}{0pt}%
}

\newcommand{\para}[1]{\textbf{#1.}}

\begin{document}
%
\date{}
\title{Robustness-Aware Evaluation and Enhancement of Mutation-Based Fuzzing for Bug Discovery}

\author{
{\rm Zirui Liu, Mengfan Xu, Juan Zhai, and Shiqing Ma}\\
University of Massachusetts Amherst\\
\texttt{\{zliu,mengfanxu,juanzhai,shiqingma\}@umass.edu}
\and
{\rm Barry Nelson}\\
Northwestern University\\
\texttt{nelsonb@northwestern.edu}
}



%


\maketitle

\begin{abstract}
Fuzzing is a powerful technique for discovering software vulnerabilities, and mutation-based fuzzing dominates in practice. Yet rigorous evaluation remains hard because the process is inherently random, and that randomness also limits bug detection. Prior work documents this variability empirically, but the community lacks a principled theoretical account of it with computable convergence and sample-complexity guarantees, even though it is central to fuzzer reliability.

We address both needs. First, we evaluate robustness by running the same fuzzer in independent campaigns and measuring how much its bug-trigger rate varies across runs after accounting for computational effort. We show that robustness estimated from $M$ campaigns of length $T$ has finite-trial error that decreases at the standard $M^{-1/2}$ rate, while the error from finite campaign length is bounded when bug correlations decay over time. We then introduce \emph{splitting}, a black-box wrapper that copies a fuzzer's queue state after a bug trigger and continues from that state in multiple branches, directing more effort toward the discovered region. For any realized split tree, branching cannot reduce the raw number of bug-triggering events relative to a single continuation path. The expected detection rate also increases when states receiving more branches tend to produce more bugs later. In a simplified setting, we further show that splitting reduces variance per unit compute when $p<\sqrt{2}-1$, where $p$ is the fraction of time spent in the bug region. At matched compute, splitting finds more real bugs per CPU-hour than the baseline in $38$ of $40$ Magma ground-truth cells (median $+52\%$, $34$ of $38$ individually significant) and never finds fewer distinct bugs. It converts CVE-2019-19926 from undetected ($0/20$ trials) to reliably detected ($20/20$; Fisher $p<10^{-4}$), with six additional detection improvements, five involving CVEs. On \emph{FuzzBench}, splitting finds more unique bugs in $53$ of $70$ pairs and reduces variation across campaigns in $66$ of $70$, with a median reduction of approximately $10\times$. Every split branch is counted toward the compute budget. With approximately $0.14\%$ overhead, splitting provides a practical way to both measure and improve fuzzers.

\end{abstract}


%

\input{section/introduction}

\input{section/related-work}
\input{section/problem-formulation}

\input{section/evaluation-framework}

\input{section/theoretical-analyses-evaluation}

\input{section/variance-reduction}

\input{section/numerical-experiments}

\input{section/conclusion-future-work}

\clearpage

\appendix

\section*{Ethical Considerations}

This work introduces new security-oriented metrics for assessing mutation-based fuzzing techniques. We conduct empirical evaluations of existing fuzzers on widely adopted, publicly accessible benchmark suites (e.g., FuzzBench), focusing on mature, legacy versions of these programs. All experiments are performed exclusively in non-production environments.

The study does not reveal previously unknown vulnerabilities in deployed real-world systems, nor does it create risks of negative impact on individuals. Consequently, responsible disclosure is not required, and the research poses minimal risk of infringing on individual rights. Overall, the work follows established responsible research standards and deliberately avoids any activities that could harm software users or operational systems.

\section*{Open Science}

In compliance with the open science policy, we make all research artifacts of this paper available. The artifacts comprise (i) the complete experiment data underlying every figure and table (FuzzBench per-trial logs for both the Monte Carlo and splitting arms, and the Magma ground-truth monitor data), (ii) the implementation of the offline and online splitting methods, (iii) all analysis and plotting scripts that reproduce every figure, table, and statistical test in the paper, and (iv) the complete per-fuzzer figure grids. All proofs and supporting derivations are in the appendices of this paper, not in the artifact. The artifacts are available anonymously at \url{https://anonymous.4open.science/r/split-fuzzing-artifact/README.md} and will be archived in a permanent repository upon publication.

\section{Implementation Details}\label{app:implementation}

\subsection{Offline Splitting Implementation (Section~\ref{sec:offline})}
\para{Modeling.}
We fix a benchmark–fuzzer pair $(b,f)$. Let $N_m(t)$ denote the cumulative bug count in trial $m$ of a \emph{prior, completed non-splitting campaign} at time $t \in [0,T]$---the historical data available to the offline policy---with across-trial mean curve $N(t) \;=\;\frac{1}{M}\sum_{m=1}^{M} N_m(t)$ and overall mean bug rate $r_{\mathrm{tot}} = N(T)/T$. The decision of \emph{when} to initiate Algorithm~\ref{alg:splitting} is guided by $N(\cdot)$ and $r_{\mathrm{tot}}$; the \emph{execution} of a split is triggered on an individual running trial.

\para{Sparsity score.}
We adopt splitting in regimes where bugs are relatively sparse. To determine whether bug discovery has entered a ``sparse tail'' after time $t$, we use the \textbf{tail-average rate estimator} and \textbf{relative sparsity ratio}
\begingroup\postdisplaypenalty=10000
\[
  r_{\mathrm{tail}}(t)\;=\;\frac{N(T)-N(t)}{T-t},   \rho(t)\;=\;\frac{r_{\mathrm{tail}}(t)}{r_{\mathrm{tot}}+\varepsilon},
  \qquad \varepsilon>0,
\]\endgroup
where $\rho(t)\ll 1$ indicates the remaining discovery rate is far below the overall average. To overcome temporary fluctuations, we impose a persistence (anti-spike) criterion over a future horizon $H$:
\[
  \rho_{H}(t)\;=\;\max_{u\in[t,\min\{t+H,T\}]}\rho(u).
\]
All predicates are evaluated on the experiment's discrete logging grid, and $\rho_H(t)$ only where the future interval is observable (no extrapolation). A richness-adaptive, fuzzer-conditioned schedule then sets each pair's earliest detection time $t_{b,f}$ and the ordered zone-entry times $z_{\theta_1},\dots,z_{\theta_K}$ at which the $K$ split stages fire, each the first logging-grid time after $t_{b,f}$ where the persistence-smoothed sparsity $\rho_H$ drops below the stage threshold $\theta_k$, with a minimum inter-stage gap $\Delta$; the closed-form schedule is given in Appendix~\ref{app:implementation}.

\iffullversion
\para{Richness-adaptive earliest detection time.}
The earliest detection time $t_{b,f}$ at which sparsity detection begins varies across pairs $(b,f)$; we compute it from the richness score
\[
\widetilde{B}_{b,f}=\log\big(1+N(T)\big).
\]
For each fixed fuzzer $f$, with reference quantiles $q_{\mathrm{low}}<q_{\mathrm{high}}$ of $\{\widetilde{B}_{b,f}\}_b$ across benchmarks, we normalize
\[
  u_{b,f}=\min\Big\{1,\max\Big\{0,\frac{\widetilde{B}_{b,f}-q_{\mathrm{low}}}{q_{\mathrm{high}}-q_{\mathrm{low}}}\Big\}\Big\},
\]
and map to $t_{b,f}\in[t_{\min},t_{\max}]$ by $t_{b,f}=t_{\min}+(t_{\max}-t_{\min})\cdot u_{b,f}^{\,p}$, a fuzzer-conditioned, monotone, continuous schedule.

\para{Zone-entry time and ordered stages.}
Given a nonincreasing threshold list $\Theta = (\theta_1,\ldots,\theta_K)$ for the $K$ stages at which Algorithm~\ref{alg:splitting} is applied, the first zone-entry time is
\[
  z_{\theta_1}\;=\;\min\Big\{t\in[t_{b,f},T): \rho_H(t)\le \theta_1\Big\},
\]
and for $k\ge 2$,
\[
  z_{\theta_k}\;=\;\min\Big\{t\in[\max\{t_{b,f},z_{\theta_{k-1}}+\Delta\},T): \rho_H(t)\le \theta_k\Big\},
\]
where the minimal gap $\Delta > 0$ prevents back-to-back splits in overlapping regions and keeps detections meaningfully separated; if the set is empty, stage~$k$ is skipped.
\fi

\para{Bug-triggered splitting on a running trial.}
If $z_{\theta_k}$ exists, the actual split on a running trial $i$ triggers at the first time the trial observes an additional bug,
\[
  s_{\theta_k,i}
  \;=\;
  \min\Big\{t\in(z_{\theta_k},\,T]: N_i(t)>N_i(z_{\theta_k})\Big\},
\]
skipping the stage if no such time exists. This ``zone-then-bug'' protocol separates \emph{when splitting becomes statistically justified} (based on $N(\cdot)$) from \emph{when a concrete opportunity to branch arises} (a new bug on a running trial)---key in extremely sparse regimes, where purely time-driven splitting would induce unnecessary branching.

\subsection{Online Splitting Implementation (Section~\ref{sec:ablation})}\label{app:ablation}

The online variant replaces the precomputed signal with a purely causal one: on each running trial $i$ it compares a short-window discovery rate to the trial's running average ($\rho_{\mathrm{on},i}$), enters the $K$ split stages when this past-only sparsity persists below the stage thresholds (warm-up $S=\max\{w,H\}$, minimum inter-stage gap $\Delta$), and splits on the next observed bug, using no prior trials or aggregate data; the closed-form online schedule is given in Appendix~\ref{app:implementation}.

\iffullversion
\para{Online sparsity score.}
To make the ablation independent of any precomputed aggregate signal or prior trials, the sparsity signal uses only \emph{past observations on the running trial} $i$ with bug curve $N_i(t)$: for a window length $w$ on the logging grid,
\[
\begin{aligned}
r_{\mathrm{avg},i}(t) &= \nicefrac{N_i(t)}{t+\varepsilon_t},
\qquad
r_{\mathrm{win},i}(t) = \nicefrac{N_i(t)-N_i(t-w)}{w},\\
\rho_{\mathrm{on},i}(t) &= \nicefrac{r_{\mathrm{win},i}(t)}{r_{\mathrm{avg},i}(t)+\varepsilon},
\end{aligned}
\]
with regularizers $\varepsilon_t,\varepsilon>0$; $\rho_{\mathrm{on},i}(t)\ll 1$ indicates the recent discovery rate is small relative to the trial's running average. Persistence uses a \emph{past} horizon:
\[
\rho^{-}_{H,i}(t)=\max_{u\in[t-H,t]}\rho_{\mathrm{on},i}(u),
\]

\para{Zone entry, ordered stages, \& splitting.}
Given stage thresholds $\Theta=(\theta_1,\ldots,\theta_K)$, the online variant defines zone-entry times
\emph{per running trial} using only running observations.
Let $S$ denote the earliest time at which the online statistic becomes well-defined.
Concretely, $S = \max\{w,H\}$ is the warm-up time required by the windowed difference $w$ and the
persistence horizon $H$.
Then, for each running trial $i$,
\[
\begin{aligned}
z_{\theta_1,i}&=\min\{t\ge S:\rho^{-}_{H,i}(t)\le \theta_1\},\\
z_{\theta_k,i}&=\min\{t\ge z_{\theta_{k-1},i}+\Delta:\rho^{-}_{H,i}(t)\le \theta_k\}
\ \ (k\ge2),
\end{aligned}
\]
with the same minimum gap $\Delta$ and stage-skipping rule (when the set is empty).
The execution rule remains the same: for a running trial $i$, the split at stage $k$ is triggered at the first logged time $s_{\theta_k,i}$ after $z_{\theta_k,i}$ when a new bug is observed:
\[
s_{\theta_k,i}=\min\{t>z_{\theta_k,i}: N_i(t)>N_i(z_{\theta_k,i})\}.
\]


\fi

The ablation preserves the splitting mechanism, modifies only the sparsity score, and is broadly deployable.

\section{Spectral Gap Estimation}

\begin{algorithm}[h]
\SetAlgoLined
\caption{Spectral Gap Estimation from Trajectory Data}\label{alg:spectral}
\textbf{Input:} Binary bug-output trajectories $\{Y_R^m(w,t)\}_{t=1}^T$, $Y_R^m=\mathbf 1_B(X_R^m)$, from $M$ trials\;
\textbf{Step 1:} Count the four binary transitions across all $M$ trials to form the empirical $2\times2$ law $\hat{Q}_R$\;
\textbf{Step 2:} Compute the eigenvalues of $\hat{Q}_R$; the non-unit eigenvalue is $\hat{\lambda}_2$\;
\textbf{Step 3:} $\hat{\gamma}^Y_{R} = 1 - |\hat{\lambda}_2|$\;
\textbf{Output:} Estimated projected gap $\hat{\gamma}^Y_{R}$ and relaxation time $\hat{\delta}_R = 1/\hat{\gamma}^Y_{R}$
\end{algorithm}

\subsection{Assumptions}\label{app:assumption}

We discuss the justification and scope of our assumptions on the fuzzer's transition process here.

\para{Time-homogeneity (stationarity)} We assume the transition matrix $P_R^X$ is fixed across time steps. This holds for mutation-based fuzzers whose mutation rules are fixed in code: given the current input $x$, the distribution over the next input $R(x)$ depends only on $x$ and the fuzzer's fixed mutation operator, not on the time step $t$. Coverage-guided fuzzers (e.g., AFL++, libFuzzer) maintain internal state such as a corpus and scheduling weights, but these are deterministic functions of the trajectory $\{X_R(w,s)\}_{s \leq t}$ and can be folded into an augmented state. For fuzzers or regimes that genuinely violate stationarity (e.g., time-dependent temperature schedules, LLM-based mutation, or the slow drift of the output surrogate that Section~\ref{subsec:empirical-assumptions} measures as easy bugs deplete), we handle the non-stationary case via windowed estimation with random restarting: partition the campaign into approximately stationary windows, estimate within windows, and aggregate. Section~\ref{subsec:empirical-assumptions} quantifies the practical extent of this idealization on the output surrogate.

\para{Geometric ergodicity}
The theoretical results assume that the \emph{input} chain is geometrically ergodic, so that the stationary autocovariances of the bug indicator obey a geometric envelope $|\Gamma(h)|\le C\rho^{\,h}$ (Lemma~\ref{lem:gap-transfer}); the binary output process is never assumed to be Markov. This excludes only degenerate ``trapped'' fuzzers that remain stuck in a single input region indefinitely, and random restarting (a standard feature of most fuzzers) makes it plausible within a campaign phase. Fixed implementation code makes time-homogeneity and geometric decay plausible but does not prove them: they are explicit assumptions of the theory. Ergodicity is \emph{not} required for the Markov-chain modeling itself, only for the convergence guarantees, and the autocorrelation analysis of Section~\ref{subsec:empirical-assumptions} is a model diagnostic consistent with geometric decay over the observed lag range, not a proof of geometric ergodicity.

\subsection{Additional Metrics and Results}\label{app:metrics}

Additionally, we evaluate (i) how well the computable variance estimates approximate the underlying asymptotic variance, and (ii) the discrepancy between the population variance and the same target:
\[
A_R(T) = \bigl|T\,S^2_{M,R}(T) - \sigma^2\bigr|, \quad B_R(T) = \bigl|T\, v_R(T) - \sigma^2\bigr|
\]
where $\sigma^2 = \lim_{T\to\infty}T\,v_R(T)$ and $v_R(T)=\operatorname{Var}[Z_{1,R}(T)]$ is the variance of a single campaign's outcome of a single independent campaign. It is distinct from the variance of the average over campaigns, $\operatorname{Var}[\Tilde{P}_R(T)]=v_R(T)/M$; the two differ by exactly a factor $M$, and this paper measures the former throughout.

\begin{theorem}[Finite-$M$, finite-$T$ error of $T\,S^2_{M,R}(T)$]\label{thm:main_body}
Let the $M$ root campaigns be independent and identically distributed, and write
$W_{m,T}=\sqrt T\bigl(Z_{m,R}(T)-\mathbb{E}Z_{m,R}(T)\bigr)$, $v^\star_T=T\,v_R(T)$ and $\mu_{4,T}=\mathbb{E}[W_{m,T}^4]$.
\begin{enumerate}
\item \emph{Finite-trial error.} The unbiased variance across campaigns satisfies the exact identity
\[
\operatorname{Var}\bigl(T\,S^2_{M,R}(T)\bigr)
=\frac1M\Bigl[\mu_{4,T}-\tfrac{M-3}{M-1}\,(v^\star_T)^2\Bigr],
\]
so for any $\alpha\in(0,1)$, with probability at least $1-\alpha$,
\[
\bigl|T\,S^2_{M,R}(T)-v^\star_T\bigr|
\;\le\;
E_M(\alpha):=\sqrt{\frac{\mu_{4,T}-\frac{M-3}{M-1}(v^\star_T)^2}{M\,\alpha}} .
\]
The rate is $M^{-1/2}$ and the confidence level appears explicitly.
\item \emph{Finite-horizon bias.} If the stationary autocovariances of the bug indicator obey
$|\Gamma(h)|\le C\rho^{\,h}$ for some $\rho\in(0,1)$, then, using
$T\,v_R(T)=\Gamma(0)+2\sum_{h=1}^{T-1}\bigl(1-\tfrac hT\bigr)\Gamma(h)$,
\[
\begin{aligned}
B_R(T)&=\bigl|T\,v_R(T)-\sigma^2\bigr|\\
&\le\;U_{B,R}(T):=\frac{2C\rho}{T(1-\rho)^2}+\frac{2C\rho^{T}}{1-\rho}.
\end{aligned}
\]
\item \emph{Initialization.} Let
\[
b_{\mathrm{start}}(T):=\lvert T\operatorname{Var}_{\pi_0}(Z_{1,R}(T))-T\operatorname{Var}_{\pi}(Z_{1,R}(T))\rvert
\]
be the effect of starting from $\pi_0$ rather than from the stationary law; it vanishes for a stationary start. If the start also satisfies $\lvert\operatorname{Cov}_{\pi_0}(Y_s,Y_t)-\Gamma(t-s)\rvert\le C_0\rho^{\,s-1}\rho^{\,t-s}$ for $1\le s\le t$ and some finite $C_0$, then
\[
b_{\mathrm{start}}(T)\;\le\;\frac{C_0}{T(1-\rho)}+\frac{2C_0\rho}{T(1-\rho)^{2}} .
\]
In the two-state model the same term is known exactly, with no such constant. Write $\pi_1$ for the stationary probability that a step triggers a bug, $\lambda\in(-1,1)$ for the second eigenvalue, $p_0$ for the probability that the first step triggers a bug, and $D=p_0-\pi_1$ for the difference between the two. Then\label{eq:exactstart}
\[
\begin{aligned}
b_{\mathrm{start}}(T)&=\tfrac1T\bigl|D(1-2\pi_1)L_T-D^{2}Q_T\bigr|,\\
L_T&=\sum_{t=1}^{T}(2t-1)\lambda^{t-1},\qquad
Q_T=\Bigl(\tfrac{1-\lambda^{T}}{1-\lambda}\Bigr)^{2},
\end{aligned}
\]
and $\max_{p_0\in[0,1]}b_{\mathrm{start}}(T)$ is available in closed form, so the term can be reported without estimating $p_0$.

\item \emph{Combined.} With probability at least $1-\alpha$,
\[
A_R(T)\;\le\;U_{A,R}(T):=E_M(\alpha)+b_{\mathrm{start}}(T)+U_{B,R}(T).
\] Applying part~1 at level $\alpha/|\mathcal R|$ gives the guarantee simultaneously for all $R\in\mathcal R$ with probability at least $1-\alpha$.
\end{enumerate}
(Proof in Appendix~\ref{app:pf-thm-main}.)
\end{theorem}

No general envelope on the autocovariances controls $b_{\mathrm{start}}$, because the part of the start discrepancy that is quadratic in the starting offset does not decay with the mixing rate; the remark below bounds it under one extra condition, and within the two-state model Theorem~\ref{thm:main_body} gives it exactly, which is what Table~\ref{tab:spectral} evaluates.

The three terms are separate and each is carried by one quantity: $E_M(\alpha)$ is finite-trial estimation error and is the \emph{only} term in which $M$ appears; $b_{\mathrm{start}}$ is initialization; $U_{B,R}$ is finite-horizon autocovariance bias. Nothing above assumes the binary output process is itself Markov. The envelope $(C,\rho)$ is a property of the true autocovariances, supplied for a reversible input chain by Lemma~\ref{lem:gap-transfer} with $C=\pi(1)\pi(0)\le\tfrac14$ and $\rho=1-\gamma^X_R$.

\para{Computable two-state instance.} Table~\ref{tab:spectral} evaluates the same three terms inside the phase-wise two-state model, with $(C,\rho)=(\tfrac12,1-\hat\gamma^Y_R)$ and the finite-trial term at its two-state calibrated value $1/\sqrt M$. That instance is a computed inside the model interval within the stated two-state model, not a guarantee for the general chain. For a reversible input chain the projected output gap is never smaller than the input-chain gap, $\gamma^Y_R\ge\gamma^X_R$, so it is not in general a conservative substitute for it.

We remark that (i) under strict stationarity the second term carries a $1/T$ factor (it arises as $C_2/T$ in the proof), which tightens $U_{A,R}$ to within $\approx 3\times$ of the observed deviation on the low-drift pairs. We state and compute the looser form throughout, which only weakens the bound; for the lone high-drift cell (Honggfuzz/stb, $|\Delta\hat\gamma|=0.86$, $\hat\pi_1:0.39\to0.05$) we flag $|\Delta\hat\gamma|$ and treat the campaign estimate as a phase-average handled by windowed estimation (Appendix~\ref{app:assumption}), where the unconditional bound continues to certify $A_R\le U_{A,R}$. (ii) The concentration is in the number of trials $M$, at the standard $M^{-1/2}$ statistical rate; the horizon $T$ enters only through the bias terms (the spectral-gap and autocovariance-tail contributions), not through the high-probability rate.

\subsection{Assumptions and Proofs of the Core Guarantees}\label{app:sketches}

\noindent The condition on the start in Part 3 is assumed rather than derived from the mixing rate, since a mixing-rate argument controls only the part of the discrepancy that is linear in the starting offset. In the two-state model the quadratic part is $D^2Q_T/T$ with $Q_T=\bigl((1-\lambda^T)/(1-\lambda)\bigr)^2$, which tends to $D^2/T$ as $\lambda\to0$.

\para{Which form of the starting-state term the reported numbers use.} Table~\ref{tab:spectral} uses the exact two-state expression of part~3, taken at its largest over every start, so it needs neither $C_0$ nor an estimate of $p_0$. On two-state parameters spanning our fits that expression is up to $59\times$ smaller than the general bound in the same part.

\para{Why a gap-only initialization term cannot be tightened.} The $1/(\gamma^X_R)^2$ growth of the second term is the true worst-case rate, not looseness that could be removed. Take the symmetric two-state chain $P=\bigl(\begin{smallmatrix}1-\varepsilon & \varepsilon\\ \varepsilon & 1-\varepsilon\end{smallmatrix}\bigr)$ with $\varepsilon=\gamma^X_R/2$; its non-unit eigenvalue is $1-\gamma^X_R$ (gap exactly $\gamma^X_R$) and $\pi=(\tfrac12,\tfrac12)$. For the centered indicator $f(x)=\mathbf 1\{x{=}1\}-\tfrac12$, a worst-case point-mass start gives transient correlation sum $S=\sum_{k\ge1}(1-\gamma^X_R)^k=(1-\gamma^X_R)/\gamma^X_R$, so the $\pi_0$-vs-$\pi$ effort-normalized-variance discrepancy is $C_2=(S/2)^2=(1-\gamma^X_R)^2/(4(\gamma^X_R)^2)$. Hence $C_2=\Theta(1/(\gamma^X_R)^2)$ with $C_2(\gamma^X_R)^2\to\tfrac14$ as $\gamma^X_R\to0$, a matching lower bound: \emph{no} bound expressed through $\gamma^X_R$ alone can improve the $1/(\gamma^X_R)^2$ rate. Conversely a campaign that starts already in the long-run state gives discrepancy $0$, and a start within $O(\gamma^X_R)$ of stationary gives $O(1)$; thus at fixed $\gamma^X_R$ the true deviation ranges from $O(1)$ to $\Theta(1/(\gamma^X_R)^2)$ with the trajectory's distance from stationarity, which $\gamma^X_R$ does not carry. So a bound written through $\gamma^X_R$ alone cannot separate a campaign that starts in the long-run state from one that starts far from it; the remaining start-to-stationarity distance is what the exact two-state start term of Theorem~\ref{thm:main_body} supplies, and under near-stationary operation the $C_2/T$ form already lies within a small factor of it.

\noindent The identity above agrees with direct evaluation of the $O(T^2)$ double sum to $7\times10^{-16}$.

\para{The bound needs no Markov assumption on the bug indicator.}\label{thm:direct-nonlumped}
Theorem~\ref{thm:main_body} is stated through the autocovariances $\Gamma(h)$ of the bug indicator itself and a bound $|\Gamma(h)|\le C\rho^{\,h}$ on them. Nothing in it assumes that $\{Y_t\}$ is a Markov chain; the envelope is a property of the true autocovariances, whatever process produced them. Lemma~\ref{lem:gap-transfer} supplies such an envelope for a reversible input chain, with $C=\pi(1)\pi(0)\le\tfrac14$ and $\rho=1-\gamma^X_R$, so the hypothesis is met without any assumption on the output. Substituting the two-state values recovers the computable instance of Table~\ref{tab:spectral} term for term; since $\pi(1)\pi(0)\le\tfrac14<\tfrac12$, that instance is the larger of the two whenever the fitted gap does not exceed the input-chain gap.

\begin{lemma}[Geometric output-autocovariance]\label{lem:gap-transfer}
Let the input chain be reversible with stationary law $\pi$ and absolute spectral gap $\gamma\in(0,1]$. For $f=\mathbf 1_B$ and every $t\ge 0$,
$\bigl|\mathrm{Cov}_\pi(f(X_0),f(X_t))\bigr|\le \mathrm{Var}_\pi(f)\,(1-\gamma)^t$,
so $\sigma^2=\mathrm{Var}_\pi(f)+2\sum_{t\ge1}\mathrm{Cov}_\pi(f(X_0),f(X_t))\le \mathrm{Var}_\pi(f)\,\tfrac{2-\gamma}{\gamma}$ is finite, and the autocovariance-tail term of Theorem~\ref{thm:main_body} holds at this geometric rate, with no assumption that $\{Y_R(w,t)\}$ is Markov.
(Proof in Appendix~\ref{app:pf-lem-gap}.)
\end{lemma}

\noindent Three clarifications. (i) For non-reversible kernels the same qualitative conclusion, finiteness of $\sigma^2$ and a geometric autocovariance tail, holds for non-reversible kernels through the spectral gap of~\cite{choi2020metropolis}, at a possibly halved rate; in that case $\gamma^X_R$ denotes that reversibilization's gap. It is distinct from the quantity Algorithm~\ref{alg:spectral} reports, the projected bug-output gap $\hat\gamma^Y_R$. (ii) The lemma reproduces the geometric \emph{rate} of Theorem~\ref{thm:main_body}'s tail; the exact leading constant there ($c_\pi=\tfrac14+\pi(1)/(4\max\{\pi(0),\pi(1)\})\le\tfrac12$) comes from the two-state surrogate, whereas the lemma's constant is $\mathrm{Var}_\pi(f)=\pi(1)\pi(0)$. (iii) The two-state estimate $\hat\gamma$ (Algorithm~\ref{alg:spectral}) equals $\gamma$ when the bug indicator happens to behave exactly like a two-state chain, and otherwise estimates an effective decay rate; we do \emph{not} claim an analytical bound on how far it departs from one, but report that departure small and stable empirically (Section~\ref{subsec:empirical-assumptions}). Hence $U_{A,R}$ is a plug-in estimate of the population bound, not a quantity presupposing a Markov output.

\para{Assumptions for Theorem~\ref{thm:splitting_var}.}
\emph{A1 (post-split dynamics):} each post-split replica $k\in\{1,2\}$ evolves by an AR(1)-type update $X_{k,t}=\phi(X_{k,t-1})\cdot W_{k,t}+Z_{k,t}\cdot W_{k,t}$, where $W_{k,t}$ selects mutation locations and $Z_{k,t}$ injects values.
\emph{A2 (domain structure):} the bug set $B$ is a discrete set of inputs, each bug tied to specific location/value patterns.
\emph{A3 (correlation control):} $\mathbb{P}(X_{1,t}\in B,\,X_{2,t}\in B)\le(1-\rho)\,\mathbb{P}(X_{1,t}\in B)\,\mathbb{P}(X_{2,t}\in B)$ for some $\rho\in[0,1]$; a coupled-noise construction attains $\rho=\min\{1,1/(|D|\pi(1))\}$, while the independent forks our implementation uses correspond to $\rho=0$, where the variance reduction stems from the $1/k_t$ covariance averaging alone.
\emph{A4 (geometric ergodicity)} and \emph{A5 (rare-bug regime, $\pi(1)\ll\pi(0)$)}.
\emph{Clustered regime:} the autocovariance contributions to $\sigma^2$ are dominated by bug-region visits, and splitting raises $k_t$ precisely on those time indices.

\input{section/supplementary}
\input{section/exp}

\input{section/proof}

\clearpage 
\newpage 
\bibliographystyle{plainurl}
\bibliography{main}
\end{document}

%% file: section/introduction.tex
\section{Introduction}

Fuzzing, iteratively generating inputs and feeding them to a program to trigger bugs, is among the most effective techniques for discovering software vulnerabilities~\citep{klees2018evaluating, schloegel2024sok, li2021unifuzz}.
Yet fuzzing is inherently stochastic: when we evaluated AFL++ on poppler across 20 independent 23-hour campaigns, the number of unique bug signatures per trial ranged from 1 to 13, a 13-fold difference.
This variability is not an anomaly; it is the norm across fuzzers and benchmarks~\citep{klees2018evaluating, schloegel2024sok, madadi2026bugs}. It raises a fundamental question: \emph{which result should a practitioner trust, and how many trials are enough?}

Despite extensive work on fuzzer evaluation~\citep{klees2018evaluating, li2021unifuzz, schloegel2024sok}, the community lacks a principled answer. Effectiveness is measured by bug counts or code coverage and stability by confidence intervals or ad-hoc variance estimates, none of them guaranteed to capture the true stochastic variation of the fuzzing process~\citep{schloegel2024sok}; and coverage, the most common proxy, does not reliably correlate with bug-finding ability~\citep{schloegel2024sok}. A rigorous, bug-centric framework for quantifying fuzzer robustness remains an open problem.

Variability in fuzzing plays a dual role. It undermines reliability: practitioners must run many expensive trials to draw sound conclusions. But it also reflects how a fuzzer traverses the input space, jumping between regions under stochastic mutation. When bugs are clustered, a fuzzer that stays in a bug-rich region finds more bugs; one that drifts away wastes effort. So \emph{controlling} variability, not merely measuring it, can directly improve bug-finding. Yet no principled tools for this exist.


This paper addresses both. We model fuzzing as a Markov chain over the \emph{input} space to capture stochasticity and define robustness as the effort-normalized variance of the bug-detection rate. Inspired by rare-event simulation for stochastic systems, we then introduce \emph{splitting} to enhance fuzzers: when a trial discovers a bug, we carefully fork the trajectory into independent descendant chains. This black-box plug-in exploits bug-rich regions, reducing variability while concentrating effort according to how bugs cluster, thereby illuminating the dual role of variability. 

\subsection{Main Contributions}
\begin{enumerate}[leftmargin=*]
\item \textbf{Robustness metric with theoretical guarantees.} We model fuzzing as a stochastic process and define robustness by how much the bug-trigger rate varies across independent fuzzing campaigns, normalized by computational effort. With $M$ independent campaigns, the sample variance is unbiased for this finite-horizon variability, and its estimation error decreases at the standard $M^{-1/2}$ rate. We further bound the error due to finite campaign length when temporal correlations decay geometrically, without requiring the observed output (bug/no-bug)  process itself to be Markov (Theorem~\ref{thm:main_body}). For the two-state approximation used in our diagnostics, we derive an exact correction for initialization that is up to $46\times$ tighter than the worst-case bound and provides practical guidance on the number and length of trials (empirically $M\geq20$, $T\geq12$\,h). Both effectiveness and robustness are thus measured from the same bug-detection-rate object: its mean captures effectiveness, while its effort-normalized variance captures robustness.


\item \textbf{Splitting-based enhancement.} We introduce \emph{splitting}, a black-box enhancement for mutation-based fuzzing that uses bug signals directly, requires no hand-tuned importance function, and leaves fuzzer internals unchanged. To our knowledge, it is the first method to combine these properties and the first effort-reallocation method that guarantees no loss in raw bug-event count relative to a continuation path. When a bug is detected, splitting forks the current trajectory and directs additional effort from that state. We prove that, for \emph{any} realized split tree, its raw bug-event count is no smaller than that of an embedded continuation path (Theorem~\ref{bug_rate_splitting}). We further show that its expected bug-detection rate increases when states selected for more branching also tend to produce more bugs in the future (Proposition~\ref{prop:sizebias}). Empirically, splitting increases the bug-event rate in $39$ of $40$ Magma cells and the terminal weighted unique-bug rate in all $70$ FuzzBench pairs. In an idealized two-level setting, we show that splitting reduces variance per unit compute when the fraction $p$ of time spent in the bug region satisfies $p<\sqrt{2}-1$ (Theorem~\ref{thm:splitting_var}). At strictly matched CPU, the deployed adaptive tree reduces variance in $66$ of $70$ pairs, with a median reduction of approximately $10\times$. Unlike power scheduling~\citep{AFLFast, lyu2019mopt}, splitting requires no modification to fuzzer internals and can be composed with any mutation-based fuzzer. Unlike independent replication, it allocates additional effort only after a bug is observed and starts that effort from the discovered state. At matched compute, splitting finds more unique bugs than independent replication in $53$ of $70$ FuzzBench pairs and achieves a higher bug-event rate than an eight-worker synchronizing pool on all eight Magma programs (Section~\ref{subsec:corpussync}).


\item \textbf{Comprehensive empirical validation.} On \emph{Magma}~\citep{hazimeh2020magma}, splitting achieves a higher real-bug rate in $38$ of $40$ fuzzer-program pairs at matched CPU (median $+52\%$; significant in $34$ of $38$) and never finds fewer distinct bugs. It also converts CVE-2019-19926 under AFL++ from undetected ($0/20$ trials) to reliably detected ($20/20$; Fisher $p<10^{-4}$), with six additional detection improvements, five involving CVEs. On \emph{FuzzBench} \citep{fuzzbench}, spanning $7$ fuzzers, $10$ benchmarks, and 23-hour campaigns, splitting finds more unique bugs in $53$ of $70$ pairs at strictly matched CPU (sign test $p=1.7\times10^{-8}$) and reduces across-campaign variance in $66$ of $70$ pairs, with a median reduction of approximately $10\times$. In practice, splitting is a black-box wrapper with approximately $0.14\%$ overhead and supports both offline and online deployment.

\end{enumerate}

\smallskip
\noindent\fbox{\parbox{0.97\columnwidth}{\small
\textbf{In one sentence:} we turn fuzzer robustness into a \emph{measurement with a convergence guarantee}, and add \emph{splitting}, a black-box wrapper that, at equal CPU-hours, finds more real bugs and turns undetected CVEs into reliably-detected ones. 
}}

\para{Paper organization}
Sections~\ref{sec:related work}--\ref{sec:evaluation} give related work, the formulation, the evaluation framework and its guarantees; Section~\ref{sec:splitting} introduces splitting; Section~\ref{sec:exp} reports the experiments; Sections~\ref{sec:practical-guidance}--\ref{sec:conclusion} conclude.

%% file: section/related-work.tex
\section{Related Work}\label{sec:related work}

\para{Mutation-based fuzzing}
Input generation broadly falls into grammar-based~\citep{godefroid2008grammar, wang2019superion, jero2019leveraging} and mutation-based approaches; the latter dominates, iteratively modifying inputs guided by feedback such as bug signals~\citep{moukahal2021vulnerability}, execution time~\citep{petsios2017slowfuzz}, divergence~\citep{petsios2017nezha, guo2018dlfuzz}, or coverage~\citep{nagy2019full, lemieux2018fairfuzz}, with recent LLM-guided variants~\citep{meng2024large, xia2024fuzz4all}. Our framework applies to any mutation-based fuzzer as a black-box plug-in.

\para{Fuzzer evaluation}
Reliable evaluation is crucial, yet the community lacks consensus on metrics. Klees et al.~\citep{klees2018evaluating} find that higher coverage does not necessarily yield more bugs; Li et al.~\citep{li2021unifuzz} advocate richer metrics including unique bug counts; Hazimeh et al.~\citep{hazimeh2020magma} observe that bug counts are subject to substantial randomness; subsequent studies refine practices further~\citep{ren2021empirical, jiang2023evaluating, wu2022evaluating, elder2022really}. Schl{\"o}gel et al.~\citep{schloegel2024sok} conclude from a large-scale study that existing methodologies still lack rigor---surveying prior work using 1--40 trials---particularly regarding how many trials are needed. Our work directly addresses this gap with a theoretically grounded criterion for trial selection and a formal robustness metric.

\para{Robustness in fuzzing evaluation}
Prior robustness assessments use confidence intervals of code coverage~\citep{klees2018evaluating} or median/percentile statistics~\citep{schloegel2024sok, li2021unifuzz}: empirical guidelines on indirect proxies (coverage, not bug detection) that capture only across-trial variability and never treat variance as a theoretically-backed evaluation dimension. Our framework jointly handles inter-trial and intra-trial variability and computational cost, grounded in bug-centric measures with provable convergence guarantees.

\para{Markov chain modeling in fuzzing}
B{\"o}hme et al.~\cite{AFLFast} first model coverage-based greybox fuzzing as a Markov chain over coverage states, deriving the AFLFast energy schedules. Our modeling differs fundamentally: we define the chain over the \emph{input space}, treating coverage and bug-triggering as output functions of it---advantageous because input-level mutation transitions are time-homogeneous for a fixed mutation operator, whereas coverage transitions depend on the growing corpus and are not necessarily Markovian, and because any program behavior (bugs, coverage, sanitizer signals) is a deterministic function of the input chain.

\para{Benchmarking randomness}
Madadi et al.~\cite{madadi2026bugs} empirically show that fuzzer benchmarking outcomes vary substantially with trial count, duration, and benchmark selection. The analogous question in MCMC, how many chains and run how long, is answered there by batch-means and spectral variance estimators~\cite{flegal2010batch} and multi-chain diagnostics~\cite{margossian2024nested}, under variance-bounding conditions~\cite{roberts2008variance}; we adapt that machinery to fuzzing, where the chain is the mutation process and the functional is the bug-detection rate. We complement their findings with a theoretical framework that formalizes this randomness through Markov chain theory, provides convergence guarantees, and proposes splitting to reduce the variability they identify.

\para{Rare-event simulation and splitting}
Splitting methods originate in rare-event simulation~\cite{glasserman1996splitting, glasserman1999multilevel, cerou2007adaptive}, where trajectories are progressively biased toward target regions using hand-tuned importance functions. We make it practical for fuzzing: the bug-detection signal \emph{is} the importance function, removing the long-standing obstacle; it is a black-box wrapper composing with any mutation-based fuzzer; and it carries a non-regression guarantee for any realized split tree (Theorem~\ref{bug_rate_splitting}). Forking the fuzzer's \emph{full queue state} on the bug signal places it in a class distinct from restart-from-crash, AFL's crash-exploration mode (\texttt{-C})~\cite{afl}, which inherits no queue state, and corpus-synchronizing parallel fuzzing, which shares state continuously rather than on bug discovery. To our knowledge this is the first bug-signal-guided, importance-function-free splitting method for fuzzing.

%% file: section/problem-formulation.tex
\section{Problem Formulation}\label{sec:formulation}

We model fuzzing as a Markov chain over the input space, with bug-triggering as an output function; the source of randomness is the input string, so effectiveness (bug detection rate) and robustness (variance of that rate) are treated in one object.

Let the target program be \(f\) with finite input domain \(D\). For a fuzzer \(R\), let \(\{X_R(w,t)\}_{t\ge 0}\) be the input Markov chain on \(D\), where \(X_R(w,t)\) is the input at time \(t\) starting from seed \(w\). Since \(D\) is finite we use the counting measure, writing \(\int h(x)\,dx := \sum_{x\in D} h(x)\).

Execution yields trace output \(f(x)\) and behavioral output \(g(x)\); with \(B_f\) the abnormal behaviors indicating bugs, the buggy inputs are $B=\{x\in D \mid g(x)\in B_f\}$ and the output process is $Y_R(w,t)=\mathbf{1}_{\{g(X_R(w,t))\in B_f\}}\in\{0,1\}$, indicating whether the input at time \(t\) triggers a bug.

A fuzzer \(R\in\mathcal{R}\) mutates \(x\in D\) into \(R(x)\in D\) by bit-level operations (flipping, deleting, inserting, modifying bytes). Each mutation is a transition of the input chain with probability \(P_R^X(x,R(x))\), giving transition matrix \(\{P_R^X(i,j)\}_{i,j\in D}\).

Let \(\pi_R^X\) denote a stationary distribution of the input Markov chain \(X_R\), i.e., $\pi_R^X(j)=\sum_{i\in D}\pi_R^X(i)\,P_R^X(i,j),\qquad j\in D.$ At stationarity, the induced transition law of the output process is $E_{x \sim \pi_R^X}\big[P_R\big(\mathbf{1}_{\{g(x)\in B_f\}},\,\mathbf{1}_{\{g(R(x))\in B_f\}} \,\big|\, x\big)\big]$. We denote the corresponding \(2\times 2\) matrix by \(\{P_R(m,n)\}_{m,n\in\{0,1\}}\). The stationary distribution of the output process \(Y_R\) induced by \(\pi_R^X\) is then given by $\pi_R(1)=\sum_{x\in D}\pi_R^X(x)\,\mathbf{1}_{\{x\in B\}}
=\pi(1), \pi_R(0)=1-\pi_R(1) =\pi(0).$
\(Y_R\) is a deterministic function of the input chain, hence not in general Markov (a function of a Markov chain is itself Markov only in special cases). We use \(\{P_R(m,n)\}\), the stationary-induced aggregate law, as a tractable two-state surrogate; Section~\ref{subsec:empirical-assumptions} validates it: the output autocorrelation shows the geometric decay it predicts in every cell.

We start from a single seed (multi-seed in Appendix~\ref{app:discussion}). Over \(T\) rounds the fuzzer executes each input, records buggy ones, and mutates; \(H_R^T\) is the multiset of bug-triggering events (\(Y_R(w,t)=1\)) and \(I_R^{T}\) the executed inputs.

\subsection{Modeling}\label{subsec:assumptions}
\para{Mutation scope}
The input chain \(X_R\) captures \emph{any} mutation operator (bit-flips, splicing, dictionary insertion, havoc), each a transition with probability \(P_R^X(x,R(x))\); the framework only requires that \(R(x)\) be drawn from a distribution fixed in code, hence independent of the time step \(t\). For stateful fuzzers with internal state \(G_t\): if \(G_t\) depends only on \(X_R^t\) the paradigm applies verbatim; if on \(G_{t-1}\), we expand the state to \(X^R_{t+1} = \mathcal{R}(X_t^R,G_t)\). The output depends only on \(X_t^R\) either way.

\para{Scope and limitations}
LLM-based fuzzers whose mutation distributions change unpredictably over time fall outside our current scope. Our modeling requires \emph{two mild, standard conditions} from the Markov chain community~\citep{down1995exponential}, both met by practical fuzzers: \emph{time-homogeneity}, which follows from the update rules being fixed in code; and \emph{geometric ergodicity} of the input chain, which only excludes a fuzzer trapped forever in one input region, and whose observable consequence, geometric autocorrelation decay, holds in every evaluated cell (Section~\ref{subsec:empirical-assumptions}), the adaptive schedulers AFL++ and MOpt included. Output Markovianity is \emph{not} required (Section~\ref{sec:evaluation}). Regimes whose mutation distribution shifts over a campaign are handled by windowed estimation over approximately stationary segments (Appendix~\ref{app:assumption}).

\textbf{Bug detection rate.} We define the \emph{bug detection rate} (BDR), which quantifies the fraction of executed inputs that trigger bugs:
$P_R = \lim_{T \to \infty} \frac{|H_R^T|}{|I_R^T|} = \lim_{T \to \infty} \frac{1}{T}\sum_{t=1}^TY_R(w,t)$
and its finite-horizon counterpart $P_R(T) = E[\frac{|H_R^T|}{|I_R^T|}] = E[\frac{1}{T}\sum_{t=1}^TY_R(w,t)].$
It is monotone in $|H_R^T|$, the total number of bug-triggering executions, the quantity behind the bug/crash counts used in existing evaluation~\citep{li2021unifuzz, schloegel2024sok, klees2018evaluating}, while additionally charging the effort over $T$ rounds. Being a functional of the chain's realization, $|H_R^T|$ is random, and its variability has not been previously quantified; we therefore measure the stability of BDR via variance.

\textbf{Robustness.} We consider both terminating and non-terminating regimes, and begin with the non-terminating regime, focusing on the asymptotic behavior of the Markov chains. We define the performance measure in the non-terminating regime as
\begin{align*}
& S_R = \lim_{T\to\infty} Var(\tfrac{|H_R^T|}{|I_R^T|}); \\
& \sigma^2
:= \lim_{T\to\infty} T\cdot Var(\tfrac{|H_R^T|}{|I_R^T|}) = \lim_{T\to\infty} Var(\tfrac{1}{\sqrt{T}}\sum_{t=1}^T Y_R(w,t)),
\end{align*}
where \(Var(A)\) denotes the variance of \(A\). The scaling by \(T\) charges the effort required to reach a given level of variability, so the criterion reflects computational cost as well as stability. In practice mutations run for finitely many rounds, so we also define a finite-horizon measure: 
\begin{align*}
   & S_R(T) = Var(\tfrac{|H_R^T|}{|I_R^T|}); \\
   & \sigma^2(T) = T \cdot Var(\tfrac{|H_R^T|}{|I_R^T|})  = T \cdot Var(\tfrac{1}{T}\sum_{t=1}^TY_R(w,t)).
\end{align*}

\begin{remark}[BDR vs.\ unique bug counts]\label{rmk:bdr-vs-unique}
Practitioners may also care about the number of \emph{unique} bugs found, which our experiments report directly. At the theoretical level, we study the underlying bug-trigger process that drives these discoveries: splitting branches from bug-triggering states and reallocates effort toward regions with higher future bug yield. The bug-detection rate (BDR) provides an effort-normalized, steady-state characterization of this process, while unique-bug discovery additionally accounts for deduplication against the previously discovered set. The two are closely connected empirically: splitting increases bug-triggering events in 39 of 40 Magma cells and yields at least as many distinct real bugs in all 40, while bug-event and distinct-bug counts move together within trials on both Magma and FuzzBench (Section~\ref{subsec:magma}).
\end{remark}


%% file: section/evaluation-framework.tex
\section{Evaluation Framework}\label{sec:evaluation-methods}

Figure~\ref{fig:overview} provides an overview of our framework. All experiments use real fuzzer executions on real programs. The term ``Monte Carlo'' refers to the evaluation methodology (repeated independent trials), not a synthetic simulator.

\begin{figure}[t]
  \centering
  \includegraphics[width=0.92\columnwidth]{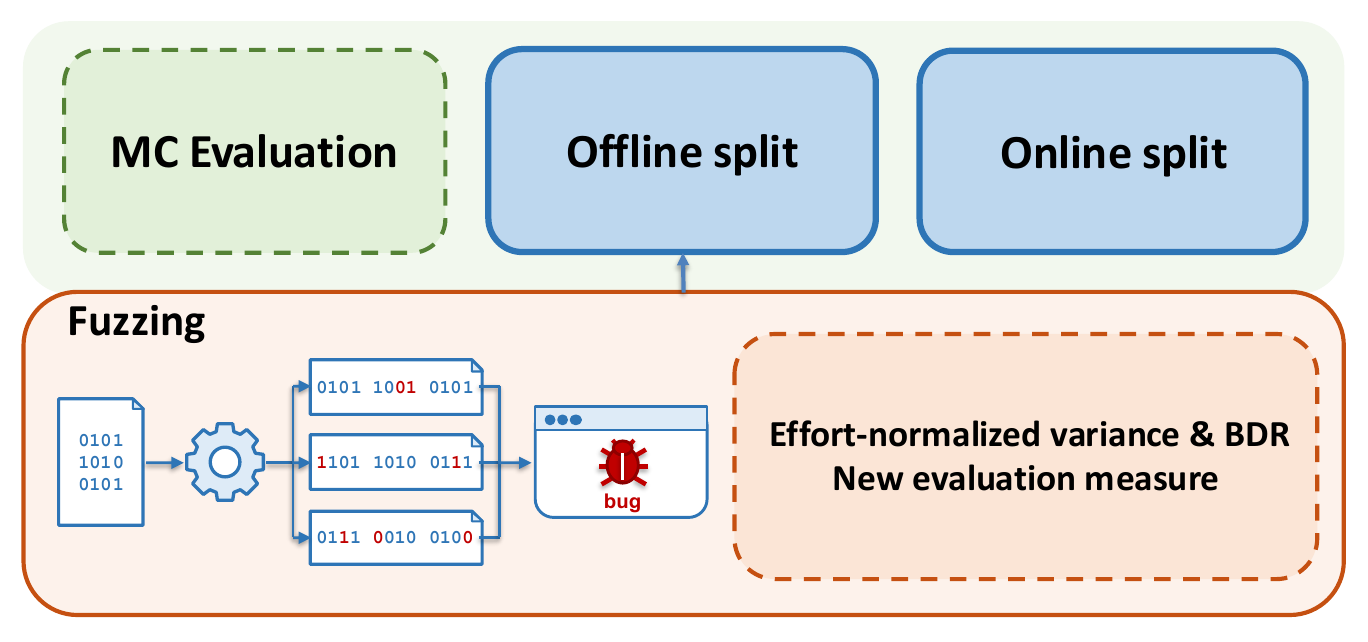}
  \caption{Overview of the framework. Real fuzzer executions generate \(M\) independent trajectories, from which we compute effort-normalized variance (robustness) and the bug-detection rate. The MC procedure evaluates fuzzers; splitting enhances them, implemented online or offline depending on prior-data availability.}
  \label{fig:overview}
\end{figure}

We now estimate the Section~\ref{sec:formulation} measures from fuzzing data.

\para{Simulation algorithm} The Monte Carlo (MC) evaluation runs $M$ independent campaigns, each replicating the fuzzing process for $T$ mutation steps from an initial input $w$: the fuzzer $R$ iteratively mutates the current input per its transition matrix and executes the result, so the $m$-th run yields one trajectory $\{(X_R^m(w,t), Y_R^m(w,t))\}_{t=1}^{T}$, and $M$ repetitions give $M$ i.i.d.\ sample trajectories (pseudocode in Algorithm~\ref{alg:mc}).

The simulation algorithm is classical; the contribution is the estimators built on its $M$ trajectories, which quantify robustness and guide the joint choice of $M$ and $T$.

\para{What counts as one repetition} Our replication unit is an independent \emph{root campaign}. For the $m$-th one, $Z_{m,R}(T)=\frac{1}{T}\sum_{t=1}^{T}Y_R^m(w,t)$ is its bug-trigger event rate, so $T\,Z_{m,R}(T)=N_T^m$ is its terminal event count.

\para{Point estimators} We construct two point estimators. The first, \(\Tilde{P}_R(T)\), is the average over campaigns of the rate and estimates the bug detection rate \(P_R\):
\begin{align*}
\widetilde{P}_R(T)
&= \frac{1}{M}\sum_{m=1}^{M}Z_{m,R}(T)
= \frac{1}{MT}\sum_{m=1}^{M}\sum_{t=1}^{T}Y_R^m(w,t),
\end{align*}
proportional to the common bug-count metric $\sum_{m} N_T^m$, while accounting for the effort $T$.

The second, our novel sample variance \emph{across campaigns}, captures how much root campaigns differ from one another:
\[
S^2_{M,R}(T)
= \frac{1}{M-1}\sum_{m=1}^{M}\bigl(Z_{m,R}(T) - \widetilde{P}_R(T)\bigr)^2 .
\]
Let $v_R(T)=\operatorname{Var}\bigl[Z_{1,R}(T)\bigr]$ denote the variance of a single campaign's outcome. Then $\mathbb{E}[S^2_{M,R}(T)]=v_R(T)$ while $\operatorname{Var}[\widetilde{P}_R(T)]=v_R(T)/M$, (every reported sample variance uses this denominator), which are two \emph{different} quantities: the first is how unstable independent campaigns are, the second how uncertain their grand mean is. Everything below evaluates the first.

\para{Evaluation metrics} Robustness is the effort-normalized variance of a single campaign's outcome $\sigma^2=\lim_{T\to\infty}T\,v_R(T)$, measured by $T \cdot S^2_{M,R}(T)$, which depends only on simulation outputs and is therefore directly computable.

We next quantify the accuracy of this estimator through
\[
A_R(T) = \bigl|T\,S^2_{M,R}(T) - \sigma^2\bigr|, \quad B_R(T) = \bigl|T\, v_R(T) - \sigma^2\bigr|,
\]
so that $A_R$ carries the finite-$M$ estimation error and $B_R$ the bias from running for a finite time rather than forever.

We also evaluate how closely the computable variance tracks the quantity it estimates, and defer that to Appendix~\ref{app:metrics}.

Lastly, we also evaluate bug-detection ability via \(\widetilde{P}_R(T)\) (BDR), equivalently the total bug count $N(T) = \sum_{m=1}^{M}N_T^m.$

%% file: section/theoretical-analyses-evaluation.tex
\section{Theoretical Guarantees for Evaluation}\label{sec:evaluation}

In this section, we evaluate the fuzzers using the proposed metrics and their theoretical guarantees. We defer detailed proofs and the supporting derivations to the appendices.\footnote{\label{fn:codebase}\url{https://anonymous.4open.science/r/split-fuzzing-artifact/README.md}}


\textbf{Our theory provides guarantees for variance reduction and robustness estimation.}
First, under an idealized two-level allocation, Theorem~\ref{thm:splitting_var} gives the closed-form sufficient condition $p<\sqrt{2}-1$ for splitting to reduce variance, where $p$ is the fraction of time spent in the bug region. Proposition~\ref{prop:tree} extends this covariance analysis to general split trees. Our matched-CPU experiments then verify that these adaptive trees reduce variance in practice. Second, Theorem~\ref{thm:main_body} decomposes the estimation error of $T S^2_{M,R}(T)$ into finite-campaign sampling error at the standard $O(M^{-1/2})$ rate, initialization error, and finite-horizon autocovariance bias. The finite-horizon bound requires only geometrically decaying autocovariances of the binary output process. For phase-specific diagnostics, its two-state specialization provides an exact initialization term that is up to $46\times$ tighter than the corresponding worst-case bound.


\begin{theorem}[Informal; formalized as Theorem~\ref{thm:main_body} in the Appendix]\leavevmode 
Write $\gamma^X_R\in(0,1]$ for the spectral gap of fuzzer $R$'s input chain: a larger gap means that the chain forgets its starting point faster. Let \(U_{A,R}(T)\) denote the explicit error margin formalized in Theorem~\ref{thm:main_body}. Ignoring constants and taking the initialization term at its largest value for a given spectral gap, this margin has order
\(
U_{A,R}(T)
=
O\!\left(
\frac{1}{\sqrt{M}}
+
\frac{1-\gamma^X_R}{(\gamma^X_R)^{2}}
+
\frac{1}{T}\frac{1}{(\gamma^X_R)^{2}}
+
\frac{(1-\gamma^X_R)^T}{\gamma^X_R}
\right).
\)
For any \(\alpha\in(0,1)\), with probability at least \(1-\alpha\), \(A_R(T)=|T\cdot S^2_{M,R}(T)-\sigma^2|\le U_{A,R}(T)\) for every fuzzer \(R\in\mathcal R\) simultaneously. The confidence level enters the finite-trial term explicitly, whose rate is \(M^{-1/2}\) for fixed \(T\); the remaining terms quantify initialization and finite-horizon effects.
(Proof in Appendix~\ref{app:pf-thm-main}.)
\end{theorem}

Hence, the directly computable \(T\cdot S^2_{M,R}(T)\) estimates the robustness \(\sigma^2\) within the margin \(U_{A,R}(T)\). If the resulting intervals for two fuzzers do not overlap, ranking them by \(T\cdot S^2_{M,R}(T)\) also ranks them by their underlying robustness. A larger spectral gap reduces the effects of slow mixing, while the $1/T$ term retains the finite-horizon contribution of the starting distribution. Within the two-state model, the exact starting-state expression of Theorem~\ref{thm:main_body} replaces the worst-case initialization term by a supremum over every possible start.

\noindent Taking all three terms at their largest values over the same $95\%$ confidence region for the two estimated parameters $\hat\pi_1$ and $\hat\lambda$ gives the per-cell interval in Table~\ref{tab:spectral}. It is up to $46\times$ tighter than the worst-case bound (Honggfuzz on openssl, $17.3\to0.37$) and at least $1.8\times$ tighter in every cell.

\noindent\textbf{How many trials suffice.} Write $b(T)=b_{\mathrm{start}}(T)+U_{B,R}(T)$ for the part of the bound controlled by the horizon rather than the number of trials. To obtain a margin $\varepsilon$ for all $|\mathcal R|$ fuzzers at confidence $1-\alpha$, choose $T$ such that $b(T)<\varepsilon$ and then the smallest $M$ satisfying $E_M(\alpha/|\mathcal R|)\le\varepsilon-b(T)$. On our suites, this rule stabilizes near $M=20$ and $T=12$\,h (Table~\ref{tab:samplecomplexity}).

\textbf{Modeling scope.} The guarantee requires the input chain to forget its starting point at a geometric rate; it does not assume that the binary bug indicator is itself a Markov chain. Appendix~\ref{app:sketches} establishes the bound for the true bug indicator and specifies the quantity measured by the estimated gap. Table~\ref{tab:spectral} reports diagnostics computed inside the two-state model. Under behavioural drift, these diagnostics characterize locally consistent campaign phases, and Section~\ref{subsec:empirical-assumptions} reports how the estimated gap changes across phases.

The same root campaigns provide both \(\widetilde{P}_R(T)\), the effort-normalized bug detection rate, and \(T\cdot S^2_{M,R}(T)\), the corresponding robustness estimator. Their mean and effort-normalized variance capture effectiveness and robustness under the same mutation effort, respectively. This connection motivates the splitting-based enhancement studied next.

%% file: section/variance-reduction.tex
\section{Enhancement of Fuzzing}\label{sec:splitting}

Robustness is a property of the fuzzer--benchmark pair (Section~\ref{sec:evaluation}), which raises the follow-up: can we \emph{improve} both robustness and bug detection? We propose a splitting-based method that enhances any fuzzer by \textbf{reallocating mutation effort} toward bug-rich regions: when a trial discovers a bug, we fork the trajectory into descendants continuing from that state. It is a black-box plug-in guided directly by bug signals, with no handcrafted importance function (Section~\ref{sec:related work}).

\subsection{Relationship to Power Scheduling}\label{subsec:splitting-vs-scheduling}

Splitting and power scheduling (AFLFast~\cite{AFLFast}, MOpt~\cite{lyu2019mopt}) both concentrate effort in promising regions, but differ in mechanism, trigger signal, and guarantees. Two distinctions matter.

First, splitting never loses bug events (Theorem~\ref{bug_rate_splitting}), raises the expected rate when allocation and future yield move together (Proposition~\ref{prop:sizebias}), and lowers variance per unit CPU (Theorem~\ref{thm:splitting_var}) under stated clustered-bug conditions, whereas power scheduling's analysis~\cite{AFLFast} concerns path discovery and says nothing about bug-detection variance. Second, we have a direct head-to-head on the same ancestor fuzzer at equal compute: splitting-wrapped AFL attains a higher terminal BDR than plain AFLFast and plain MOpt on \emph{all ten} benchmarks (median $+25\%$ and $+22\%$; Wilcoxon $p=0.002$ each), and it improves even AFL++, whose power schedule it composes with, by up to 125\%.

\para{Splitting is not simply ``more independent runs''}
More \emph{independent} runs explore \emph{uniformly}, blind to where bugs occur. Splitting runs one chain until a bug is \emph{triggered}, then forks descendants inheriting the corpus state \emph{at that moment}: effort is spent \emph{conditionally and only where bugs appear}, which is what Theorems~\ref{bug_rate_splitting} and~\ref{thm:splitting_var} formalize.

\para{Deployment: cost and correctness}
Splitting is a black-box wrapper needing only read access to the on-disk queue and the bug signal; it modifies no fuzzer internals, and only the queue path differs across the seven fuzzers. On a split the parent stops and $B$ instances launch from a read-only copy of its queue, each pinned to a core; a cap ($K{=}3$ depths, $B{=}2$) bounds a trial to eight concurrent leaves, each charged in full per CPU-hour. Descendants are fresh processes on a read-only seed mount, inheriting \emph{no} in-memory state. The only overhead is stop/copy/relaunch dead-time, median $0.14\%$ of wall-time over the $187$ online trials that ran the full $12$ hours. Details in Appendix~\ref{app:implementation}.

\subsection{Methodology based on Splitting}

\begin{algorithm}[!t]
\SetAlgoLined
\caption{Splitting Algorithm}\label{alg:splitting}
\textbf{Initialization:} Maximum path length $T$, input seed $w$, $t = 1$, active paths $k_t$ ($k_1 = 1$), fuzzer $F$, mutation speed $\Delta +1 \in \mathbb{N}^{+}$ (default $\Delta = 0$), splits at time $t$: $n_t \in \mathbb{N}^{+}$, bug count $N_t$ ($N_0=0$), bug-triggering inputs $E_t$, new paths $TMP_t$\;\par
 \For{$1 \leq t \leq T$}{
\For(\tcp*[f]{Simulation}){each $s_t = X(w,t) \in [k_t]$}{
    Mutate $X(w,t)$ with fuzzer $F$ and generate $X(w,t+1), \ldots, X(w,t+\Delta+1)$\; \par
    Execute and observe $Y(w,t+1), \ldots, Y(w,t+\Delta+1)$\; \par
    \eIf{$\exists\; 1 \leq i \leq \Delta+1$ s.t.\ $Y(w , t+i) = 1 $}{
    Add $X(w , t+i)$ to $E_t$; update $N_t \leftarrow N_t+1$ (with $N_t$ initialized to $N_{t-1}$ at the start of round $t$)\; \par
    Stop $s_t$; split $X(w , t+i)$ into $n_t$ copies $\to tmp_t$}
{Keep $s_t$; $tmp_t = \emptyset$}
Add $tmp_t$ to $TMP_t$ \;
}
  Update $[k_{t+1}] = TMP_t \cup \{\text{kept paths}\}$\; \par
   \For{each sample path $s_t \in [k_{t+1}]$}{
   Update path weight}
    }
\textbf{Output:} $E_T$, $N_T$ and $\{k_t\}_{t=1}^{T}$
\end{algorithm}

\para{Splitting algorithm}
We present the proposed method in Algorithm \ref{alg:splitting}; a split parent is replaced by its children rather than kept alongside them. We retain the MC simulation framework of Section~\ref{sec:evaluation-methods} but modify the procedure: whenever a bug is detected, we \emph{split} the run into multiple descendant runs, each initialized at the current input with a randomly sampled seed, and each continuing to fuzz independently (our implementations additionally gate splits by a sparsity-zone signal; Section~\ref{sec:offline} and Appendix~\ref{app:implementation}). As in Section~\ref{sec:evaluation-methods}, each run yields one input--output trajectory and $M$ runs yield $M$ i.i.d.\ trajectories.

\para{Point estimators} Splitting modifies the Section~\ref{sec:evaluation-methods} estimators only through the branch weighting $1/k_t^m$:
\[
\begin{aligned}
\hat{P}_R^{m}(T)&=\frac{1}{T}\sum_{t=1}^{T}\frac{1}{k_t^m}\sum_{s=1}^{k_t^m}Y_{s,t},\\
\Tilde{P}_R&=\frac1M\sum_{m}\hat{P}_R^{m},\qquad S^2_{M,R}=\frac1{M-1}\sum_m\bigl(\Tilde{P}_R-\hat{P}_R^m\bigr)^2,
\end{aligned}
\]
with $k_t^m$ the active paths at time $t$ in run $m$. The total CPU-hours are $\mathrm{CPU}_R(T)=\sum_{m,t}k_t^m$, summed over all $M$ trials and every live leaf (Monte Carlo pays $MT$; splitting pays for each leaf launched), and $\mathcal{J}_R(T)=\mathrm{CPU}_R(T)\,S^2_{M,R}(T)$: each arm is charged the compute it consumed.

\subsection{Effectiveness of Splitting}

Splitting is most effective when bugs cluster: there it reduces variability \emph{and} raises the bug-detection rate. Practitioners need not know this in advance. It triggers \emph{reactively} on an observed bug, still guarantees non-regression of the bug-event count when structure is sparse (Theorem~\ref{bug_rate_splitting}), and its variance can be monitored online. Our experiments show the clustered regime is the common case (Section~\ref{sec:exp}).

\subsubsection{Variance Comparison}

We consider effort-normalized variance, defined by
\[
\mathcal{E}(T)\;:=\;\mathrm{Var}(\text{estimator at horizon }T)\times \mathrm{CPU}(T),
\]
of which \(\hat\sigma^2(T)=\mathrm{CPU}_R(T)\,S^2_{M,R}(T)\) is the empirical estimator; \(\mathcal{E}_{\mathrm{split}}(T)\) and \(\mathcal{E}_{\mathrm{MC}}(T)\) denote it for the two arms.

\begin{theorem}[Idealized two-level allocation]\label{thm:splitting_var}\leavevmode
Consider an idealized two-level allocation with a stationary, reversible, positive input chain (Assumption~A1), so all lag-autocovariances of $g=\mathbf 1_B-\pi(B)$ are non-negative; a binary allocation $k_t\in{1,2}$ with independent post-allocation replicas; and, for the closed-form expression, the two-state surrogate. Let $C_{\mathrm{in}},C_{\mathrm{out}}$ be the in- and out-region autocovariance masses and $p=|S|/T$ the in-region time fraction. Then
\[
\begin{aligned}
\frac{\mathcal{E}_{\mathrm{split}}(T)}{\mathcal{E}_{\mathrm{MC}}(T)} &\;\le\; \frac{T+|S|}{T}\cdot\frac{C_{\mathrm{out}}+\tfrac12 C_{\mathrm{in}}}{C_{\mathrm{out}}+C_{\mathrm{in}}}\\
&\;=\;(1+p)\,\frac{2+r}{2(1+r)},\quad r:=\tfrac{C_{\mathrm{in}}}{C_{\mathrm{out}}},
\end{aligned}
\]
a ratio bound: positivity makes $\mathcal{E}_{\mathrm{MC}}$ exact, supplying the matching lower bound.
(Proof in Appendix~\ref{app:pf-thm-splitting}.)
\end{theorem}

The bound is below $1$, and hence splitting reduces variance per unit CPU, exactly when $p<C_{\mathrm{in}}/(2C_{\mathrm{out}}+C_{\mathrm{in}})$. When bugs are strongly clustered, $C_{\mathrm{in}}\gg C_{\mathrm{out}}$, this condition approaches $p<1$, and the variance ratio approaches $(1+p)/2$. The two-state output surrogate makes this relationship explicit: $r=(1-p)/p$, so rare visits to the bug region (small $p$) correspond to strong clustering (large $r$). Substituting this identity gives a ratio depending only on $p$, $(1+p)^2/2$, which is below $1$ exactly when $p<\sqrt{2}-1\approx0.414$ (Appendix~\ref{app:pf-thm-splitting}). The ratio bound also requires positivity to obtain a matching lower bound; in the two-state model, this reduces to the directly checkable condition $\hat\lambda_2\ge0$, which holds in $102$ of the $110$ evaluated cells and in all $40$ Magma cells. Theorem~\ref{thm:splitting_var} considers the binary fork $k_t\in\{1,2\}$ that yields this closed-form rule. Our deployed schedule instead forms a larger tree whose descendants remain active after leaving the bug region; the next result extends the analysis to this general tree.


\begin{proposition}[Branching-tree variance bound, tree fixed]\label{prop:tree}
Let $\mathcal{T}$ be the splitting tree, $L_t$ its leaves at time $t$, and
$\hat P_{\mathcal{T}}=\frac1T\sum_{t}\sum_{\ell\in L_t}w_{\ell,t}Y_{\ell,t}$ with $w_{\ell,t}=1/|L_t|$
the estimator we deploy, which weights every live branch equally. For leaves $\ell\neq j$ let $\tau_{\ell j}$ be the fork time of their
most recent common ancestor, a time decided by the trajectory itself, let $\nu_\tau$ be the distribution of the state at that fork, and let $\kappa$ be the largest factor by which any such fork distribution can inflate a probability relative to the stationary one ($\kappa=1$ if forks happen at stationarity). Assume the contraction $\lVert P^{n}g\rVert_\pi\le\rho^{\,n}\lVert g\rVert_\pi$ for some $\rho\in(0,1)$, where $\lVert\cdot\rVert_\pi$ is the root-mean-square size of a function under $\pi$; this holds with $\rho=1-\gamma^X_R$ for a reversible input chain. Descendants forked at $\tau$ are independent of each other once the state at the fork is fixed, and the tree's shape and fork times are fixed independently of that post-fork randomness. Write $c_t=\sum_{u=1}^{T}\bigl|\operatorname{Cov}_\pi\bigl(g(X_t),g(X_u)\bigr)\bigr|$ for the total autocovariance a single branch contributes at time $t$. Then, for that tree, 
\[
\begin{aligned}
\operatorname{Var}\bigl(\hat P_{\mathcal{T}}\bigr)\;\le\;\frac1{T^{2}}\,\mathbb{E}\Bigl[&\sum_{t=1}^{T}\frac{c_t}{|L_t|}
\;+\;\kappa\operatorname{Var}_\pi(g)\;\times\\[-2pt]
&\sum_{t,u}\ \sum_{\substack{\ell\in L_t,\,j\in L_u\\ \ell\neq j}}
w_{\ell,t}w_{j,u}\,\rho^{\,(t-\tau_{\ell j})+(u-\tau_{\ell j})}\Bigr].
\end{aligned}
\]
(Proof in Appendix~\ref{app:pf-prop-tree}.)
\end{proposition}
Proposition~\ref{prop:tree} gives a general upper bound for split trees that depends on their structure only through the ancestry times $\tau_{\ell j}$. It therefore covers trees of any shape, including descendants that persist to the horizon, with $\kappa$ accounting explicitly for the distribution at the split state. Our deployed schedule chooses its forks from observed bugs, so its total variance carries a further term from that choice, which we do not bound; its net effect is what the matched-CPU experiments measure. At strictly matched CPU, the deployed schedule lowers the CPU-weighted variance across campaigns in $66$ of $70$ pairs, with a median split-to-baseline ratio of $0.098$ (approximately $10\times$ lower; Section~\ref{sec:exp:setup}).

\subsubsection{Bug Detection Rate Comparison}

We denote by \(|B_{\text{MC}}|\) the number of bug-triggering events observed by \(\text{MC}\), so $R_{\text{MC}} = |B_{\text{MC}}|/T$, while the splitting rate is $R_{\text{Split}} = \frac{1}{T}\sum_{t=1}^{T}\frac{1}{k_t}\sum_{s=1}^{k_t}Y_{s,t}$, the $1/k_t$ weighting giving each of the $k_t$ active paths equal voice so the rate stays an unbiased per-path average under branching. Write $J(q,\xi)$ for the restart kernel that launches a fresh process from a copied queue state $q$ with process randomness $\xi$: the deployed descendants are draws from $J(q_s,\cdot)$, so taking the comparator to be one further draw from that same kernel, coupled to descendant $1$, makes the identity below exact for what is run. No in-memory state is assumed copied.

\begin{theorem}[Pathwise event containment, any $k_t$]\label{bug_rate_splitting}
Fix any realized split tree and designate one continuation child at every fork; write $\ell_0(t)$ for the branch this designated path occupies at time $t$, and let $\widehat P_{0}$ and $N_0$ be its rate and its raw event count. Then, writing the rate so that every live branch counts equally, for the rate $\widehat P_{\mathcal{T}}$ and raw count $N_{\mathcal{T}}$ of the whole tree,
\[
\begin{aligned}
\widehat P_{\mathcal{T}}-\widehat P_{0}
&=\frac1T\sum_{t=1}^{T}\frac{1}{k_t}\sum_{\ell=2}^{k_t}\bigl(Y_{\ell,t}-Y_{1,t}\bigr),\\
N_{\mathcal{T}}-N_{0}
&=\sum_{t}\sum_{\ell\neq\ell_0(t)}Y_{\ell,t}\ \ge\ 0 .
\end{aligned}
\]
(Proof in Appendix~\ref{app:pf-thm-bugrate}.)
\end{theorem}

\noindent The raw count therefore never decreases, for any $k_t\ge1$ and hence for every tree the deployed schedule produces, and the weighted rate improves exactly when the first right-hand side is non-negative. The path being compared against is the one embedded in the same tree; the comparison against a separately run Monte Carlo arm at equal CPU is empirical (Section~\ref{sec:evaluation}).

\begin{proposition}[Adaptive branching raises the expected rate]\label{prop:sizebias}\leavevmode
Let $m(x)$ be how often state $x$ splits and $V(x)=\mathbb{E}_x\bigl[\sum_{u>t}Y_u\bigr]$ how many bugs follow from it.
One reproduction step re-weights the population of live branches from $\mu$ to $\mu_m(dx)=m(x)\mu(dx)/\mathbb{E}_\mu[m]$, so
\[
\mathbb{E}_{\mu_m}[V]-\mathbb{E}_{\mu}[V]
=\frac{\mathbb{E}_\mu[mV]-\mathbb{E}_\mu[m]\mathbb{E}_\mu[V]}{\mathbb{E}_\mu[m]}
=\frac{\operatorname{Cov}_\mu(m,V)}{\mathbb{E}_\mu[m]} .
\]
Hence the expected bug-detection rate strictly increases if and only if $\operatorname{Cov}_\mu(m,V)>0$.
(Proof in Appendix~\ref{app:pf-prop-sizebias}.)
\end{proposition}

States that branch more also tend to have higher future bug yield. For a single fork whose descendants are independent once the state at the fork is fixed, $m$ is constant at the split state, so $\operatorname{Cov}_\mu(m,V)=0$. In the $K$-stage schedule, however, $m$ increases with observed bug activity $A$; when bugs cluster, the future yield $V$ also increases with $A$. Thus, $m$ and $V$ are non-decreasing functions of the same signal, and two quantities that both rise with the same signal cannot be negatively correlated, so $\operatorname{Cov}_\mu(m,V)\ge0$, with strict inequality when both vary on sets of positive probability. The same mechanism applies at successive branching stages.


\subsubsection{Discussions}\label{discussion_cases}
The BDR increment is bounded by \(\frac{1}{2T}(|B_{\text{split}}|-|B_{\text{MC}}|)\) and, in the strong-clustering limit, the variance-increment ratio approaches \(\frac{T+|S|}{2T}\), a \textbf{trade-off}: when the event surplus is small, splitting improves both detection and variance (\textbf{best-of-both-worlds}); when large, it trades variance for a larger detection gain. The derivations are in Appendices~\ref{app:pf-thm-splitting} and~\ref{app:pf-thm-bugrate}.

The experiments that follow confirm both at scale.

%% file: section/numerical-experiments.tex
\section{Numerical Experiments}\label{sec:exp}

We validate on two \emph{complementary} families: FuzzBench, whose deduplicated \texttt{crash\_key} count~\cite{fuzzbench} measures bug discovery and its variance, and Magma, whose planted-bug ground truth confirms bug-finding against real CVEs. The complete grids are in Appendix~\ref{app:appendix_experiments}; the data and code are in the artifact.\footref{fn:codebase}

\subsection{Experimental Setup and Data}\label{sec:exp:setup}

We evaluate seven widely used mutation-based fuzzers (AFL~\cite{afl}, AFLFast~\cite{AFLFast}, AFL++~\cite{AFL++}, AFLSmart~\cite{AFLSmart}, MOpt~\cite{lyu2019mopt}, libFuzzer~\cite{LibFuzzer}, Honggfuzz~\cite{honggfuzz}), following~\cite{schloegel2024sok}, under both Monte Carlo evaluation and splitting (Algorithm~\ref{alg:splitting}) on 10 FuzzBench benchmarks~\cite{fuzzbench}: arrow, ffmpeg, grok, libhevc, libhtp, matio, openh264, php, poppler, stb, each run for up to \(T=23\) hours. Several benchmarks contain multiple programs.

\para{Experimental protocol}
Fix a benchmark--fuzzer pair $(b,f)$.
Each \emph{trial} is a $T{=}23$\,h campaign. Both arms have $M{=}20$ independent trials per pair, splitting's descendants weighted by $1/k_t$ (Section~\ref{sec:splitting}). The arms are compute-matched: the splitting arm's total CPU averages $1.00\times$ the non-split arm's $460$ CPU-hours (median $1.10\times$; $53/70$ pairs within $\pm20\%$), so the MC arm doubles as the natural baseline of more independent instances at equal CPU. At strictly matched per-pair CPU, on the \emph{raw} union of unique bugs (no $1/k_t$ weighting), splitting beats this independent-instance baseline in $53$ of $70$ pairs ($7$ ties, $10$ losses; mean $+3.4$, $95\%$ CI $[2.3,4.6]$; sign test $p=1.7\times10^{-8}$; two-way clustered on per-benchmark and per-fuzzer medians, $p=3.9\times10^{-3}$, $7.8\times10^{-3}$), and is more stable in $66$ of $70$ (median ratio $0.098$, $\approx10\times$). Matching drops \emph{whole} root trials from the larger arm, never truncating a campaign, until its retained CPU no longer exceeds the smaller arm's, pools and deduplicates unique bugs identically on both sides, and repeats over $10{,}000$ seeded resamples. Every split branch is charged in full.
FuzzBench records every 15 minutes; we align both arms on a common 6-minute grid ($\delta=0.1$\,h) by forward-filling, and all time extrema are taken over that grid.

\subsubsection{Bug Definition and Deduplication}\label{subsec:bug-def}

\para{Sanitizer-based bug oracle}
We compile each benchmark with AddressSanitizer and UndefinedBehaviorSanitizer, so violations deterministically raise \emph{fatal} reports; each is evidence of bug-triggering behavior.

\para{What we count as a \emph{bug}}
A bug is \emph{not} a failing input: many inputs trigger one defect. We count \emph{deduplicated failure signatures}, meaning sanitizer error class plus canonicalized stack trace, once per trial, at first appearance. This is FuzzBench's own \texttt{crash\_key} deduplication, recorded as ``unique bugs found''~\cite{fuzzbench}; we report \emph{unique bugs} throughout.

\para{Relation to ground-truth bugs}
FuzzBench's count of \emph{distinct} \texttt{crash\_key}s is the field-standard unique-bug metric, and Magma provides an independent, stronger check that the bug-event-to-distinct-bug link is real: its planted-bug identifiers confirm the gains directly: splitting never reduces distinct real bugs in any of the $40$ cells, wins per CPU-hour in $38$ of $40$, and converts real CVEs from undetected to reliably detected (Section~\ref{subsec:magma}).

\para{Clarification on figure y-axis}
Two quantities appear here, kept separate. Our theory models the \emph{bug-trigger event} rate: $Y_t$ indicates whether execution $t$ triggers a bug, so repeated triggers count; Magma reports it directly, and splitting fires more bug events in $39$ of $40$ cells, by up to $29.8\times$. The deployment-facing quantity is \emph{unique bugs}, distinct signatures counted once per trial, plotted as cumulative unique bugs over elapsed hours, FuzzBench's standard normalization~\cite{fuzzbench}. The two are coupled but not equal: within-trial Spearman $\rho=0.66$ on FuzzBench and $\rho=0.55$ on Magma. CPU-time is the fair basis because execution counts differ enormously between fuzzers at equal wall-clock time.

\subsection{Implementation}\label{sec:exp:casestudy}\label{sec:offline}
Appendix~\ref{app:implementation} details Algorithm~\ref{alg:splitting}: a sparsity signal from a \emph{prior, completed non-splitting campaign}, run separately and sharing no trial with either evaluation arm, zone entry once that signal stays below a stage threshold with persistence, and a split only when a running trial observes a new bug after entry. The schedule is fixed from historical data before the splitting campaign runs, so trials remain i.i.d.\ and use no information from their own future; the caps $K{=}3$ and $B{=}2$ are global constants and no parameter is chosen from splitting-arm outcomes. Because forks follow observed bugs, the tree is not fixed in advance, and the variance reduction we report for it is measured rather than bounded. An online version requiring no history follows in Section~\ref{sec:ablation}.

\subsection{Experiment Results}

We report robustness as the \emph{CPU-weighted variance across campaigns} $\mathcal{J}_R(T)=\mathrm{CPU}_R(T)\,S^2_{M,R}(T)$ (smaller is more stable) and effectiveness as the effort-normalized rate $\tilde{P}(t)$ (larger is stronger), both on a matched compute budget.
\para{What one data point is, and what we report} The replication unit is an independent root campaign, not a branch and not the average over campaigns. With $S^2_a$ the variance across campaigns of arm $a$ and $C_a$ the total CPU that arm consumed, every descendant branch included, we report $\mathcal{J}_a=C_a S^2_a$. It is deliberately \emph{not} $C_a\operatorname{Var}(\bar Z_a)=C_aS_a^2/M_a$, a different quantity, measuring the uncertainty of the grand mean: root counts are part of the matched-total-compute design and no division by $M_a$ is intended. Once $C_{\mathrm{Split}}=C_{\mathrm{MC}}$, $\mathcal{J}_{\mathrm{Split}}/\mathcal{J}_{\mathrm{MC}}=S^2_{\mathrm{Split}}/S^2_{\mathrm{MC}}$, so every ratio reported is already an ratio of the variances across campaigns.

\subsubsection{Monte Carlo Simulation}\label{mc simulation}

We first measure baseline variability under Monte Carlo (MC, non-splitting). Figure~\ref{fig:mc-simulation} shows the CPU-weighted variance across campaigns for two representative benchmarks (all ten in Appendix~\ref{app:appendix_experiments}): each row a benchmark, left MC and right splitting, one curve per fuzzer. Variance decays on most benchmarks, consistent with the Section~\ref{sec:evaluation} prediction that the criterion stabilises toward a finite effort-normalized variance, but heterogeneously: decreasing on \textsc{libhevc}, \textsc{libhtp}, \textsc{stb}, yet flat-to-rising on \textsc{php} and \textsc{poppler}. Robustness is thus a property of the fuzzer--benchmark pair and elapsed time, not the fuzzer alone.

\begin{figure}[htb]
    \centering
    \includegraphics[width=0.46\textwidth]{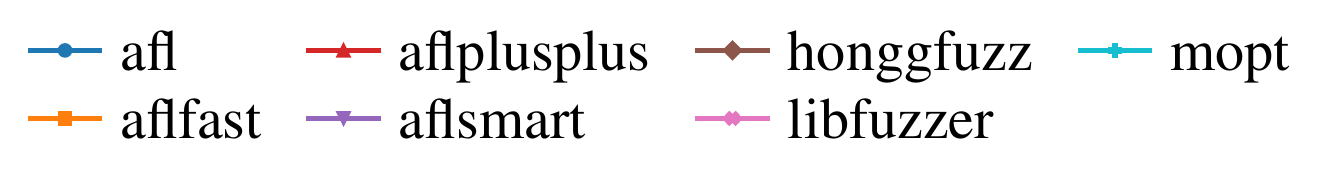}\\[3pt]
    \makebox[0.205\textwidth]{\footnotesize\textbf{Non-splitting (MC)}}\hfill\makebox[0.205\textwidth]{\footnotesize\textbf{Splitting}}\\[2pt]
    \includegraphics[width=0.205\textwidth]{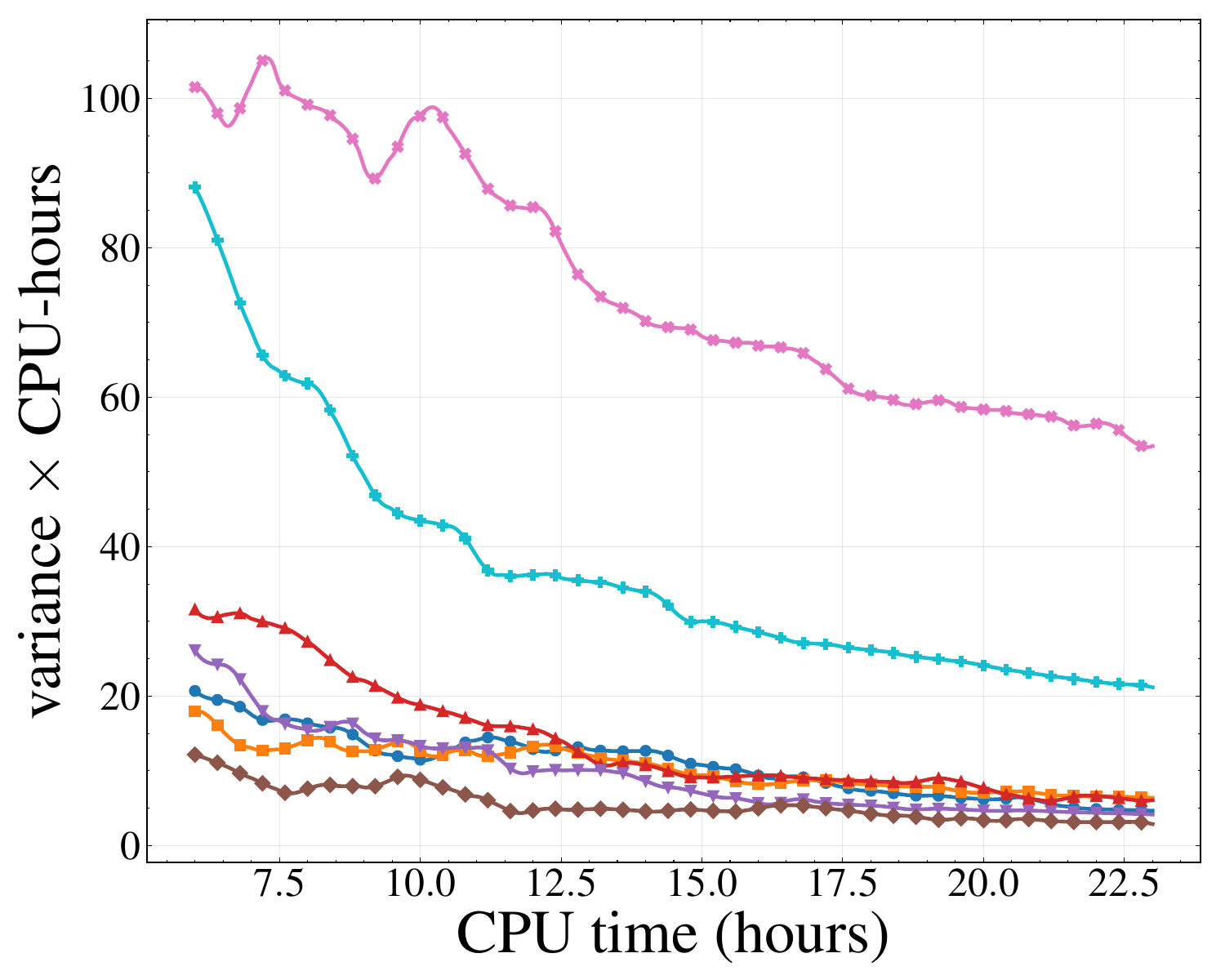}\hfill\includegraphics[width=0.205\textwidth]{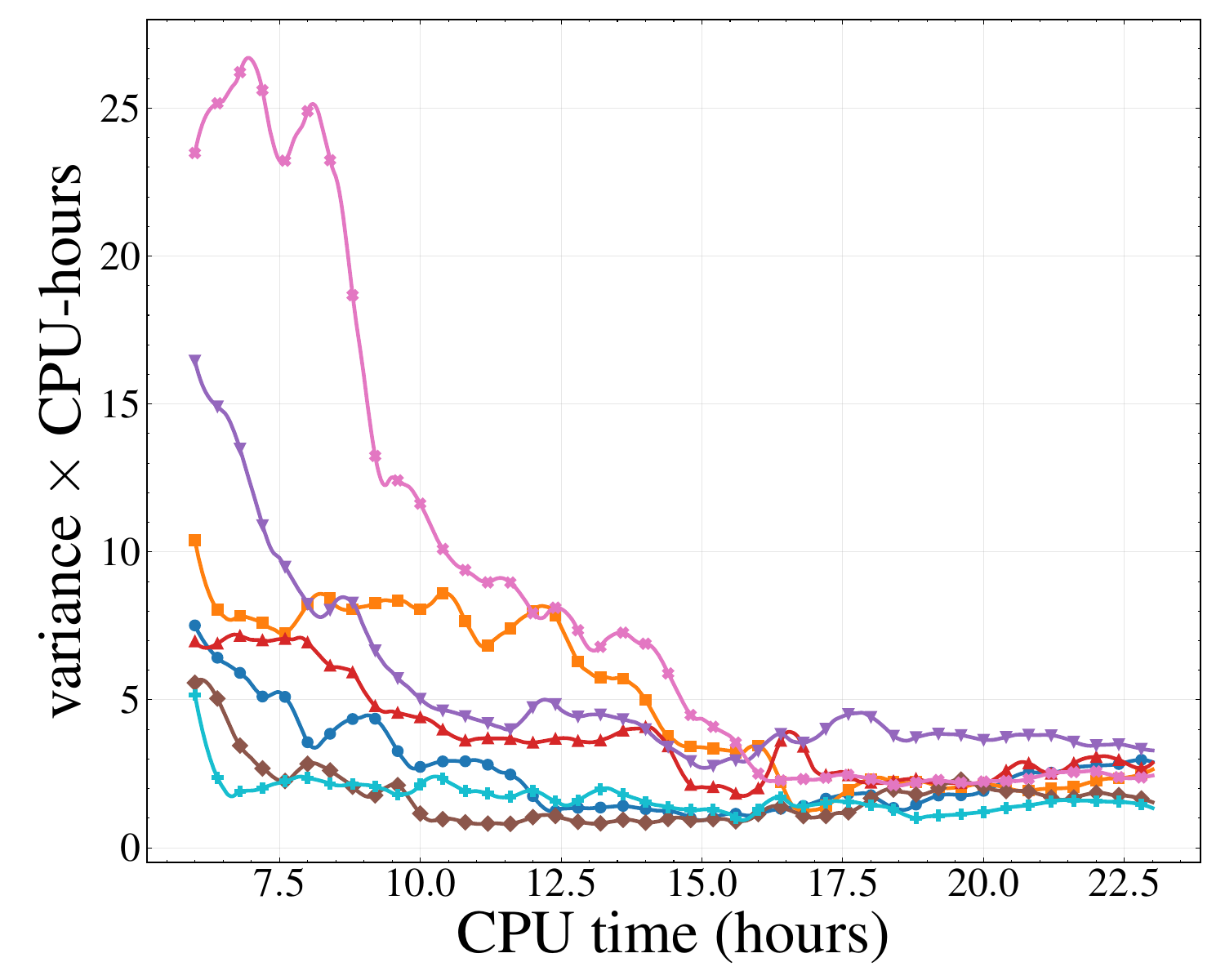}\\
    {\small (a,\,b) arrow}\\[3pt]
    \includegraphics[width=0.205\textwidth]{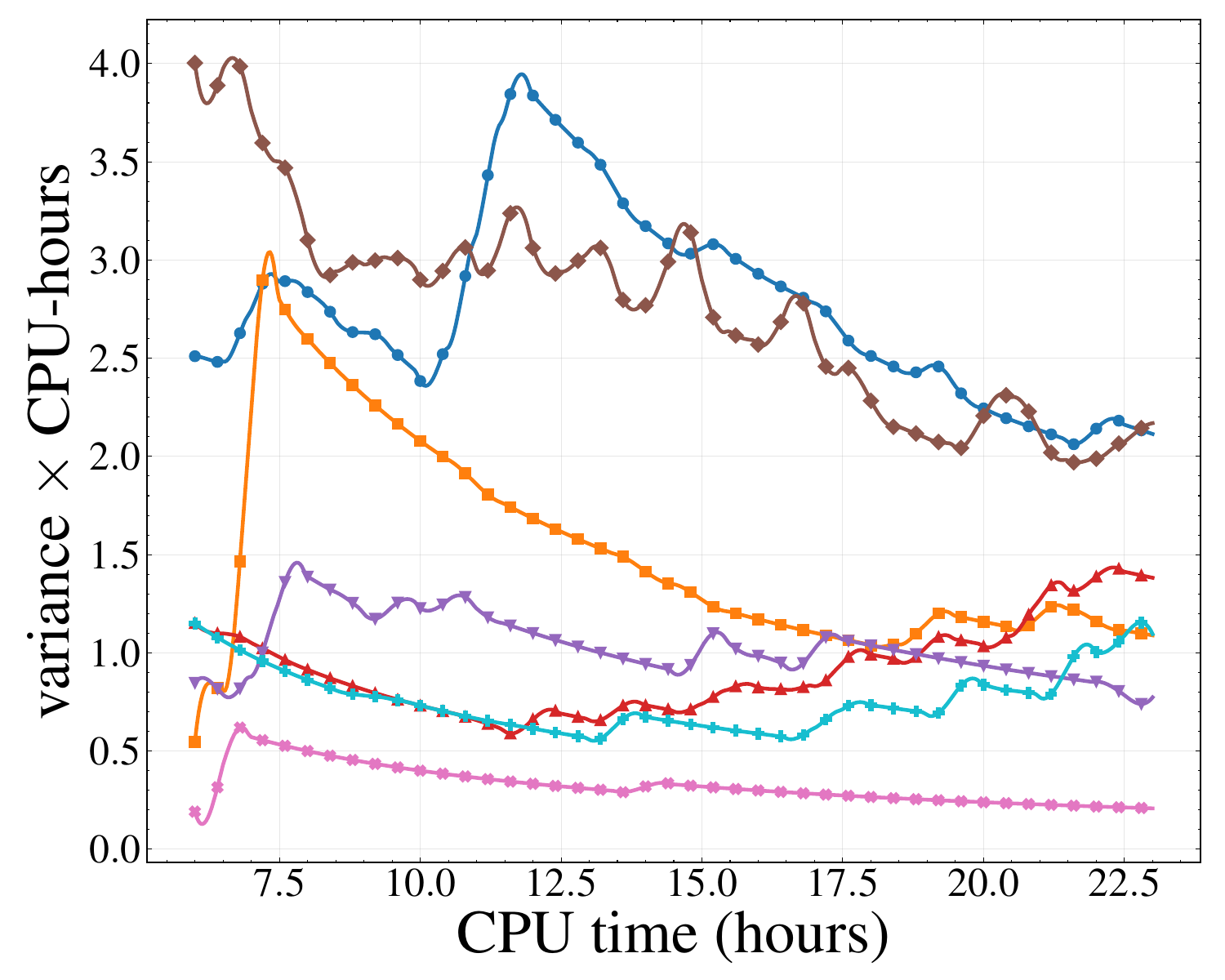}\hfill\includegraphics[width=0.205\textwidth]{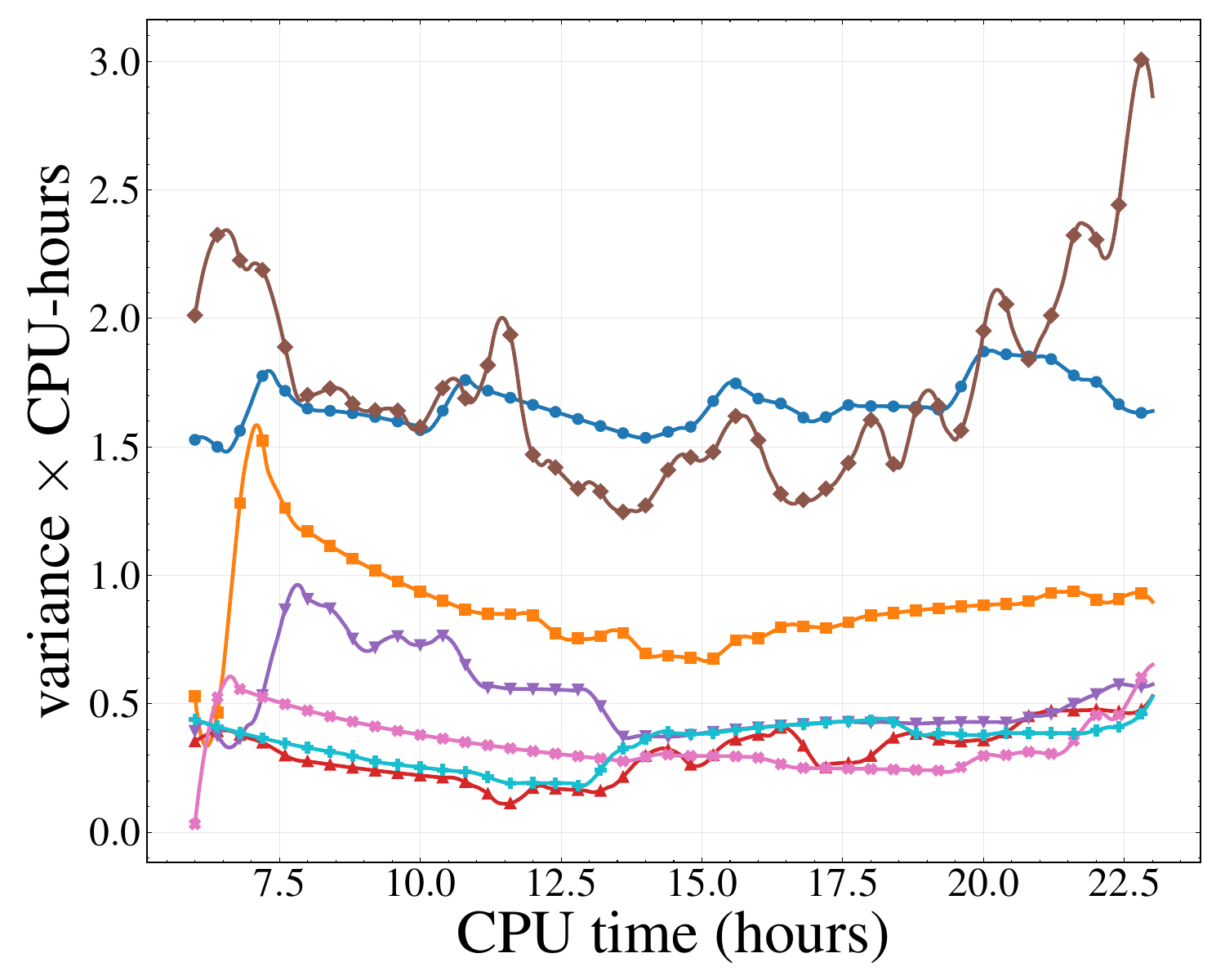}\\
    {\small (c,\,d) php}
    \caption{Variance ($\times$ CPU-hours) across the seven fuzzers under Monte Carlo (non-splitting, left) and splitting (right), for two representative benchmarks (one per row). The remaining eight benchmarks are in Appendix~\ref{app:appendix_experiments}.}
    \label{fig:mc-simulation}
\end{figure}

\subsubsection{Splitting changes the variance profile}\label{splitting simulation}

Splitting alters both the \emph{level} and \emph{temporal profile} of variance: trajectories contract earlier and stay lower for much of the run (\textsc{arrow}, \textsc{libhevc}, \textsc{libhtp}), while others (\textsc{ffmpeg}, \textsc{poppler}) fluctuate in a narrow band with mild late drift. On \textsc{libFuzzer}/\textsc{grok} and \textsc{libFuzzer}/\textsc{php}, splitting enlarges the discovered-bug set so much that variance rises, exactly the regime Theorem~\ref{thm:splitting_var} describes, and precisely where its rate gain is largest (up to $+800\%$).

\subsubsection{Comparisons}\label{variance_comparisons}

\para{Variance comparisons} We fix each fuzzer and compare its CPU-hour--normalized variance across benchmarks (The full grids appear in Appendix~\ref{app:appendix_experiments} and are also included in the artifact.); results fall into a majority regime of \emph{lower} normalized variance and corner cases where it is \emph{higher}.
In the majority regime (six fuzzers on all ten benchmarks and {libFuzzer} on eight of ten), splitting reduces the CPU-weighted variance across campaigns in $68$ of $70$ pairs over the displayed $[6,23]$-hour window; $66$ of $70$ at the terminal snapshot (the four exceptions reflect the predicted robustness--effectiveness trade-off); and under the most conservative accounting (strictly matched per-pair CPU on raw per-trial variance), $66$ of $70$, one effect under progressively less penalizing normalizations. Among winning pairs the terminal reduction has median $3.4\times$ and reaches $243\times$ (Honggfuzz on libhevc), and a Wilcoxon signed-rank test confirms it is systematic ($p = 1.9\times10^{-12}$). It is not an artifact of treating the $70$ cells as independent: aggregating to per-benchmark and per-fuzzer medians first, splitting is lower in $10$ of $10$ benchmarks ($p=9.8\times10^{-4}$) and $7$ of $7$ fuzzers ($p=7.8\times10^{-3}$). The time-resolved curves show the largest gains in the mid-phase of a run. The full grids are in Appendix~\ref{app:appendix_experiments}.

\begin{figure*}[t]
    \centering
    \includegraphics[width=0.19\textwidth]{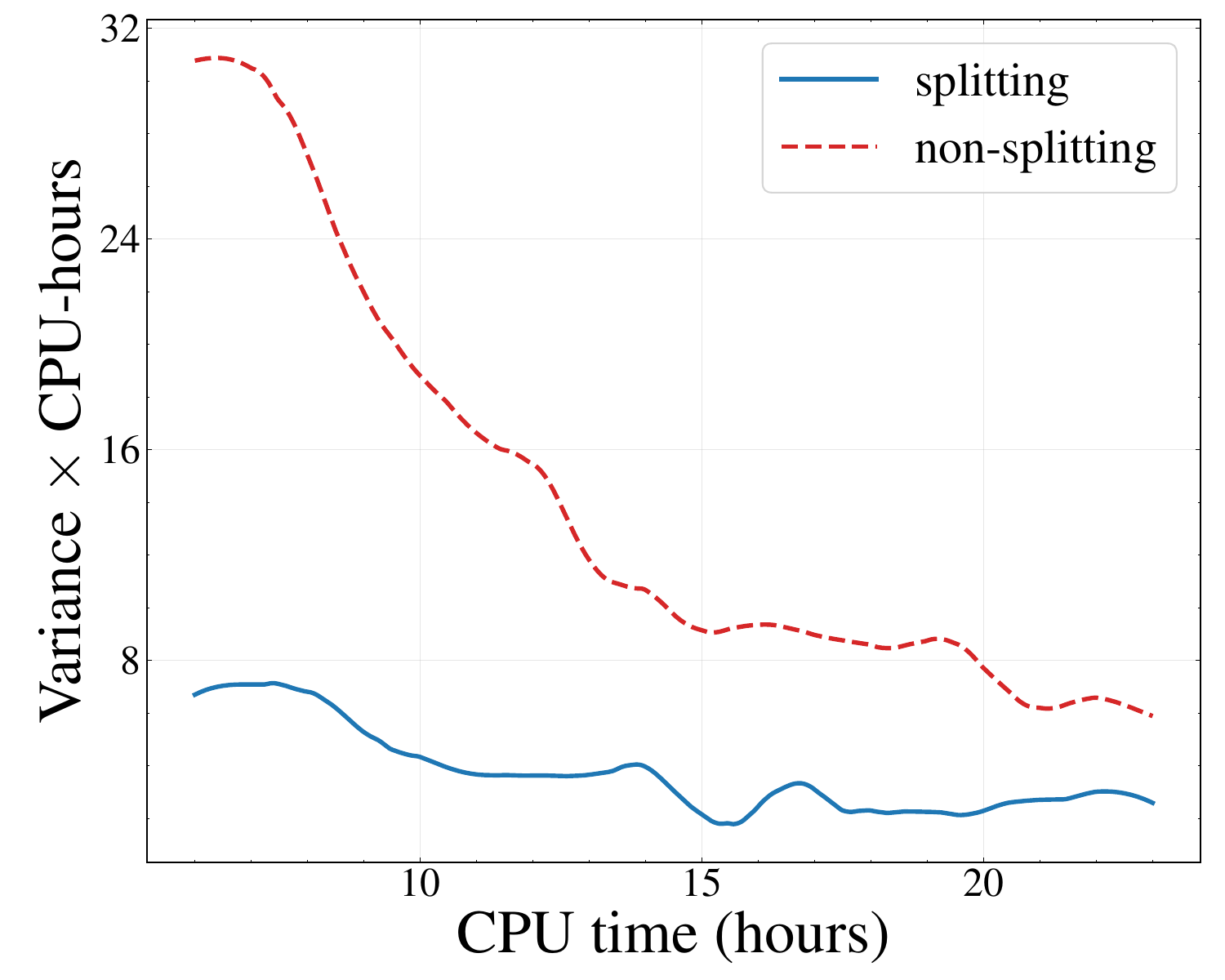}\hfill\includegraphics[width=0.19\textwidth]{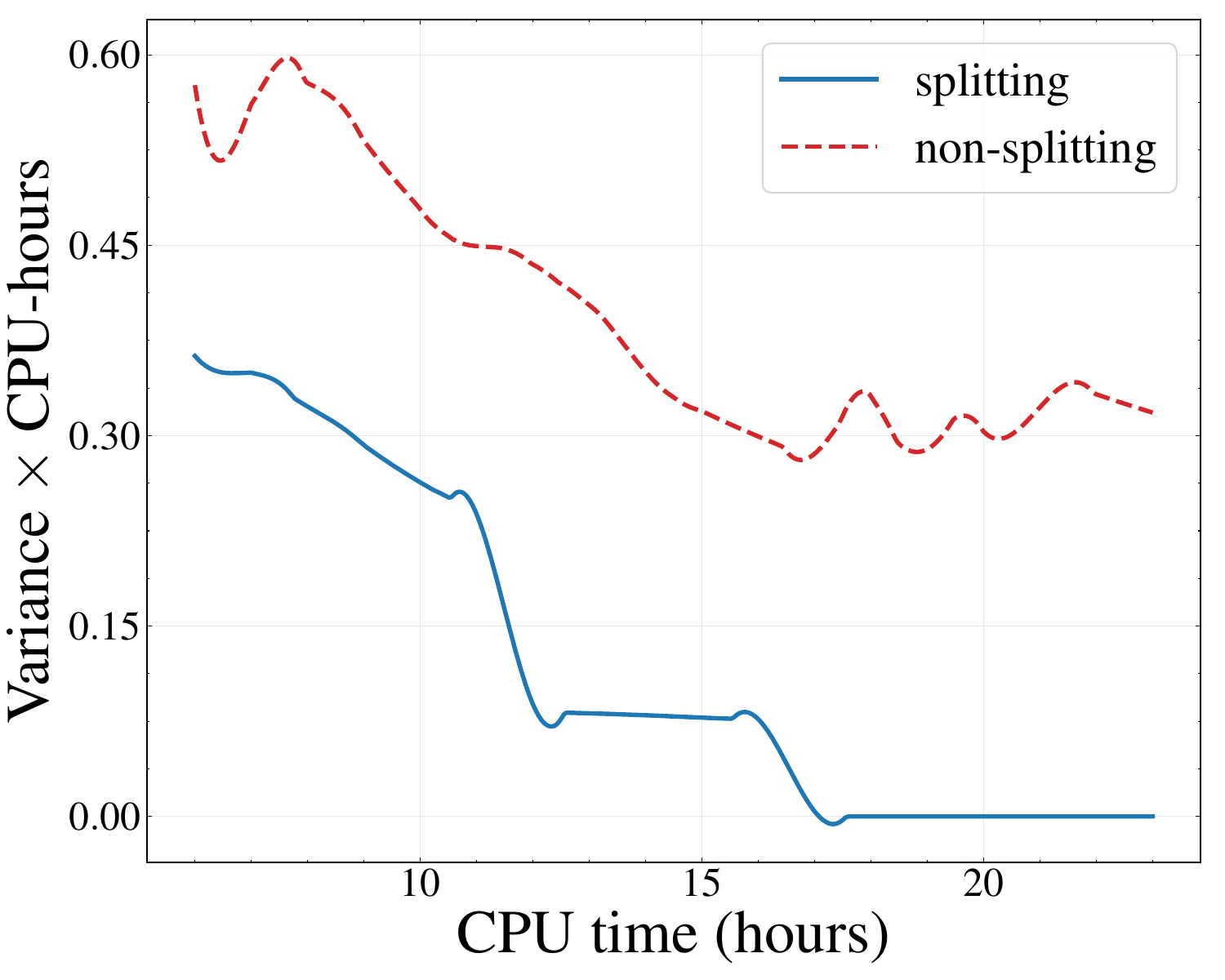}\hfill\includegraphics[width=0.19\textwidth]{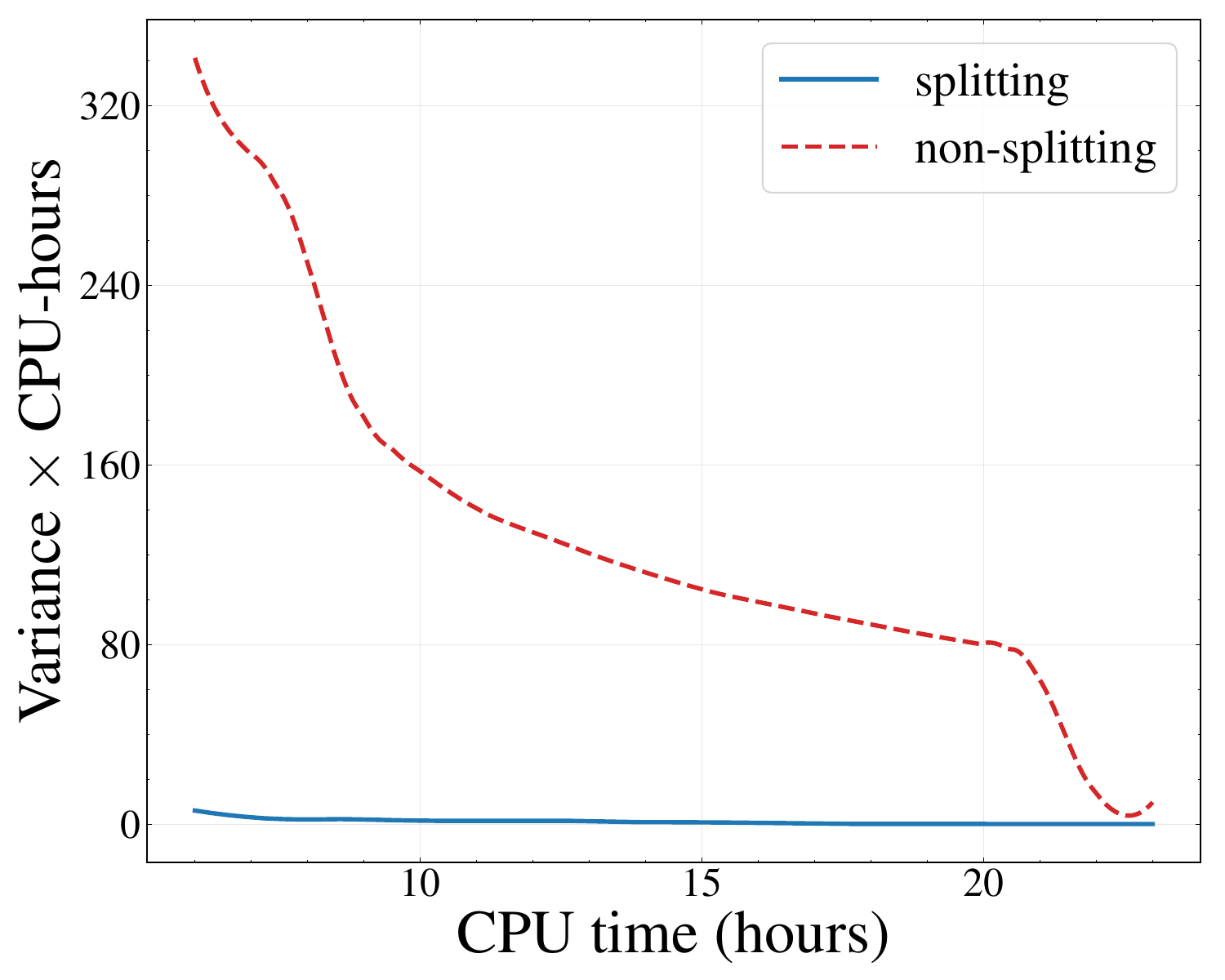}\hfill\includegraphics[width=0.19\textwidth]{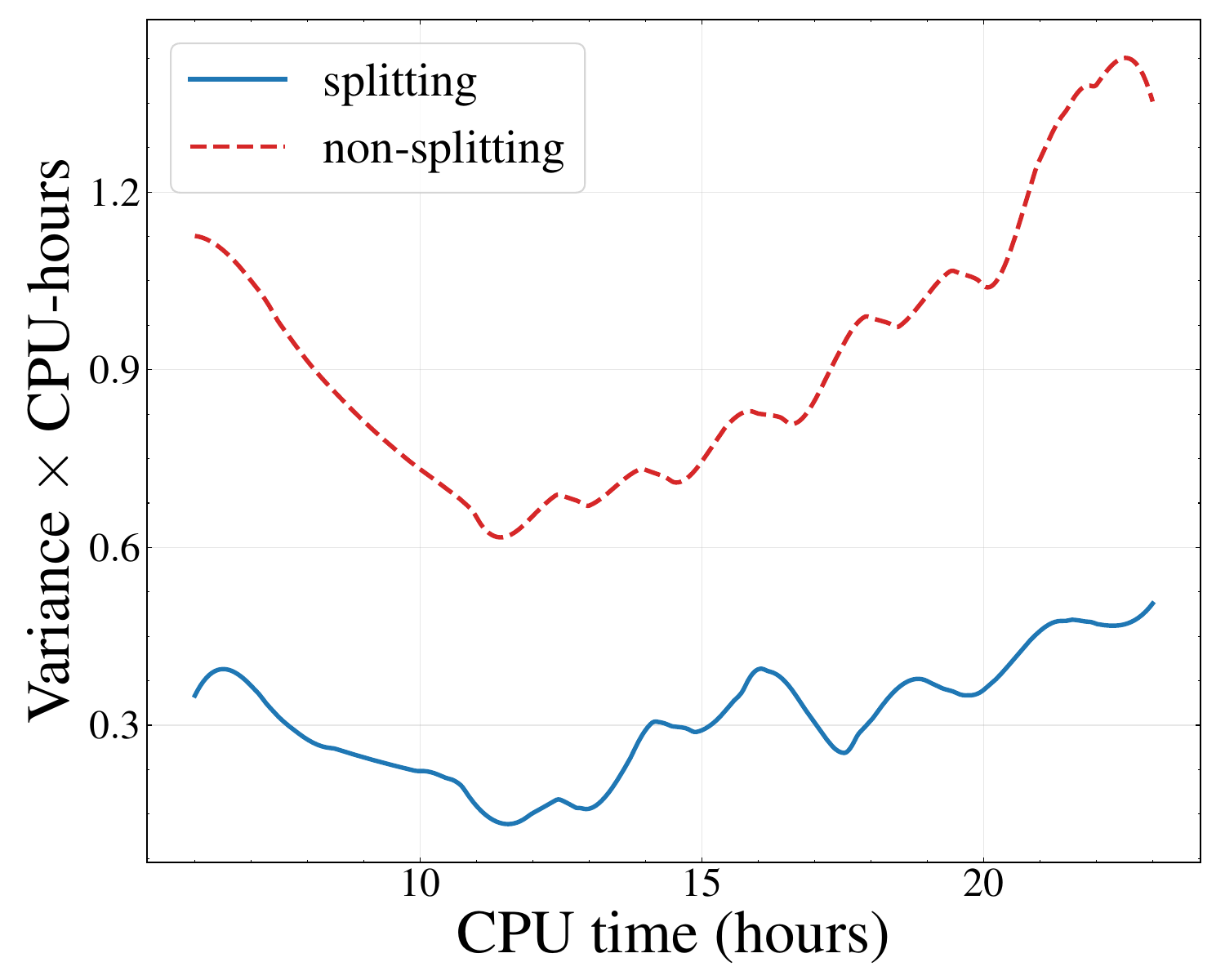}\hfill\includegraphics[width=0.19\textwidth]{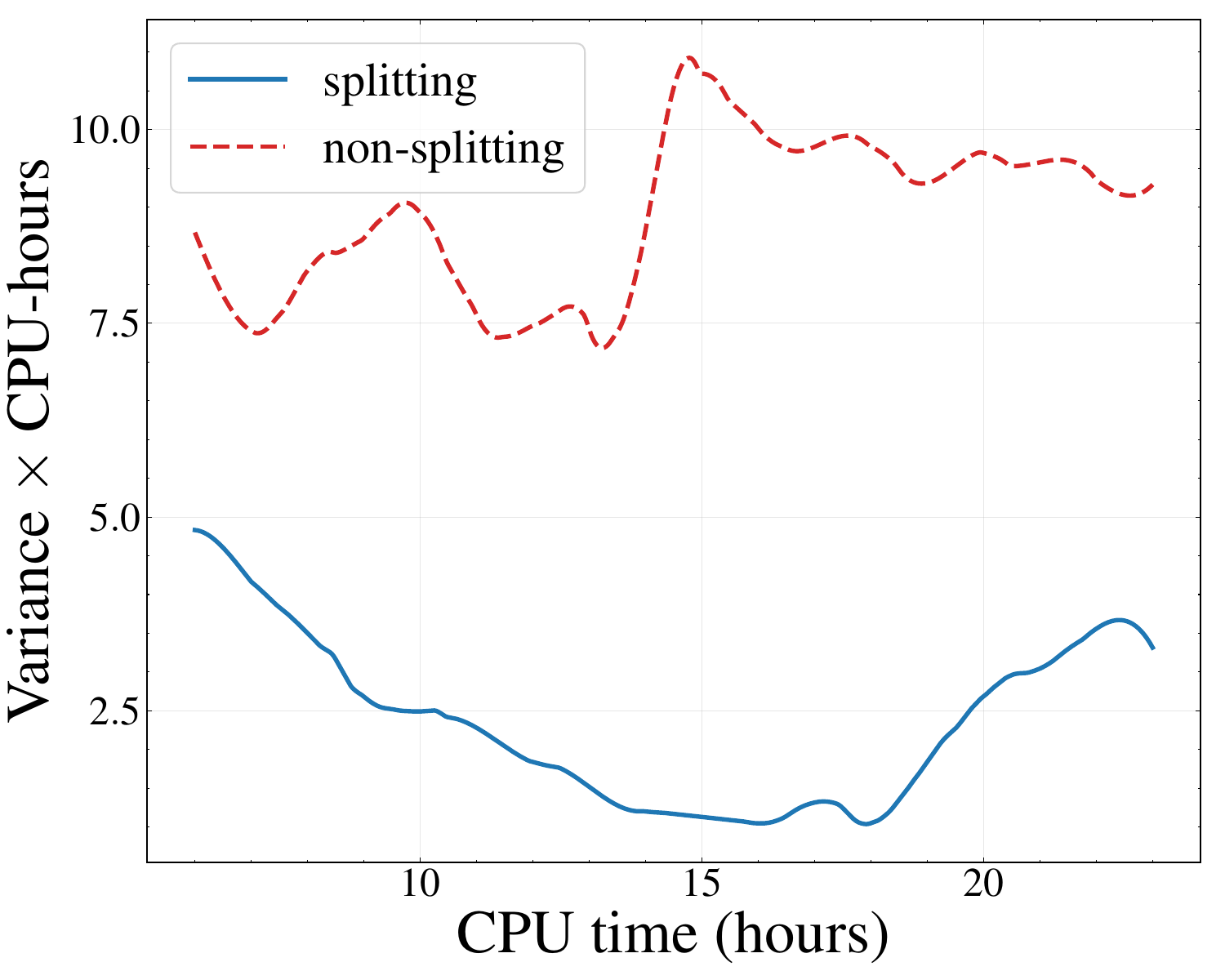}\\
    \makebox[0.19\textwidth]{\footnotesize (a) variance of arrow}\hfill\makebox[0.19\textwidth]{\footnotesize (b) variance of grok}\hfill\makebox[0.19\textwidth]{\footnotesize (c) variance of libhevc}\hfill\makebox[0.19\textwidth]{\footnotesize (d) variance of php}\hfill\makebox[0.19\textwidth]{\footnotesize (e) variance of poppler}
    \caption{Comparisons of variance (splitting vs.\ Monte Carlo) for fuzzer aflplusplus on five representative benchmarks. Full grids in Appendix~\ref{app:appendix_experiments}.}
    \label{fig:comparison_aflplusplus}
\end{figure*}

\para{Unique-bug discovery rate} Figure~\ref{fig:bdr_aflplusplus} compares the time-resolved unique-bug discovery rate for aflplusplus (solid: splitting; dashed: baseline). On FuzzBench this, not the per-execution event rate, is what can be measured, and across all $70$ pairs its terminal $1/k_t$-weighted value under splitting is strictly greater in \emph{every} pair (median $+25\%$, up to $+125\%$ for AFL++ on poppler; Wilcoxon $p = 3.6\times10^{-13}$); per-pair Mann--Whitney tests favor splitting in $51$ of $70$ pairs ($46$ after Holm; median $\hat{A}_{12}=0.85$), none favoring non-splitting. On the two pairs where splitting raises variance ({libFuzzer} on {grok}, {php}) its rate gain is largest. These gains hold on the fair per-CPU-hour basis, which charges every split leaf in full: at strictly matched per-pair CPU on the raw union of distinct bugs (no $1/k_t$ weighting), splitting still wins $53$ of $70$ pairs (Section~\ref{sec:exp:setup}), so the advantage is a genuine bug-finding gain per unit of compute, not an artifact of weighting.

\begin{figure*}[t]
    \centering
    \includegraphics[width=0.19\textwidth]{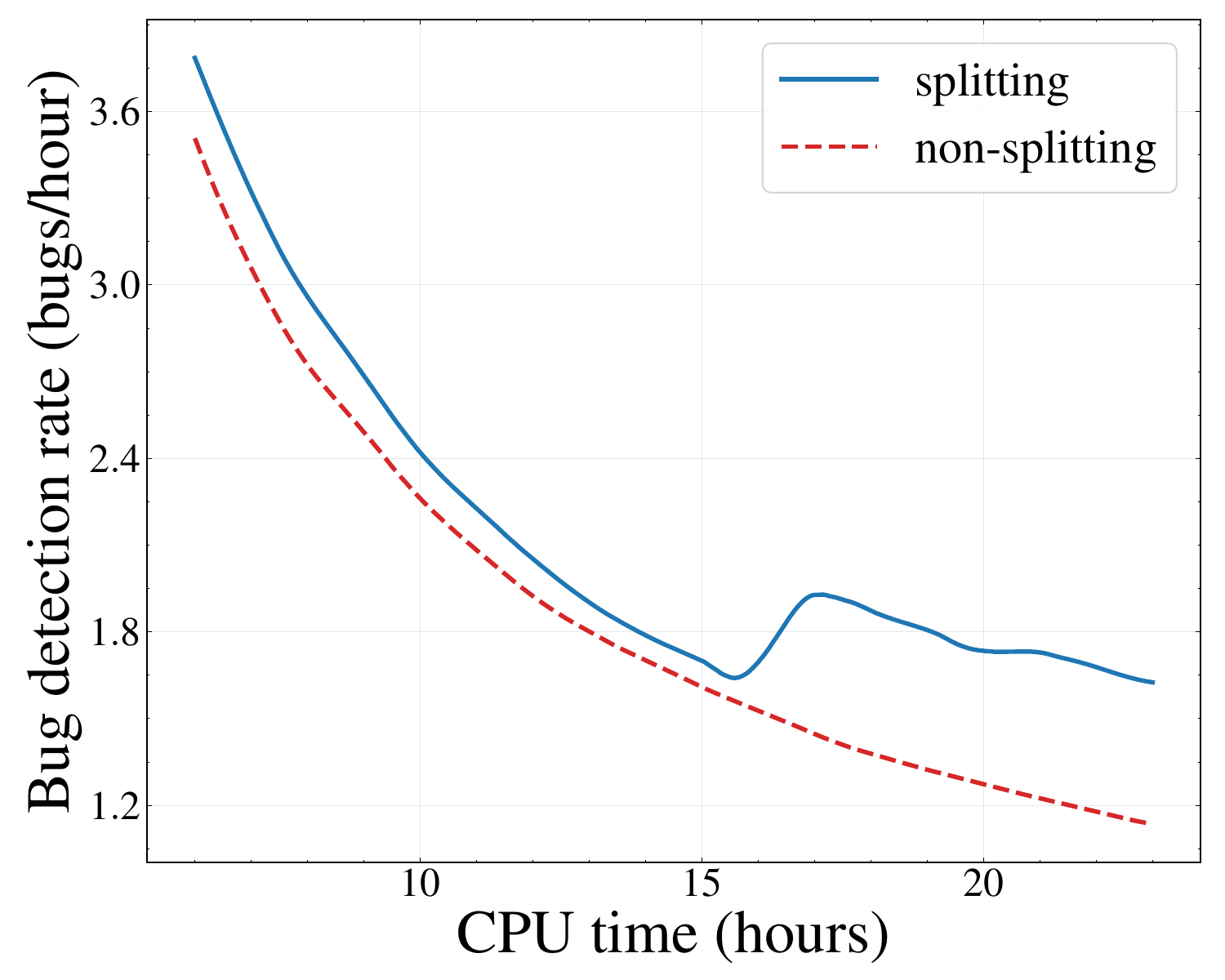}\hfill\includegraphics[width=0.19\textwidth]{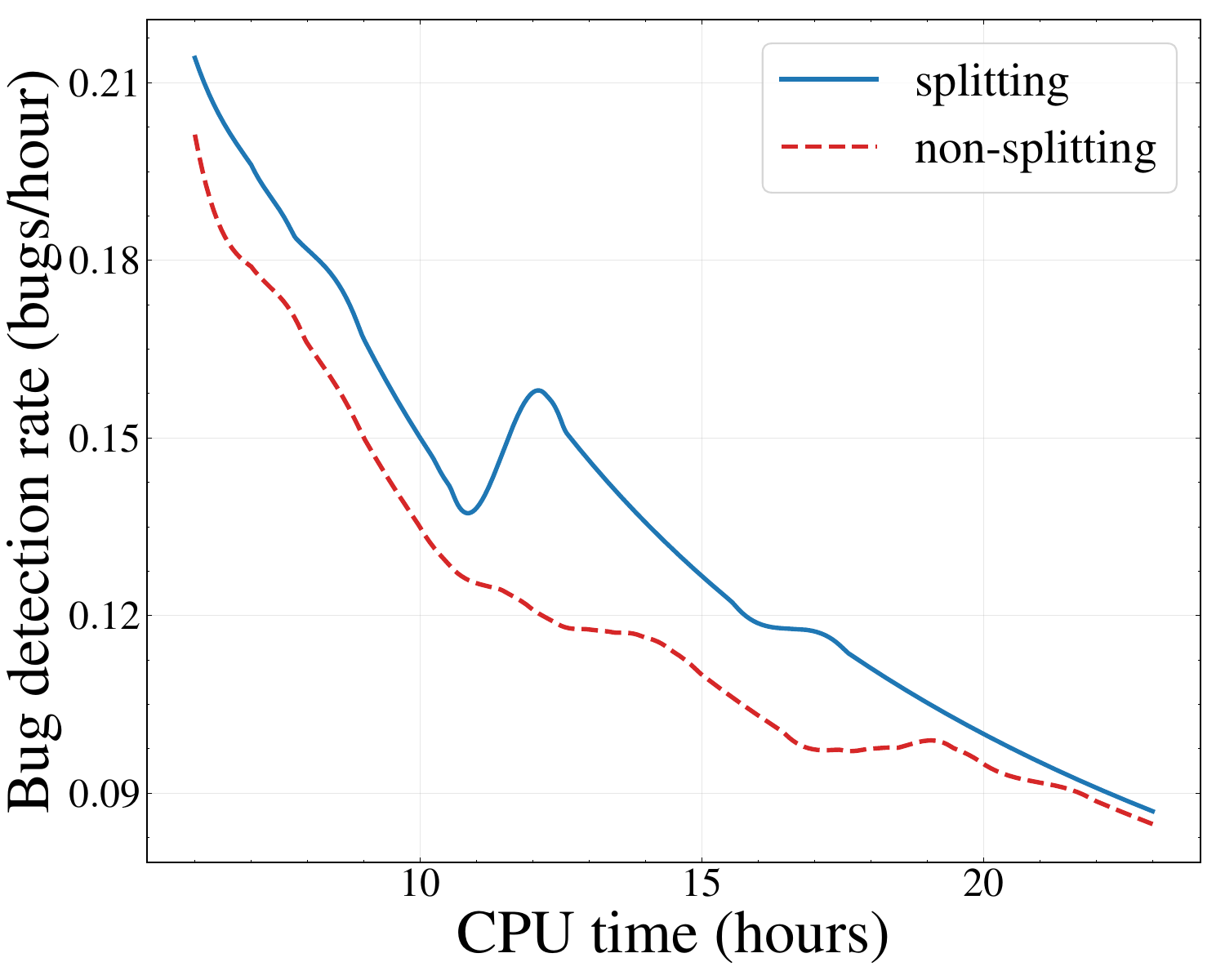}\hfill\includegraphics[width=0.19\textwidth]{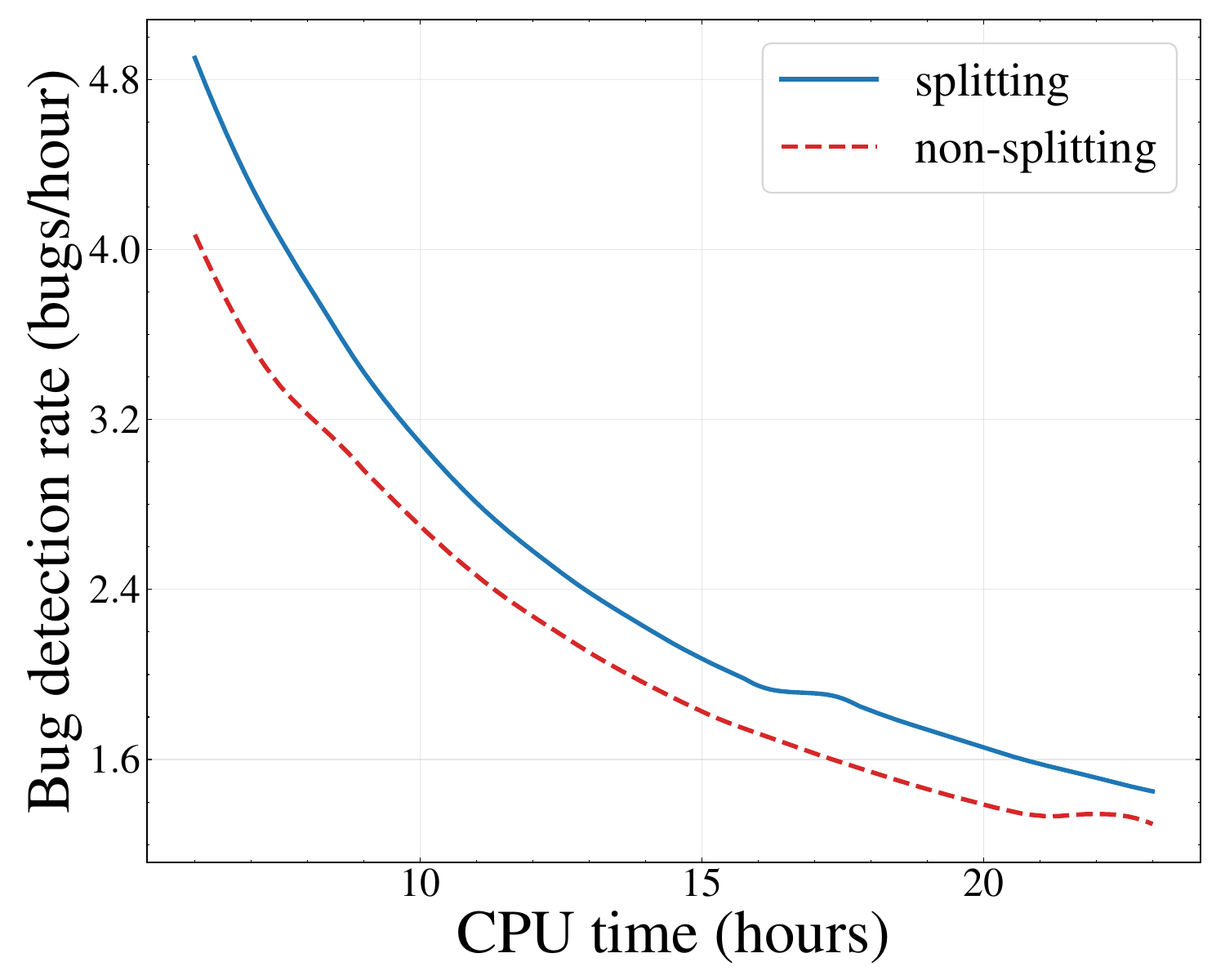}\hfill\includegraphics[width=0.19\textwidth]{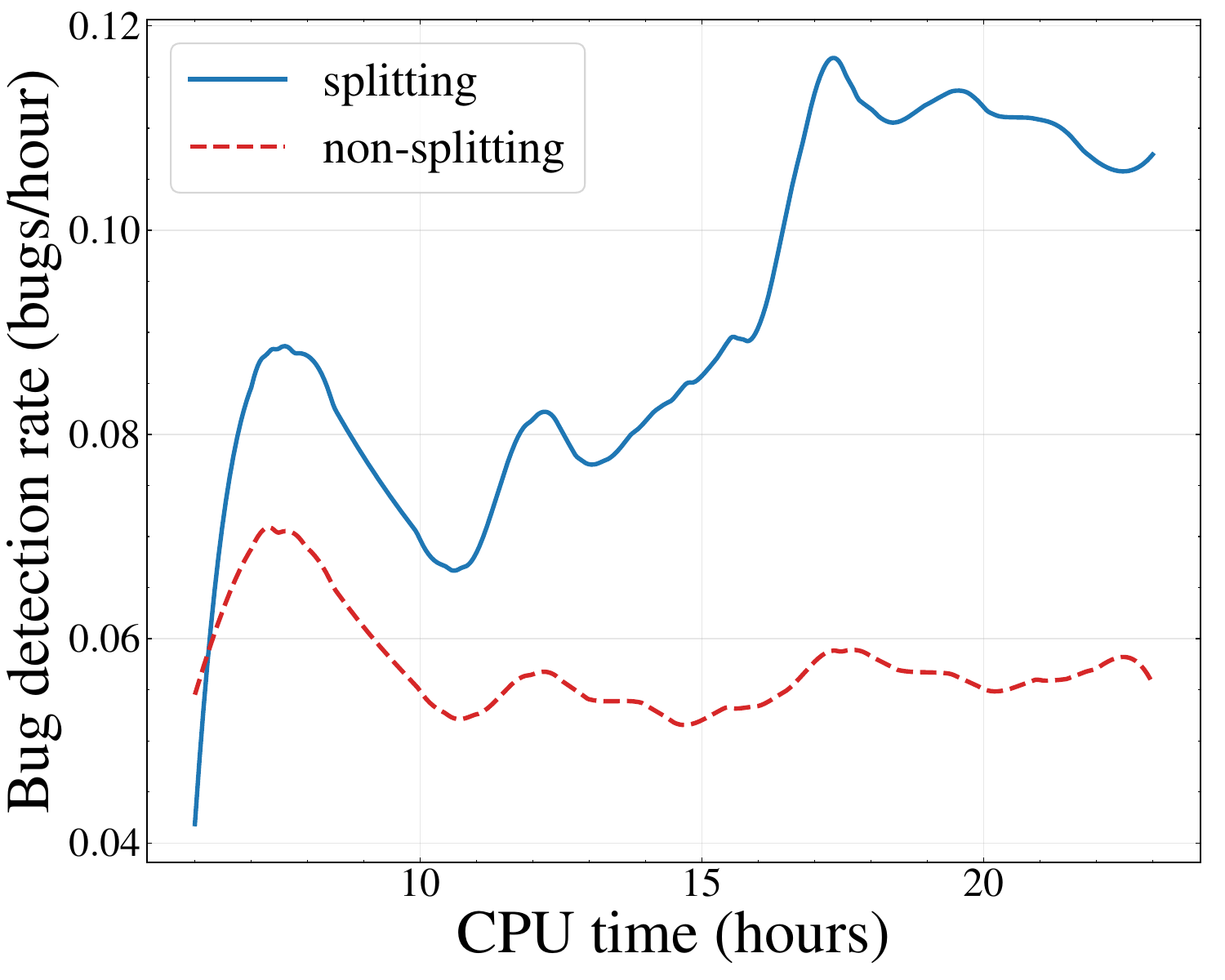}\hfill\includegraphics[width=0.19\textwidth]{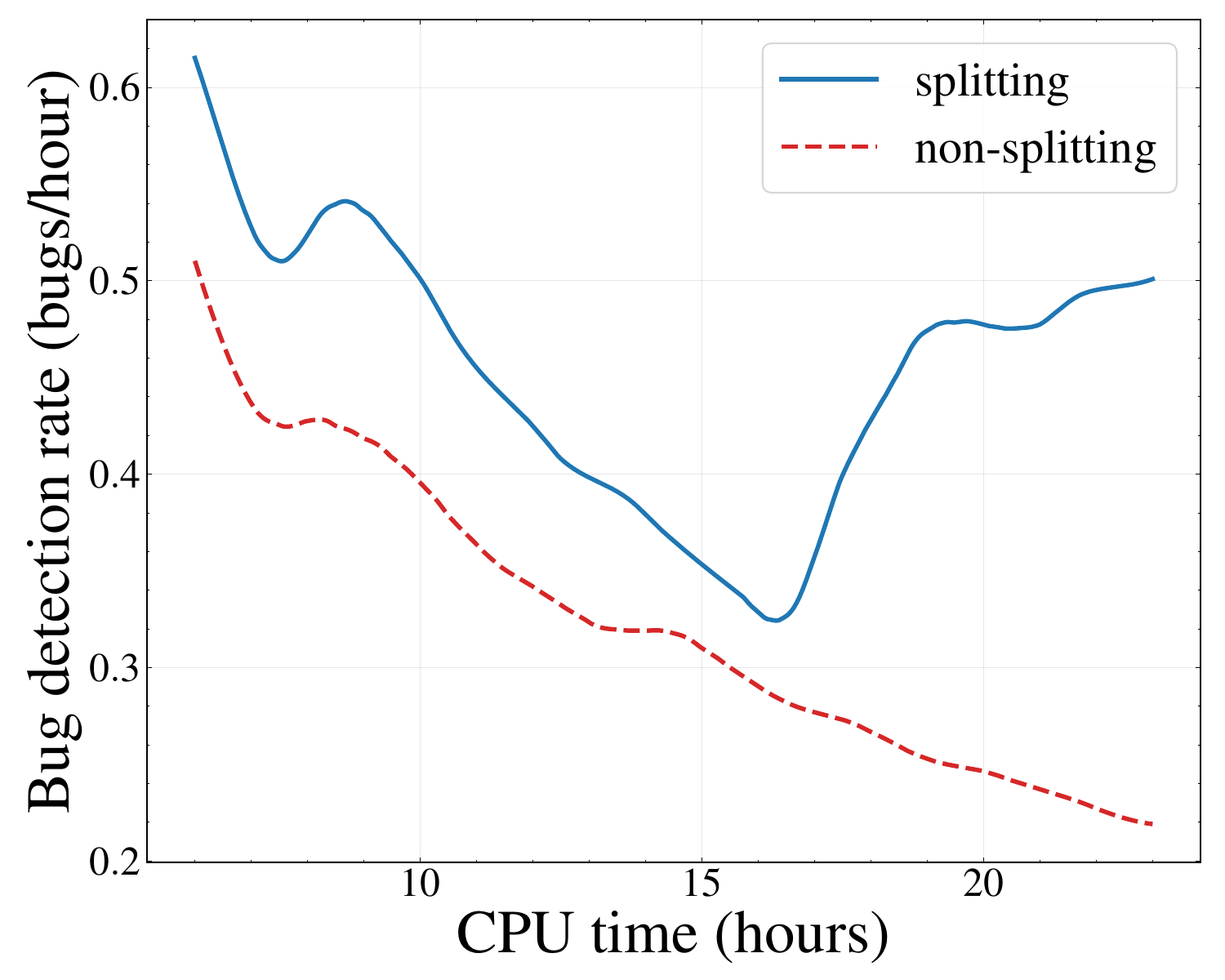}\\
    \makebox[0.19\textwidth]{\footnotesize (a) arrow}\hfill\makebox[0.19\textwidth]{\footnotesize (b) grok}\hfill\makebox[0.19\textwidth]{\footnotesize (c) libhevc}\hfill\makebox[0.19\textwidth]{\footnotesize (d) php}\hfill\makebox[0.19\textwidth]{\footnotesize (e) poppler}
    \caption{\emph{Unique-bug discovery rate} (cumulative unique signatures per elapsed hour) for aflplusplus on five representative benchmarks; the event-rate object $\widehat{P}_R$ is reported separately. Full grids in Appendix~\ref{app:appendix_experiments}.}
    \label{fig:bdr_aflplusplus}
\end{figure*}

\subsection{Ablation Study}\label{sec:ablation}

To isolate the splitting operation from the historical sparsity zone of Section~\ref{sec:exp:casestudy}, we ablate an \emph{on-the-fly} variant using only data from the running experiment, which is practical for new targets with no logged history.

\begin{figure}[htb]
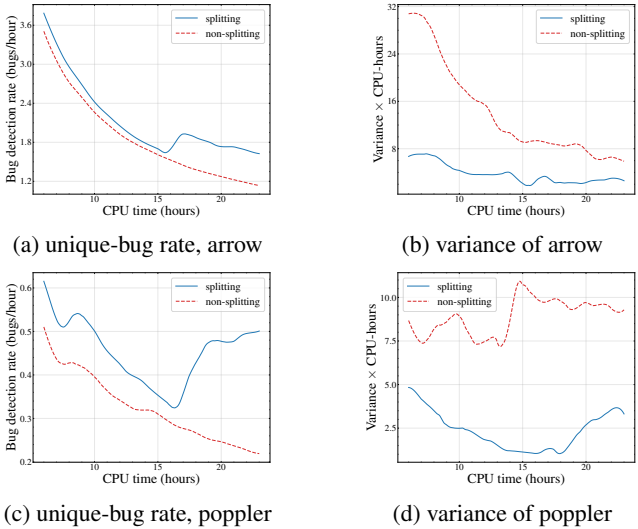

    \centering
    \includegraphics[width=0.205\textwidth]{figs/ablation/aflplusplus/arrow_parquet-arrow-fuzz__bug_detection_ratio.pdf}\hfill\includegraphics[width=0.205\textwidth]{figs/ablation/aflplusplus/arrow_parquet-arrow-fuzz__variance_times_cpu_hour.pdf}\\
    \makebox[0.205\textwidth]{\small (a) unique-bug rate, arrow}\hfill\makebox[0.205\textwidth]{\small (b) variance of arrow}\\[3pt]
    \includegraphics[width=0.205\textwidth]{figs/ablation/aflplusplus/poppler_pdf_fuzzer__bug_detection_ratio.pdf}\hfill\includegraphics[width=0.205\textwidth]{figs/ablation/aflplusplus/poppler_pdf_fuzzer__variance_times_cpu_hour.pdf}\\
    \makebox[0.205\textwidth]{\small (c) unique-bug rate, poppler}\hfill\makebox[0.205\textwidth]{\small (d) variance of poppler}
    \caption{Ablation study for aflplusplus on arrow and poppler; php and stb appear in Appendix~\ref{app:appendix_experiments}.}
    \label{fig:ablation_aflplusplus_1}
\end{figure}

\para{Implementation} The FuzzBench offline mode is \emph{history-informed}: it reads a sparsity signal from a prior campaign on the same target. The online mode needs no history and underlies \emph{every} Magma ground-truth result (all $40$ cells; Section~\ref{subsec:magma}), so this ablation isolates online-vs-offline. It computes the sparsity signal from a \emph{prefix-only} statistic instead of the tail-based signal over a completed trajectory, removing the dependence on logged data. Both variants share the fork-on-bug mechanism and differ only in the signal that \emph{selects} split times, so Theorems~\ref{bug_rate_splitting} and~\ref{thm:splitting_var} carry over unchanged; details in Appendix~\ref{app:ablation}.

\para{Ablation results}
We run the ablation on three fuzzers ({aflplusplus}, {aflsmart}, {mopt}) and four benchmarks ({arrow}, {php}, {poppler}, {stb}); its figures are in Appendix~\ref{app:appendix_experiments}.

The online variant preserves splitting's qualitative behavior: it reduces normalized variance relative to MC in all twelve settings (terminal and window-averaged) and improves the terminal rate in all twelve, so triggering on each new bug is itself effective. The offline-informed policy is nonetheless more stable, with lower window-averaged variance in 11 of 12 settings (8 of 12 terminal). Offline guidance is thus \emph{not required}, but it adds a measurable robustness gain by stabilizing when and how aggressively a run is split.

\subsection{Ground-Truth Validation with Magma}\label{subsec:magma}

To verify that signature deduplication tracks root-cause bugs and to show security impact, we use Magma~\cite{hazimeh2020magma}, whose targets carry individually identified bugs tied to real CVEs; its monitor reports which planted bug is \emph{reached} and \emph{triggered}, giving unambiguous ground truth.

\para{Setup}
We evaluate five of the seven mutation-based fuzzers (AFL, AFLFast, AFL++, MOpt, Honggfuzz; the five supported by Magma's harness) on eight Magma programs (libpng, libtiff, libsndfile, poppler, sqlite3, openssl, libxml2, php). Each fuzzer--program pair is run for $M{=}20$ independent trials of $T{=}12$ hours under both the baseline Monte Carlo (non-split) and splitting configurations. \emph{Every Magma ground-truth result below uses the \textbf{online} splitting mode}, the deployable variant that needs no prior data, at full scale ($40$ cells, $20$ trials each). Magma's monitor records the ground-truth bugs triggered per trial, from which we compute mean and union distinct-bug counts, the weighted \emph{bug-trigger event} rate $\widehat{P}_R$ of Section~\ref{sec:splitting} in bugs/h, and its effort-normalized variance, charging every live leaf in full to one per-CPU-hour budget ($\mathrm{CPU}_R(T)=\sum_{m,t} k_t^m$). \para{Bugs per CPU-hour} Splitting finds more real bugs per CPU-hour in $38$ of $40$ cells, and on raw distinct-bug count is never lower in any of the $40$.

\para{Results}
Table~\ref{tab:magma} summarizes per benchmark across the five fuzzers (full per-fuzzer table in Appendix~\ref{app:appendix_experiments}). Across all $40$ pairs splitting \emph{ties or improves} the unique-bug count with no regression: $25$ wins/$15$ ties on mean unique bugs (Wilcoxon $p=1.2\times10^{-5}$) and $18$/$22$ on total union ($p=9.6\times10^{-5}$); it raises the weighted rate $\widehat{P}_R$ on $39$ of $40$ pairs ($p=3.9\times10^{-8}$), the sole exception a \emph{rate} near-tie (AFLFast/libtiff, $1.992$ vs.\ $2.000$) that still does not regress on the raw distinct-bug count. Per-trial Mann--Whitney tests are individually significant in $22$ of the $25$ winning cells and favor non-split in none.

\para{Compute accounting}
Because splits fire only on a bug, this denominator counts exactly the compute consumed. Our principal equal-compute comparison is ground-truth bugs per CPU-hour, which uses no $1/k_t$ weighting: splitting wins $38$ of $40$ cells (median $+52\%$, $34/38$ significant after Benjamini--Hochberg).

\para{Robustness on Magma: the metric correctly flags saturation}
On Magma's saturated targets planted bugs surface within minutes of every trial, so outcomes coincide across trials and the effort-normalized variance of $\widehat{P}_R$ correctly reads exactly zero in $17$ of $40$ cells, the regime where variability is no longer the operative signal (Remark~\ref{rmk:bdr-vs-unique}). We therefore assess splitting here by Magma's ground truth, where it never finds fewer real bugs ($40$ of $40$ cells) and converts unreliable CVEs to reliable ones; the variance metric does its work on FuzzBench, where bugs surface throughout.

\para{Security impact}
Splitting also converts \emph{unreliable} detections into reliable ones. On sqlite3, AFL++ triggers CVE-2019-19926 (SQL014) in all $20/20$ splitting trials but $0/20$ non-split, a strongly significant conversion (Fisher's exact one-sided $p<10^{-4}$). The detection probabilities are sharply separated: one-sided $95\%$ bounds give $p_{\mathrm{MC}}<0.14$ and $p_{\mathrm{split}}>0.86$, so no compute-equalized pooling of the baseline's $20$ runs recovers it. On poppler, MOpt finds three bugs only under splitting (CVE-2019-14494/PDF001, CVE-2019-7310/PDF011, PDF008; total $3{\to}6$); Honggfuzz adds CVE-2016-6302 (SSL020) on openssl ($1{\to}2$) and CVE-2016-10269 (TIF006), CVE-2019-7663 (TIF009) on libtiff ($4{\to}6$), the clustered-bug regions the mechanism targets. At $M{=}20$ these CVE conversions are each individually Fisher-significant.

\begin{table}[t]
\centering
\caption{Estimated spectral parameters (non-split chains; Algorithm~\ref{alg:spectral}): second eigenvalue $\hat{\lambda}_2$, projected output gap $\hat{\gamma}^Y_{R}=1-|\hat{\lambda}_2|$, relaxation time $\hat{\delta}_R=1/\hat{\gamma}^Y_{R}$. All three right-hand columns are computed inside the phase-wise two-state model and are not guarantees for the general chain: $U_{A,R}$ evaluates the three terms of Theorem~\ref{thm:main_body} at that model's parameters; the \emph{two-state calibrated} column evaluates all three terms of Theorem~\ref{thm:main_body} inside that model, the finite-trial term as the exact quantile of the deviation under the two-state count law, the starting-state term at its worst over every start, and the finite-horizon term, each taken at its largest over a bootstrap confidence set for $(\hat\pi_1,\hat\lambda_2)$, and is how far the measured variability can sit from the true one; and $A_R^{\mathrm{ref}}$ is the deviation of $TS^2_{M,R}(T)$ from the \emph{fitted} model reference $\hat\sigma^2_{\mathrm{ref}}=\hat\pi_1(1-\hat\pi_1)(1+\hat\lambda_2)/(1-\hat\lambda_2)$, not from an independently known $\sigma^2$. $\dagger$ marks the lone cell that keeps the unconditional $U_{A,R}$.}\label{tab:spectral}
\resizebox{\columnwidth}{!}{%
\begin{tabular}{lcccccc}
\toprule
Fuzzer / Benchmark & $\hat{\lambda}_2$ & $\hat{\gamma}^Y_{R}$ & $\hat{\delta}_R$ & $U_{A,R}$ & Two-state calibrated & $A_R^{\mathrm{ref}}$ \\
\midrule
\multicolumn{7}{l}{\emph{FuzzBench}}\\
\quad AFLSmart / ffmpeg & 0.33 & 0.67 & 1.48 & 1.43 & 0.34 & 0.00 \\
\quad Honggfuzz / stb & 0.41 & 0.59 & 1.68 & 2.80 & $2.80^\dagger$ & 0.30 \\
\quad AFL++ / arrow & 0.22 & 0.78 & 1.28 & 0.57 & 0.31 & 0.12 \\
\midrule
\multicolumn{7}{l}{\emph{Magma}}\\
\quad Honggfuzz / openssl & 0.83 & 0.17 & 5.88 & 17.29 & 0.37 & 0.16 \\
\quad MOpt / poppler & 0.34 & 0.66 & 1.53 & 0.68 & 0.37 & 0.14 \\
\quad AFL++ / libxml2 & 0.16 & 0.84 & 1.19 & 0.31 & 0.02 & 0.01 \\
\bottomrule
\end{tabular}}
\end{table}

\para{Detection-rate regime} The condition $p<\sqrt2-1$ constrains the bug-region dwell fraction $p=|S|/T$, which on the two-state surrogate is the settled bug-step fraction $\hat\pi_1$ measured on the logging grid, which is coarser than a per-execution indicator. It is a \emph{sufficient} condition marking the rare-bug regime in which the variance metric is informative; empirically splitting reduces per-trial variance in $66$ of $70$ FuzzBench pairs at matched CPU. On the saturated Magma cells bug steps are dense, the variance degenerates to zero, and we assess the benefit by ground-truth bugs instead.

\providecommand{\nCSruns}{twenty}

\subsection{Comparison with corpus-synchronizing parallel fuzzing}
\label{subsec:corpussync}

The strongest deployed alternative shares state \emph{unconditionally}: a
corpus-synchronizing pool, at the scale hardest to beat, eight synchronized AFL++
workers, \nCSruns\ campaigns of $T{=}12$\,h per program, same setup as
Section~\ref{subsec:magma}, both arms billed per CPU-hour.

\begin{table}[t]\centering\small
\caption{Splitting versus a corpus-synchronizing eight-worker pool, per CPU-hour
(AFL++, $T{=}12$\,h).}
\label{tab:corpussync8}
\begin{tabular}{lcc@{\hskip 8pt}cc}
\toprule
& \multicolumn{2}{c}{$\widehat{P}_R$/CPU-h} & \multicolumn{2}{c}{bugs/CPU-h}\\
\cmidrule(r){2-3}\cmidrule(l){4-5}
program & split & pool-8 & split & pool-8\\
\midrule
libsndfile & \textbf{.0874} & .0396 & \textbf{.0764} & .0730\\
libxml2    & \textbf{.0721} & .0375 & .0426 & .0438\\
libtiff    & \textbf{.0676} & .0479 & \textbf{.0752} & .0667\\
sqlite3    & \textbf{.0653} & .0417 & .0379 & .0417\\
poppler    & \textbf{.0568} & .0334 & \textbf{.0361} & .0355\\
libpng     & \textbf{.0548} & .0229 & \textbf{.0318} & .0313\\
php        & \textbf{.0530} & .0188 & \textbf{.0327} & .0313\\
openssl    & \textbf{.0409} & .0216 & .0107 & .0216\\
\bottomrule
\end{tabular}
\end{table}

Splitting attains the higher bug-event rate on \emph{every} program by
$1.41$--$2.82\times$ (median $1.90\times$; sign test $p=7.8\times10^{-3}$),
and leads on distinct bugs per CPU-hour on five of eight. Ninety percent of the pool's bugs
arrive within $2.8$\,h of $12$, and doubling the pool from four to eight workers
raises its mean rate by only $4\%$.
Synchronization also requires fuzzer support; splitting is more universal,
wrapping \emph{any} mutation-based fuzzer.

Ties concentrate on saturated pairs where both methods recover every quickly-reachable bug (php: all five fuzzers tie on union counts, yet splitting still raises $\widehat{P}_R$). Figure~\ref{fig:bug_event_bug} makes the bug-rate$\,\to\,$unique-bug link explicit: splitting fires more bug-triggering events, the quantity the theory governs (Theorem~\ref{bug_rate_splitting}), in $39$ of $40$ Magma cells, spanning several orders of magnitude (panel a), and those additional events yield at least as many distinct bugs in all $40$ cells ($25$ strictly more; panel b). The within-trial link holds on \emph{both} suites: a trial's bug-event count moves with its distinct-bug count across all Magma trials (Spearman $\rho=0.55$, $p=6.7\times10^{-33}$) and across $3{,}859$ FuzzBench crashing trials (panel d; Spearman $\rho=0.66$, $p<10^{-300}$), as our framework requires.

\begin{figure*}[t]
\centering
\includegraphics[width=0.82\textwidth]{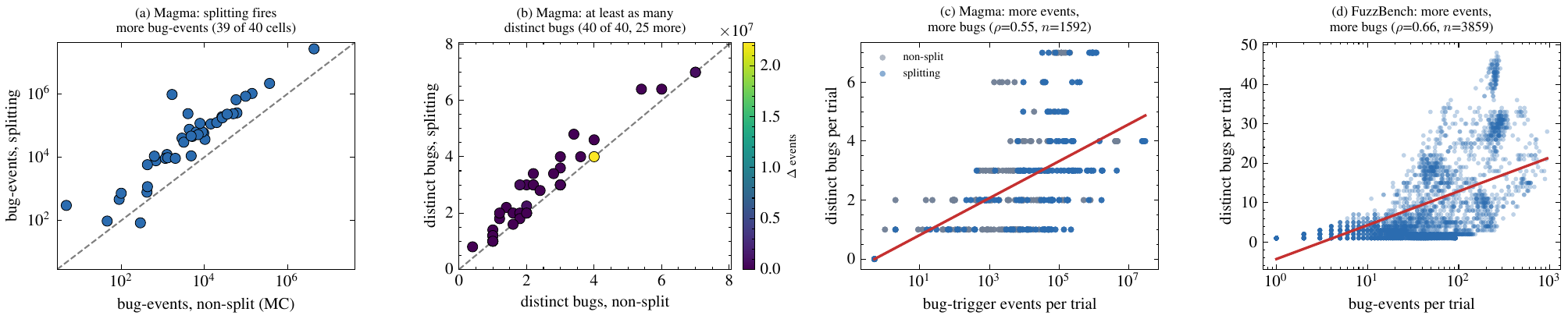}
\caption{Bug-trigger events and distinct bugs, splitting vs.\ Monte Carlo, on \emph{both} suites. \textbf{(a)}~Magma, one point per fuzzer--program cell ($M{=}20$ trials each): splitting fires more bug-trigger events, the quantity the theory models, in $39$ of $40$ cells. \textbf{(b)}~Magma: those additional events yield at least as many distinct bugs in all $40$ cells ($25$ strictly more); color is the per-cell event increase. \textbf{(c)}~Magma, one point per trial ($M{=}20$; $1{,}592$ of the $1{,}600$ trials over both arms, the $8$ omitted having failed a pre-declared infrastructure check): a trial's bug-trigger event count predicts its distinct-bug count over seven decades of events (Spearman $\rho=0.55$); the line is the median with the interquartile band. \textbf{(d)}~FuzzBench, one point per trial ($3{,}859$ crashing trials): the same link on the second suite (Spearman $\rho=0.66$). The bug-event-to-distinct-bug link our framework relies on is present on both suites.}
\label{fig:bug_event_bug}
\end{figure*}

\begin{table}[t]
\centering
\caption{Magma ground-truth validation (summary): five fuzzers on each of eight Magma programs, $M{=}20$ trials of $12$\,h, splitting vs.\ non-split (MC). \emph{Win/Tie/Lose} (denoted by W/T/L) counts, over the five fuzzers, how often splitting wins/ties/loses on mean unique bugs. $\Sigma$\,mean uniq and mean $\widehat{P}_R$ are summed / averaged over the five fuzzers (\emph{splitting\,/\,non-split} denoted by s/n); \textbf{bold} marks the better value. The full per-fuzzer breakdown is in Appendix~\ref{app:appendix_experiments}.}\label{tab:magma}
\small
\setlength{\tabcolsep}{5pt}
\begin{tabular}{lccc}
\toprule
Benchmark & W/T/L & $\Sigma$\,mean uniq (s/n) & mean $\widehat{P}_R$ (s/n) \\
\midrule
libpng     & 4/1/0 & \textbf{11.0}\,/\,9.0 & \textbf{4.80}\,/\,1.76 \\
libtiff    & 4/1/0 & \textbf{19.1}\,/\,15.8 & \textbf{3.93}\,/\,2.96 \\
libsndfile & 1/4/0 & \textbf{24.4}\,/\,24.0 & \textbf{6.55}\,/\,4.04 \\
poppler    & 3/2/0 & \textbf{14.2}\,/\,11.2 & \textbf{5.03}\,/\,2.00 \\
sqlite3    & 4/1/0 & \textbf{11.6}\,/\,9.0 & \textbf{4.16}\,/\,1.76 \\
openssl    & 4/1/0 & \textbf{8.2}\,/\,6.6  & \textbf{4.36}\,/\,1.32 \\
libxml2    & 4/1/0 & \textbf{12.4}\,/\,10.4 & \textbf{4.88}\,/\,2.08 \\
php        & 1/4/0 & \textbf{13.0}\,/\,12.0 & \textbf{5.62}\,/\,2.32 \\
\midrule
\textbf{All 8} & \textbf{25/15/0} & \textbf{113.9}\,/\,98.0 & \textbf{4.92}\,/\,2.28 \\
\bottomrule
\end{tabular}
\end{table}

\subsection{Model Diagnostics for the Assumptions}\label{subsec:empirical-assumptions}

To empirically examine the time-homogeneity and geometric ergodicity assumptions discussed in Section~\ref{subsec:assumptions}, we analyze the autocorrelation structure and spectral properties of real fuzzer trajectories from our experiments.

\para{Autocorrelation analysis}
We compute the empirical autocorrelation of $Y_R(w,t)$ across lags for all $110$ evaluated fuzzer--benchmark cells. In every cell it decays geometrically toward the estimation-noise floor, with a strictly positive gap throughout ($\hat{\gamma}\in[0.02,1.00]$; per-cell values in the artifact), the adaptive schedulers AFL++ and MOpt included, consistent with geometric decay over the observed lag range and supporting the two-state phase model as an empirical approximation. Figure~\ref{fig:autocorrelation} shows six representative cells. The positivity the closed-form rule of Theorem~\ref{thm:splitting_var} needs is directly checkable on the surrogate: $\hat\lambda_2\ge0$ holds in $102$ of the $110$ cells and in all $40$ on Magma (AFL++ $18/18$, MOpt $17/18$), and six of the eight exceptions have $|\hat\lambda_2|<0.05$.

\begin{figure}[h]
    \centering
    \includegraphics[width=0.43\columnwidth]{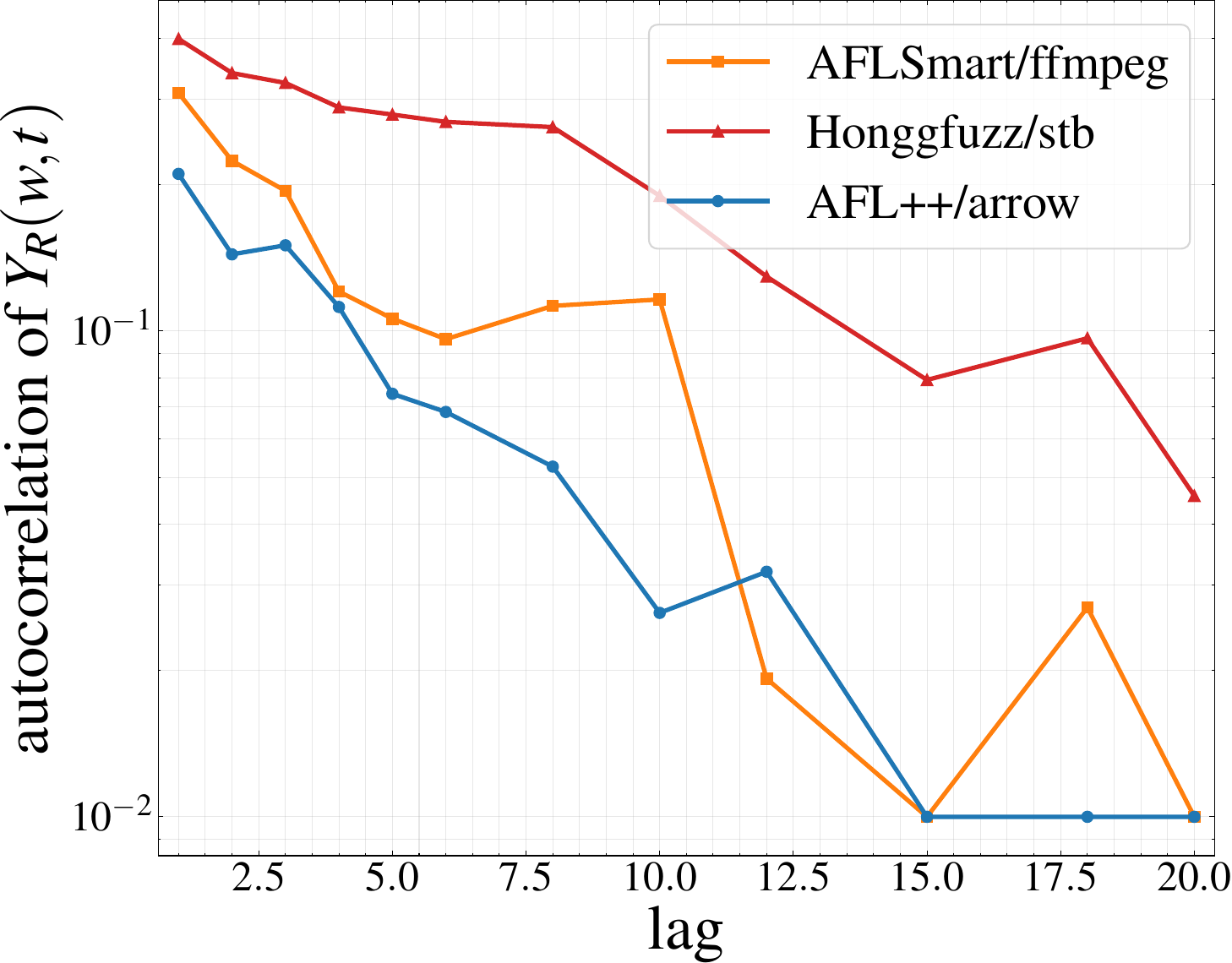}\hfill\includegraphics[width=0.43\columnwidth]{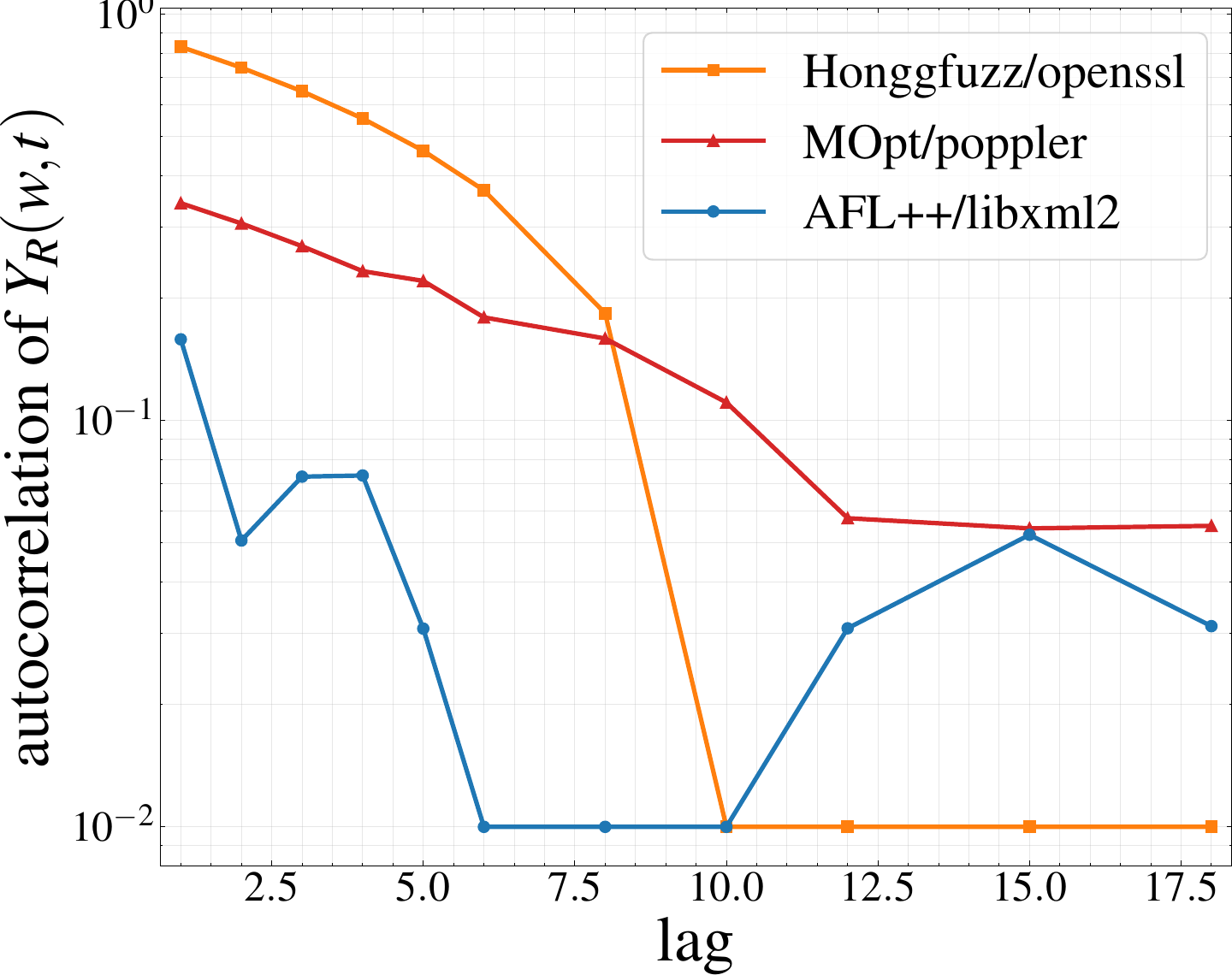}\\
    \makebox[0.43\columnwidth]{\small (a) FuzzBench}\hfill\makebox[0.43\columnwidth]{\small (b) Magma}
    \caption{Empirical autocorrelation of $Y_R(w,t)$ versus lag (log vertical axis) for six representative fuzzer--benchmark pairs (three per suite). The approximately geometric decay toward the estimation-noise floor is consistent with geometric decay over the observed lag range; this is a model diagnostic, not a proof of geometric ergodicity.}
    \label{fig:autocorrelation}
\end{figure}

\para{Spectral gap estimation}
Using Algorithm~\ref{alg:spectral}, we estimate the spectral gap $\hat{\gamma}^Y_{R}=1-|\hat{\lambda}_2|$ and relaxation time $\hat{\delta}_R=1/\hat{\gamma}^Y_{R}$ from the empirical $2\times2$ output transition matrix (Table~\ref{tab:spectral}). Relaxation times are short in Table~\ref{tab:spectral} ($1.2$--$5.9$), and the gap stays positive in all $110$ cells (Section~\ref{subsec:empirical-assumptions}). Evaluated inside the two-state model, $U_{A,R}$ serves as a relative sample-complexity ranking; replacing its worst-case start term by the exact starting-state term of Theorem~\ref{thm:main_body} gives the tight per-cell two-state calibrated value, up to $46\times$ smaller (Honggfuzz on openssl, $17.3\to0.37$, against the $0.16$ actually observed in that cell) and tighter by $\ge1.8\times$ on every cell; the one heavy-drift cell ($\dagger$) keeps the unconditional bound. These measurement-grid times (15-min FuzzBench, 60-s Magma snapshots) quantify snapshot-level, not per-execution, mixing; and the check $A_R^{\mathrm{ref}}\le U_{A,R}$ holds for every pair shown, consistent with the fitted two-state model (compare those two columns).

\para{Time-homogeneity window check}
The firing rate $\hat{\pi}_1$ drifts over a campaign as the corpus and scheduling state evolve, so we estimate the spectral quantities per campaign phase via windowed estimation (Appendix~\ref{app:assumption}): re-estimating on early vs.\ late halves shifts $\hat{\gamma}$ by $|\Delta\hat{\gamma}|=0.08$--$0.86$, tracking $\hat{\pi}_1$'s decline, while adaptive-scheduler pairs stay homogeneous ($|\Delta\hat{\gamma}|=0.003,0.063$). Windowed estimates are thus phase-specific by construction, and the geometric autocorrelation decay (Figure~\ref{fig:autocorrelation}) holds within each phase.

%% file: section/conclusion-future-work.tex
\section{Implications for Fuzzing Practice}\label{sec:practical-guidance}

\textbf{Evaluating fuzzers.}
On our evaluated suites the separation analysis, computed inside the model, stabilizes around $M=20$ trials and $T=12$\,h, where about two-thirds of pairwise robustness rankings separate and more horizon adds little; $M=10$ suffices for exploration. This is an empirical recommendation on these suites, not a universal claim; the separation table (Table~\ref{tab:samplecomplexity}) shows horizon helping up to $12$\,h and trials helping at every horizon. Report both effectiveness ($\widetilde{P}_R(T)$) and robustness ($T\cdot S^2_{M,R}(T)$): Honggfuzz ranks first in effectiveness but last in robustness.

\begin{table}[t]
\centering\small
\caption{Share of FuzzBench pairwise robustness rankings whose certificates separate, over a fixed set of $51$ fuzzer pairs. Trials are the first $M$ of the $20$; horizons are prefixes of the $23$\,h campaign, with the two-state parameters and the two-state calibrated interval re-estimated on each prefix.}
\label{tab:samplecomplexity}
\begin{tabular}{r|cccccc}
\toprule
$M\;\backslash\;T$ & $6$\,h & $9$\,h & $12$\,h & $15$\,h & $18$\,h & $23$\,h \\
\midrule
$5$  & 27\% & 43\% & 45\% & 45\% & 47\% & 55\% \\
$10$ & 25\% & 49\% & 49\% & 55\% & 57\% & 55\% \\
$15$ & 31\% & 49\% & 57\% & 65\% & 63\% & 63\% \\
$20$ & \textbf{33\%} & \textbf{53\%} & \textbf{65\%} & 71\% & 71\% & 67\% \\
\bottomrule
\end{tabular}
\end{table}

\textbf{Using splitting.}
Splitting helps most when bugs cluster into bursts separated by long silences (Appendix~\ref{app:implementation}). Prefer the offline variant when history exists, the online variant otherwise.

\section{Conclusion}\label{sec:conclusion}
We introduce a bug-guided splitting wrapper and a root-campaign robustness framework. At strictly matched compute, splitting improves unique-bug discovery in $53$ of $70$ FuzzBench pairs and lowers CPU-weighted variance across campaigns in $66$ of $70$; its deployable online mode improves ground-truth bugs per CPU-hour in $38$ of $40$ Magma cells with no loss in raw distinct-bug counts, and converts CVE-2019-19926 from $0/20$ to $20/20$. The theory separates finite-trial error from finite-horizon bias and adds a pathwise identity.

%% file: section/supplementary.tex
\section{Supplementary Tables, Figures and Derivations}\label{app:supplementary}

\subsection{From Section: Problem Formulation}

Table~\ref{tab:notation} summarizes the key notation used throughout.

\begin{table}[t]
\centering
\caption{Summary of Key Notation}\label{tab:notation}
\resizebox{\columnwidth}{!}{%
\begin{tabular}{ll}
\toprule
Symbol & Meaning \\
\midrule
$D$ & Input domain (finite set of bit strings) \\
$f$, $g$ & Program trace and behavioral output functions \\
$B_f$, $B$ & Set of abnormal behaviors; set of buggy inputs \\
$R \in \mathcal{R}$ & A fuzzer from the fuzzer set \\
$X_R(w,t)$ & Input Markov chain at time $t$ from seed $w$ \\
$Y_R(w,t)$ & Output process: $\mathbf{1}\{g(X_R(w,t)) \in B_f\}$ \\
$P_R^X$, $P_R$ & Transition matrices of input and output chains \\
$\pi_R^X$, $\pi_R$ & Stationary distributions of input and output chains \\
$T$ & Total mutation rounds (horizon) \\
$M$ & Number of independent trials (replications) \\
$P_R$ & Bug detection rate (BDR): $\lim_{T\to\infty} N_T/T$ \\
$\widetilde{P}_R(T)$ & Finite-horizon BDR estimator \\
$\sigma^2$ & Effort-normalized asymptotic variance (robustness) \\
$S_R(T)$, $S^2_{M,R}(T)$ & Finite-horizon variance; sample variance estimator \\
$\gamma^X_R$, $\delta_R$ & Spectral gap; relaxation time ($\delta_R = 1/\gamma^X_R$) \\
$k_t^m$ & Active paths at time $t$ in trial $m$ (splitting) \\
\bottomrule
\end{tabular}}
\end{table}

\begin{remark}[Extension to coverage-oriented benchmarks]\label{rmk:coverage-extension}
On targets where bugs are unknown or extremely rare, one can replace $B_f$ with a set of coverage targets $B_c$ (e.g., uncovered branches) and define $Y_R(w,t) = \mathbf{1}\{g(X_R(w,t)) \in B_c\}$; since $Y_R$ remains a deterministic function of $X_R$, all analyses carry over, and the robustness metric quantifies stability of coverage-target discovery, keeping the framework informative even when $\text{BDR} = 0$.
\end{remark}

\subsection{From Section: Evaluation Framework}
Algorithm~\ref{alg:mc} gives the Monte Carlo evaluation procedure of Section~\ref{sec:evaluation-methods} in full.

\begin{algorithm}[h]
\SetAlgoLined
\caption{Monte Carlo Simulation Algorithm}\label{alg:mc}
 \textbf{Input:} replications $M$; run length $T$; input seed $w$; start time $t=1$; fuzzer $R$; executed-input and bug-triggering-input sets $I_R^T$, $H_R^T$.\;\par
  \For{$1 \leq m \leq M$}{
  $N_T^m = 0$, $H_R^{m,T}, I_R^{m,T}= \emptyset$ \par
 \For(\tcp*[f]{Simulation}){$1 \leq t \leq T$}{
    Mutate $s_t = X(w,t)$ with fuzzer $R$ and generate $X(w,t+1)$\; \par 
    Add $X(w , t+1)$ to $I_R^T$ \par
    The software executes $X(w,t+1)$ and generates the feedback;\tcp*[f]{Env}\; \par
    Observe $y_{R}^m(t) = Y(w,t+1)$\; \par
    \If{$y_{R}^m(t) = 1$}{ Add $X(w , t+1)$ to $H_R^T$ and update $N_T^m = N_{T}^m+1$\; \par
}
$s_t = X(w, t+1)$
}
   }
\textbf{Output:} $\{H_R^{m,T}\}_{m=1}^{M}$, $\{I_R^{m,T}\}_{m=1}^{M}$, $\{N_T^m\}_{m=1}^{M}$, $\{y_{R}^m(t)\}_{m, t}$
\end{algorithm}

\subsection{From Section: Theoretical Guarantees}

\begin{table}[t]
\centering\footnotesize
\setlength{\tabcolsep}{4pt}
\caption{Dependency map of the guarantees, for one-pass reading. Positivity/reversibility is needed \emph{only} for the variance lower bound; the never-worse identity assumes nothing; the concentration bound needs only a geometric autocovariance envelope and carries the confidence level explicitly.}
\label{tab:theorem-map}
\begin{tabular}{@{}p{0.20\columnwidth}p{0.29\columnwidth}p{0.25\columnwidth}p{0.15\columnwidth}@{}}
\toprule
Result & Delivers & Powered by & Tightness \\
\midrule
Thm~\ref{bug_rate_splitting} & bug-event count never decreases; exact identity for the rate difference & the original lineage is continued & exact \\
Prop~\ref{prop:sizebias} & expected rate rises when allocation and future yield move together & positive association & conditional \\
Thm~\ref{thm:splitting_var} & variance per unit compute falls (idealized) & positivity \& reversibility; derived $p{<}\sqrt2{-}1$ & mechanism \\
Thm~\ref{thm:main_body} & finite-$M$ and finite-$T$ error of $T\,S^2_{M,R}$ & geometric autocovariance envelope & explicit $\alpha$ \\
Lem~\ref{lem:gap-transfer} & holds without assuming the bug indicator is Markov & input chain forgets its start geometrically & --- \\
Lem~\ref{lem:gap-transfer} & supplies the envelope Thm~\ref{thm:main_body} needs, with no assumption on the output & reversible input chain & --- \\
\bottomrule
\end{tabular}
\end{table}

\noindent\emph{Worked example (one fuzzer--target pair, end to end).} Take Honggfuzz on openssl. From its trajectory we estimate the two-state output gap $\hat\gamma^Y_{R}=0.17$ and bug-step fraction $\hat\pi_1=0.99$ (Table~\ref{tab:spectral}). Evaluating the three terms of Theorem~\ref{thm:main_body} inside that two-state model gives the gap-only accuracy value $U_{A,R}=17.29$, useful for \emph{ranking}. Computing all three terms inside that model instead, with the exact starting-state term at its worst over every start, sharpens it to a two-state calibrated interval of $0.37$, within $\approx 3\times$ of the observed deviation $A_R^{\mathrm{ref}}=0.16$ against the fitted model reference. The $5\%$ is split as $2.5\%$ for the sampling quantile and $2.5\%$ for the parameter set, so the interval covers at least $95\%$ without assuming the two are independent. Both numbers are computed inside the stated two-state model, not a guarantee for the general chain.

\textbf{What is estimated under drift.} Under piecewise homogeneity (windows of length $W$, per-window gap $\gamma_w$, robustness $\sigma_w^2$), the windowed statistic estimates the effort-weighted phase average $\bar\sigma^2=\sum_w a_w\sigma_w^2$ ($a_w=W/T$); a union bound gives $A_R^{(w)}\le U(\gamma_w)$ simultaneously, and the phase-average-to-window discrepancy carries an explicit drift-proportional bias (Appendix~\ref{app:gap-robust}). This is a defined quantity with a quantified bias, small for homogeneity-consistent pairs, large only for the high-drift ones we flag.

\subsection{From Section: Splitting}

\begin{table}[t]
\centering
\small
\setlength{\tabcolsep}{2pt}
\caption{Splitting vs.\ power scheduling}
\label{tab:splitting-vs-scheduling}
\begin{tabularx}{\columnwidth}{lXX}
\toprule
& \textbf{AFLFast / MOpt} & \textbf{Splitting (Ours)} \\
\midrule
Trigger signal & Coverage-based heuristic & Bug-triggering output \\
Mechanism & Re-prioritize fixed seed pool & Fork new chain from current state \\
Modifies fuzzer & Yes, internal scheduler & No, black-box plug-in \\
Theoretical guarantee & Path-hitting under a coverage Markov model~\cite{AFLFast}; none on bug detection or variance & On BDR and variance: Thm.~\ref{thm:splitting_var},~\ref{bug_rate_splitting} \\
Effect on robustness & Not analyzed & Reduces variance in the clustered-bug regime \\
Composable & Within-fuzzer scheduler & Yes, wraps any fuzzer \\
\bottomrule
\end{tabularx}
\end{table}

In other words, splitting is most beneficial when it increases \(k_t\) primarily in clustered bug regions, thereby reducing the dominant covariance contributions by more than the corresponding marginal increase in CPU cost. We scope the analysis device explicitly: our implementation forks \emph{independent} descendants, and the clustered-regime mechanism above, averaging the dominant within-cluster covariance over the $k_t$ paths, applies to independent forks; the coupled-noise construction, under its own separate assumption set, is a stylized device quantifying an \emph{additional} reduction available to anti-correlated descendants, which our implementation does not exploit.

\subsection{From Section: Numerical Experiments}

\begin{table}[t]
\centering
\caption{Bug signature statistics per benchmark. Counts are per-trial unique signatures over the 23-hour campaigns, aggregated across all 7 fuzzers and 20 trials. FuzzBench commit: \texttt{90e59b6}. Bugs/h = median count divided by the 23-hour horizon.}
\label{tab:bug-counts}
\resizebox{\columnwidth}{!}{%
\begin{tabular}{lcccccccccc}
\toprule
& arrow & ffmpeg & grok & libhevc & libhtp & matio & openh264 & php & poppler & stb \\
\midrule
Median/trial & 39 & 7 & 1 & 29 & 3 & 10 & 4 & 1 & 12 & 18 \\
Max/trial    & 47 & 36 & 3 & 34 & 6 & 16 & 7 & 7 & 22 & 22 \\
Bugs/h & 1.70 & 0.30 & 0.04 & 1.26 & 0.13 & 0.43 & 0.17 & 0.04 & 0.52 & 0.78 \\
\bottomrule
\end{tabular}}
\end{table}

\para{What we claim} Precisely: \emph{(i) Baseline.} Splitting is compared against the natural baseline, more independent runs of the \emph{same} fuzzer at equal total compute (the non-split arm). \emph{(ii) Bug-finding on both suites.} Splitting finds more unique bugs on \emph{both} FuzzBench (its standard \texttt{crash\_key} metric; $53/70$ pairs) and Magma (planted ground truth; $38/40$ cells per CPU-hour), with the Magma gains tied to specific CVEs. \emph{(iii) Magma.} The ground-truth conclusion (more bugs per CPU-hour in $38$ of $40$ cells, $34$ of $38$ individually significant; never lower in raw distinct-bug count in any of the $40$) is established at $M{=}20$. \emph{(iv) Theory.} The $\approx10\times$ variance reduction is \emph{measured} ($66$ of $70$ pairs); Theorem~\ref{thm:splitting_var} \emph{explains} it in the clustered-bug regime: the \emph{derived} condition $p<\sqrt2-1$, which marks the non-saturated rare-bug regime, exactly where the variance metric is informative, and the operative accuracy guarantee is the per-fuzzer certificate the two-state calibrated interval (Theorem~\ref{thm:main_body}).

\subsection{From Section: Conclusion}
Table~\ref{tab:fuzzer-ranking} ranks the fuzzers by terminal bug-detection rate on the non-splitting arm, averaged over the ten benchmarks.

\begin{table}[t]
\centering
\caption{Fuzzer ranking from the MC (non-splitting) arm: mean over the 10 benchmarks of the terminal (23\,h) bug-detection rate (bugs/h; higher is better) and of the effort-normalized variance (variance $\times$ CPU-h; lower is better). Parenthesized rank, 1 = best.}
\label{tab:fuzzer-ranking}
\resizebox{\columnwidth}{!}{%
\begin{tabular}{lccccccc}
\toprule
& Hongg. & AFL & MOpt & AFLFast & AFLSm. & AFL++ & LibF. \\
\midrule
BDR & 0.595 (1) & 0.552 (2) & 0.542 (3) & 0.537 (4) & 0.529 (5) & 0.511 (6) & 0.280 (7) \\
Var & 28.7 (7) & 11.7 (5) & 10.7 (4) & 3.4 (1) & 12.1 (6) & 5.9 (2) & 7.4 (3) \\
\bottomrule
\end{tabular}}
\end{table}

%% file: section/exp.tex
\section{Additional Numerical Results}\label{app:appendix_experiments}

All experiments are conducted on a server running Ubuntu 24.04.2 LTS (Noble) with 2.0 TB of RAM and Intel Xeon Platinum 8480+ processors (56 cores / 112 threads per socket; 2.00 GHz base frequency). The full per-fuzzer Magma ground-truth table and the complete per-fuzzer figure grids for Section~\ref{sec:exp} appear in this appendix and are also included in the artifact.

\iffullversion
\begin{table}[tp]
\centering
\caption{Magma ground-truth validation (full): five fuzzers $\times$ eight Magma programs, aggregated over $M{=}20$ trials of $12$\,h, splitting vs.\ no-split (MC). Each cell is \emph{splitting\,/\,no-split}: mean unique bugs per trial; total (union) unique bugs over the $20$ trials; weighted bug-detection rate $\widehat{P}_R$ (bugs/h); effort-normalized $\mathrm{Var}(\widehat{P}_R)$ (lower is better). \textbf{Bold} marks the better value; ties unbolded. Note the saturation in the variance column: the no-split variance is exactly $0$ in 17 of 40 cells (identical outcomes across trials), the regime where the variance metric is uninformative (Section~\ref{subsec:magma}).}\label{tab:magma-full}
\footnotesize
\setlength{\tabcolsep}{4pt}
\resizebox{\columnwidth}{!}{%
\begin{tabular}{lcccc}
\toprule
Fuzzer & mean uniq & total uniq & $\widehat{P}_R$ & $\mathrm{Var}(\widehat{P}_R)$ \\
       & (s\,/\,n) & (s\,/\,n)  & (s\,/\,n)        & (s\,/\,n) \\
\midrule
\multicolumn{5}{l}{\emph{libpng}}\\
\quad AFL & \textbf{1.2}/1 & \textbf{2}/1 & \textbf{4.025}/1 & 7.23e-08/\textbf{0} \\
\quad AFLFast & 1/1 & 1/1 & \textbf{4}/1 & 0/0 \\
\quad AFL++ & \textbf{3}/2.2 & 3/3 & \textbf{5.167}/2 & \textbf{5.87e-06}/1.16e-05 \\
\quad MOpt & \textbf{1.8}/1.2 & \textbf{3}/2 & \textbf{4.25}/1.2 & \textbf{1.63e-06}/4.63e-06 \\
\quad Honggfuzz & \textbf{4}/3.6 & 4/4 & \textbf{6.55}/3.6 & 7.52e-06/\textbf{6.94e-06} \\
\midrule
\multicolumn{5}{l}{\emph{libtiff}}\\
\quad AFL & 1/1 & 2/2 & \textbf{1.27}/1 & 2.72e-05/\textbf{0} \\
\quad AFLFast & \textbf{2.25}/2 & \textbf{5}/3 & 1.992/\textbf{2.000} & \textbf{6.41e-06}/2.31e-05 \\
\quad AFL++ & \textbf{6.4}/5.4 & 7/7 & \textbf{5.755}/4.4 & 3.76e-05/\textbf{6.94e-06} \\
\quad MOpt & \textbf{4.8}/3.4 & \textbf{5}/4 & \textbf{4.683}/3.4 & \textbf{3.89e-06}/6.94e-06 \\
\quad Honggfuzz & \textbf{4.6}/4 & \textbf{6}/4 & \textbf{5.968}/4 & 1.17e-05/\textbf{0} \\
\midrule
\multicolumn{5}{l}{\emph{libsndfile}}\\
\quad AFL & 2/2 & 2/2 & \textbf{4.16}/2 & 1.32e-06/\textbf{0} \\
\quad AFLFast & 2/2 & 2/2 & \textbf{4.06}/2 & 1.71e-07/\textbf{0} \\
\quad AFL++ & 7/7 & 7/7 & \textbf{8.009}/5.2 & \textbf{9.28e-06}/1.62e-05 \\
\quad MOpt & 7/7 & 7/7 & \textbf{7.533}/5 & 8.69e-05/\textbf{9.92e-36} \\
\quad Honggfuzz & \textbf{6.4}/6 & \textbf{7}/6 & \textbf{8.969}/6 & 2.90e-06/\textbf{0} \\
\midrule
\multicolumn{5}{l}{\emph{poppler}}\\
\quad AFL & 1.8/1.8 & \textbf{3}/2 & \textbf{4.05}/1.8 & \textbf{2.89e-07}/4.63e-06 \\
\quad AFLFast & 2/2 & 2/2 & \textbf{4.307}/1 & 1.08e-06/\textbf{0} \\
\quad AFL++ & \textbf{3.4}/2.2 & \textbf{5}/3 & \textbf{5.342}/2.2 & \textbf{2.78e-06}/4.63e-06 \\
\quad MOpt & \textbf{4}/3 & \textbf{6}/3 & \textbf{6.175}/2.8 & \textbf{1.74e-06}/4.63e-06 \\
\quad Honggfuzz & \textbf{3}/2.2 & 3/3 & \textbf{5.289}/2.2 & 5.10e-06/\textbf{4.63e-06} \\
\midrule
\multicolumn{5}{l}{\emph{sqlite3}}\\
\quad AFL & \textbf{0.8}/0.4 & \textbf{3}/2 & \textbf{1.158}/0.4 & 9.98e-05/\textbf{1.85e-05} \\
\quad AFLFast & 1.6/1.6 & \textbf{3}/2 & \textbf{3.36}/1.6 & 8.78e-05/\textbf{6.94e-06} \\
\quad AFL++ & \textbf{3}/1.8 & \textbf{3}/2 & \textbf{5.166}/1.8 & \textbf{1.64e-06}/4.63e-06 \\
\quad MOpt & \textbf{2.8}/2.4 & \textbf{4}/3 & \textbf{5.285}/2.2 & \textbf{6.88e-06}/1.62e-05 \\
\quad Honggfuzz & \textbf{3.4}/2.8 & \textbf{5}/3 & \textbf{5.828}/2.8 & \textbf{1.33e-06}/4.63e-06 \\
\midrule
\multicolumn{5}{l}{\emph{openssl}}\\
\quad AFL & \textbf{2}/1.2 & 2/2 & \textbf{4.8}/1.2 & \textbf{6.51e-07}/4.63e-06 \\
\quad AFLFast & \textbf{2}/1.6 & 2/2 & \textbf{4.85}/1.6 & \textbf{4.34e-07}/6.94e-06 \\
\quad AFL++ & 1/1 & 1/1 & \textbf{3.825}/1 & 4.07e-07/\textbf{0} \\
\quad MOpt & \textbf{2}/1.8 & 2/2 & \textbf{4.592}/1.8 & \textbf{4.30e-07}/4.63e-06 \\
\quad Honggfuzz & \textbf{1.2}/1 & \textbf{2}/1 & \textbf{3.749}/1 & 5.10e-07/\textbf{0} \\
\midrule
\multicolumn{5}{l}{\emph{libxml2}}\\
\quad AFL & \textbf{1.2}/1 & \textbf{2}/1 & \textbf{4}/1 & 0/0 \\
\quad AFLFast & \textbf{1.4}/1 & \textbf{2}/1 & \textbf{4.029}/1 & 6.86e-08/\textbf{0} \\
\quad AFL++ & 4/4 & 4/4 & \textbf{6.775}/4 & 6.15e-07/\textbf{0} \\
\quad MOpt & \textbf{3.6}/3 & 4/4 & \textbf{5.35}/3 & \textbf{9.77e-07}/2.31e-05 \\
\quad Honggfuzz & \textbf{2.2}/1.4 & \textbf{3}/2 & \textbf{4.25}/1.4 & \textbf{5.43e-07}/6.94e-06 \\
\midrule
\multicolumn{5}{l}{\emph{php}}\\
\quad AFL & 3/3 & 3/3 & \textbf{5.996}/3 & 5.79e-06/\textbf{0} \\
\quad AFLFast & 3/3 & 3/3 & \textbf{6.45}/3 & 1.13e-05/\textbf{0} \\
\quad AFL++ & \textbf{3}/2 & 3/3 & \textbf{4.864}/2 & 1.46e-06/\textbf{0} \\
\quad MOpt & 3/3 & 3/3 & \textbf{6.775}/2.6 & 1.74e-05/\textbf{6.94e-06} \\
\quad Honggfuzz & 1/1 & 1/1 & \textbf{4}/1 & 0/0 \\
\bottomrule
\end{tabular}}
\end{table}
\fi

\iffullversion

\begin{figure}[htb]
    \centering
    \includegraphics[width=0.46\textwidth]{figs/main/legend_fuzzers.pdf}\\[3pt]
    \makebox[0.235\textwidth]{\small\textbf{Non-splitting (MC)}}\hfill\makebox[0.235\textwidth]{\small\textbf{Splitting}}\\[2pt]
    \includegraphics[width=0.235\textwidth]{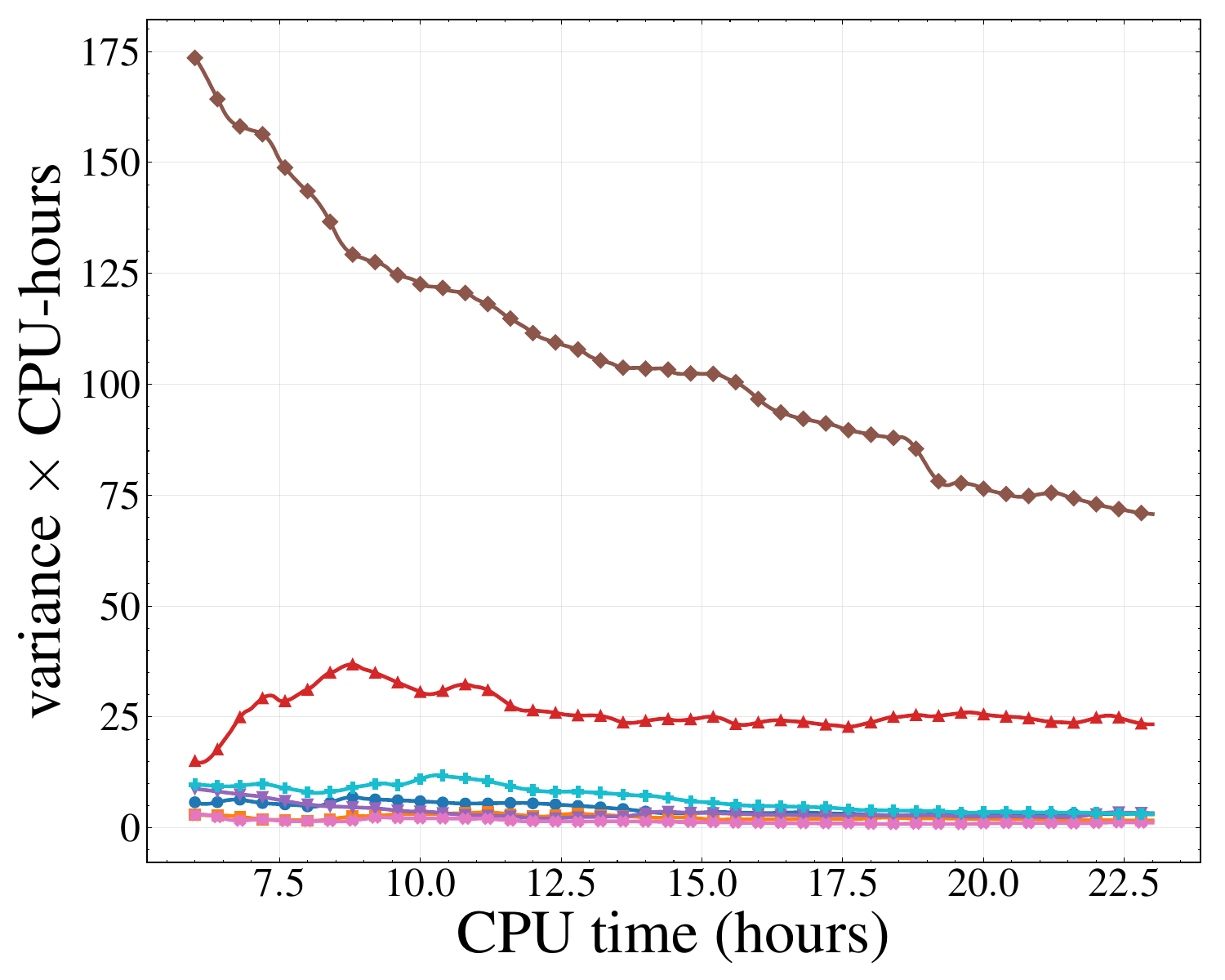}\hfill\includegraphics[width=0.235\textwidth]{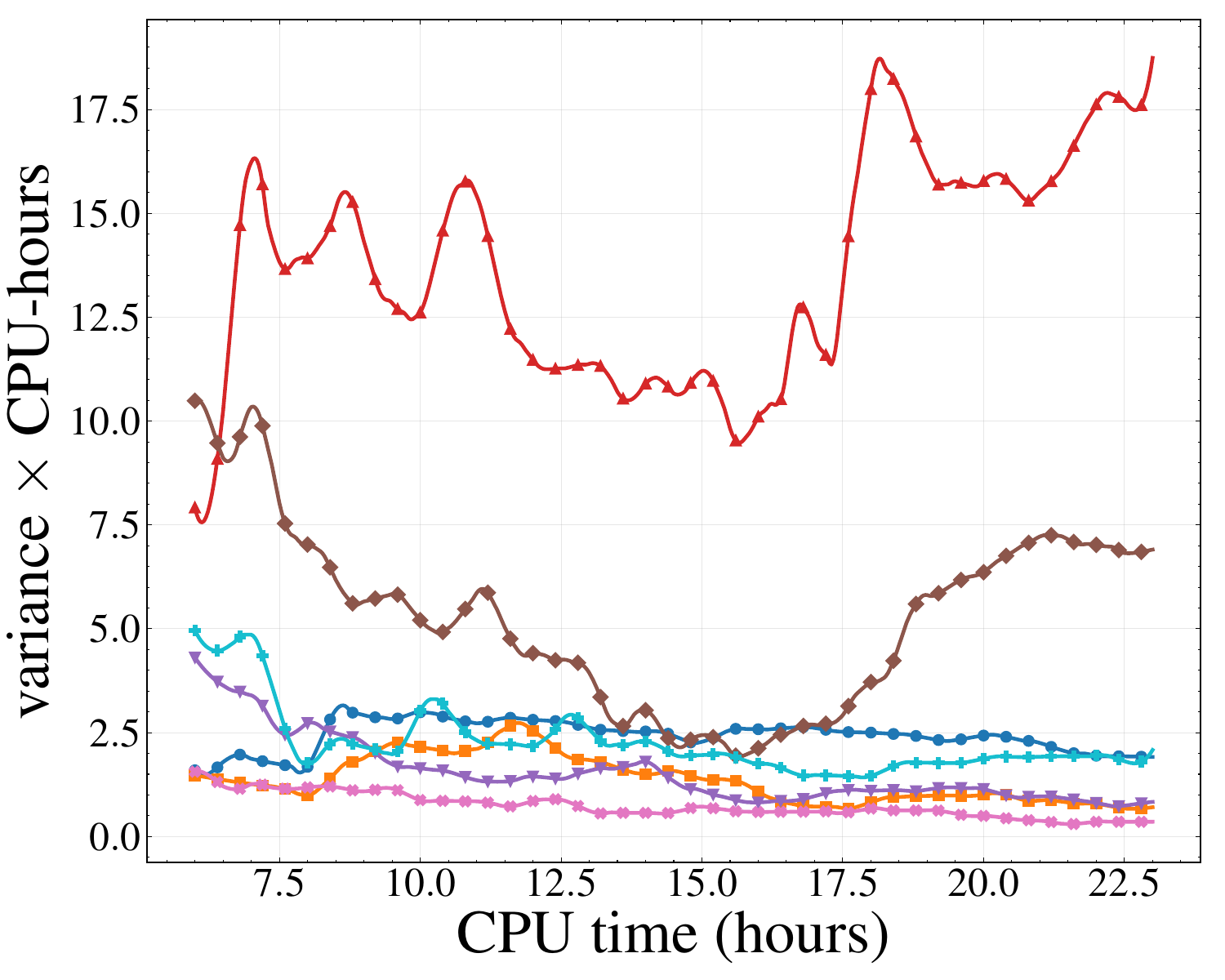}\\
    {\small (a,\,b) ffmpeg}\\[3pt]
    \includegraphics[width=0.235\textwidth]{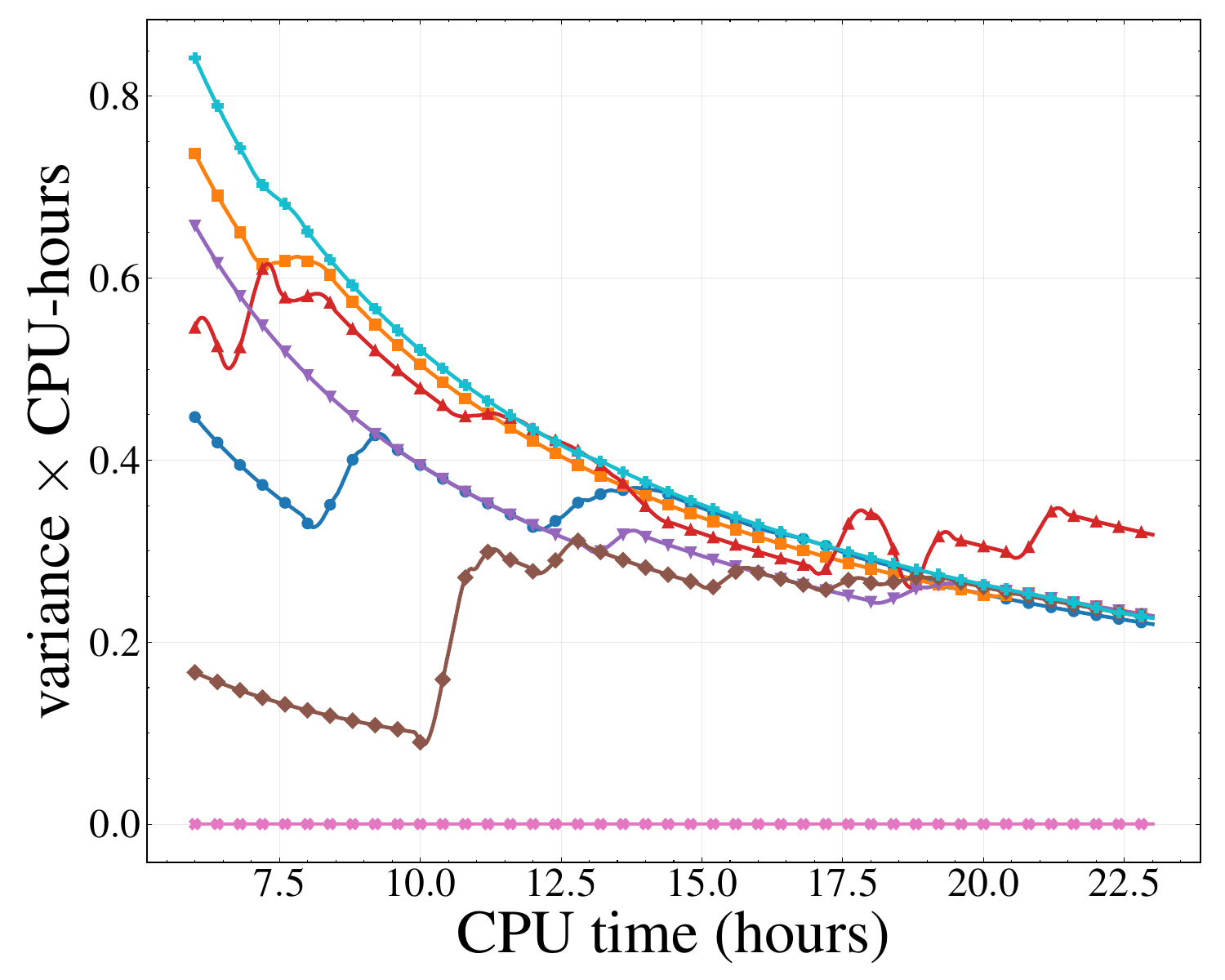}\hfill\includegraphics[width=0.235\textwidth]{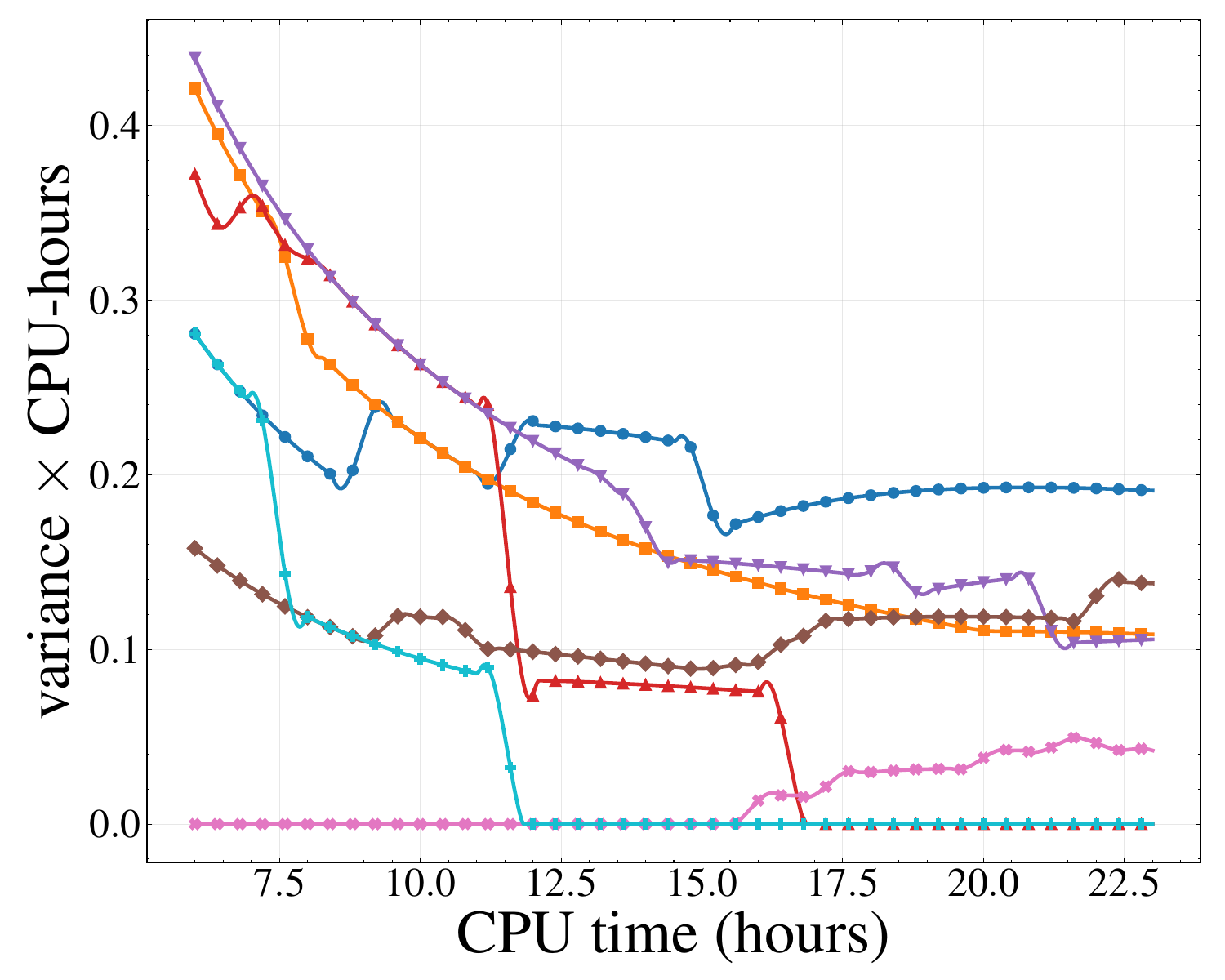}\\
    {\small (c,\,d) grok}\\[3pt]
    \includegraphics[width=0.235\textwidth]{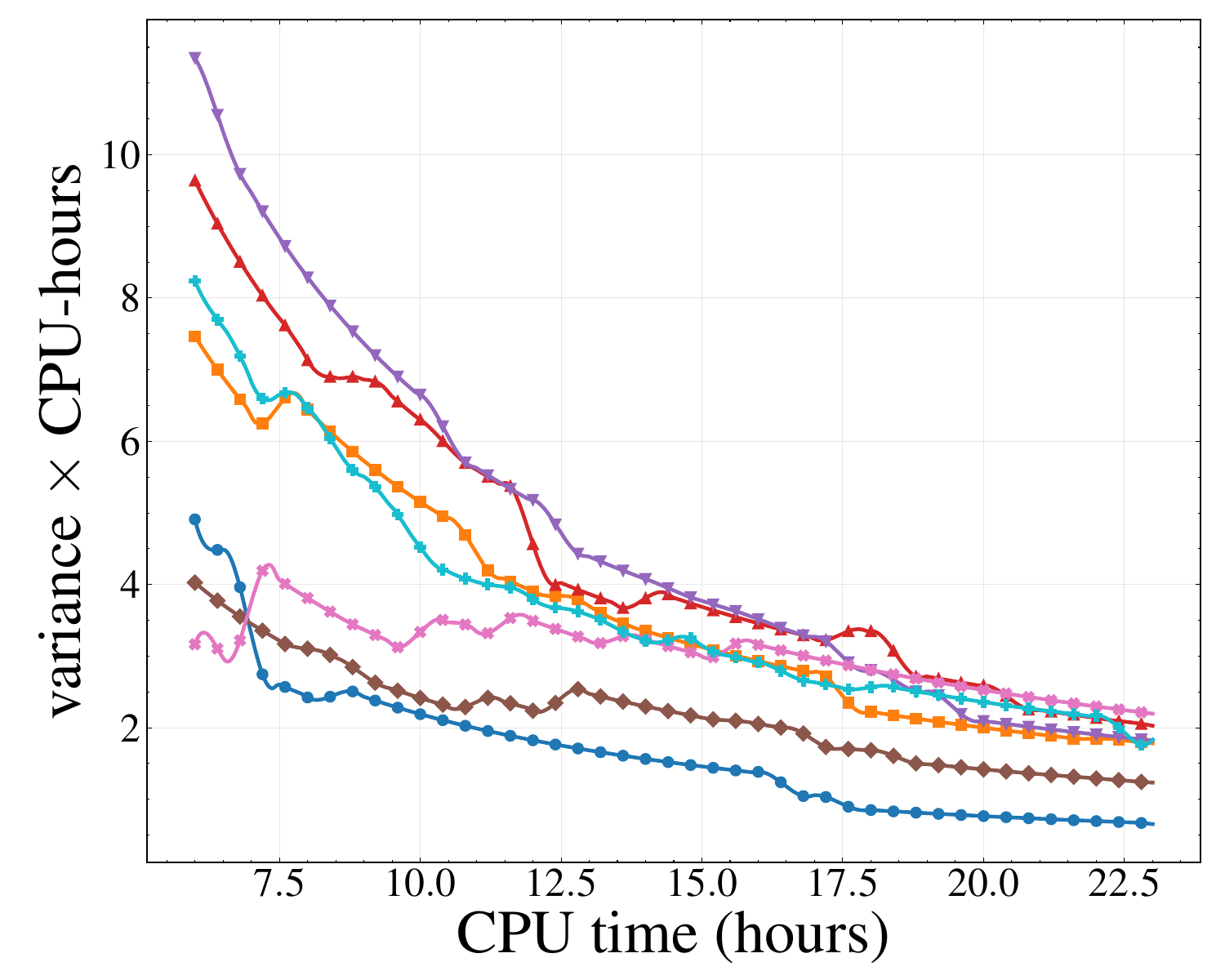}\hfill\includegraphics[width=0.235\textwidth]{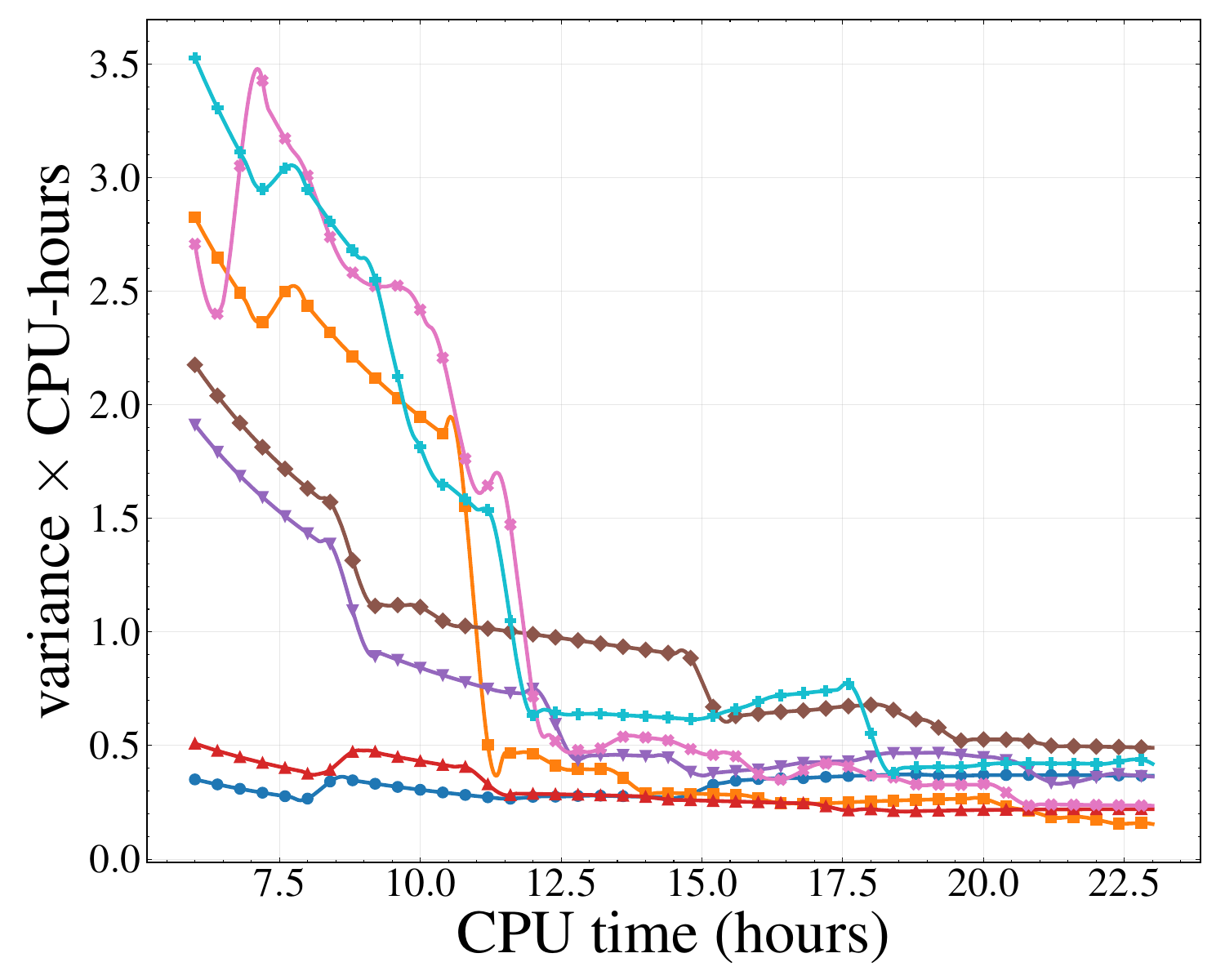}\\
    {\small (e,\,f) libhtp}
    \caption{Variance ($\times$ CPU-hours) across the seven fuzzers under Monte Carlo (non-splitting, left column) and splitting (right column), for three of the remaining benchmarks (part 1 of 3). Companion to Figure~\ref{fig:mc-simulation}.}
    \label{fig:var-appendix}
\end{figure}

\begin{figure}[htb]
    \centering
    \includegraphics[width=0.46\textwidth]{figs/main/legend_fuzzers.pdf}\\[3pt]
    \makebox[0.235\textwidth]{\small\textbf{Non-splitting (MC)}}\hfill\makebox[0.235\textwidth]{\small\textbf{Splitting}}\\[2pt]
    \includegraphics[width=0.235\textwidth]{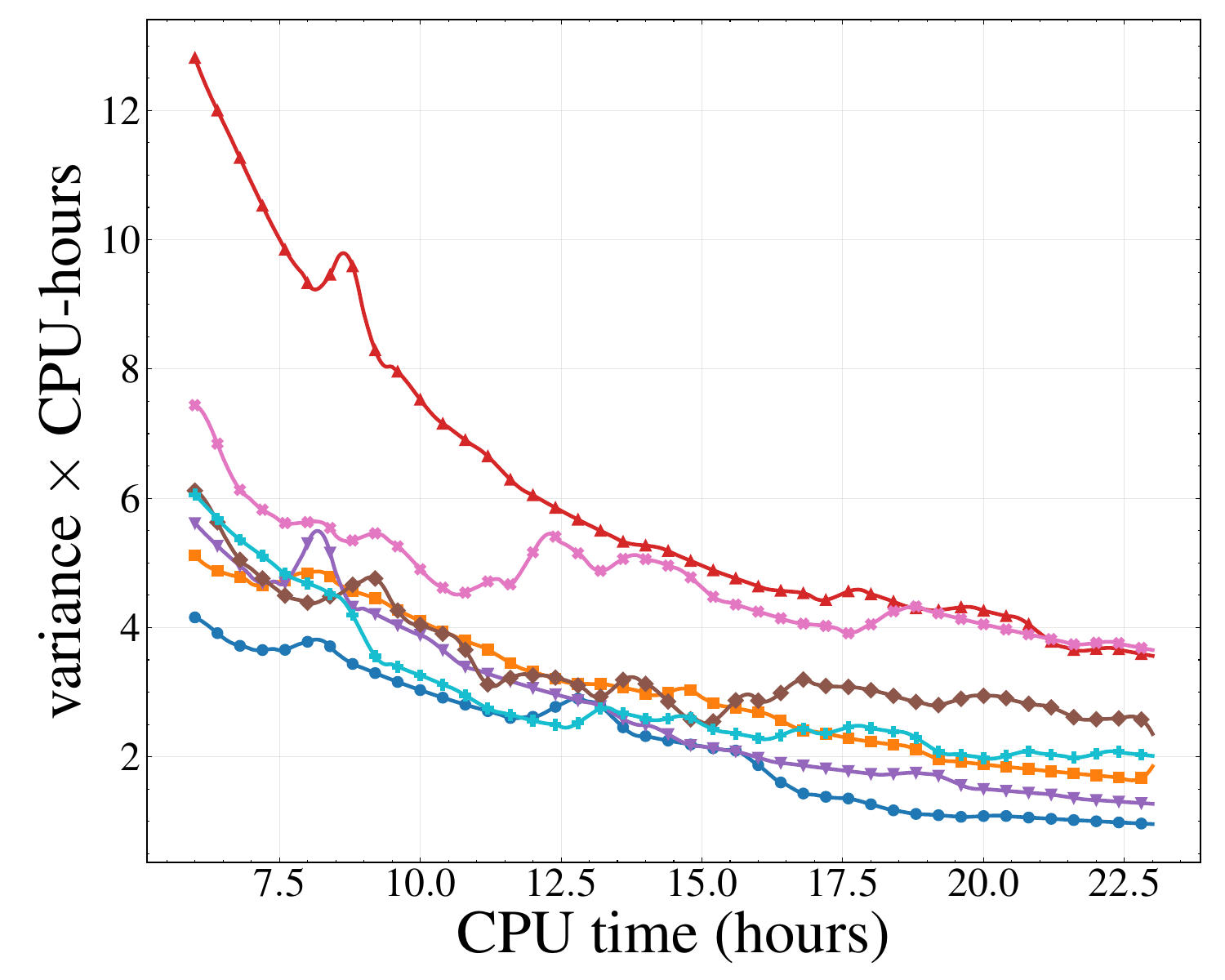}\hfill\includegraphics[width=0.235\textwidth]{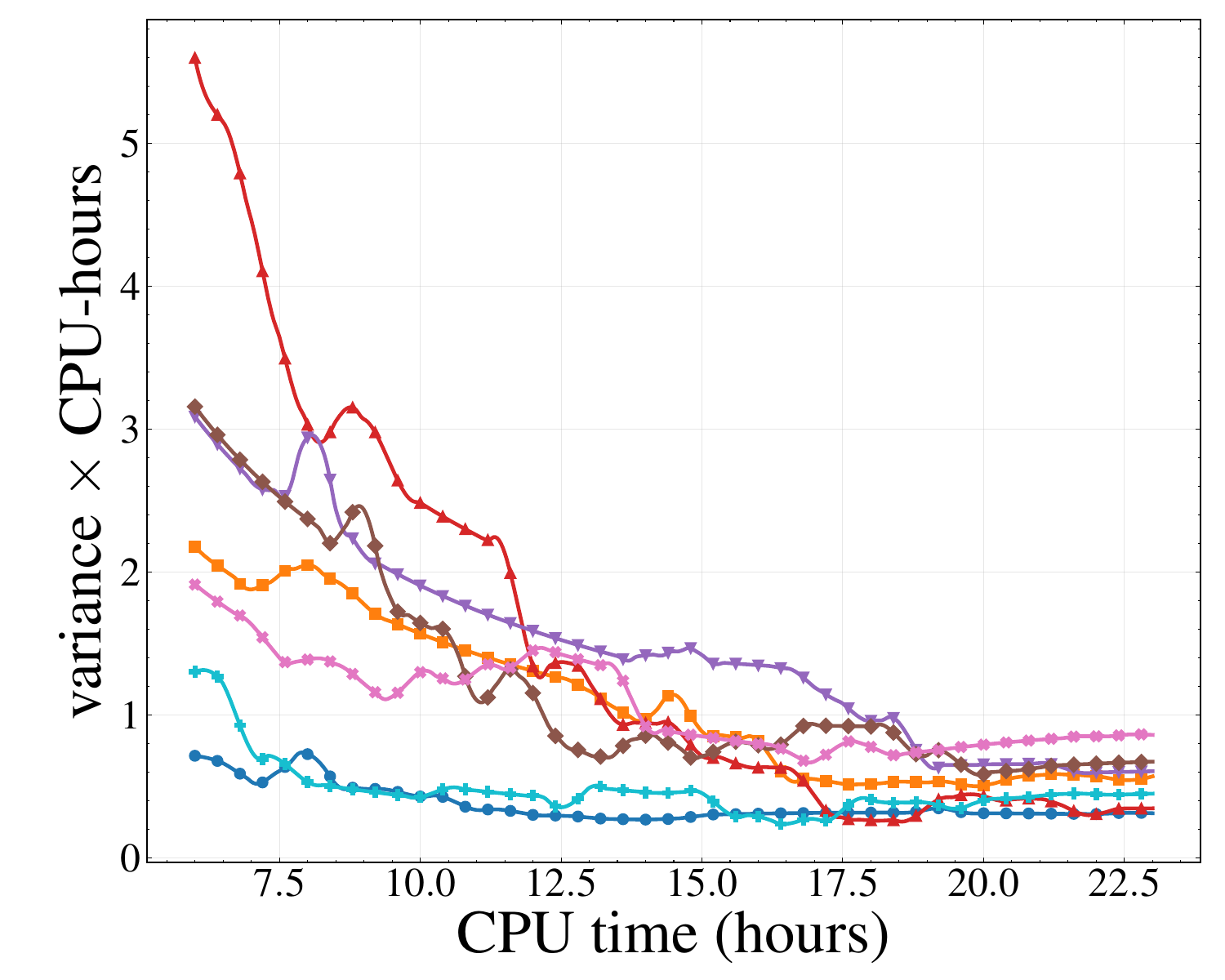}\\
    {\small (a,\,b) matio}\\[3pt]
    \includegraphics[width=0.235\textwidth]{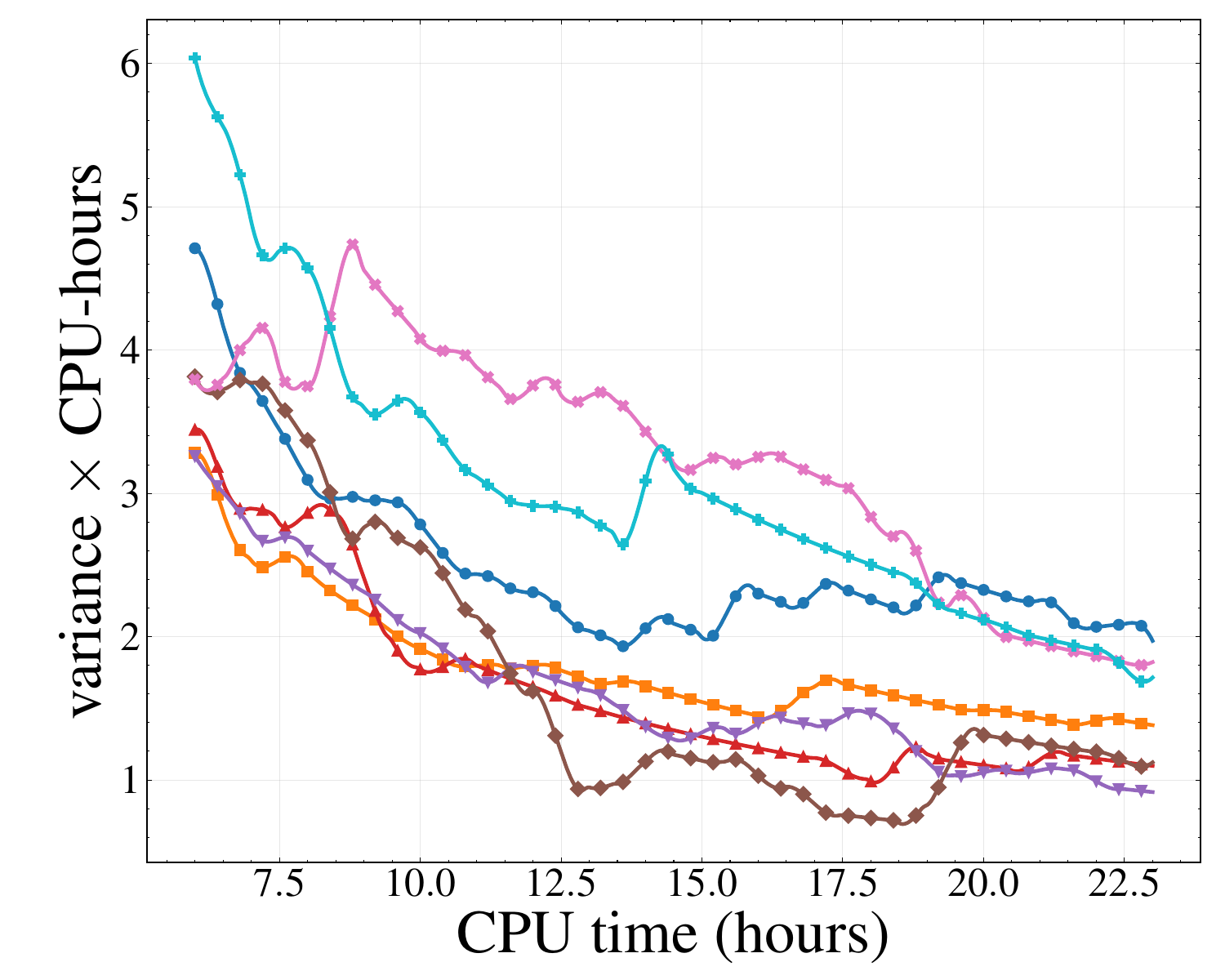}\hfill\includegraphics[width=0.235\textwidth]{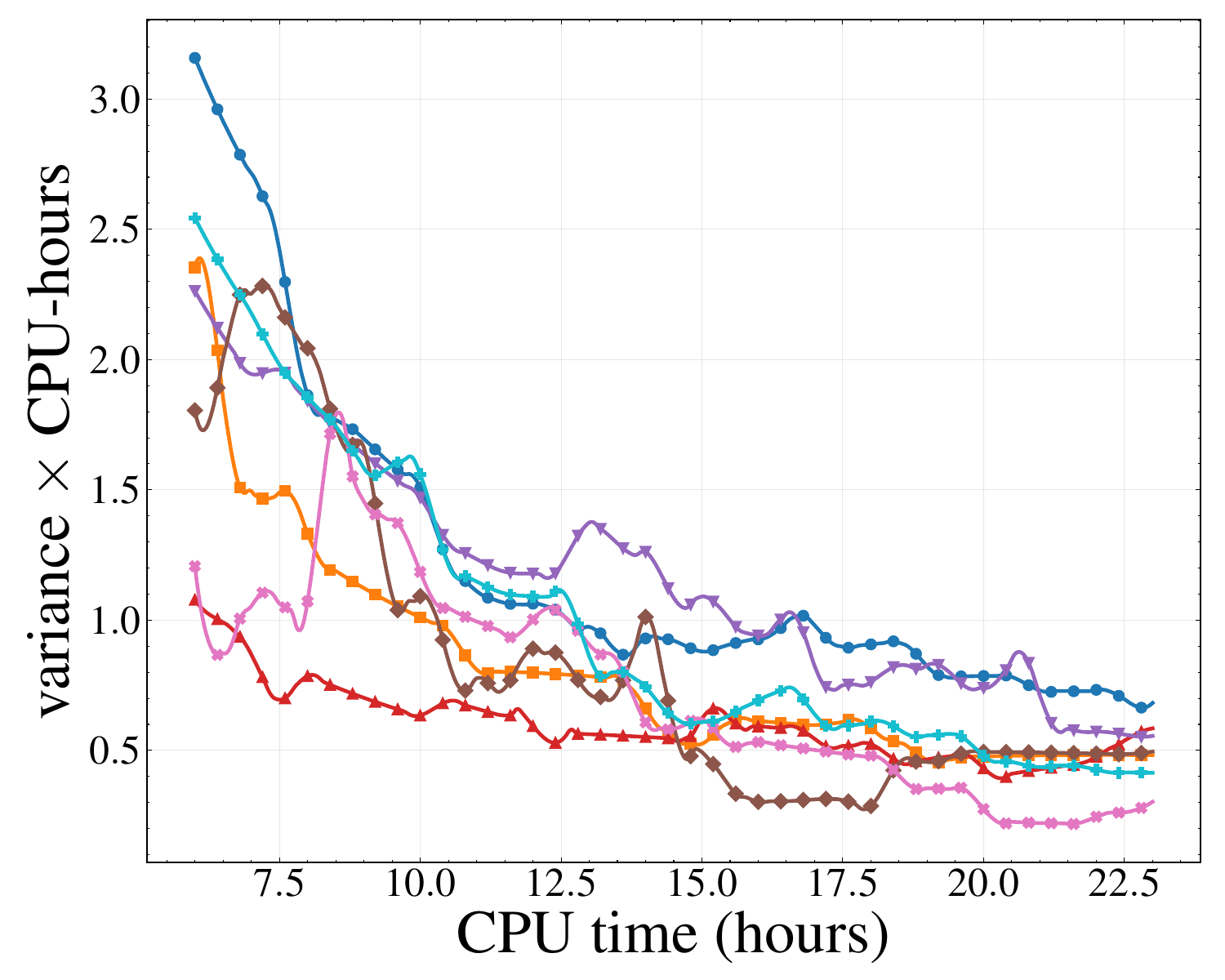}\\
    {\small (c,\,d) openh264}\\[3pt]
    \includegraphics[width=0.235\textwidth]{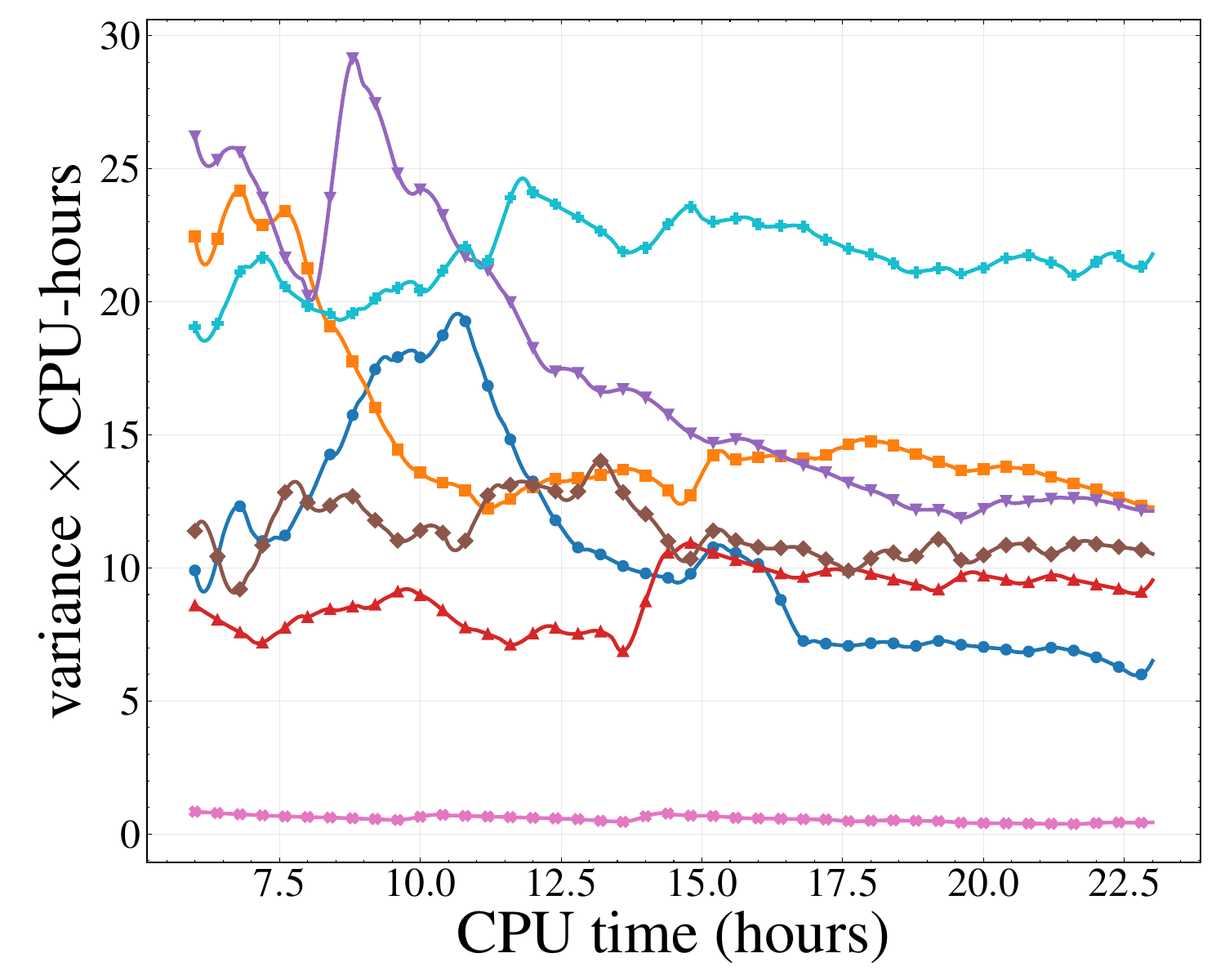}\hfill\includegraphics[width=0.235\textwidth]{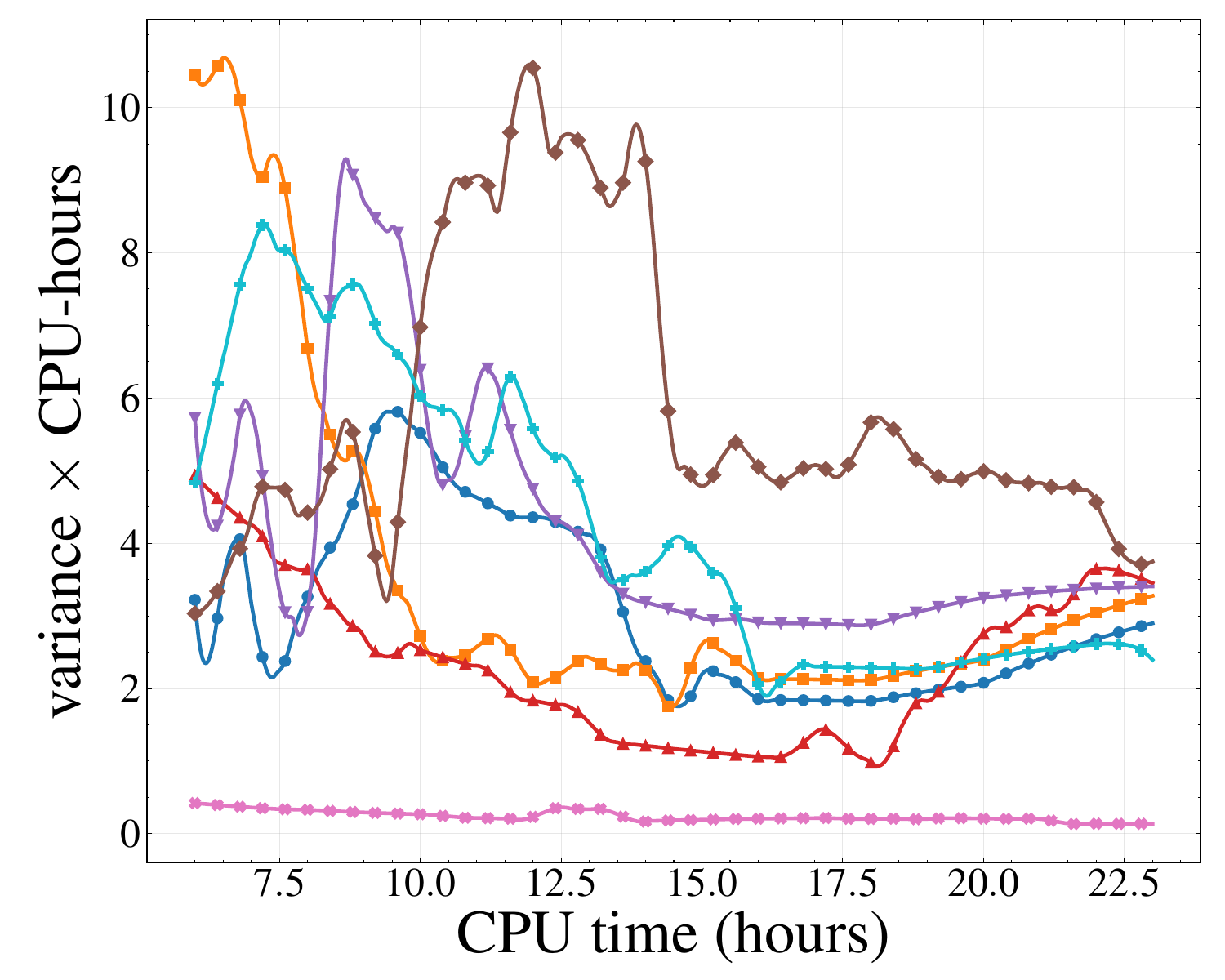}\\
    {\small (e,\,f) poppler}
    \caption{Variance ($\times$ CPU-hours) across the seven fuzzers under Monte Carlo (non-splitting, left column) and splitting (right column), for three of the remaining benchmarks (part 2 of 3). Companion to Figure~\ref{fig:mc-simulation}.}
    \label{fig:var-appendix2}
\end{figure}

\begin{figure}[htb]
    \centering
    \includegraphics[width=0.46\textwidth]{figs/main/legend_fuzzers.pdf}\\[3pt]
    \makebox[0.235\textwidth]{\small\textbf{Non-splitting (MC)}}\hfill\makebox[0.235\textwidth]{\small\textbf{Splitting}}\\[2pt]
    \includegraphics[width=0.235\textwidth]{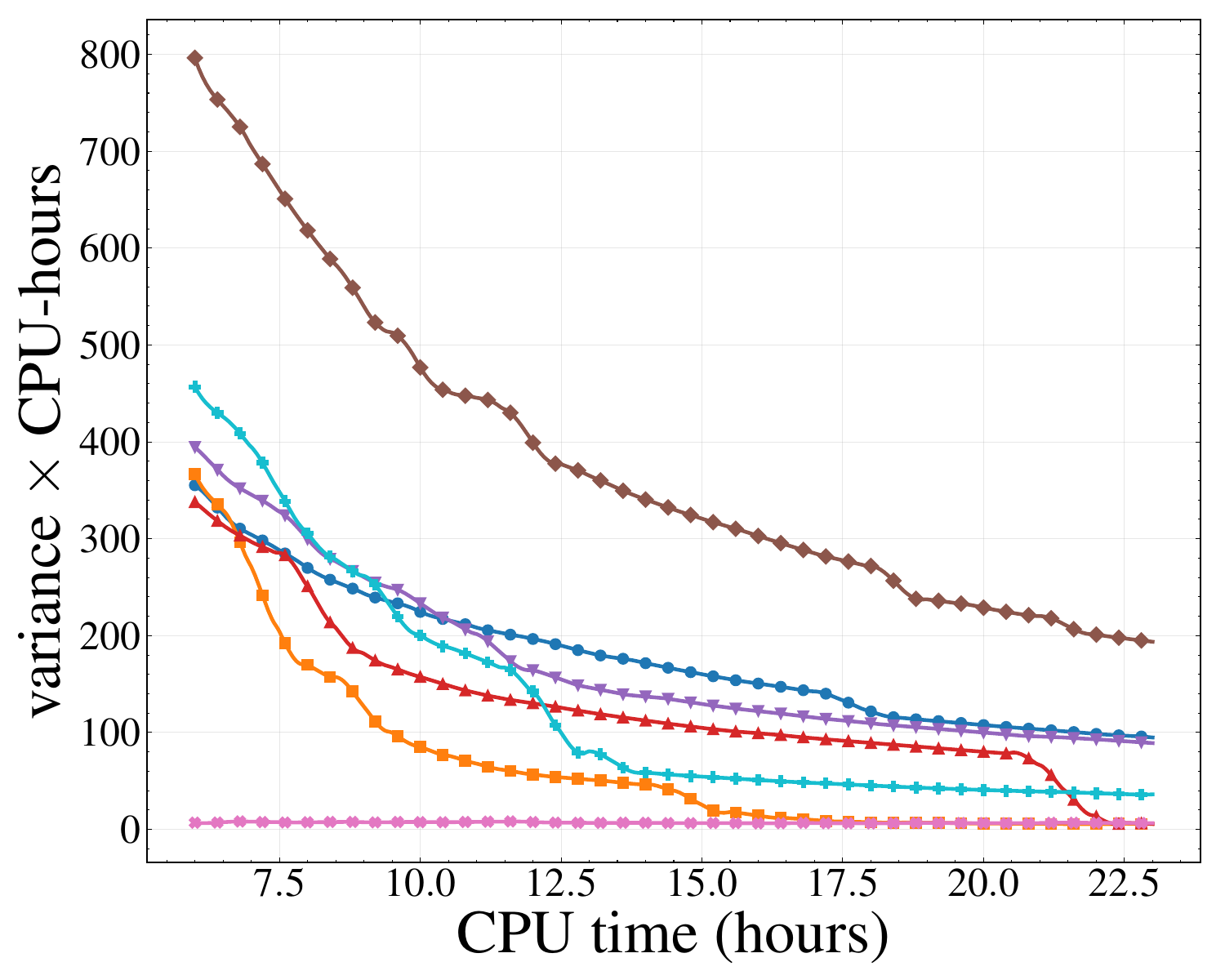}\hfill\includegraphics[width=0.235\textwidth]{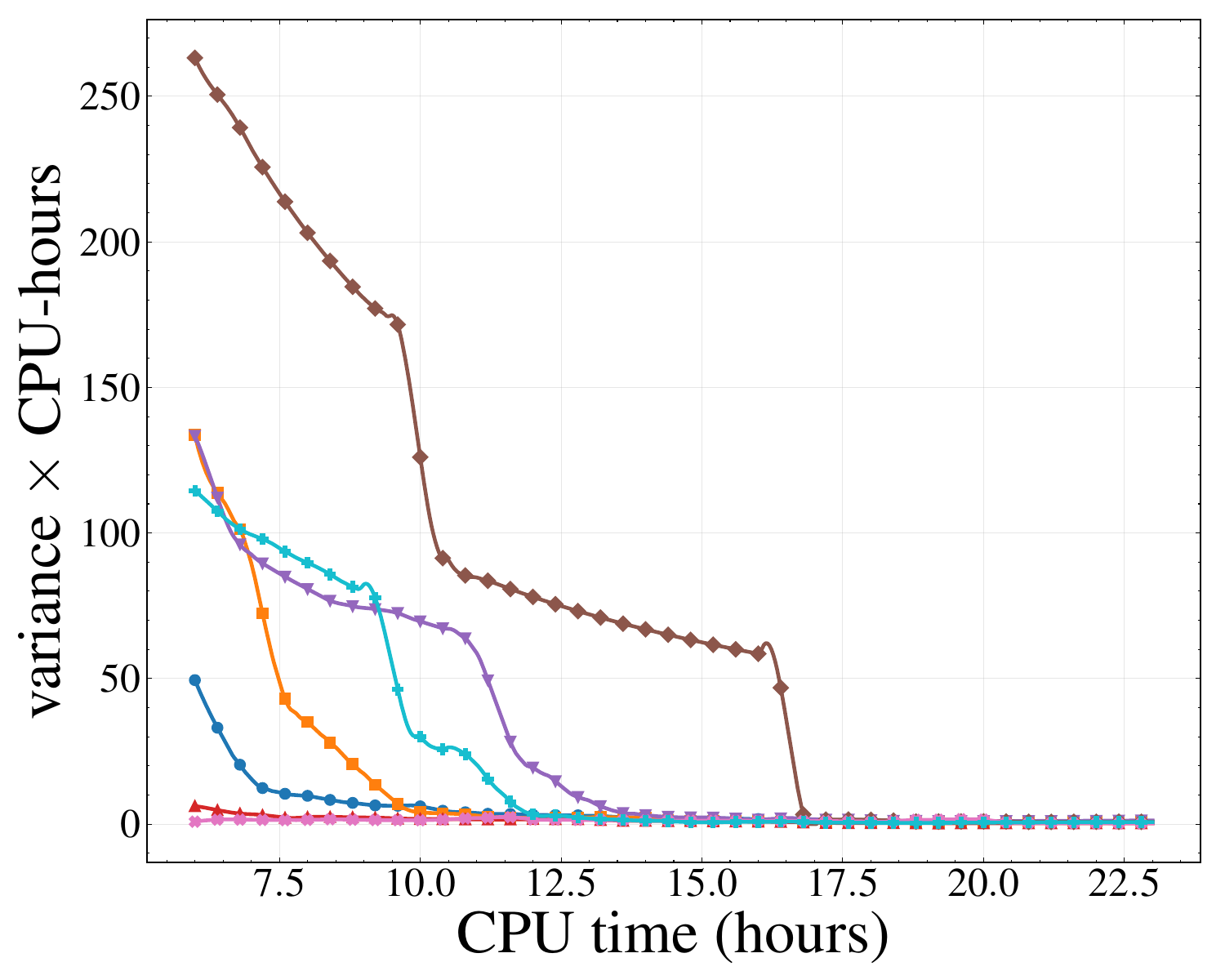}\\
    {\small (a,\,b) libhevc}\\[3pt]
    \includegraphics[width=0.235\textwidth]{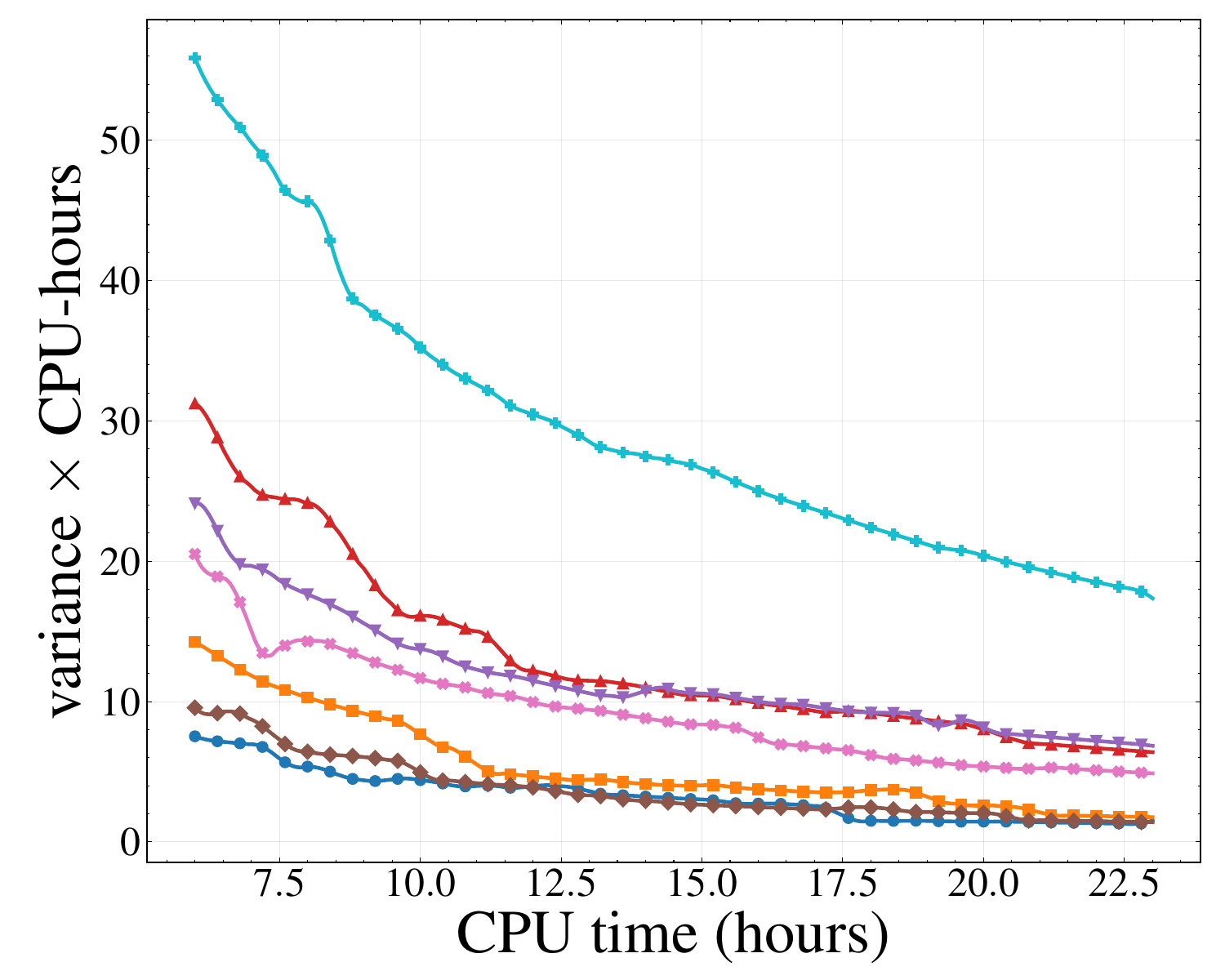}\hfill\includegraphics[width=0.235\textwidth]{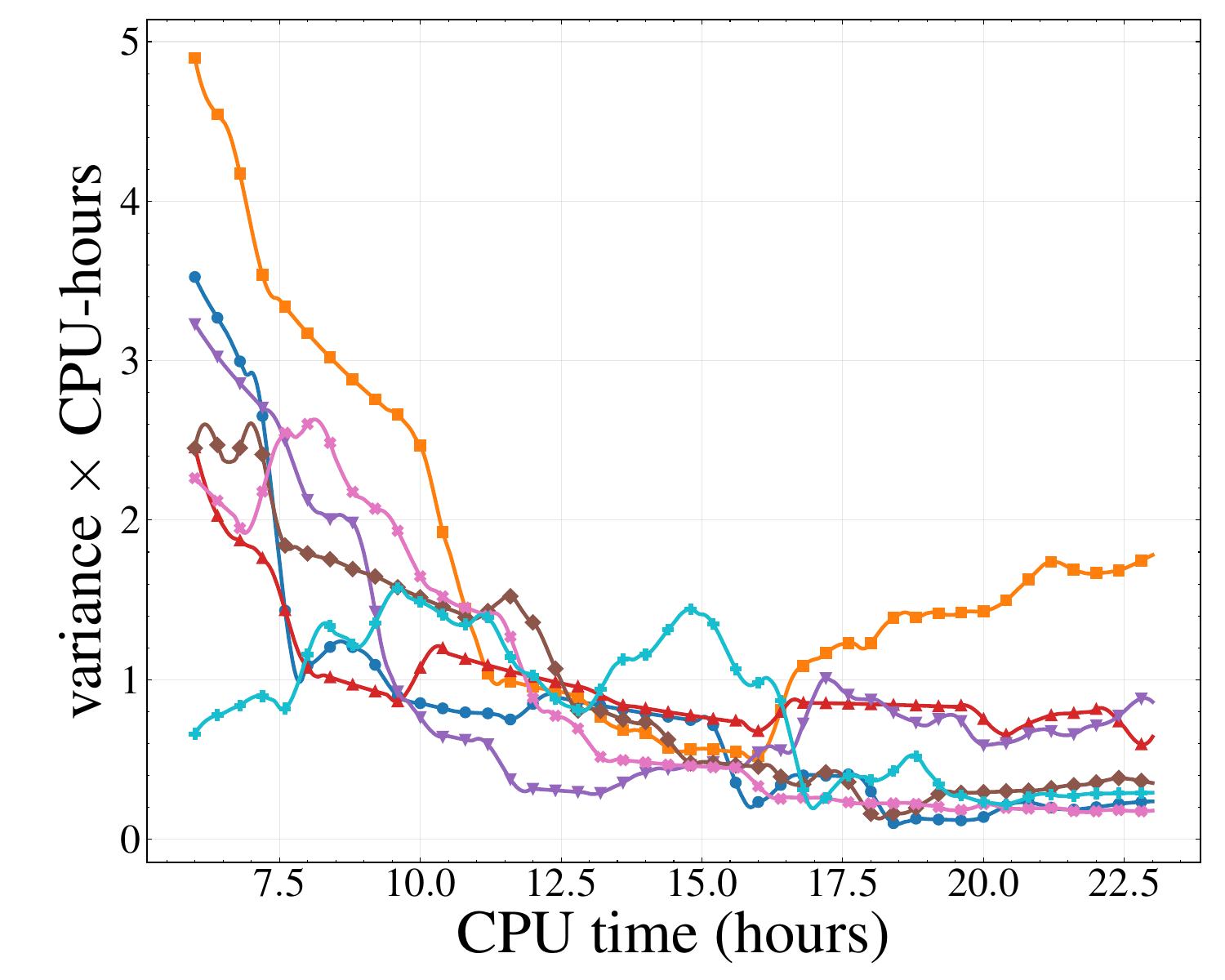}\\
    {\small (c,\,d) stb}
    \caption{Variance ($\times$ CPU-hours) across the seven fuzzers under Monte Carlo (non-splitting, left column) and splitting (right column), for the last two remaining benchmarks (part 3 of 3). Companion to Figure~\ref{fig:mc-simulation}.}
    \label{fig:var-appendix3}
\end{figure}

\begin{figure*}[tp]
    \centering
    \includegraphics[width=0.19\textwidth]{figs/comparison/aflplusplus/arrow_parquet-arrow-fuzz__variance_times_cpu_hour.pdf}\hfill\includegraphics[width=0.19\textwidth]{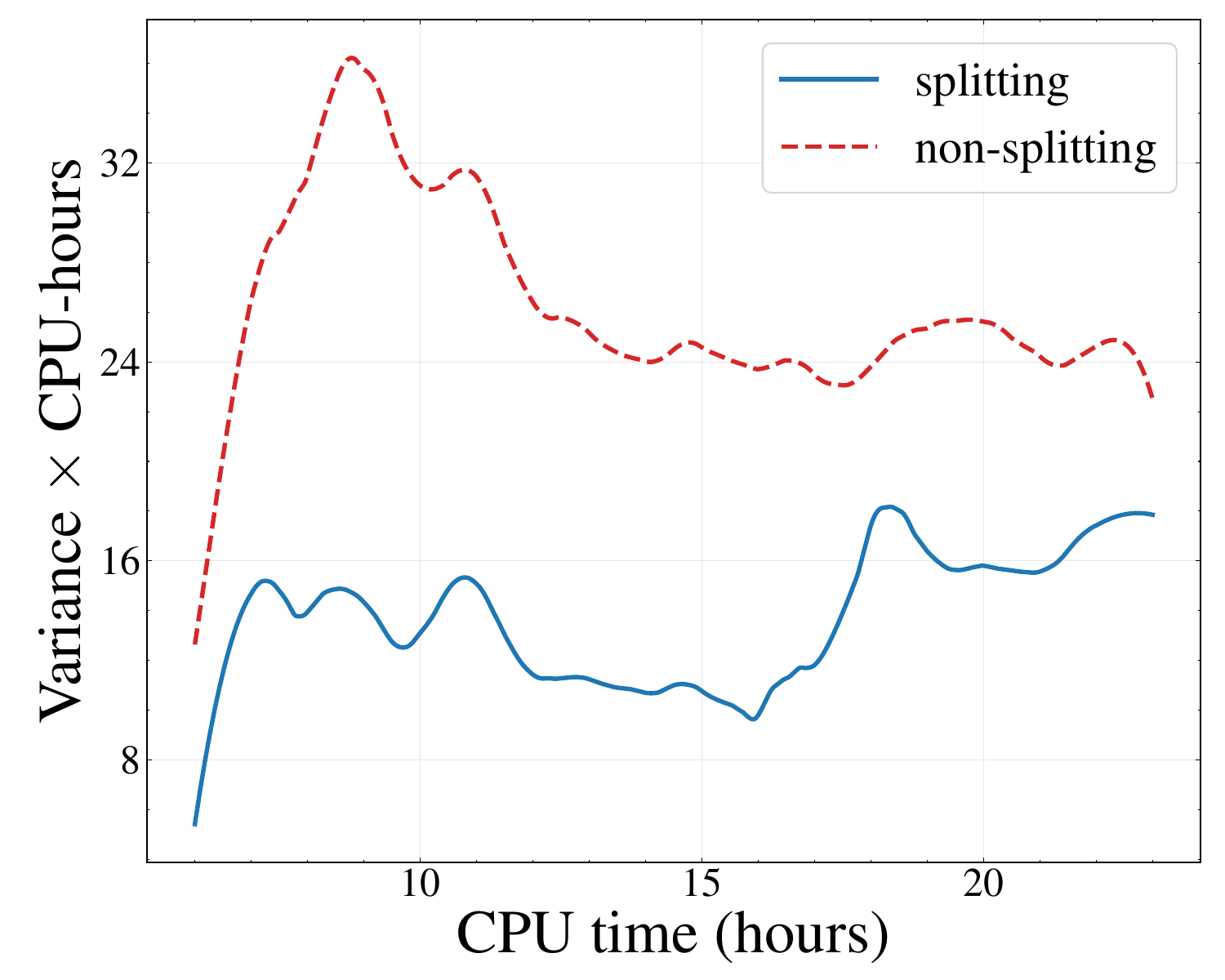}\hfill\includegraphics[width=0.19\textwidth]{figs/comparison/aflplusplus/grok_grk_decompress_fuzzer__variance_times_cpu_hour.pdf}\hfill\includegraphics[width=0.19\textwidth]{figs/comparison/aflplusplus/libhevc_hevc_dec_fuzzer__variance_times_cpu_hour.pdf}\hfill\includegraphics[width=0.19\textwidth]{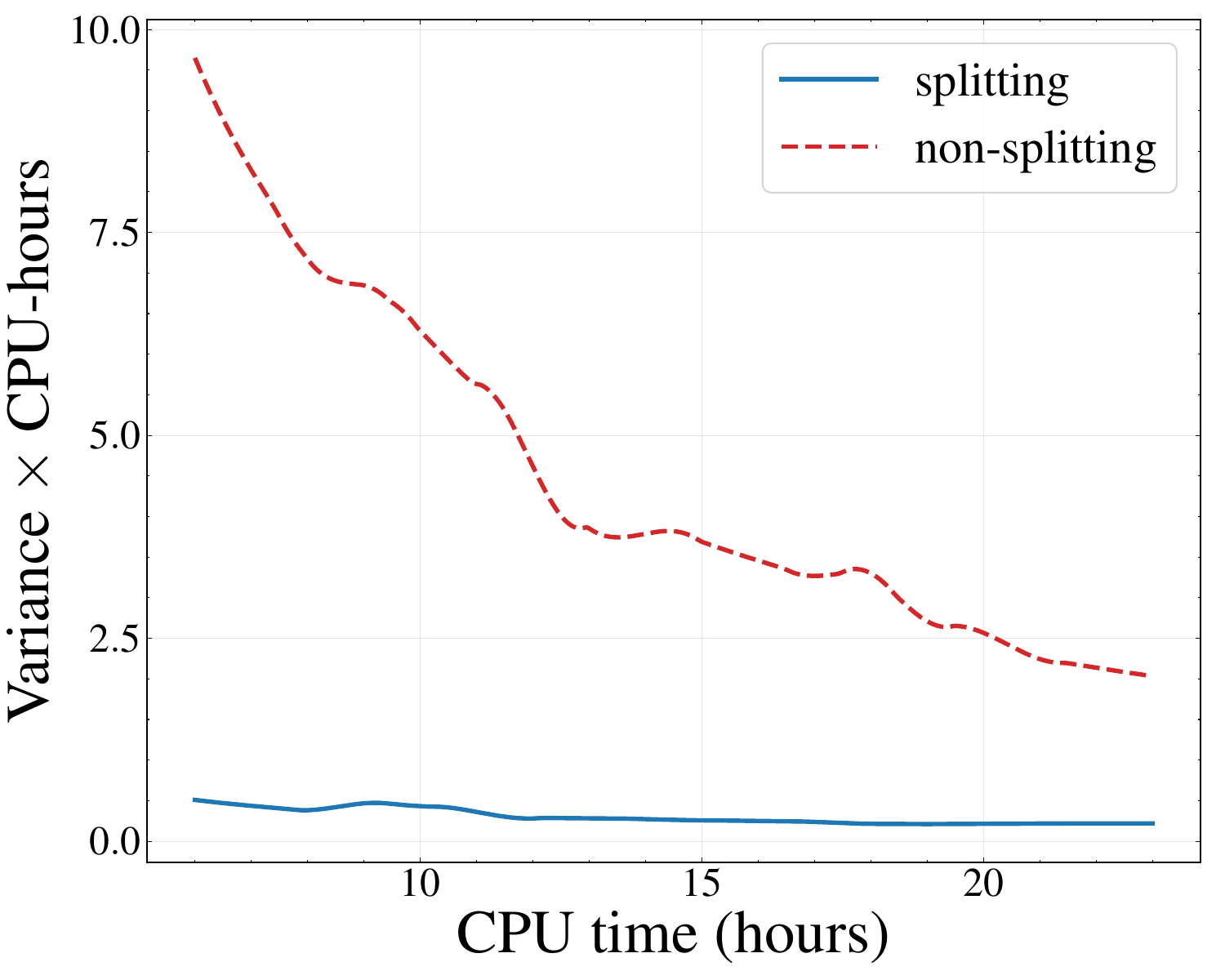}\\
    \makebox[0.19\textwidth]{\footnotesize (a) variance of arrow}\hfill\makebox[0.19\textwidth]{\footnotesize (b) variance of ffmpeg}\hfill\makebox[0.19\textwidth]{\footnotesize (c) variance of grok}\hfill\makebox[0.19\textwidth]{\footnotesize (d) variance of libhevc}\hfill\makebox[0.19\textwidth]{\footnotesize (e) variance of libhtp}\\[3pt]
    \includegraphics[width=0.19\textwidth]{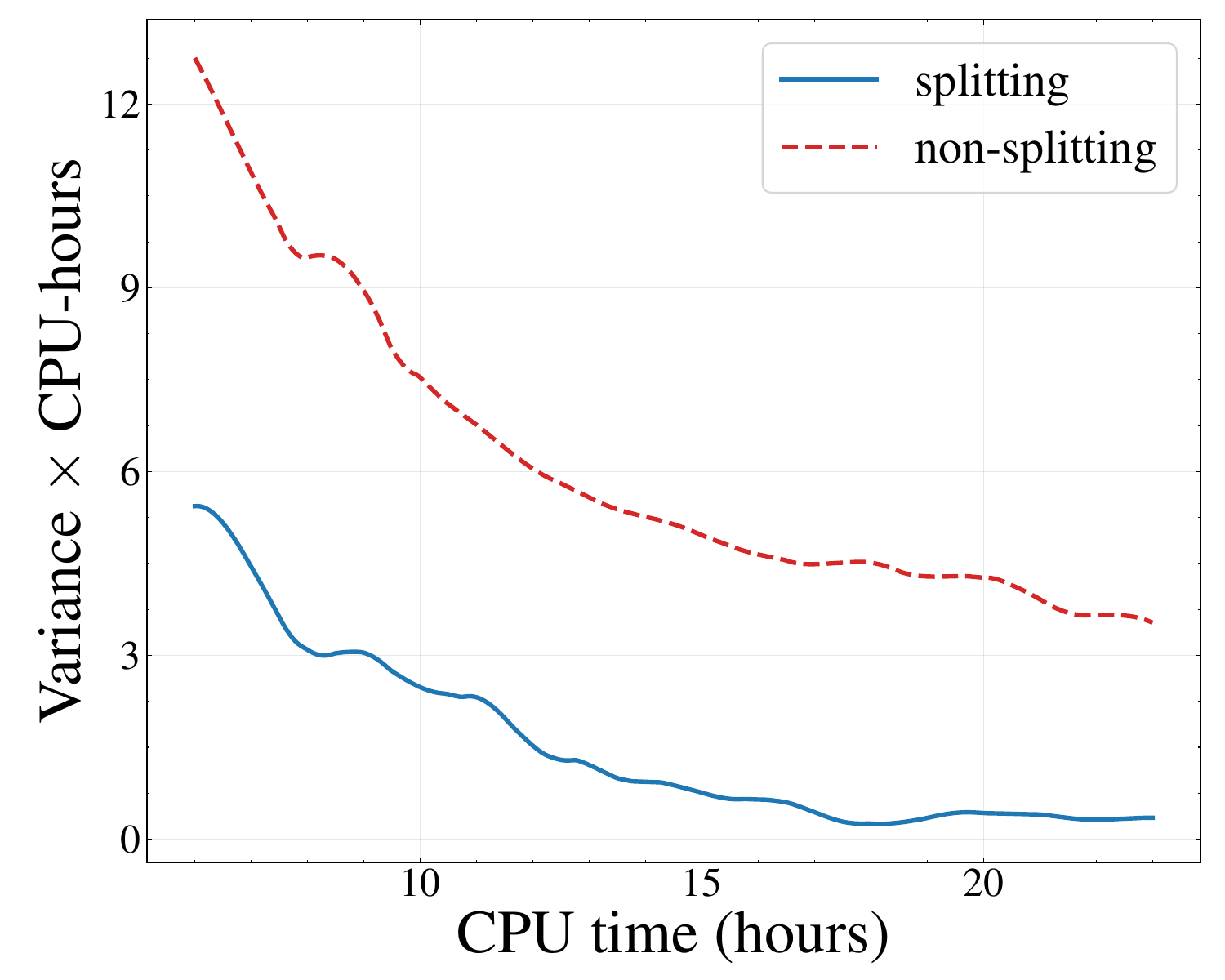}\hfill\includegraphics[width=0.19\textwidth]{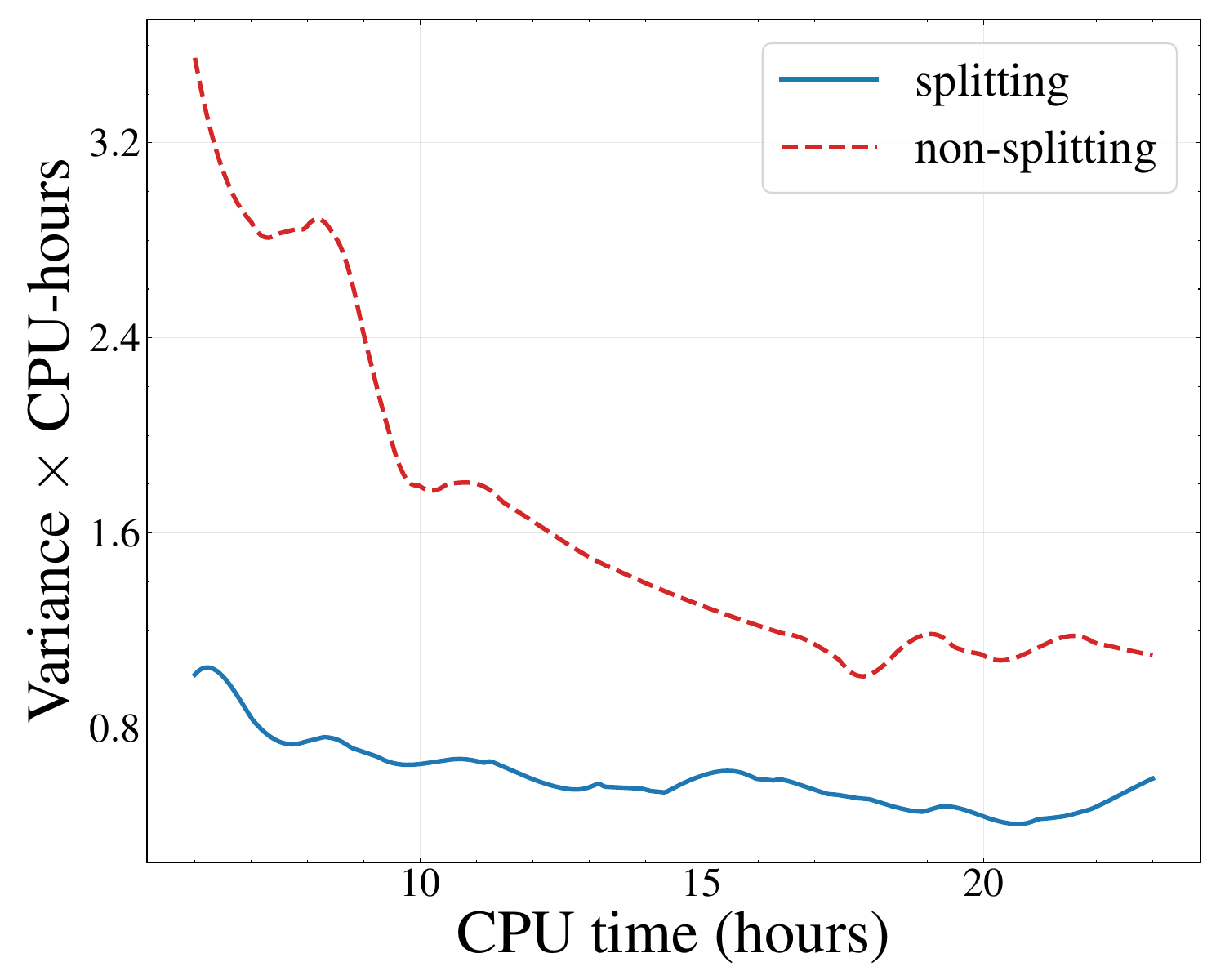}\hfill\includegraphics[width=0.19\textwidth]{figs/comparison/aflplusplus/php_php-fuzz-parser-2020-07-25__variance_times_cpu_hour.pdf}\hfill\includegraphics[width=0.19\textwidth]{figs/comparison/aflplusplus/poppler_pdf_fuzzer__variance_times_cpu_hour.pdf}\hfill\includegraphics[width=0.19\textwidth]{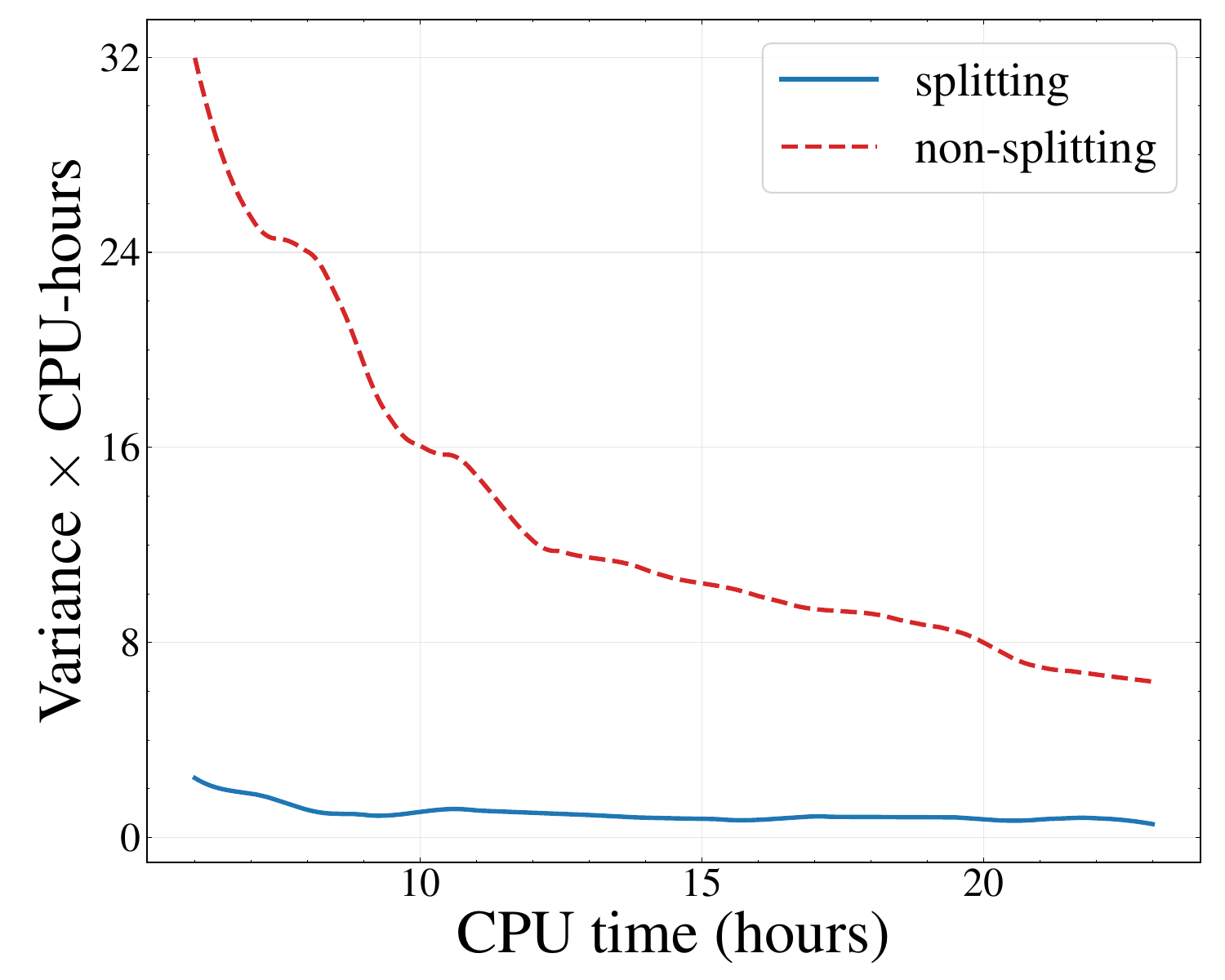}\\
    \makebox[0.19\textwidth]{\footnotesize (f) variance of matio}\hfill\makebox[0.19\textwidth]{\footnotesize (g) variance of openh264}\hfill\makebox[0.19\textwidth]{\footnotesize (h) variance of php}\hfill\makebox[0.19\textwidth]{\footnotesize (i) variance of poppler}\hfill\makebox[0.19\textwidth]{\footnotesize (j) variance of stb}
    \caption{Comparisons of variance across all 10 benchmarks for fuzzer aflplusplus (full version of Figure~\ref{fig:comparison_aflplusplus}).}
    \label{fig:comparison_aflplusplus_full}
\end{figure*}

\begin{figure*}[tp]
    \centering
    \includegraphics[width=0.19\textwidth]{figs/comparison/aflplusplus/arrow_parquet-arrow-fuzz__bug_detection_ratio.pdf}\hfill\includegraphics[width=0.19\textwidth]{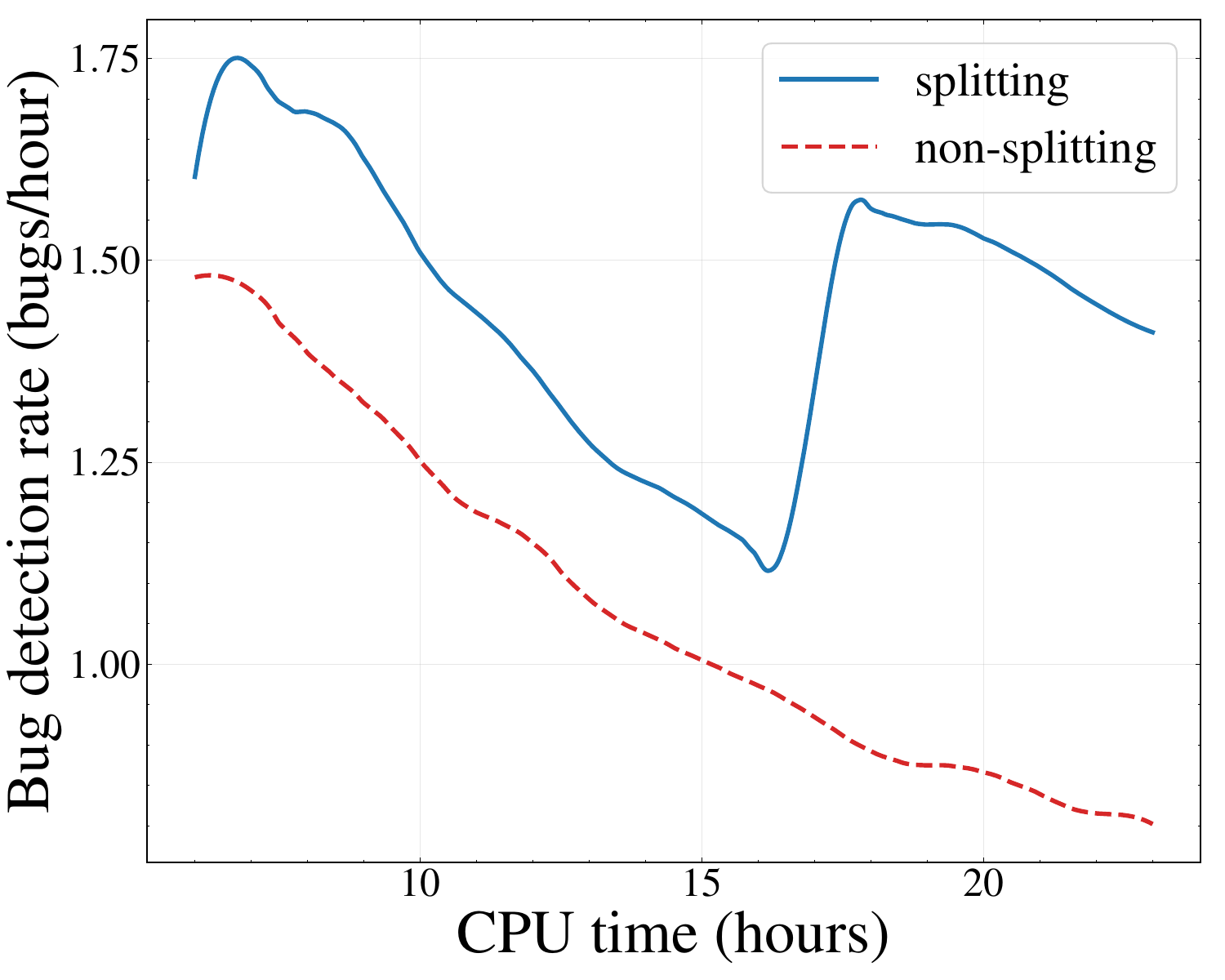}\hfill\includegraphics[width=0.19\textwidth]{figs/comparison/aflplusplus/grok_grk_decompress_fuzzer__bug_detection_ratio.pdf}\hfill\includegraphics[width=0.19\textwidth]{figs/comparison/aflplusplus/libhevc_hevc_dec_fuzzer__bug_detection_ratio.pdf}\hfill\includegraphics[width=0.19\textwidth]{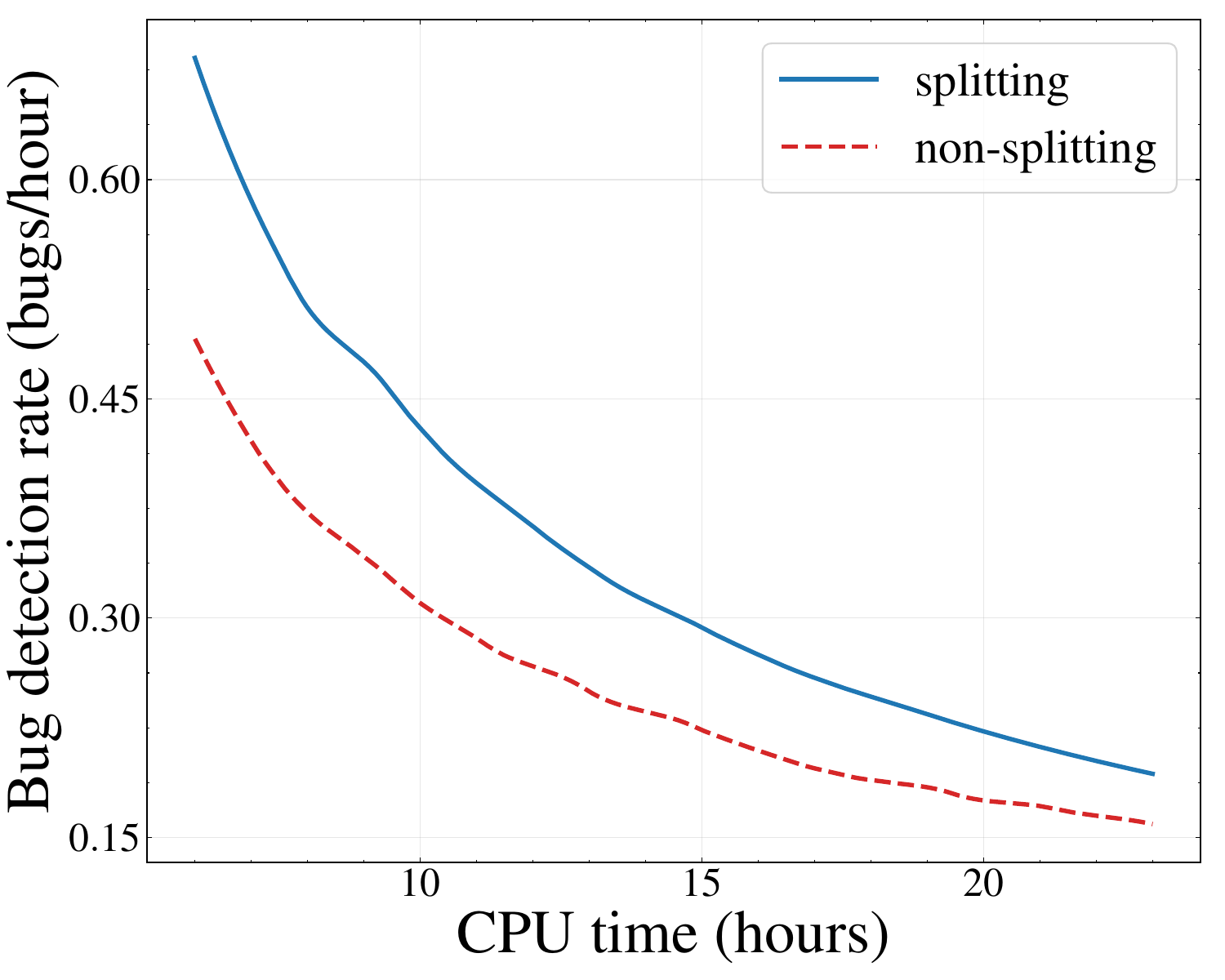}\\
    \makebox[0.19\textwidth]{\footnotesize (a) BDR of arrow}\hfill\makebox[0.19\textwidth]{\footnotesize (b) BDR of ffmpeg}\hfill\makebox[0.19\textwidth]{\footnotesize (c) BDR of grok}\hfill\makebox[0.19\textwidth]{\footnotesize (d) BDR of libhevc}\hfill\makebox[0.19\textwidth]{\footnotesize (e) BDR of libhtp}\\[3pt]
    \includegraphics[width=0.19\textwidth]{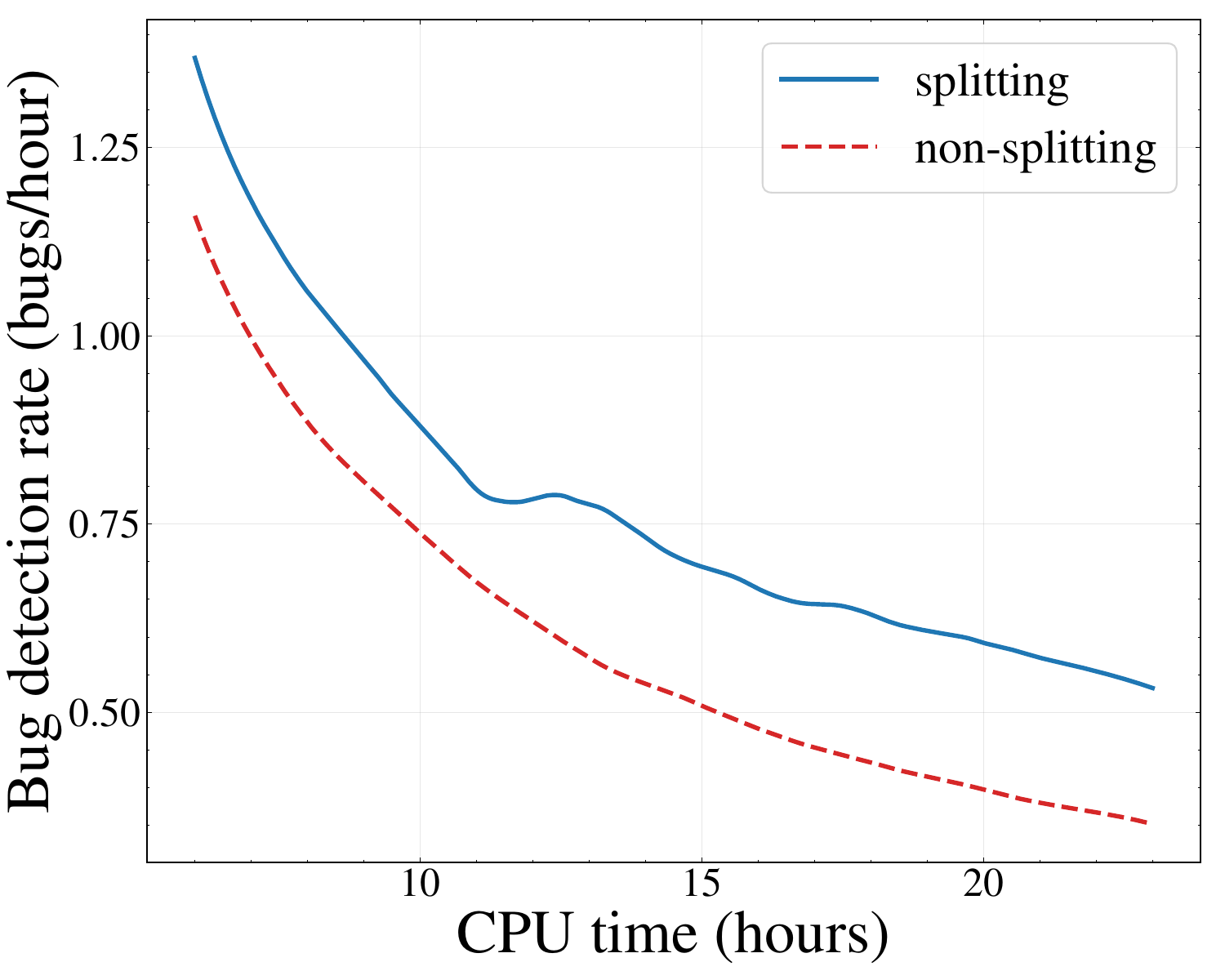}\hfill\includegraphics[width=0.19\textwidth]{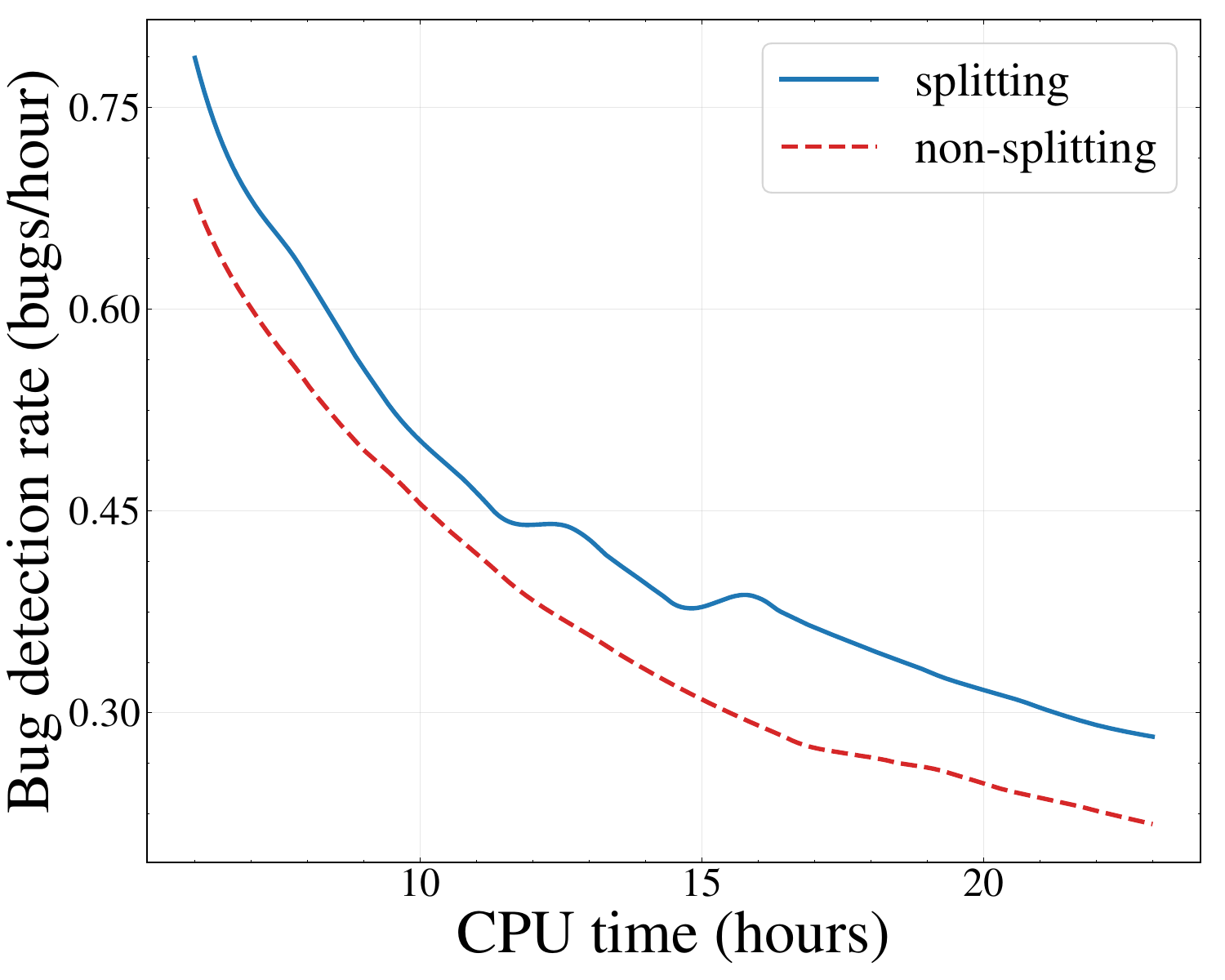}\hfill\includegraphics[width=0.19\textwidth]{figs/comparison/aflplusplus/php_php-fuzz-parser-2020-07-25__bug_detection_ratio.pdf}\hfill\includegraphics[width=0.19\textwidth]{figs/comparison/aflplusplus/poppler_pdf_fuzzer__bug_detection_ratio.pdf}\hfill\includegraphics[width=0.19\textwidth]{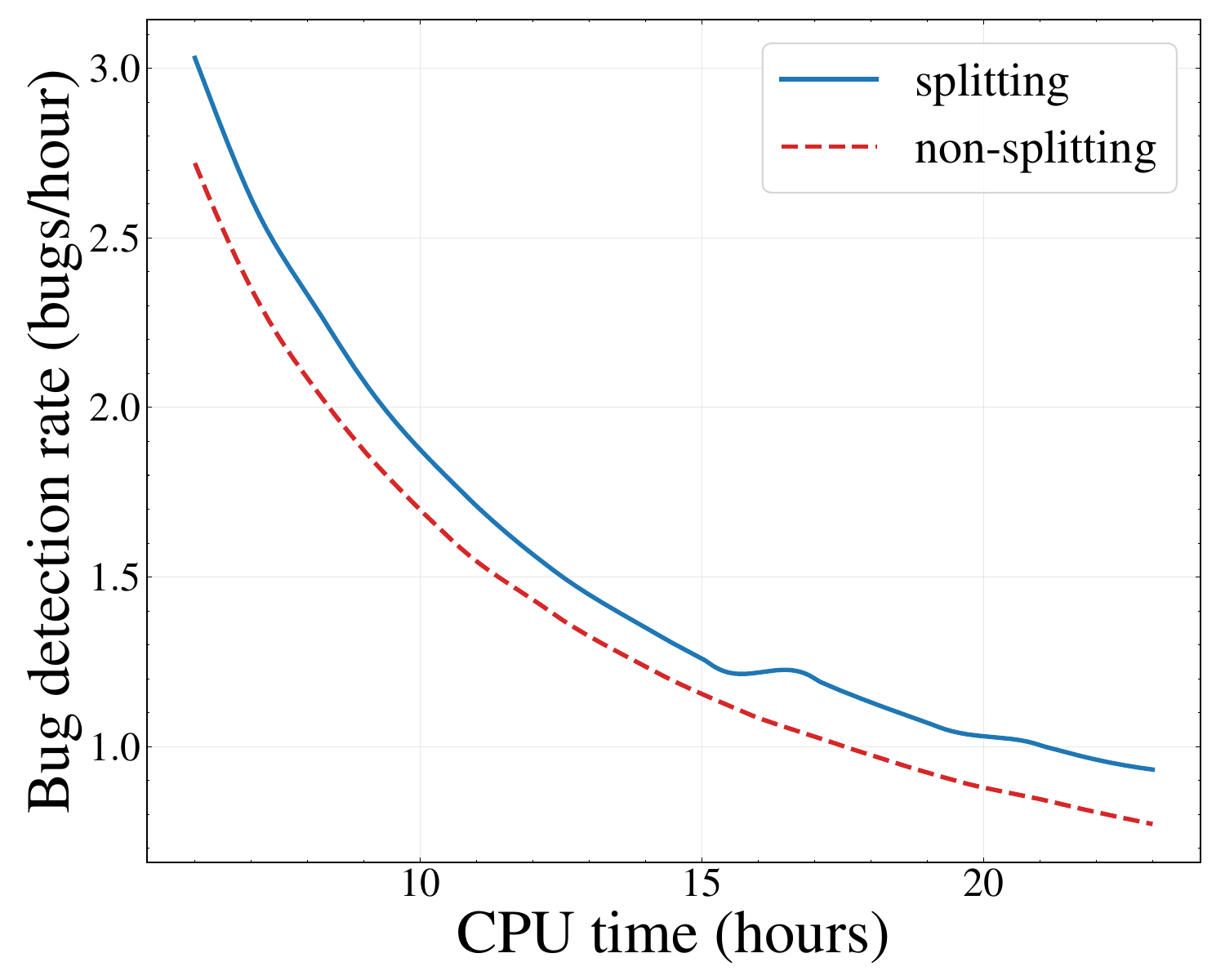}\\
    \makebox[0.19\textwidth]{\footnotesize (f) BDR of matio}\hfill\makebox[0.19\textwidth]{\footnotesize (g) BDR of openh264}\hfill\makebox[0.19\textwidth]{\footnotesize (h) BDR of php}\hfill\makebox[0.19\textwidth]{\footnotesize (i) BDR of poppler}\hfill\makebox[0.19\textwidth]{\footnotesize (j) BDR of stb}
    \caption{Comparisons of bug detection rate across all 10 benchmarks for fuzzer aflplusplus (full version of Figure~\ref{fig:bdr_aflplusplus}).}
    \label{fig:bdr_aflplusplus_full}
\end{figure*}

\begin{figure*}[tp]
    \centering
    \includegraphics[width=0.19\textwidth]{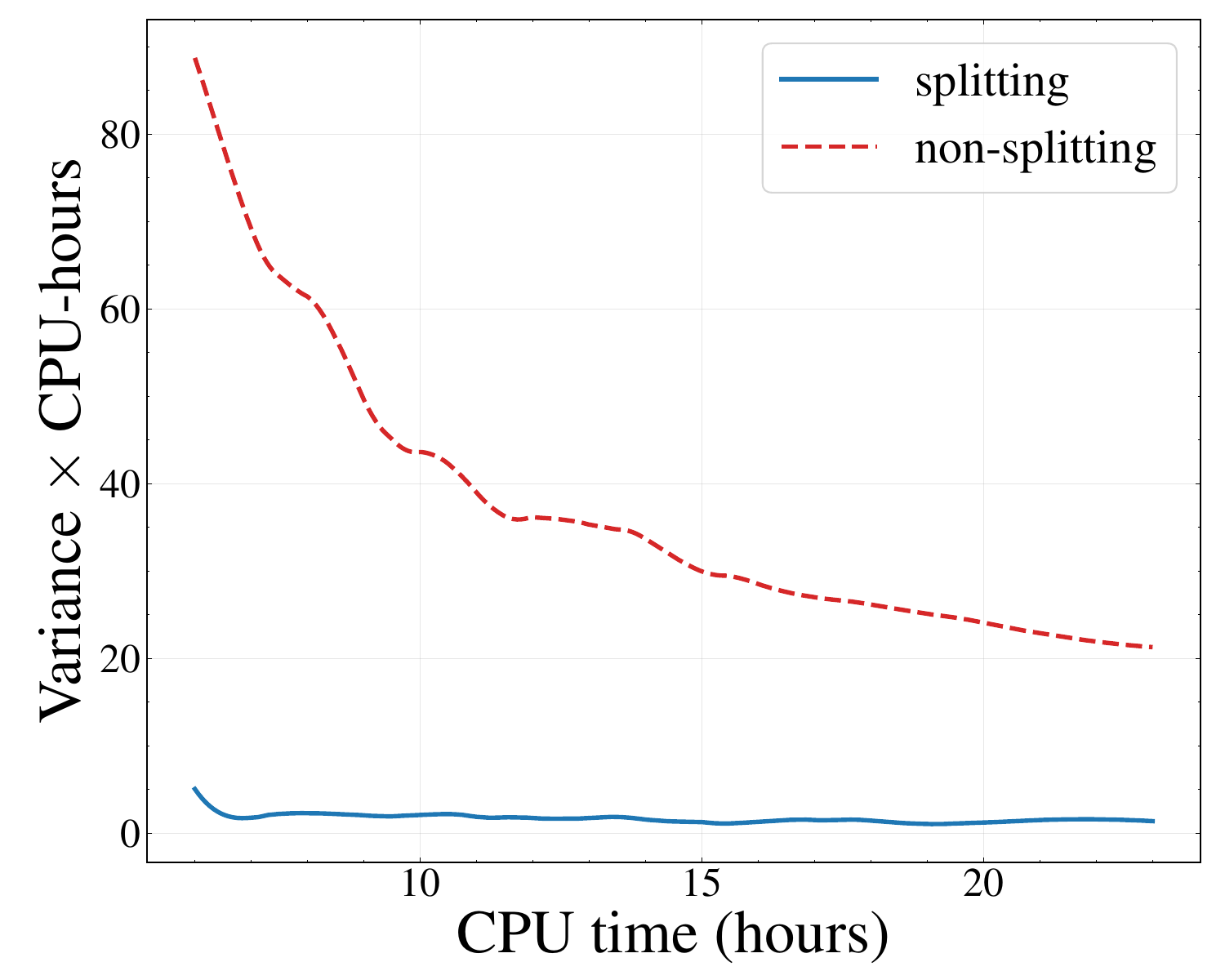}\hfill\includegraphics[width=0.19\textwidth]{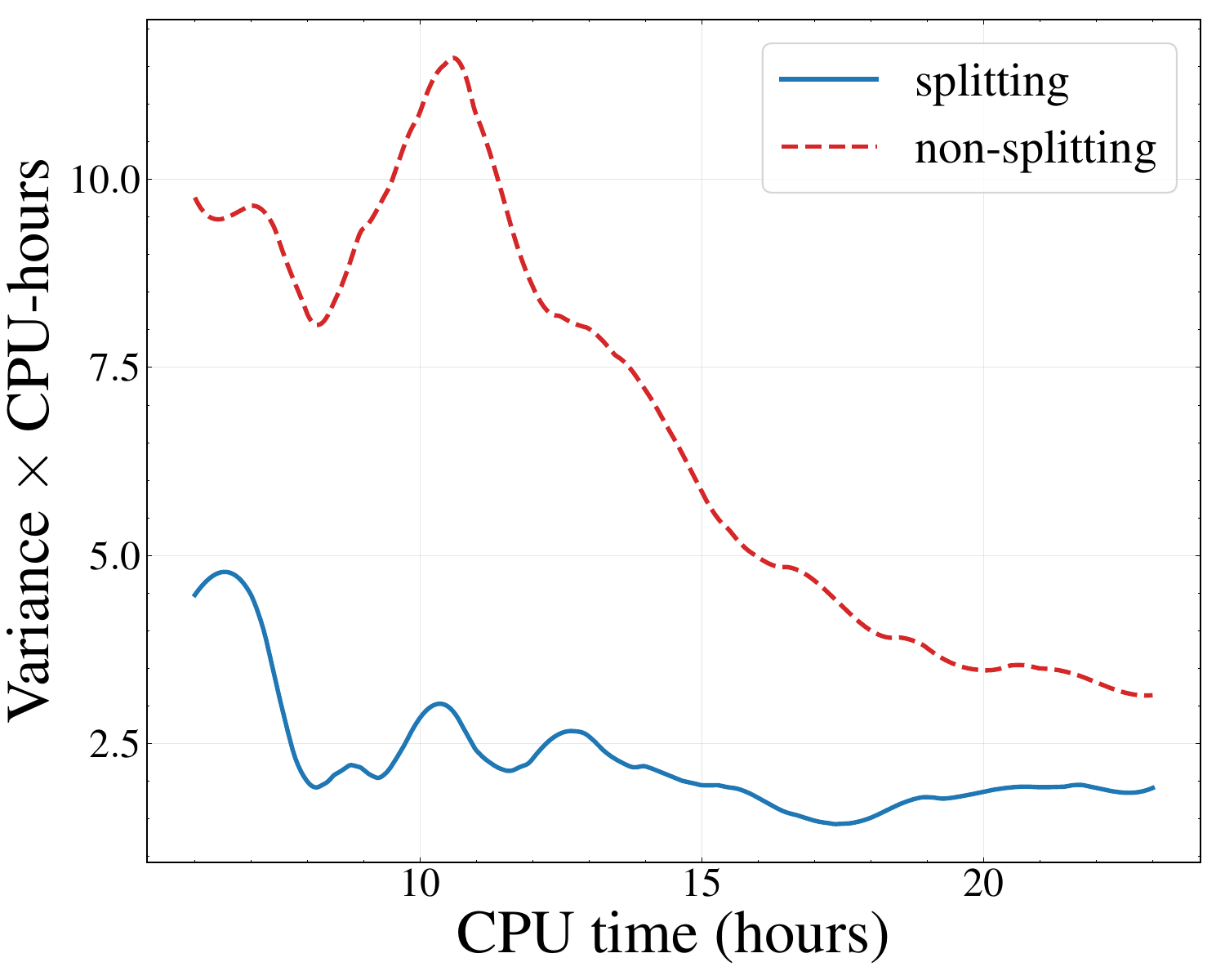}\hfill\includegraphics[width=0.19\textwidth]{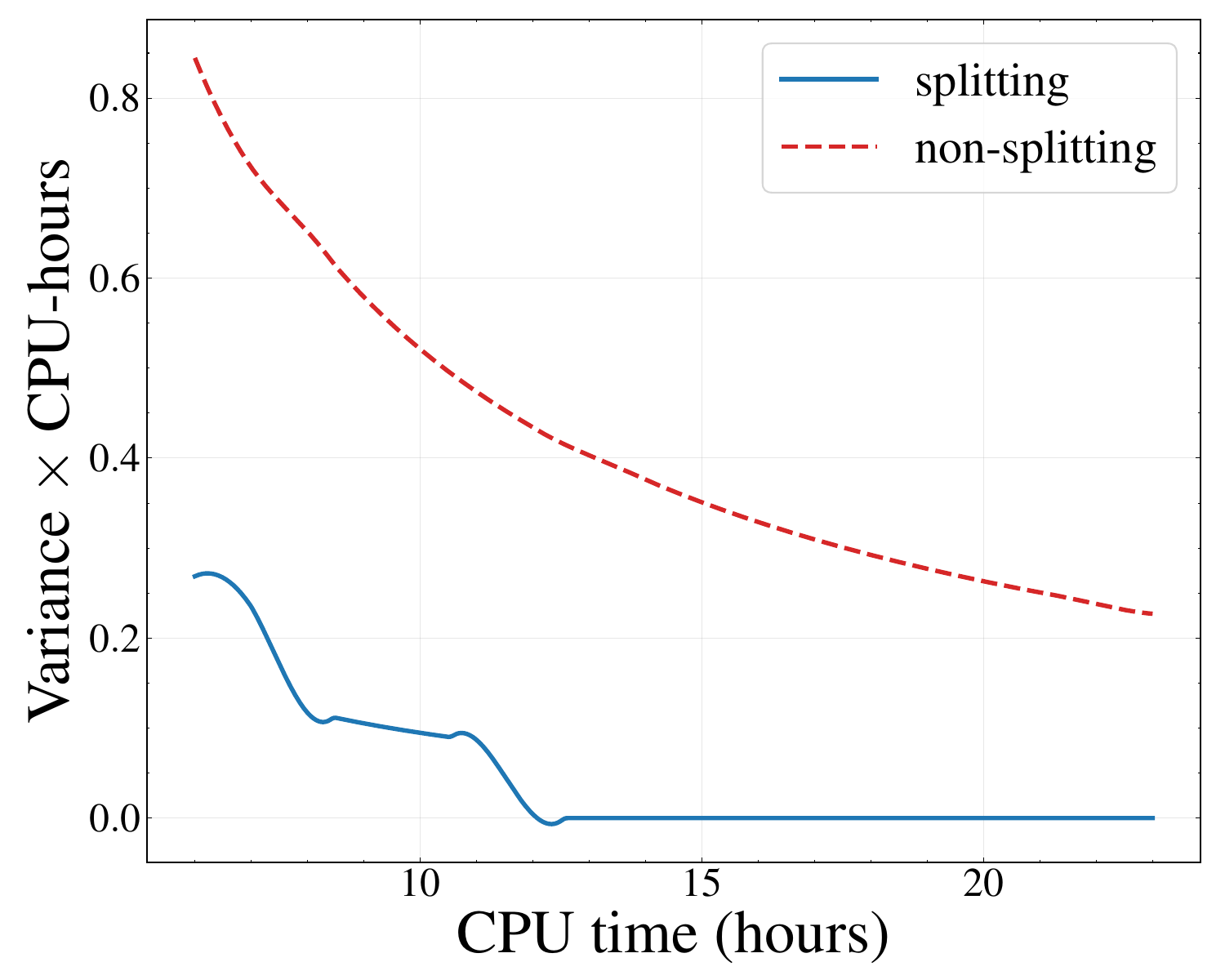}\hfill\includegraphics[width=0.19\textwidth]{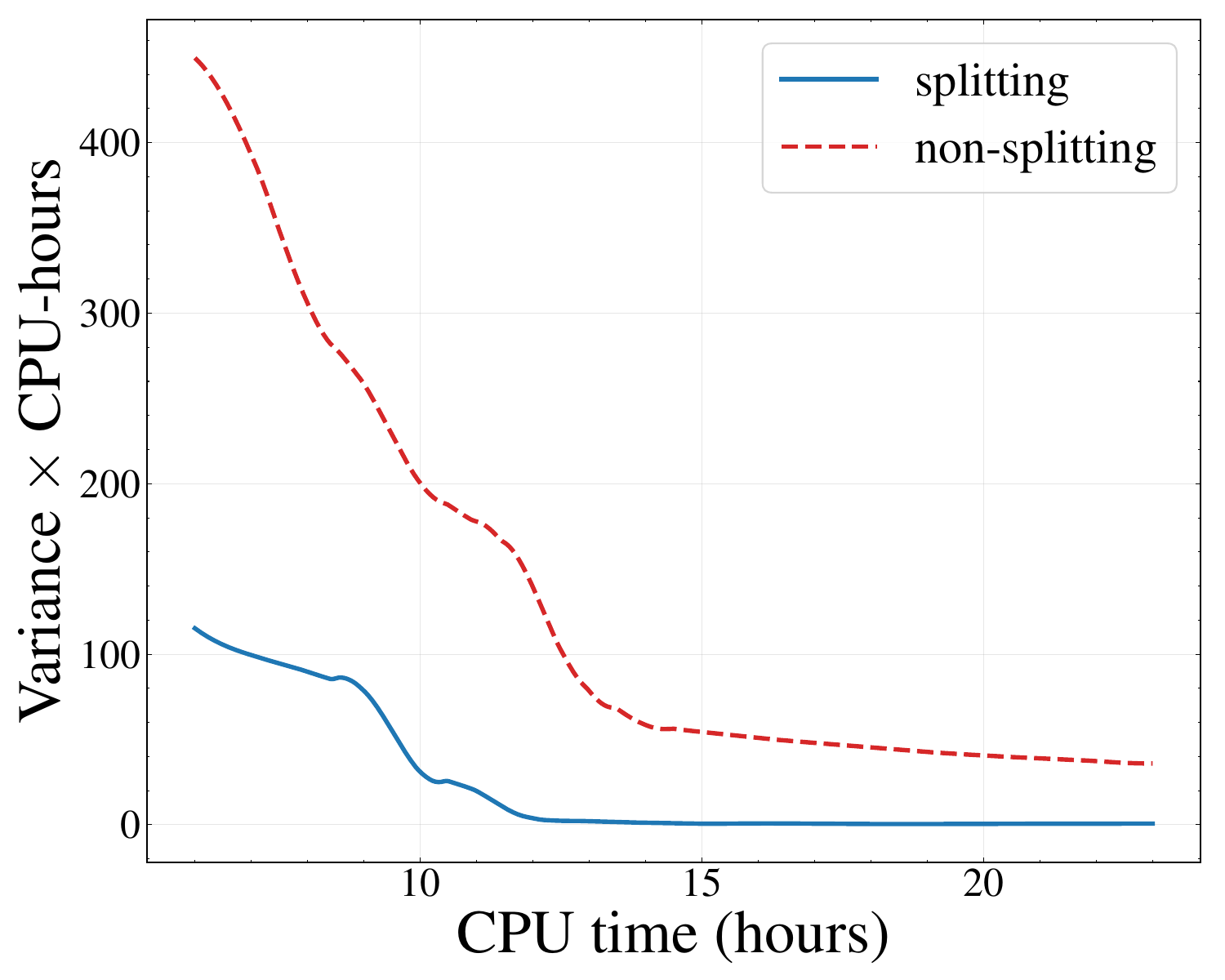}\hfill\includegraphics[width=0.19\textwidth]{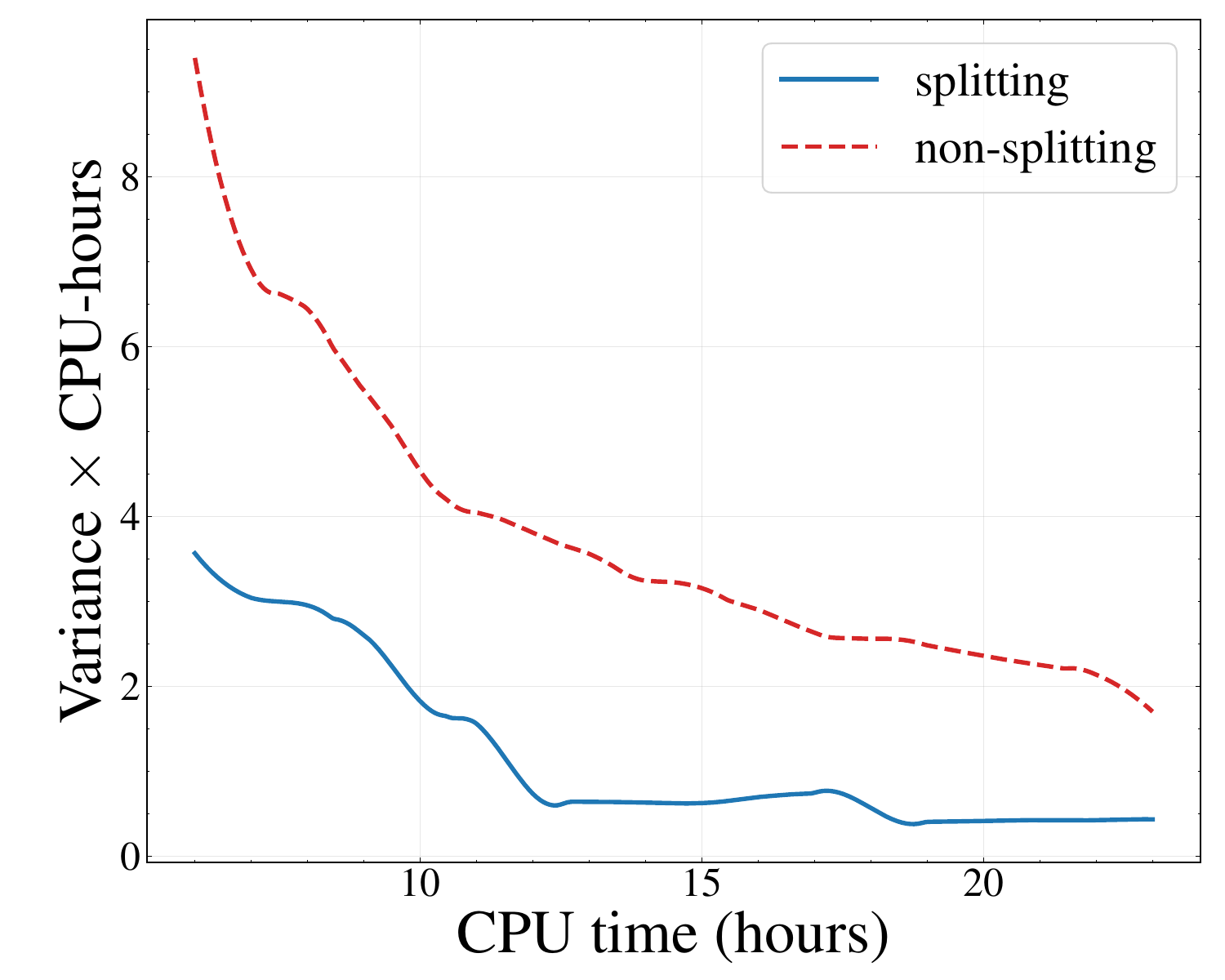}\\
    \makebox[0.19\textwidth]{\footnotesize (a) variance of arrow}\hfill\makebox[0.19\textwidth]{\footnotesize (b) variance of ffmpeg}\hfill\makebox[0.19\textwidth]{\footnotesize (c) variance of grok}\hfill\makebox[0.19\textwidth]{\footnotesize (d) variance of libhevc}\hfill\makebox[0.19\textwidth]{\footnotesize (e) variance of libhtp}\\[3pt]
    \includegraphics[width=0.19\textwidth]{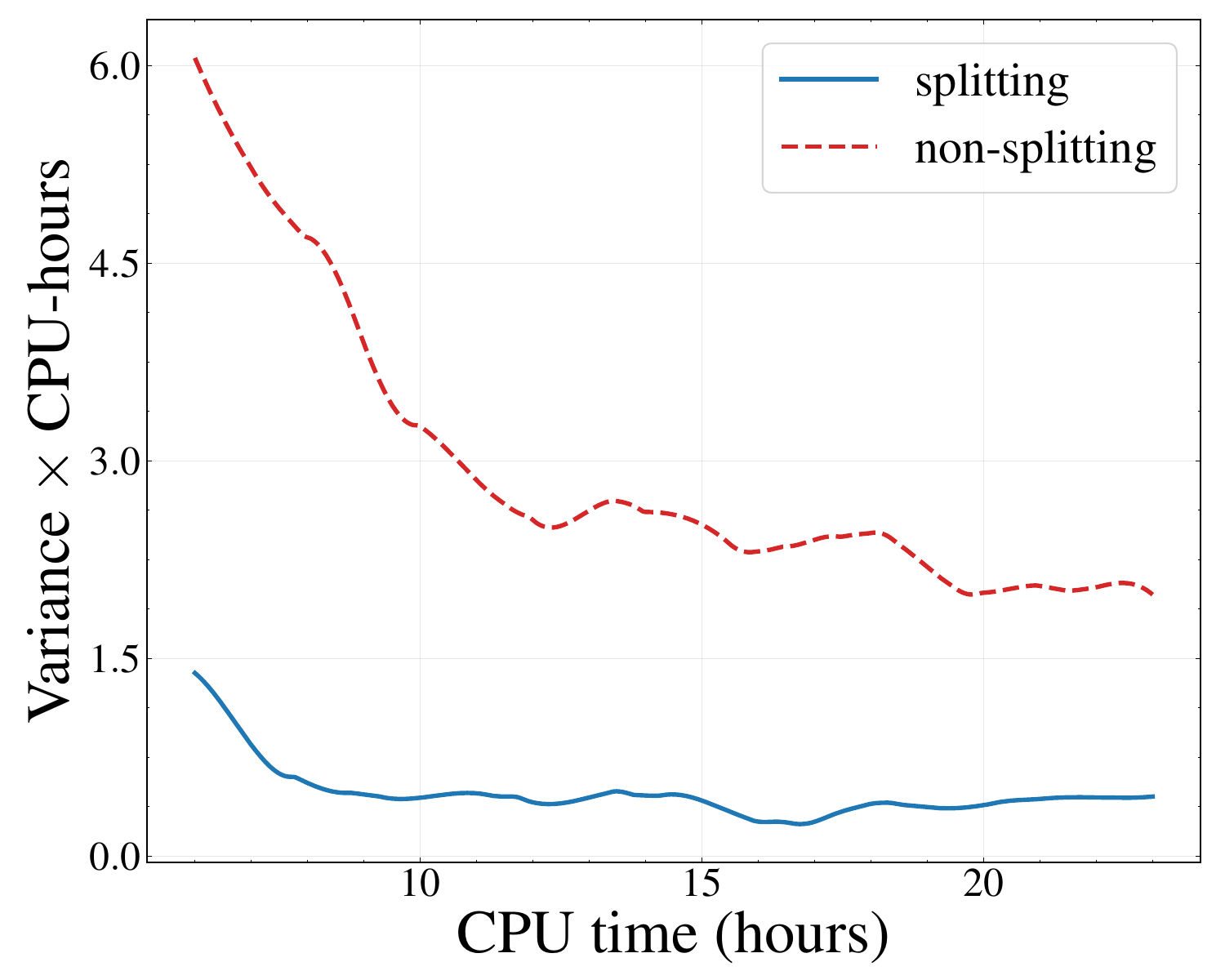}\hfill\includegraphics[width=0.19\textwidth]{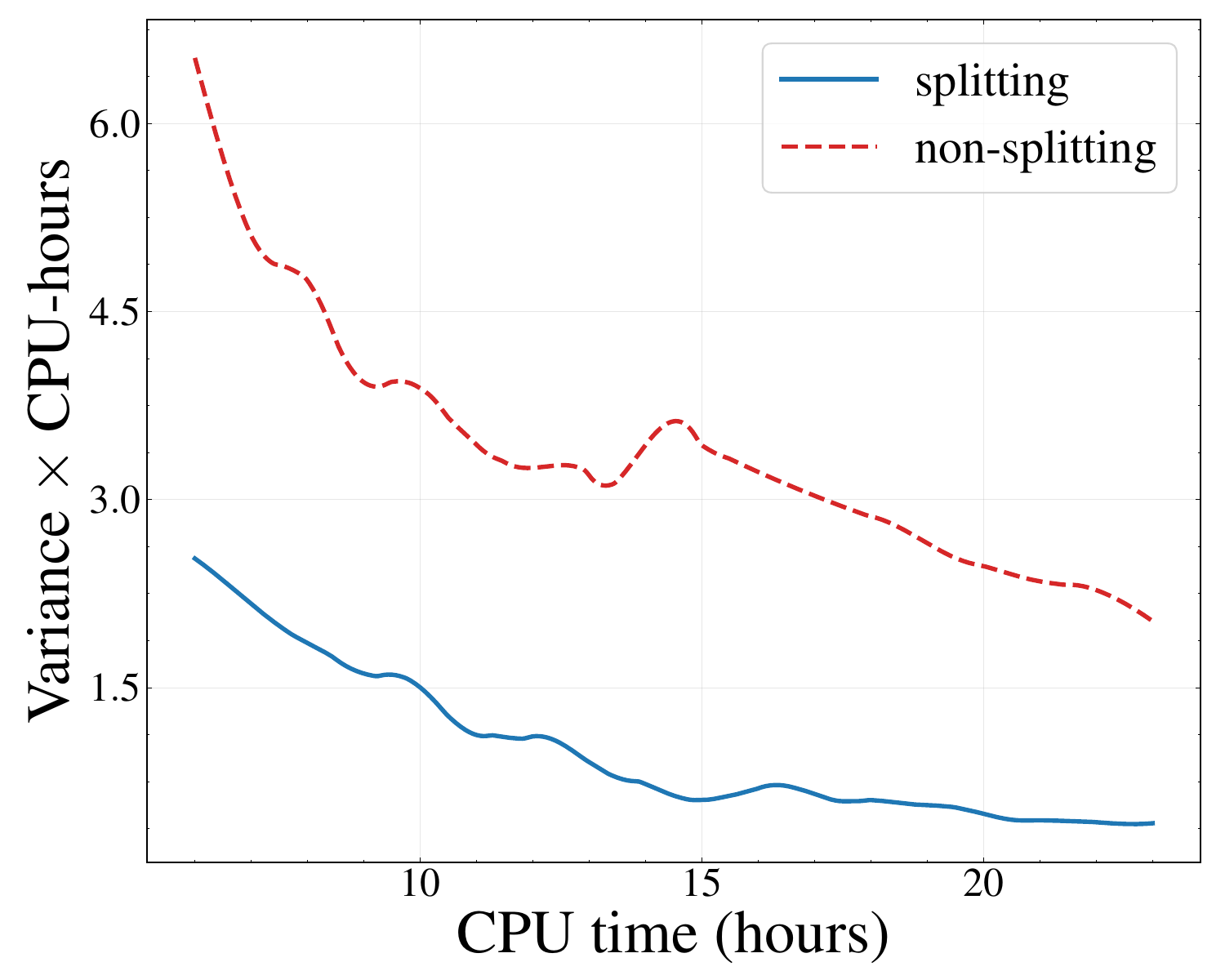}\hfill\includegraphics[width=0.19\textwidth]{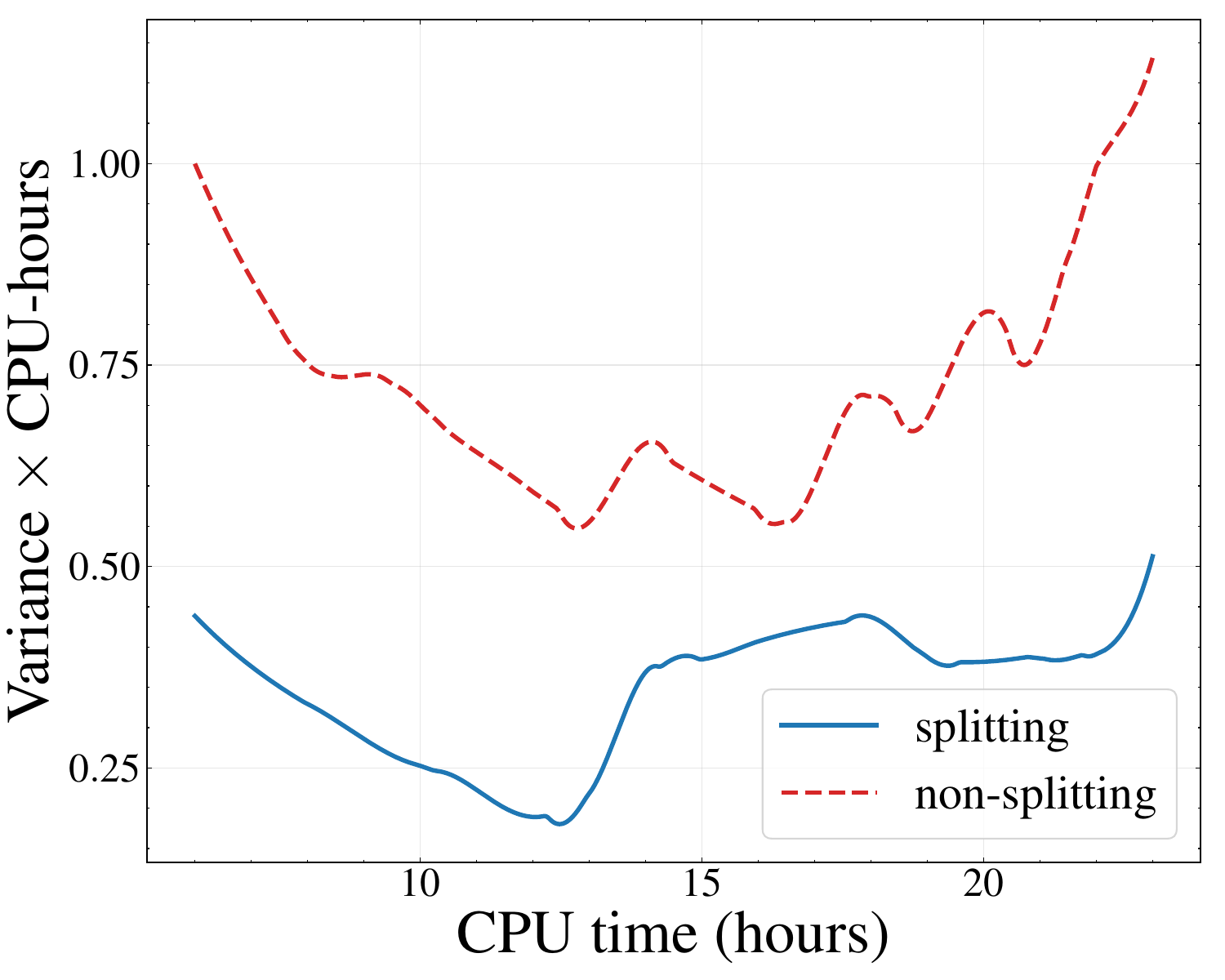}\hfill\includegraphics[width=0.19\textwidth]{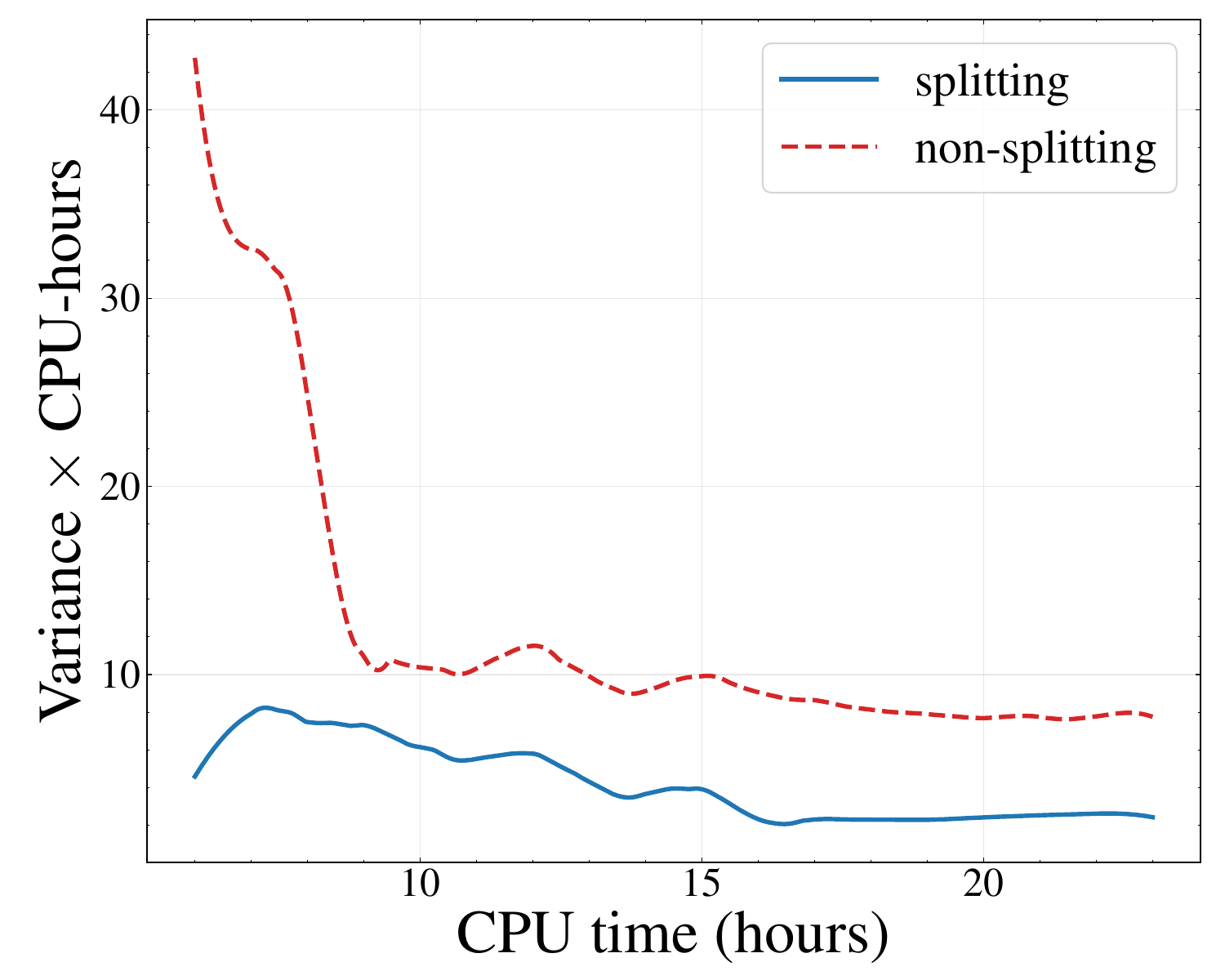}\hfill\includegraphics[width=0.19\textwidth]{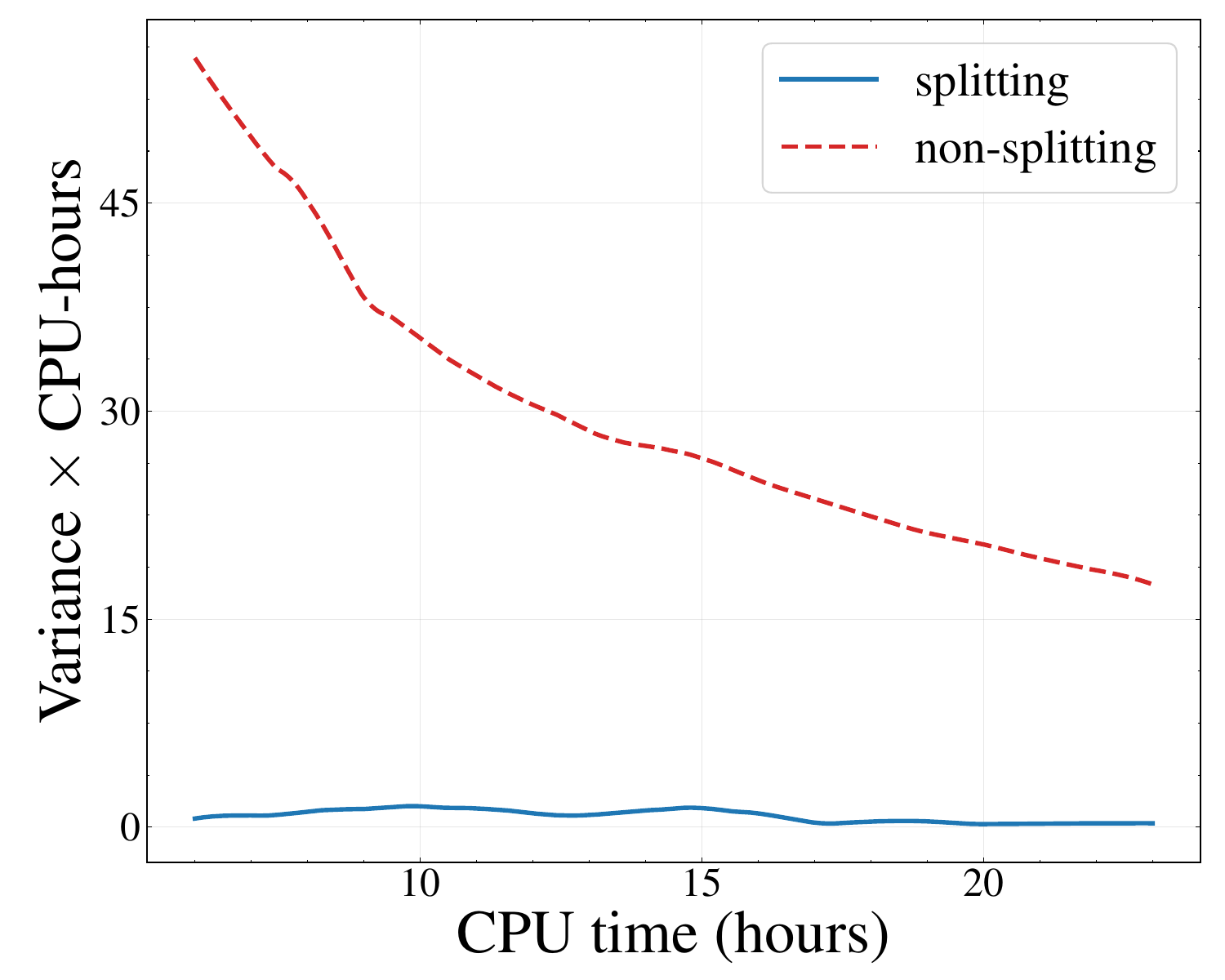}\\
    \makebox[0.19\textwidth]{\footnotesize (f) variance of matio}\hfill\makebox[0.19\textwidth]{\footnotesize (g) variance of openh264}\hfill\makebox[0.19\textwidth]{\footnotesize (h) variance of php}\hfill\makebox[0.19\textwidth]{\footnotesize (i) variance of poppler}\hfill\makebox[0.19\textwidth]{\footnotesize (j) variance of stb}
    \caption{Comparisons of variance across 10 benchmarks for fuzzer MOpt.}
    \label{fig:comparison_mopt}
\end{figure*}

\begin{figure*}[tp]
    \centering
    \includegraphics[width=0.19\textwidth]{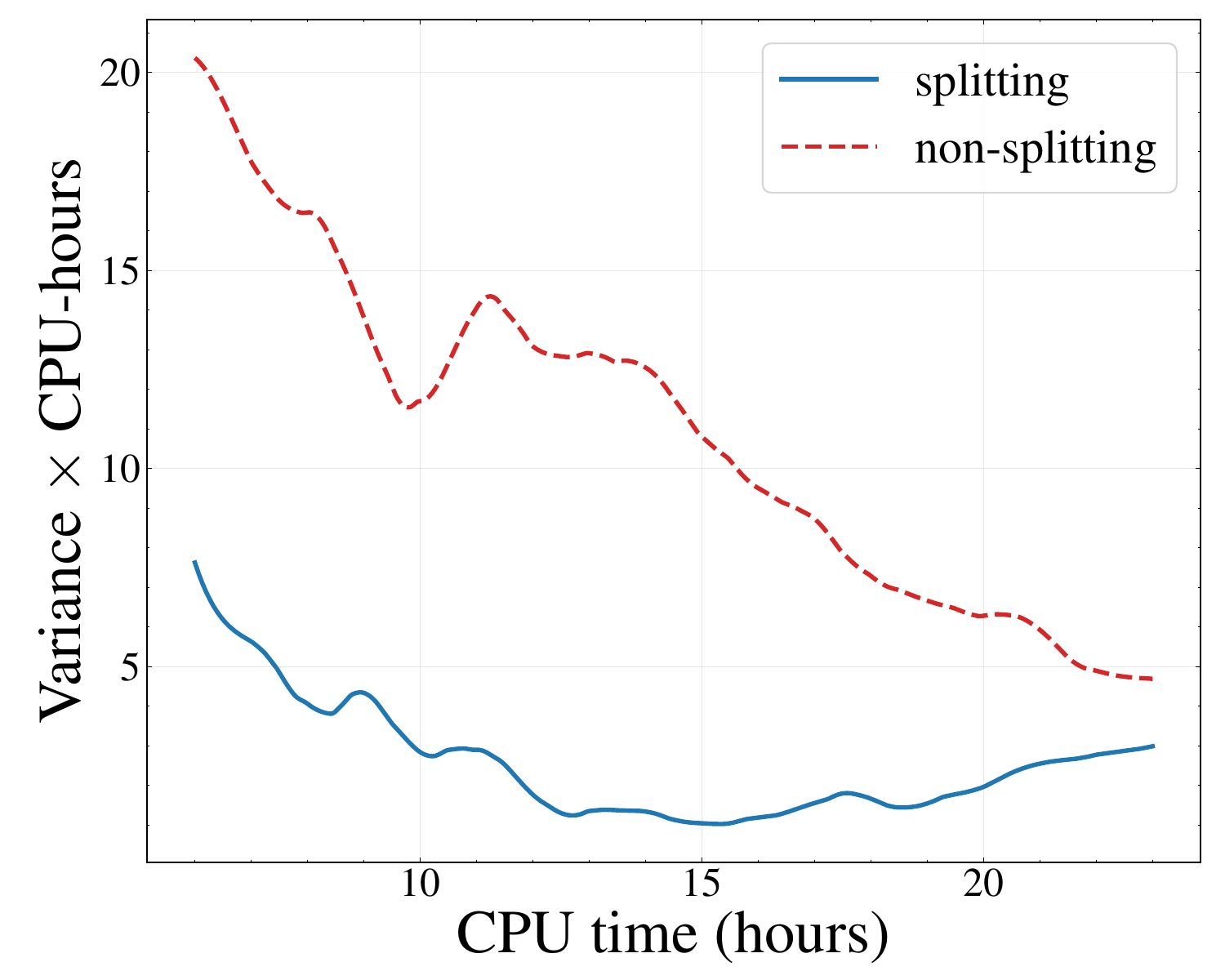}\hfill\includegraphics[width=0.19\textwidth]{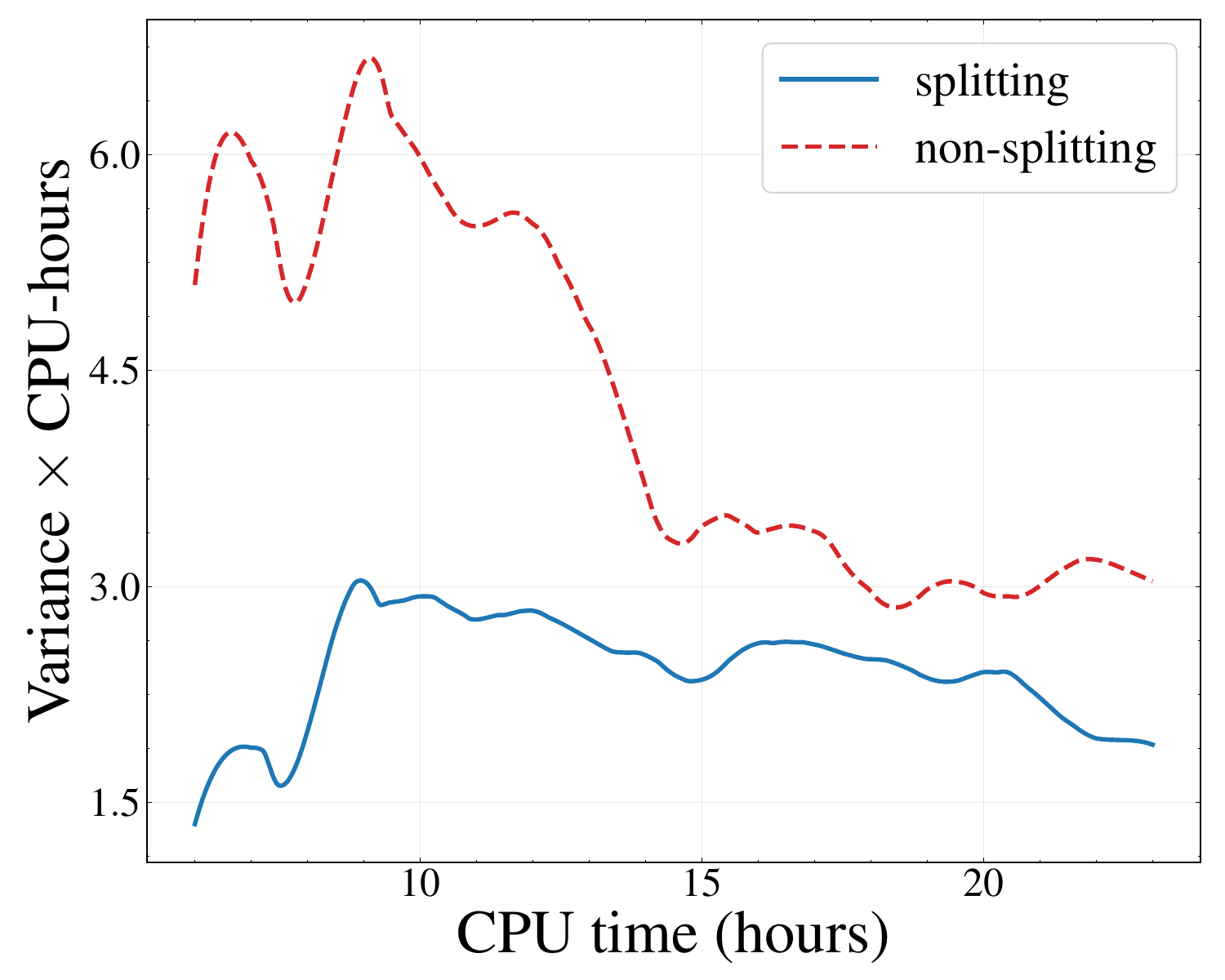}\hfill\includegraphics[width=0.19\textwidth]{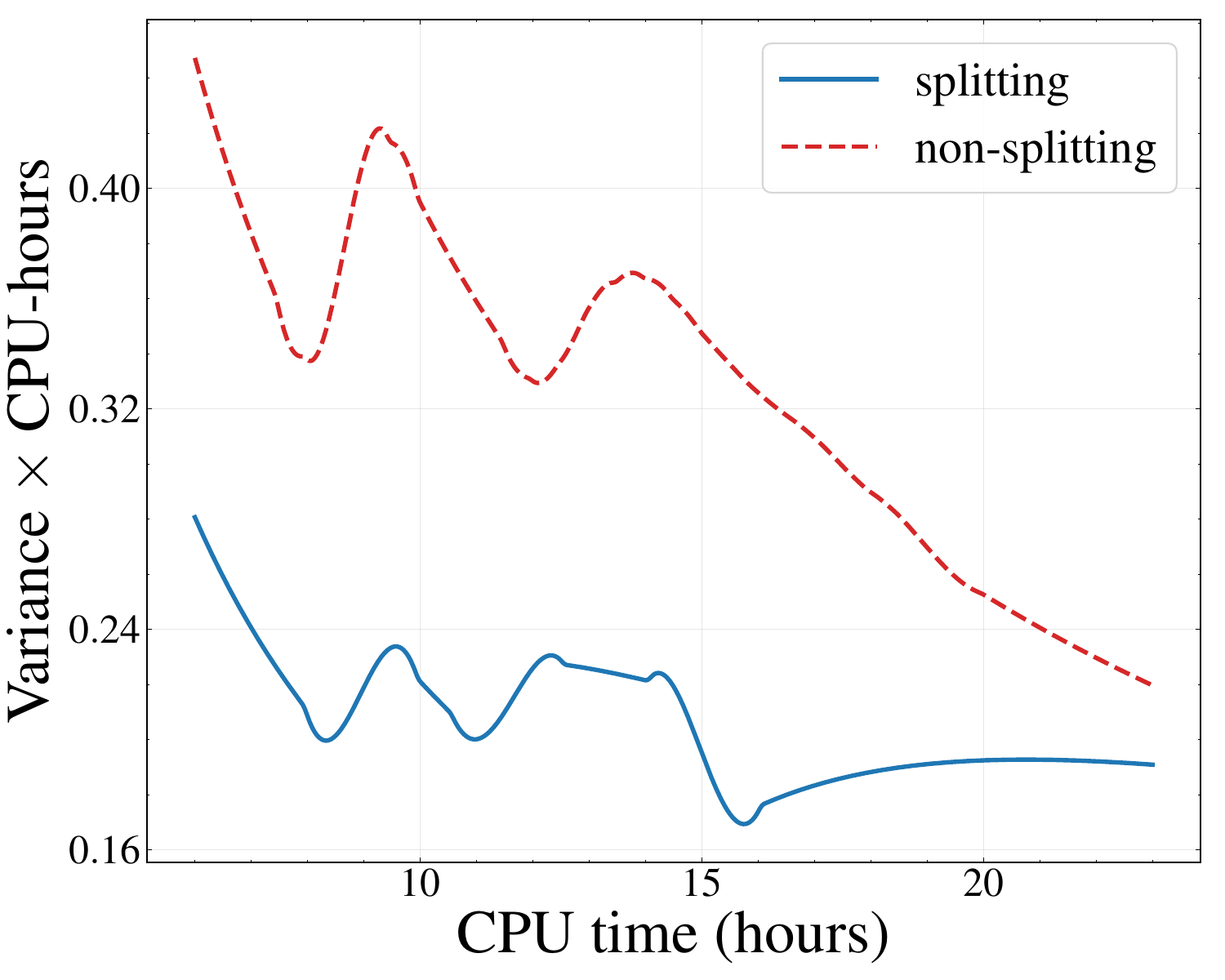}\hfill\includegraphics[width=0.19\textwidth]{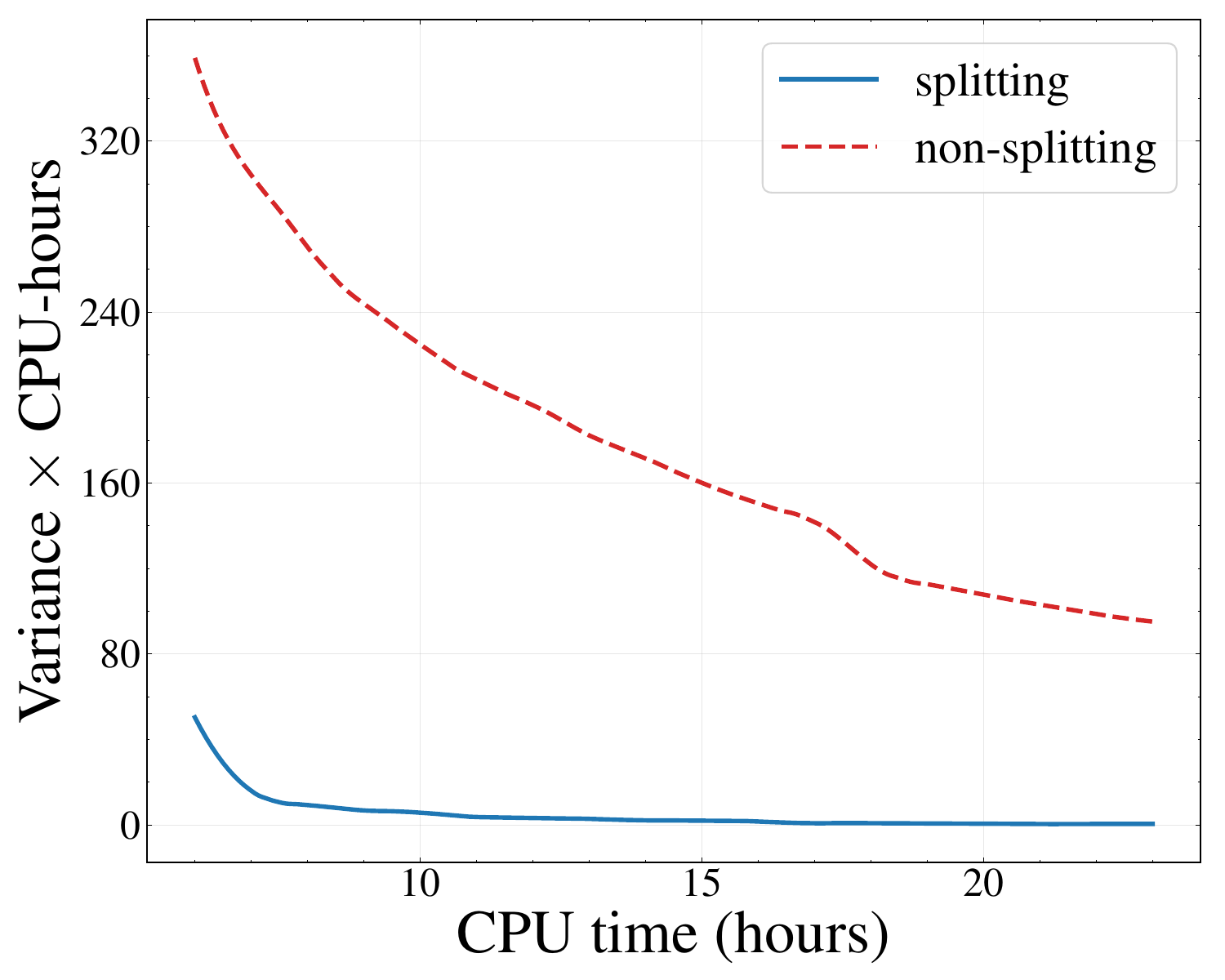}\hfill\includegraphics[width=0.19\textwidth]{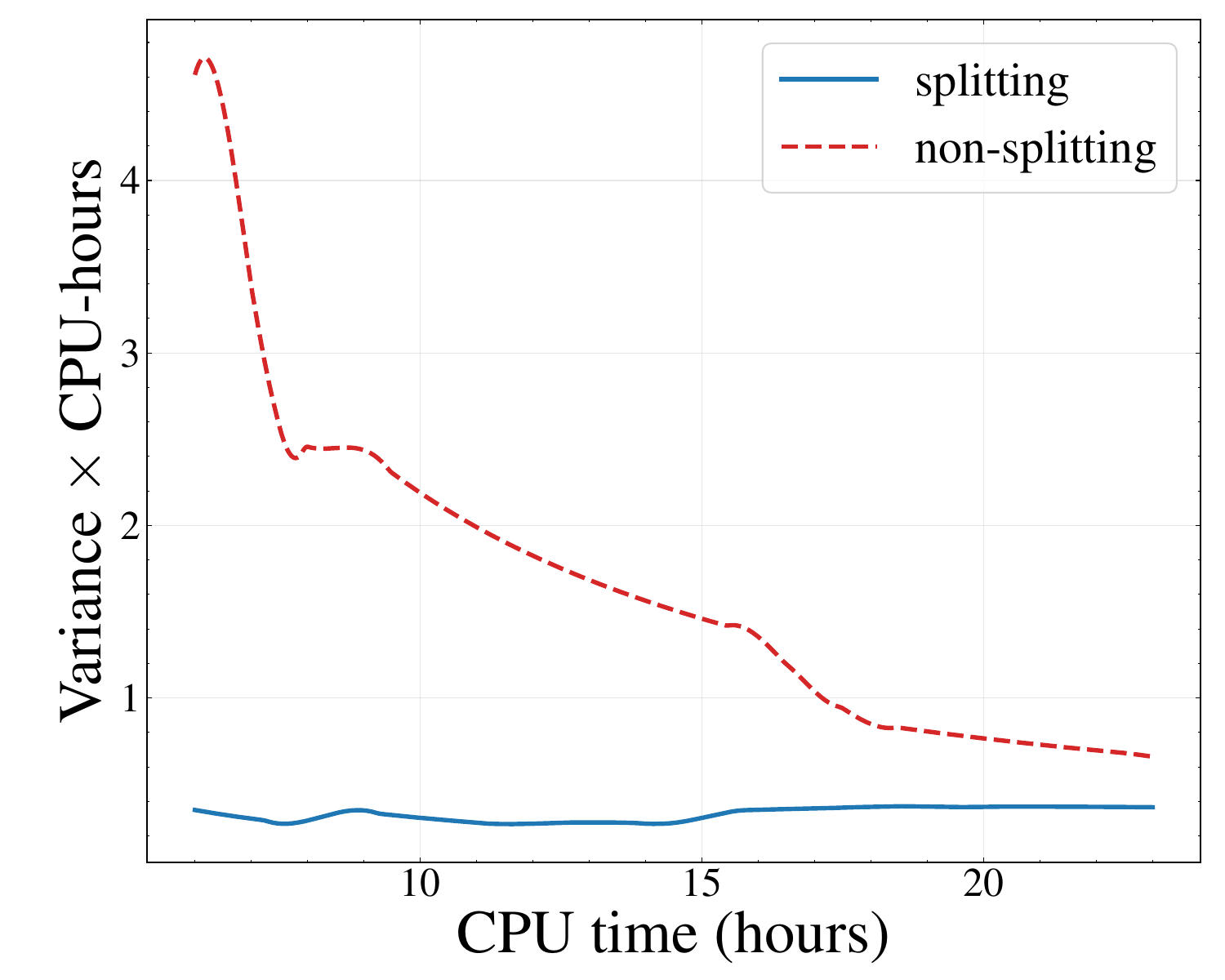}\\
    \makebox[0.19\textwidth]{\footnotesize (a) variance of arrow}\hfill\makebox[0.19\textwidth]{\footnotesize (b) variance of ffmpeg}\hfill\makebox[0.19\textwidth]{\footnotesize (c) variance of grok}\hfill\makebox[0.19\textwidth]{\footnotesize (d) variance of libhevc}\hfill\makebox[0.19\textwidth]{\footnotesize (e) variance of libhtp}\\[3pt]
    \includegraphics[width=0.19\textwidth]{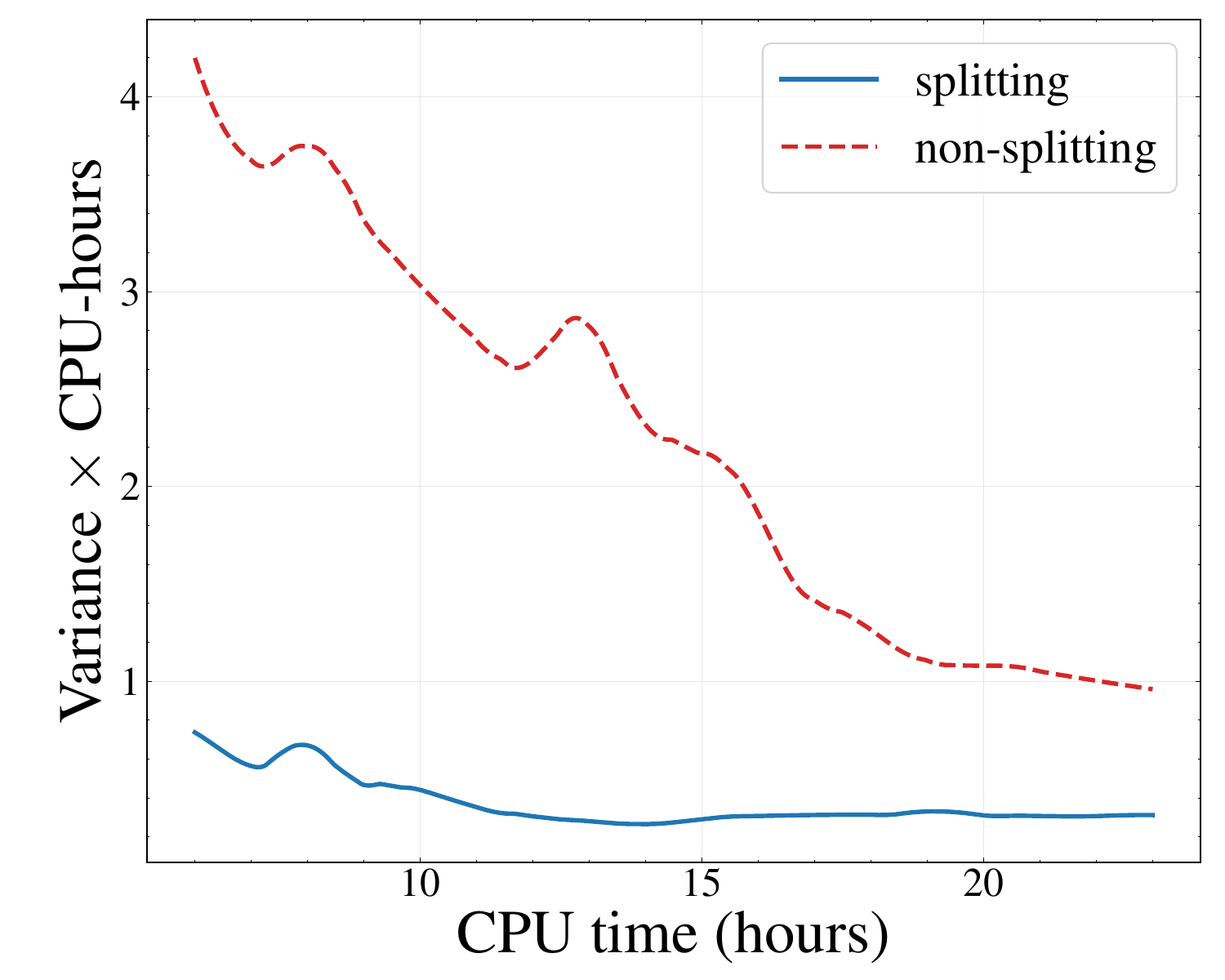}\hfill\includegraphics[width=0.19\textwidth]{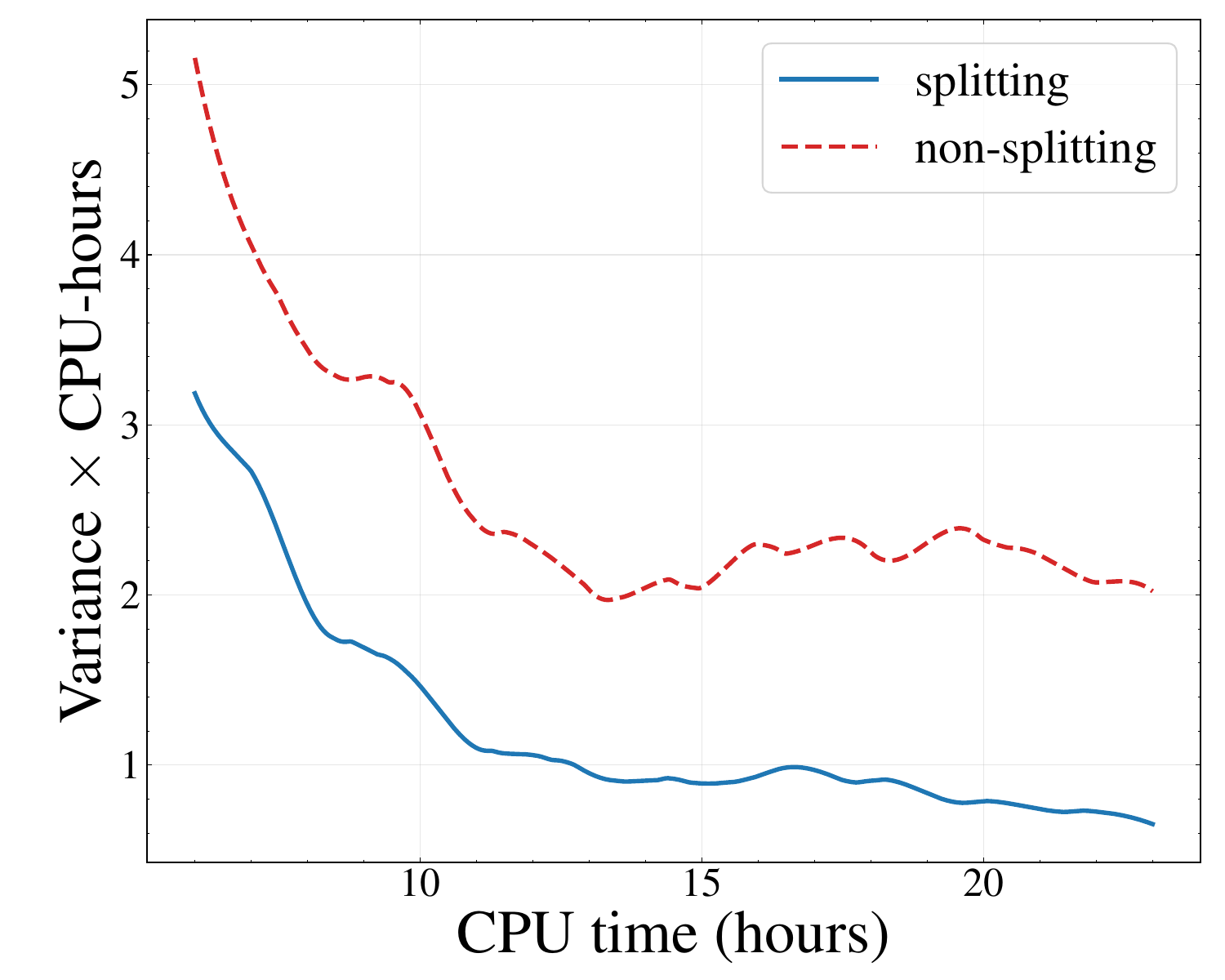}\hfill\includegraphics[width=0.19\textwidth]{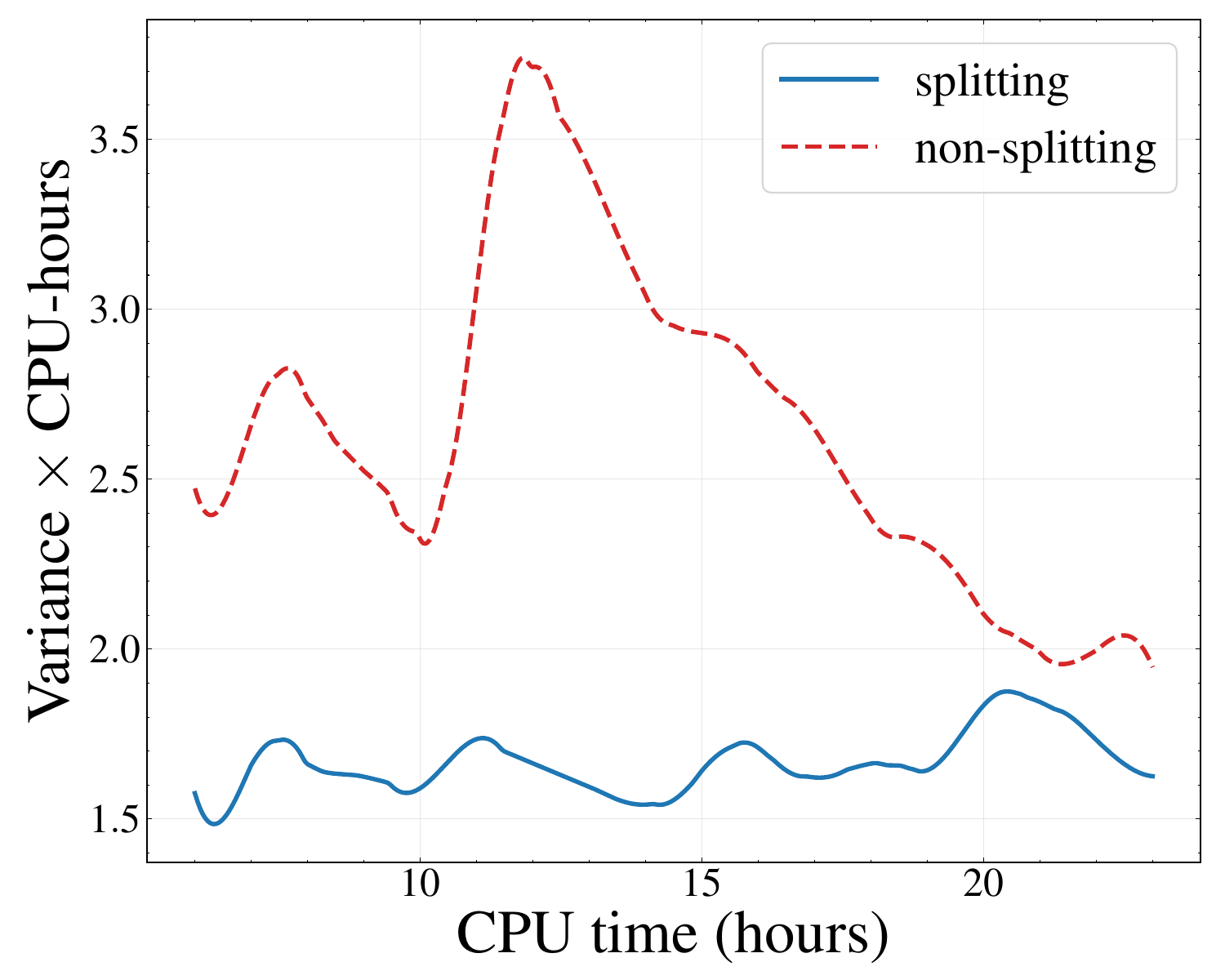}\hfill\includegraphics[width=0.19\textwidth]{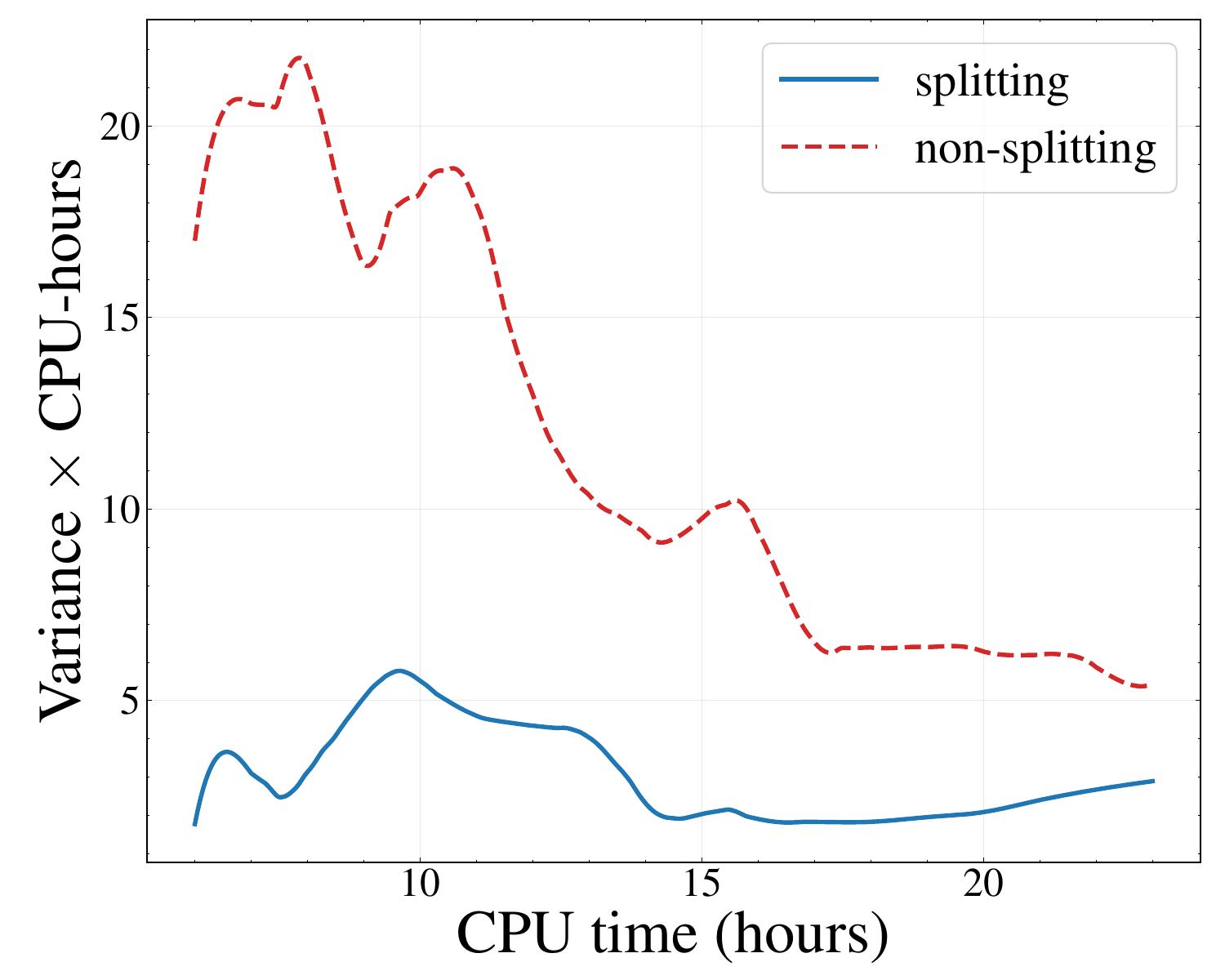}\hfill\includegraphics[width=0.19\textwidth]{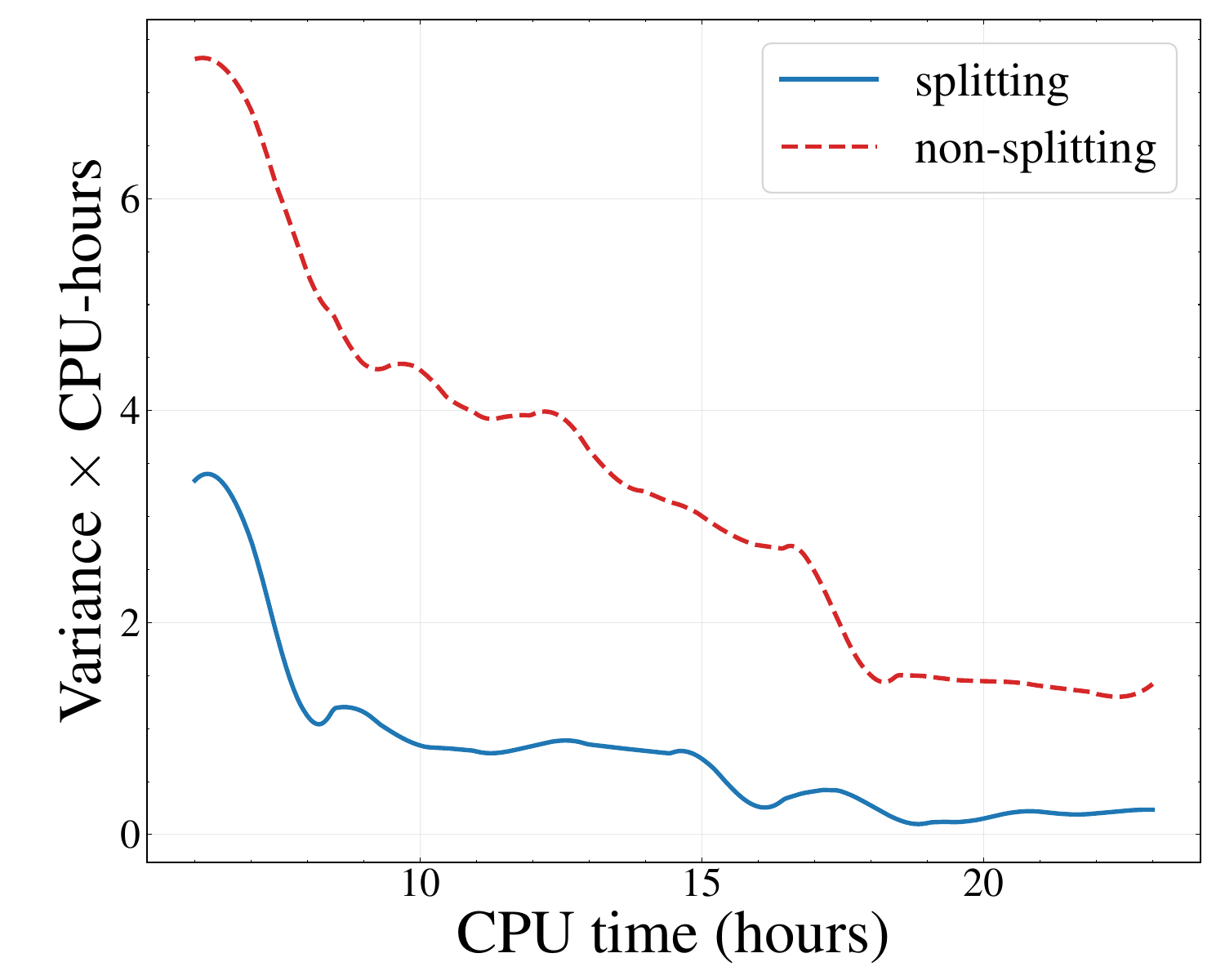}\\
    \makebox[0.19\textwidth]{\footnotesize (f) variance of matio}\hfill\makebox[0.19\textwidth]{\footnotesize (g) variance of openh264}\hfill\makebox[0.19\textwidth]{\footnotesize (h) variance of php}\hfill\makebox[0.19\textwidth]{\footnotesize (i) variance of poppler}\hfill\makebox[0.19\textwidth]{\footnotesize (j) variance of stb}
    \caption{Comparisons of variance across 10 benchmarks for fuzzer AFL.}
    \label{fig:comparison_afl}
\end{figure*}

\begin{figure*}[tp]
    \centering
    \includegraphics[width=0.19\textwidth]{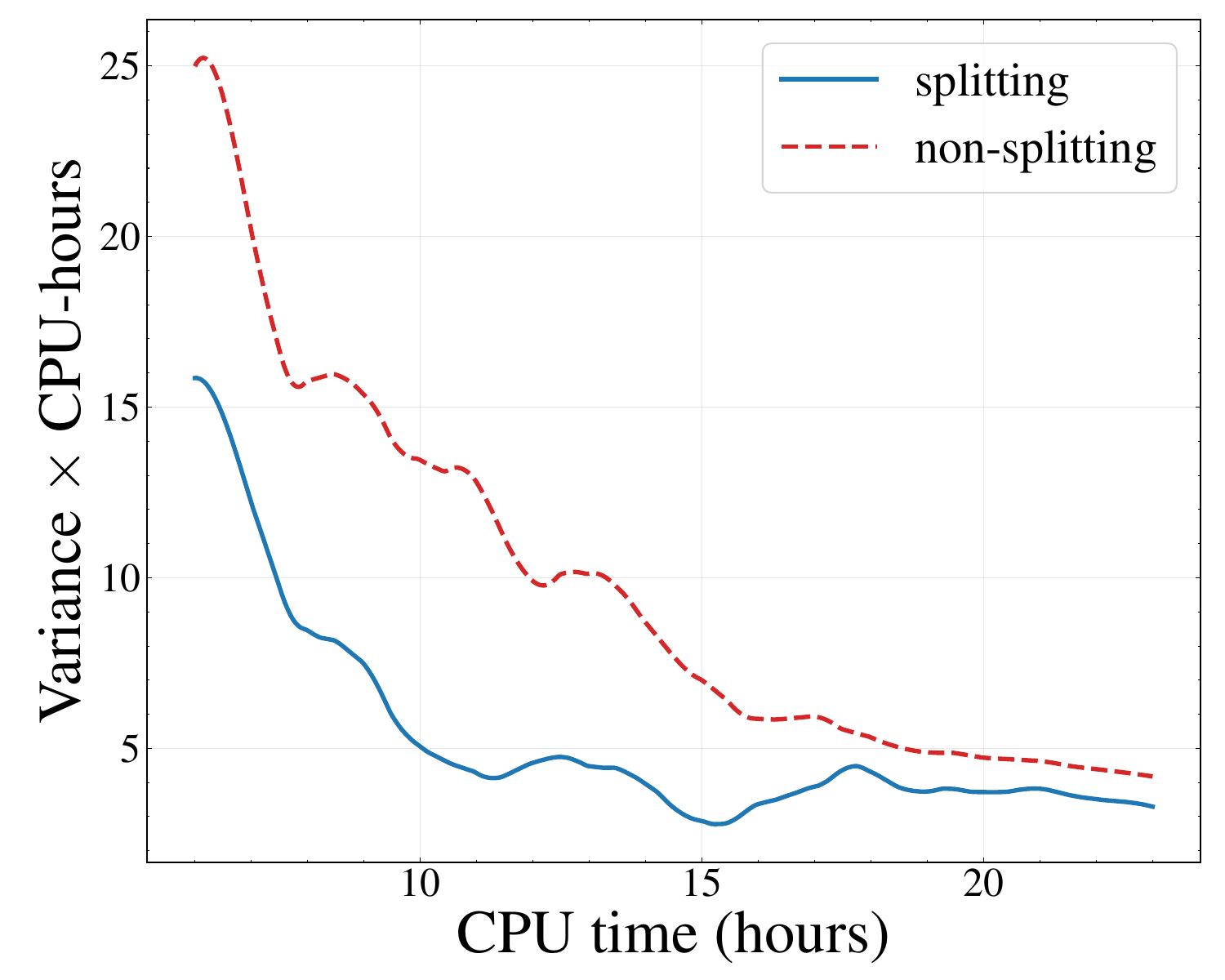}\hfill\includegraphics[width=0.19\textwidth]{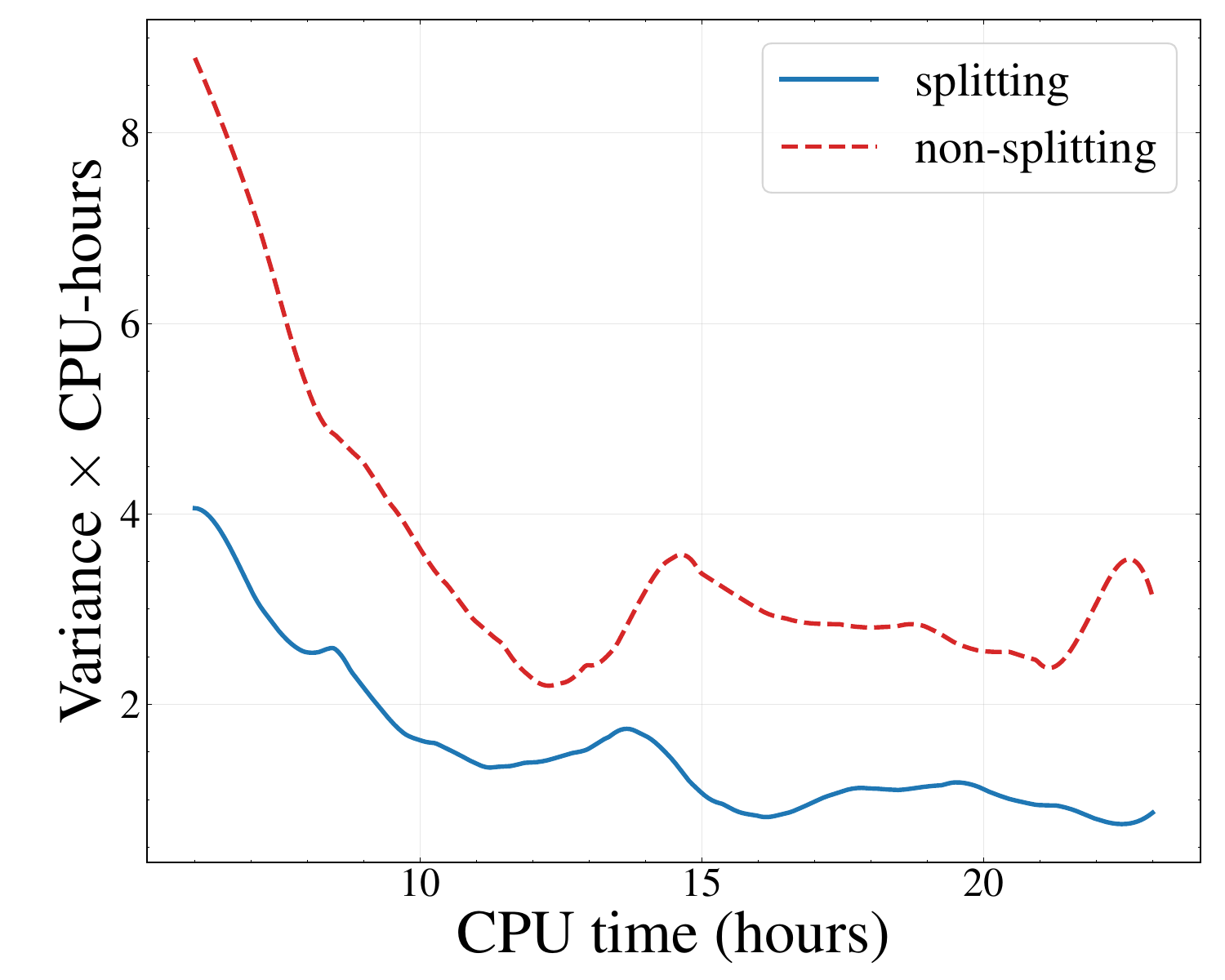}\hfill\includegraphics[width=0.19\textwidth]{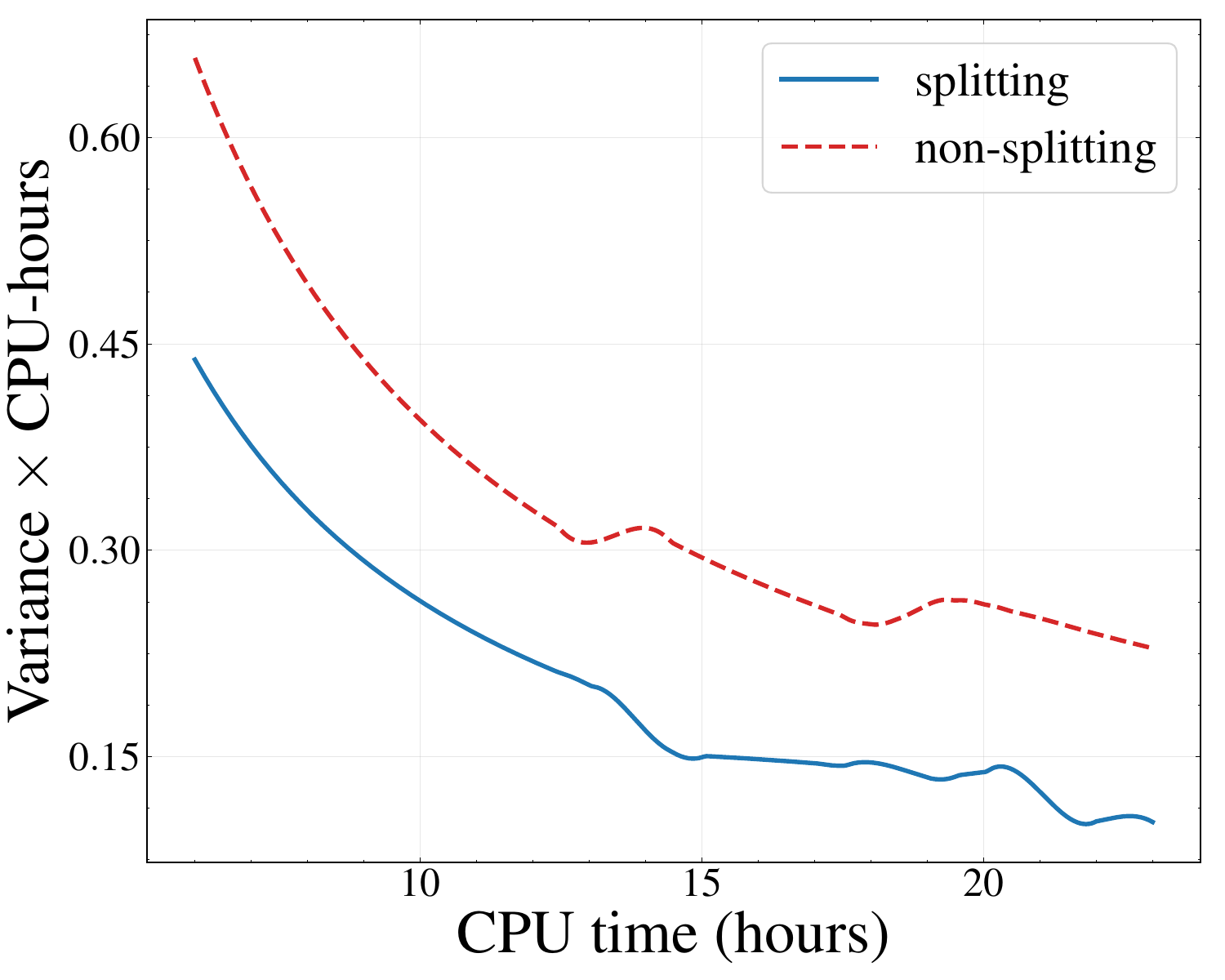}\hfill\includegraphics[width=0.19\textwidth]{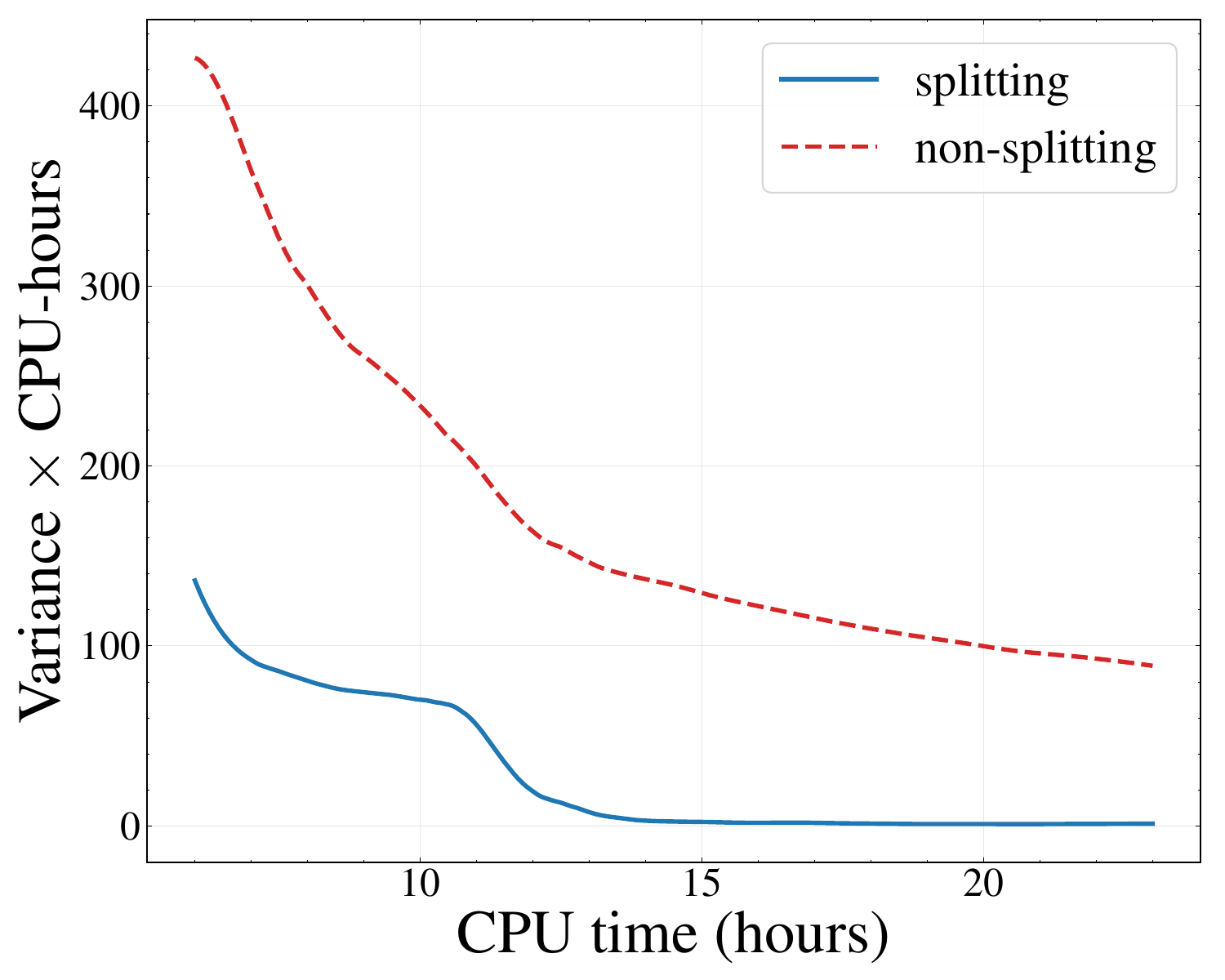}\hfill\includegraphics[width=0.19\textwidth]{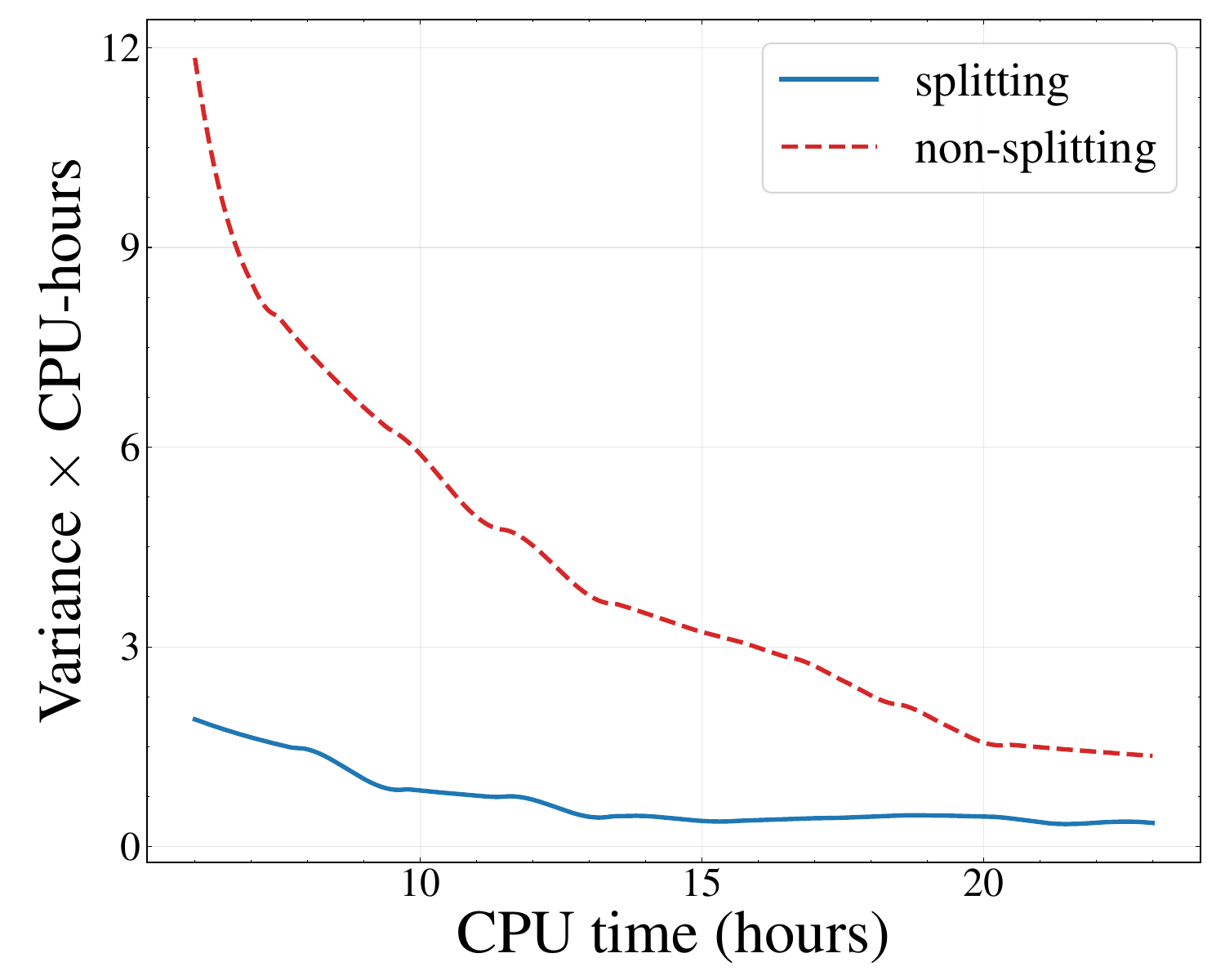}\\
    \makebox[0.19\textwidth]{\footnotesize (a) variance of arrow}\hfill\makebox[0.19\textwidth]{\footnotesize (b) variance of ffmpeg}\hfill\makebox[0.19\textwidth]{\footnotesize (c) variance of grok}\hfill\makebox[0.19\textwidth]{\footnotesize (d) variance of libhevc}\hfill\makebox[0.19\textwidth]{\footnotesize (e) variance of libhtp}\\[3pt]
    \includegraphics[width=0.19\textwidth]{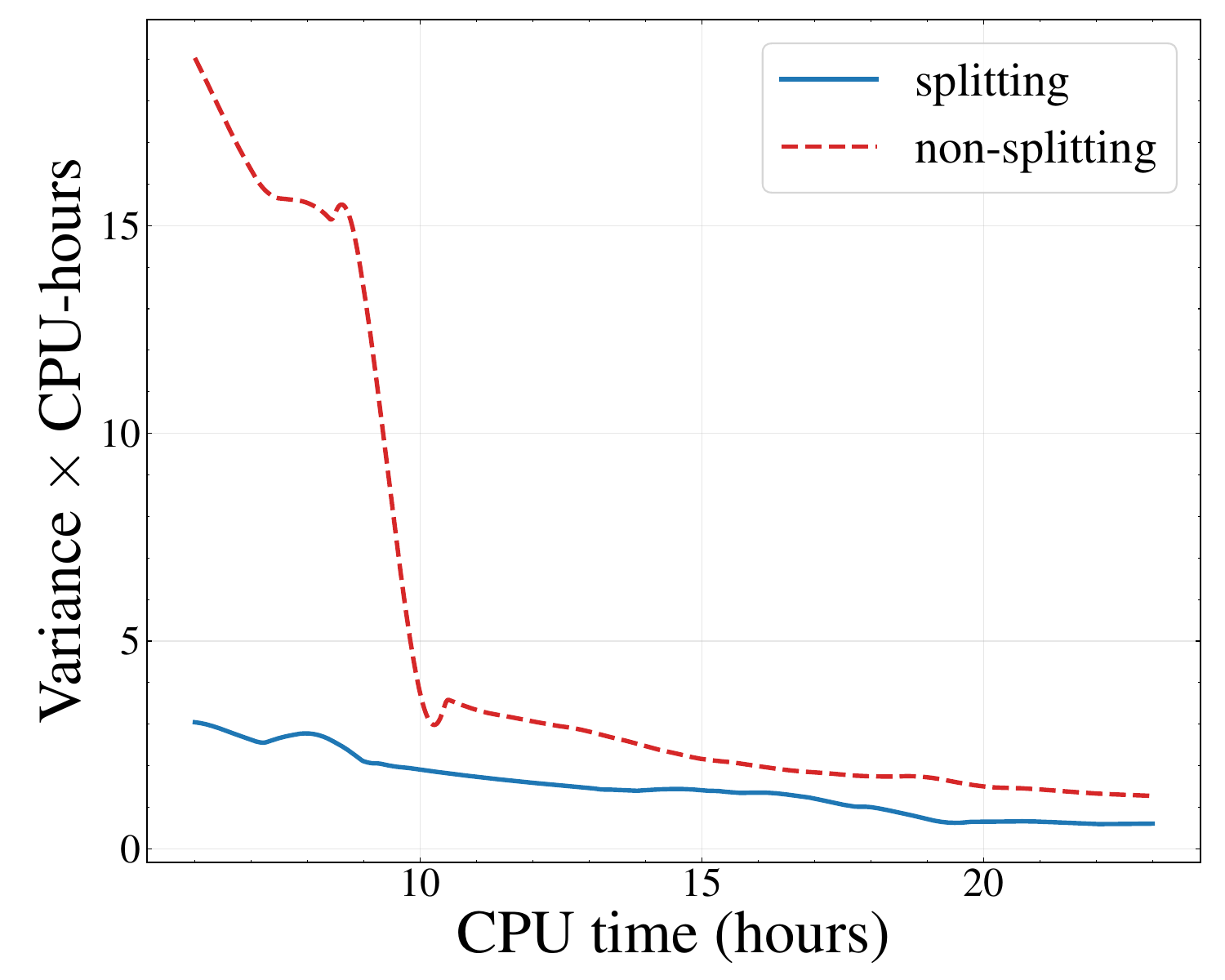}\hfill\includegraphics[width=0.19\textwidth]{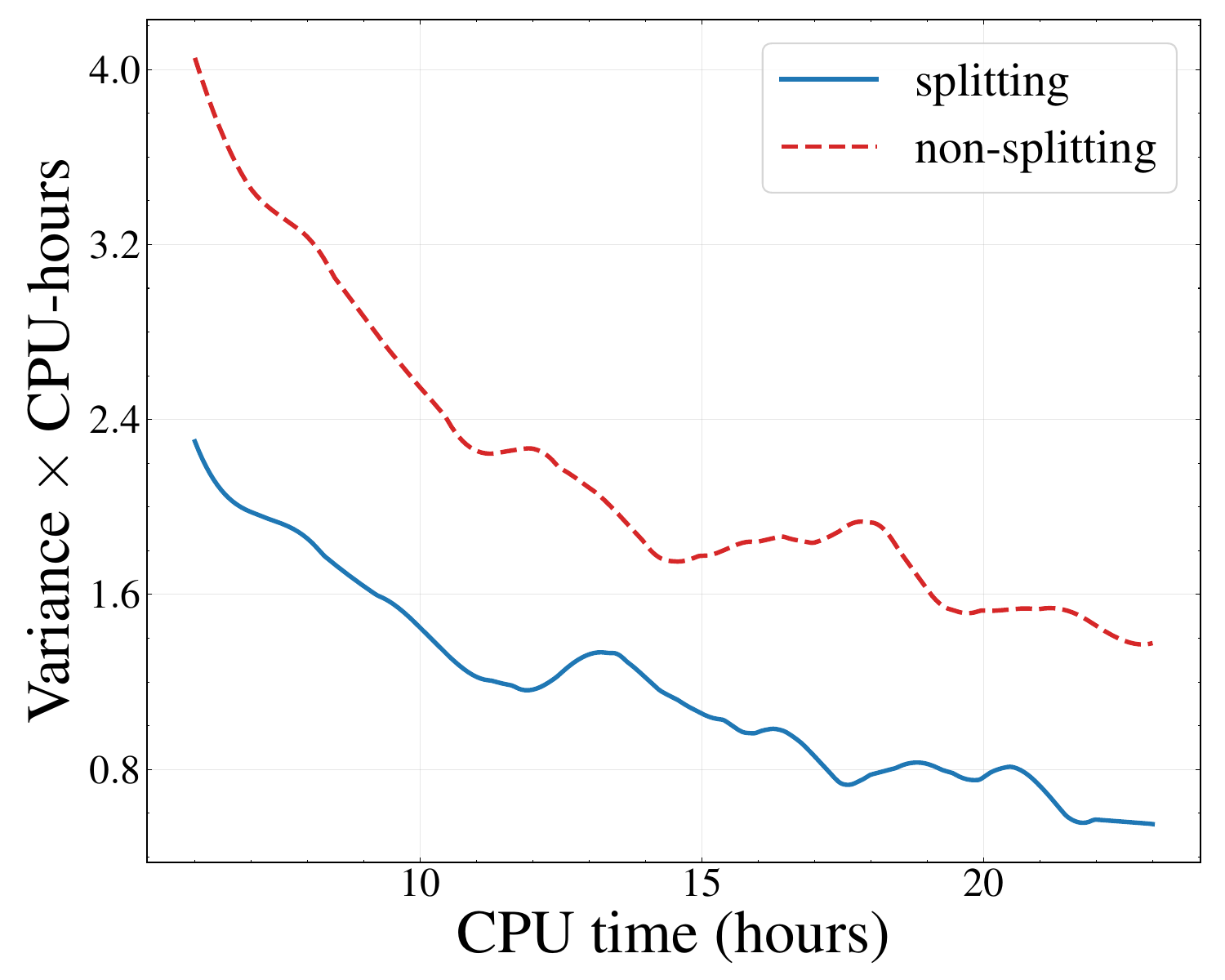}\hfill\includegraphics[width=0.19\textwidth]{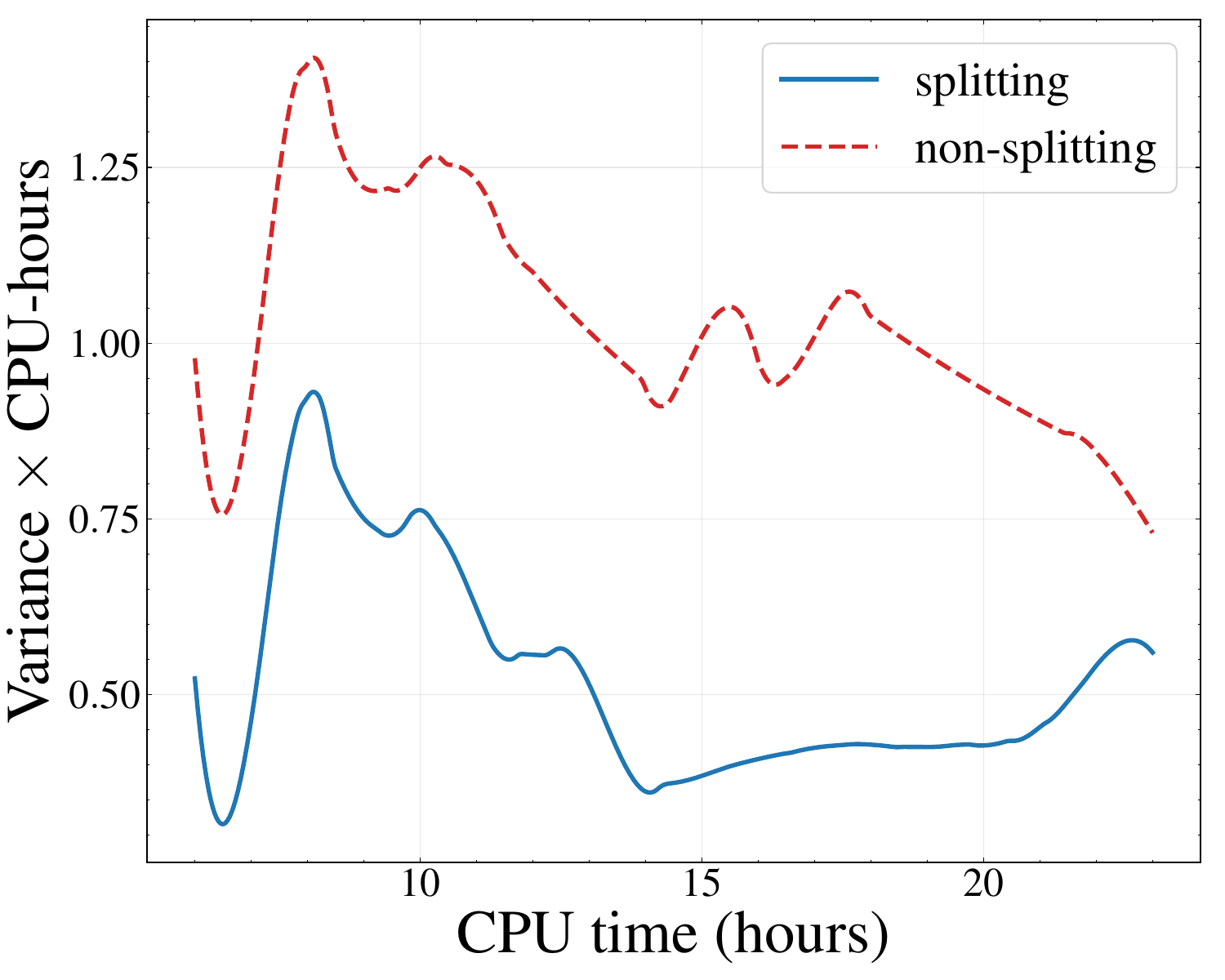}\hfill\includegraphics[width=0.19\textwidth]{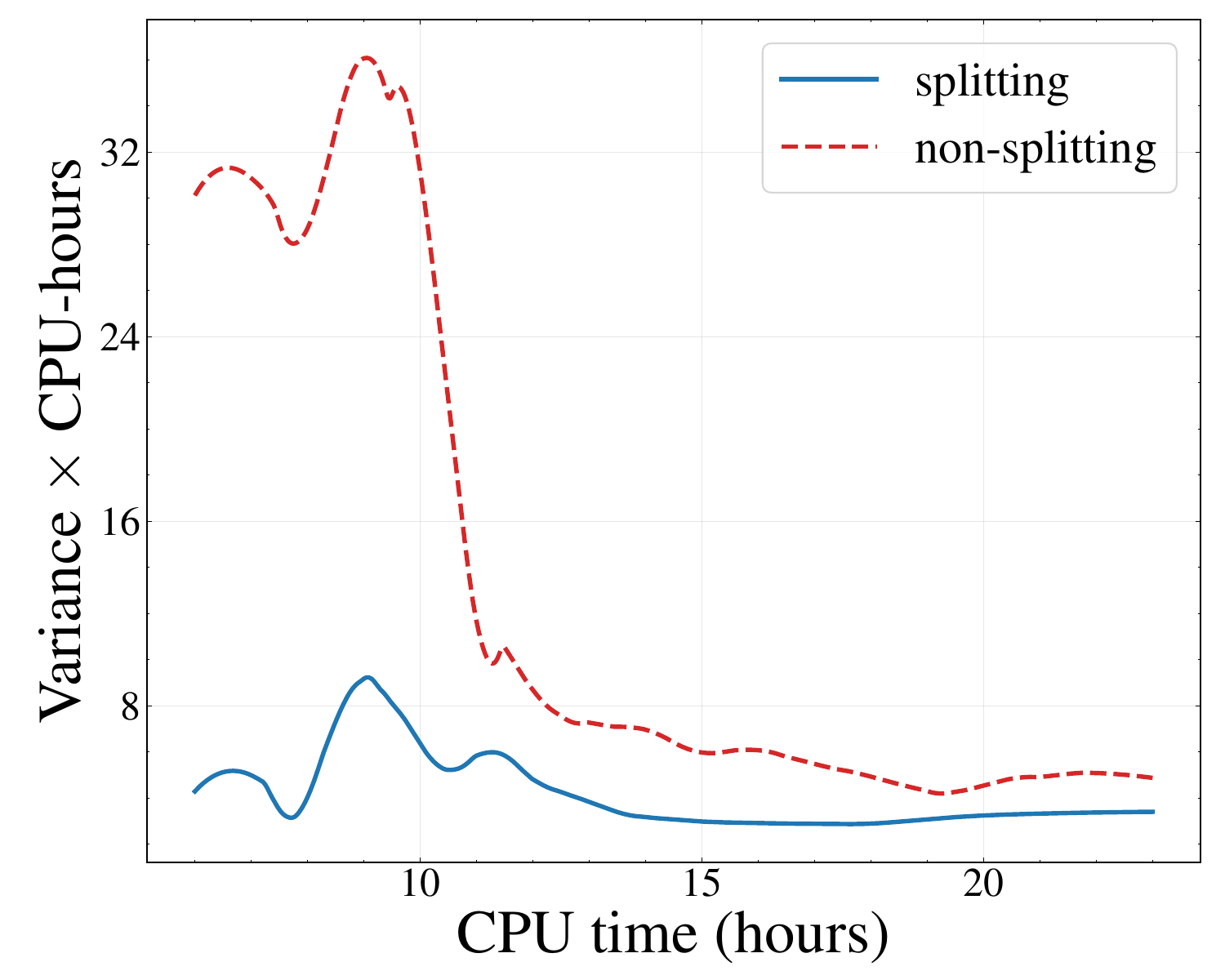}\hfill\includegraphics[width=0.19\textwidth]{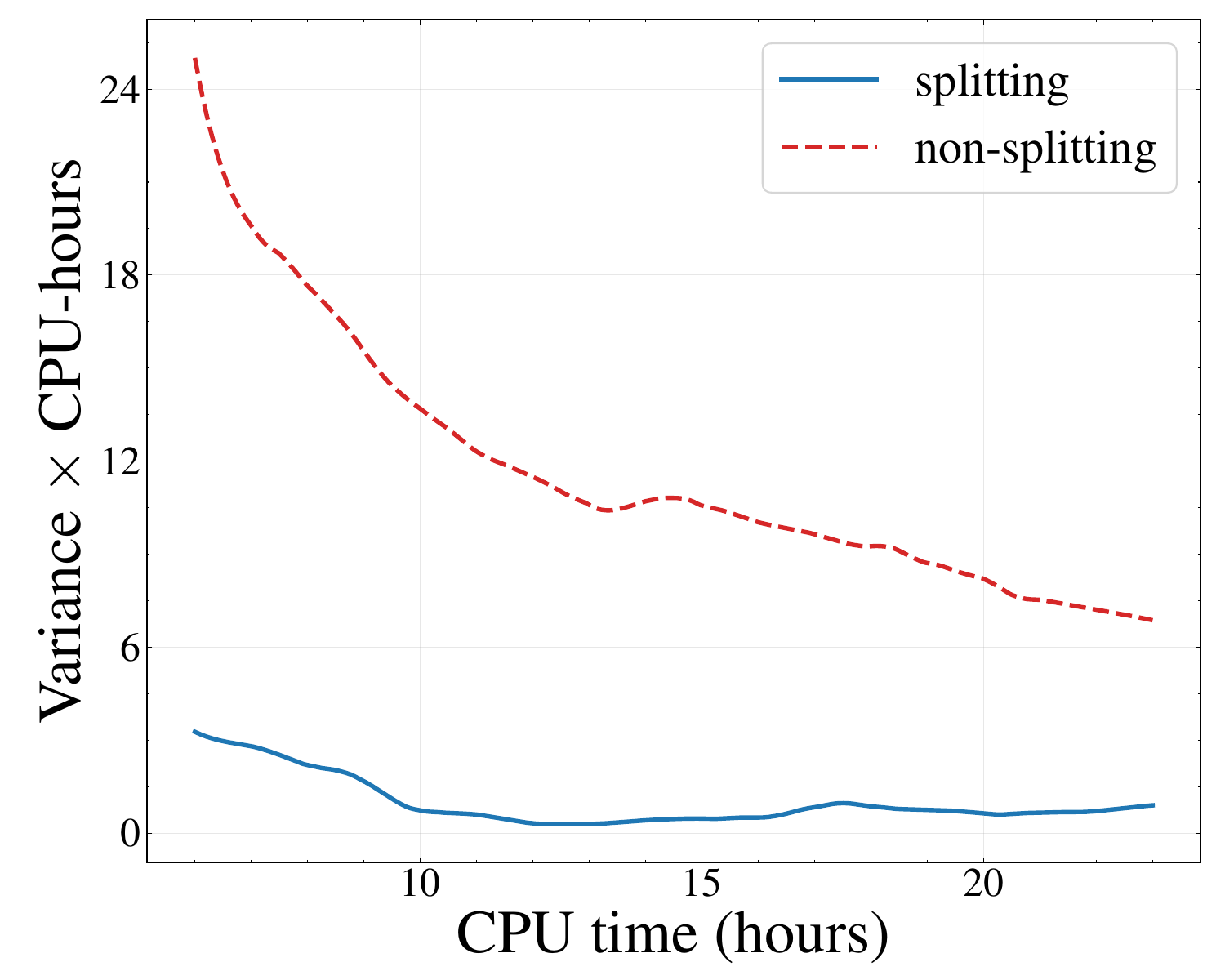}\\
    \makebox[0.19\textwidth]{\footnotesize (f) variance of matio}\hfill\makebox[0.19\textwidth]{\footnotesize (g) variance of openh264}\hfill\makebox[0.19\textwidth]{\footnotesize (h) variance of php}\hfill\makebox[0.19\textwidth]{\footnotesize (i) variance of poppler}\hfill\makebox[0.19\textwidth]{\footnotesize (j) variance of stb}
    \caption{Comparisons of variance across 10 benchmarks for fuzzer AFLSmart.}
    \label{fig:comparison_aflsmart}
\end{figure*}

\begin{figure*}[tp]
    \centering
    \includegraphics[width=0.19\textwidth]{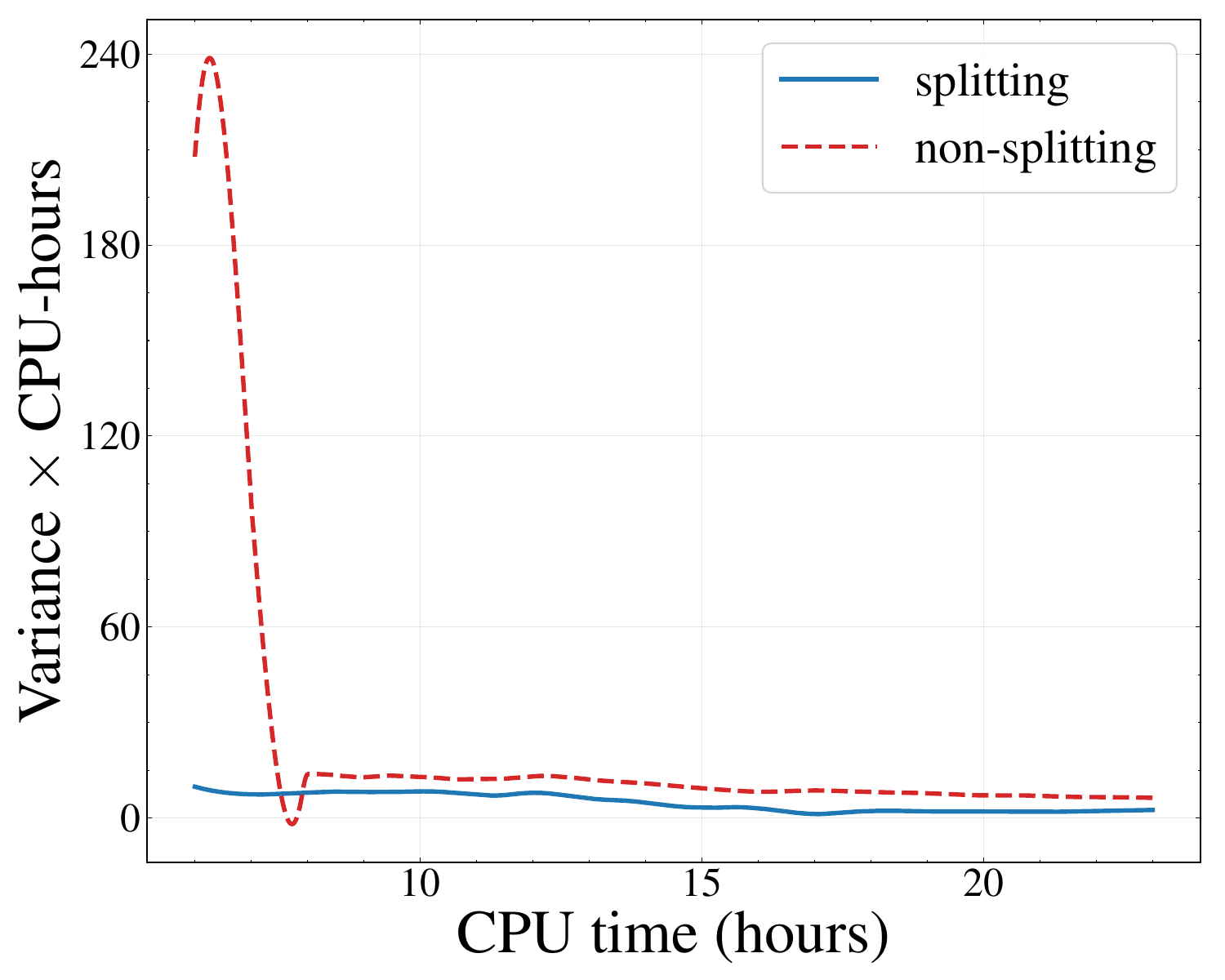}\hfill\includegraphics[width=0.19\textwidth]{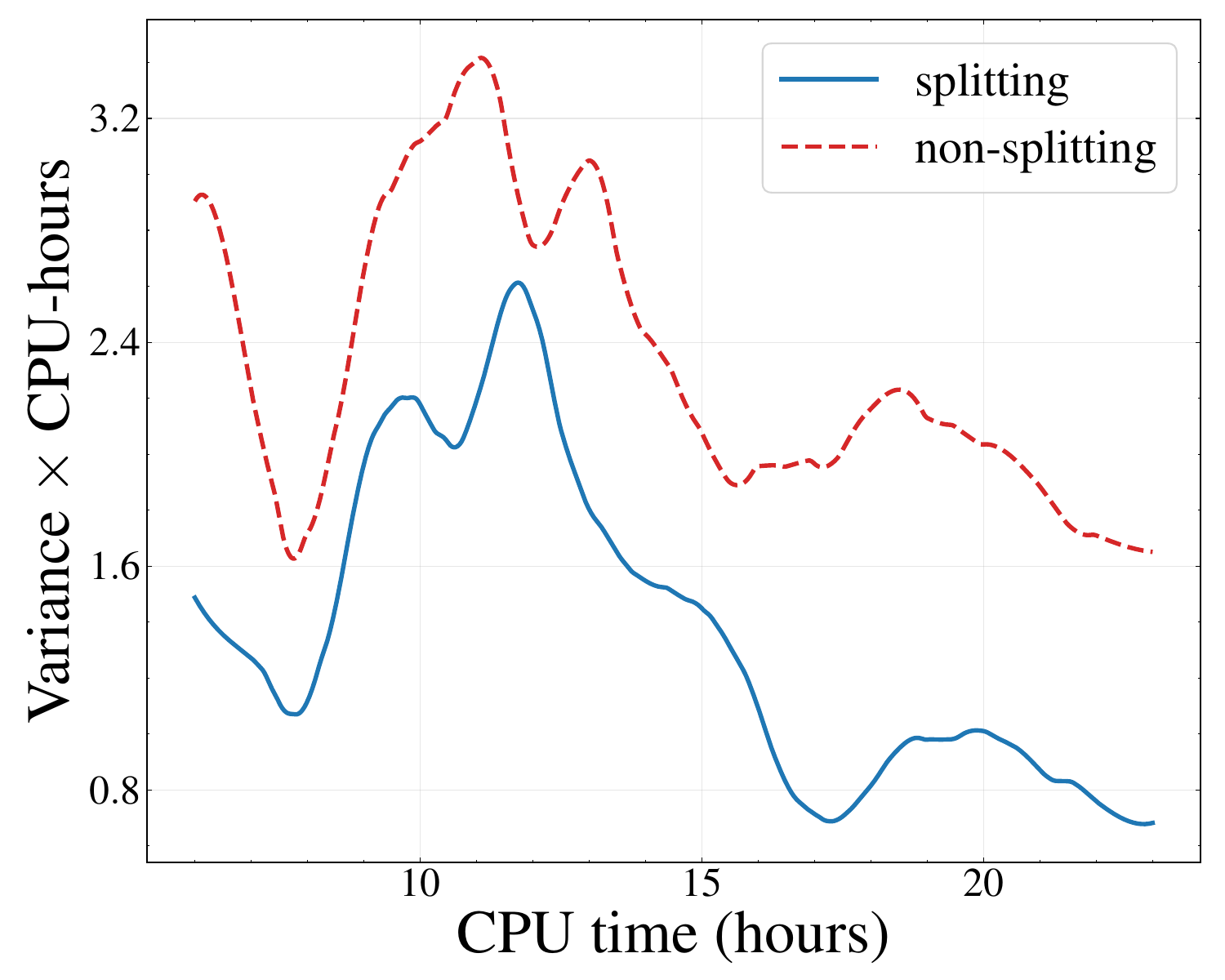}\hfill\includegraphics[width=0.19\textwidth]{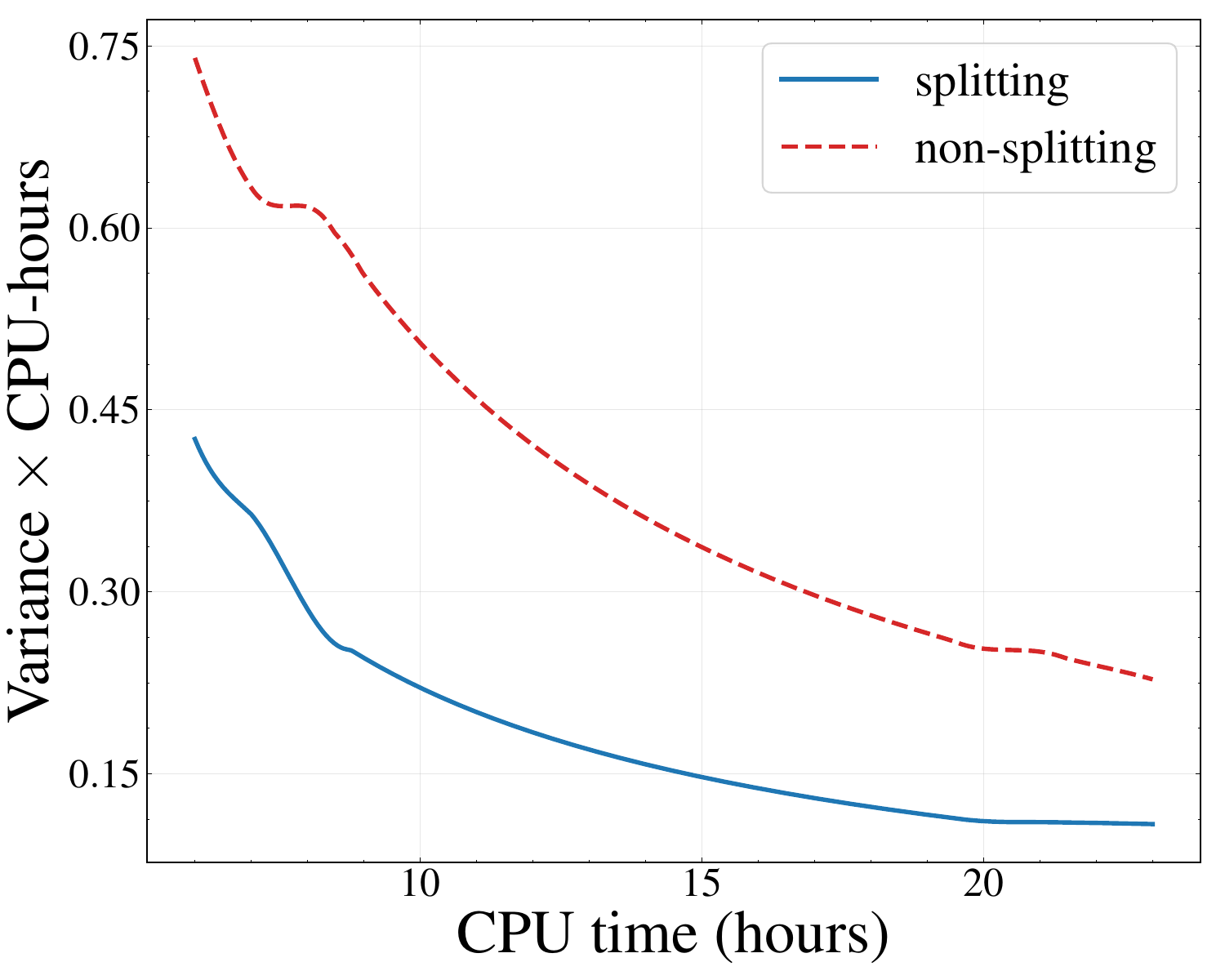}\hfill\includegraphics[width=0.19\textwidth]{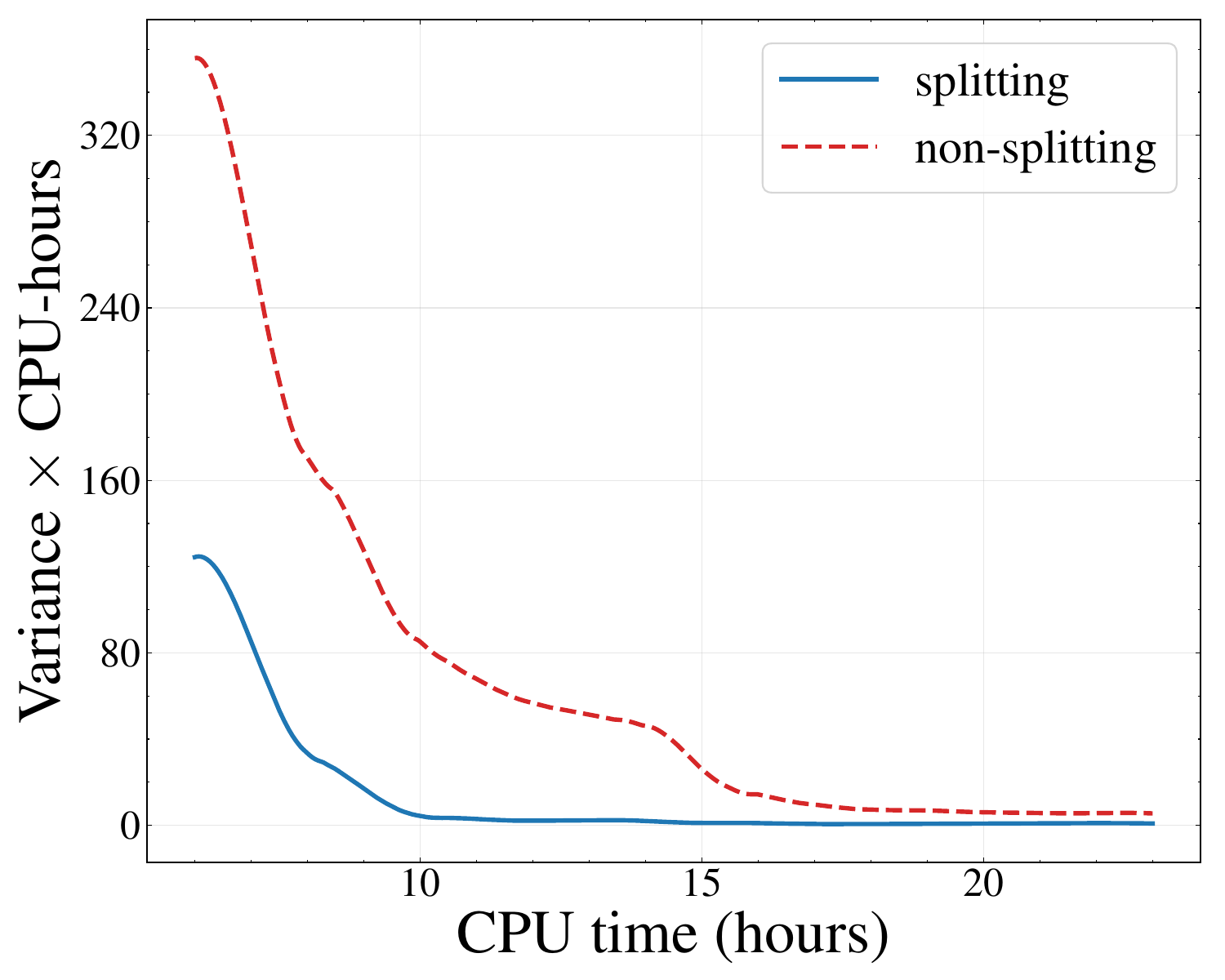}\hfill\includegraphics[width=0.19\textwidth]{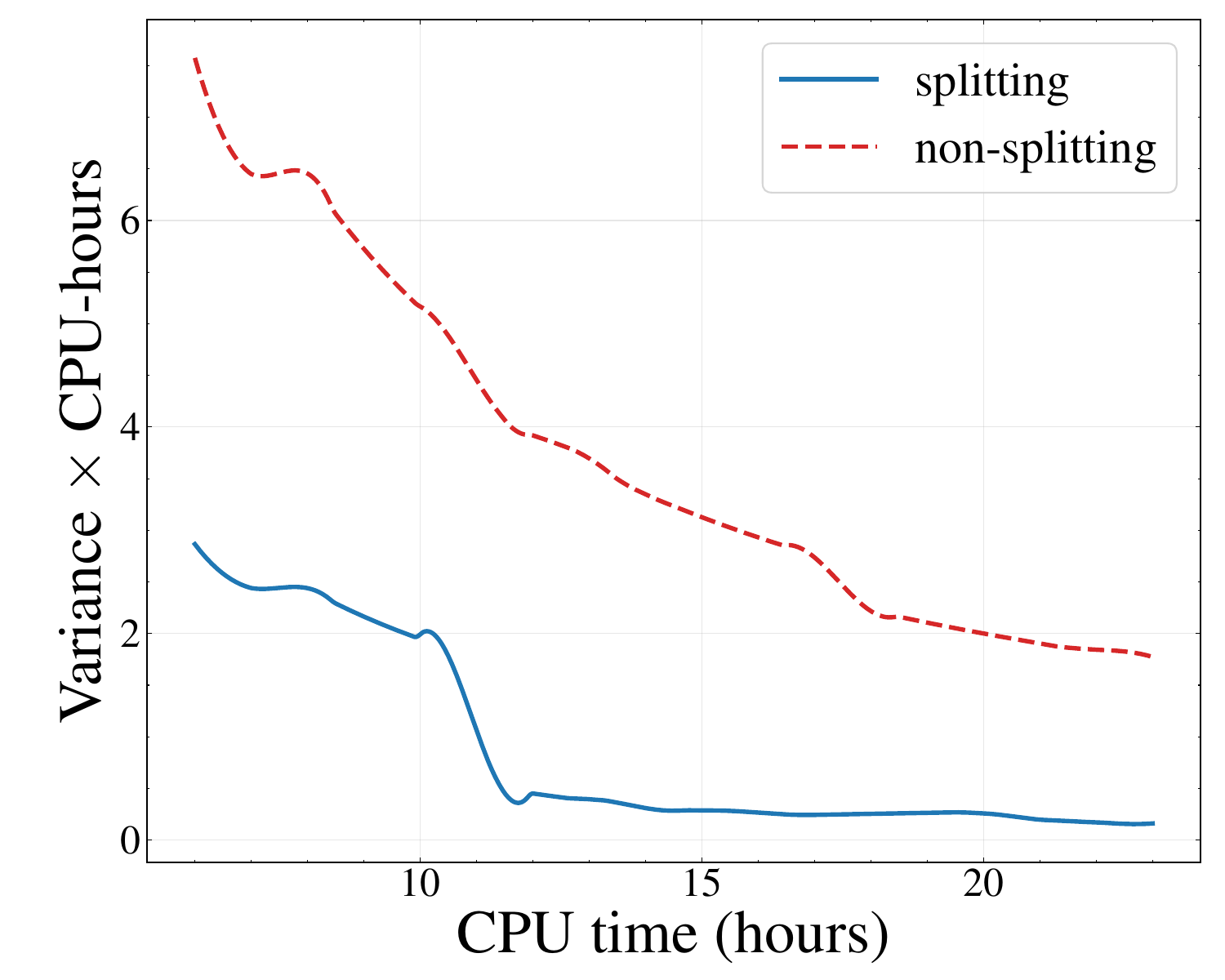}\\
    \makebox[0.19\textwidth]{\footnotesize (a) variance of arrow}\hfill\makebox[0.19\textwidth]{\footnotesize (b) variance of ffmpeg}\hfill\makebox[0.19\textwidth]{\footnotesize (c) variance of grok}\hfill\makebox[0.19\textwidth]{\footnotesize (d) variance of libhevc}\hfill\makebox[0.19\textwidth]{\footnotesize (e) variance of libhtp}\\[3pt]
    \includegraphics[width=0.19\textwidth]{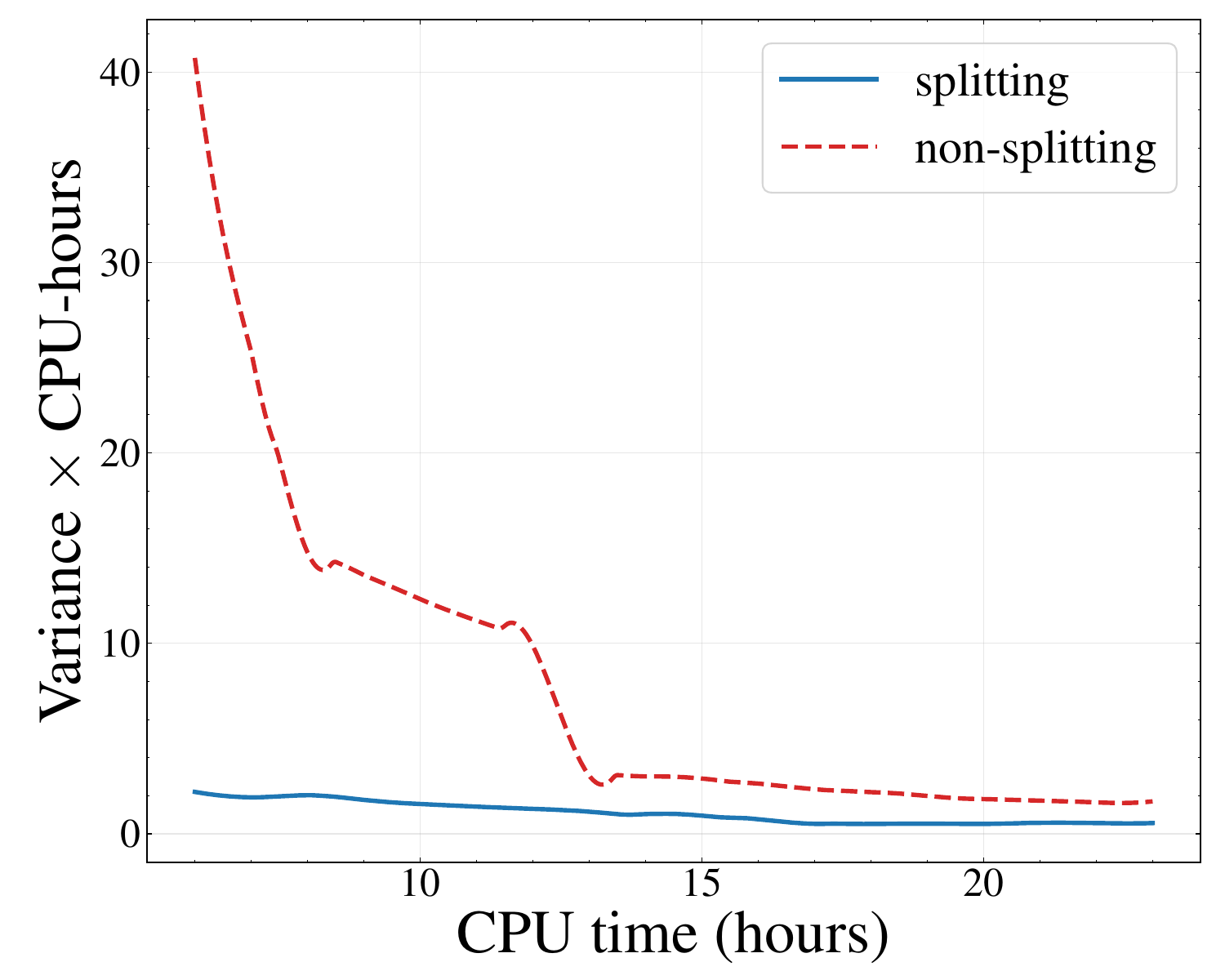}\hfill\includegraphics[width=0.19\textwidth]{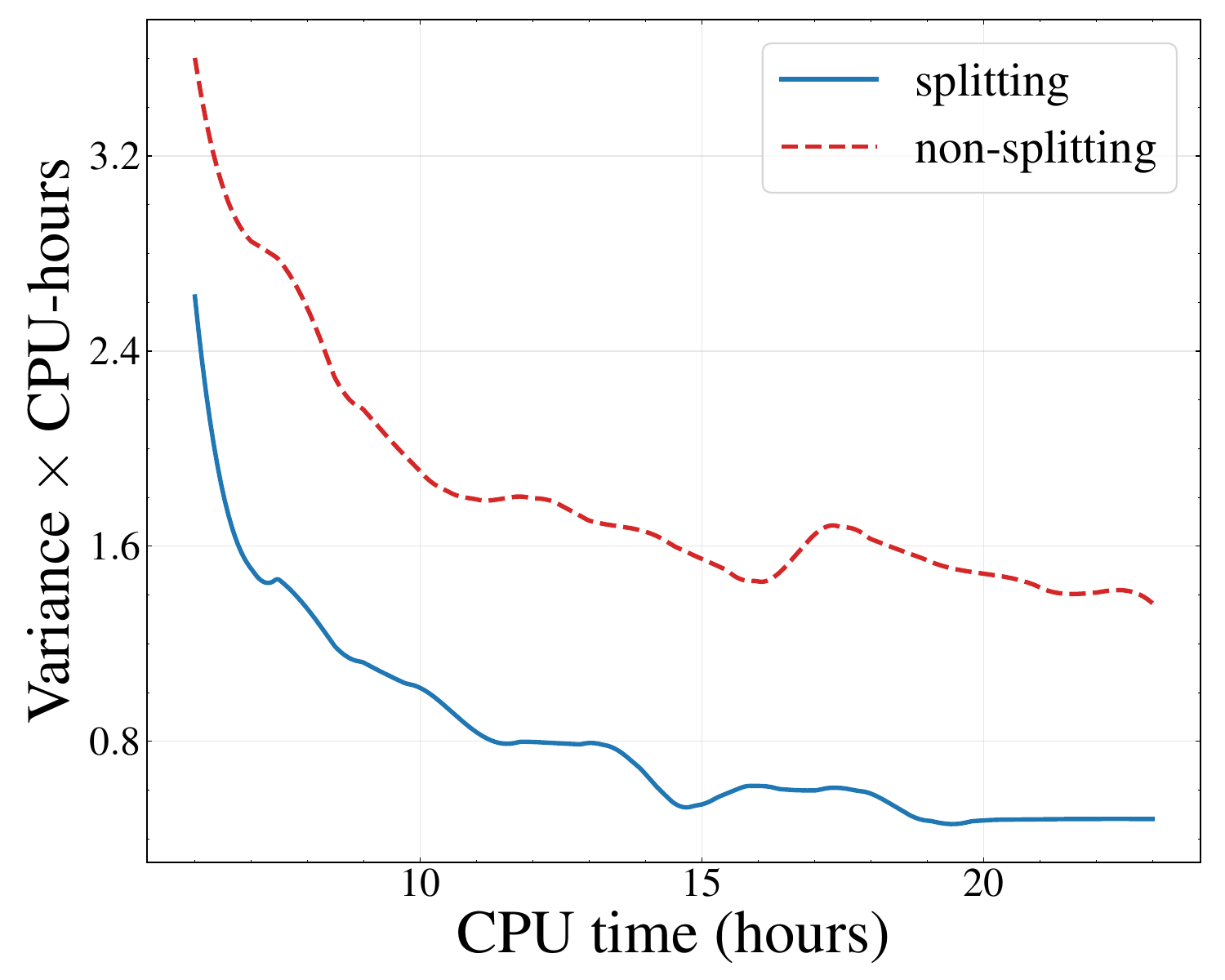}\hfill\includegraphics[width=0.19\textwidth]{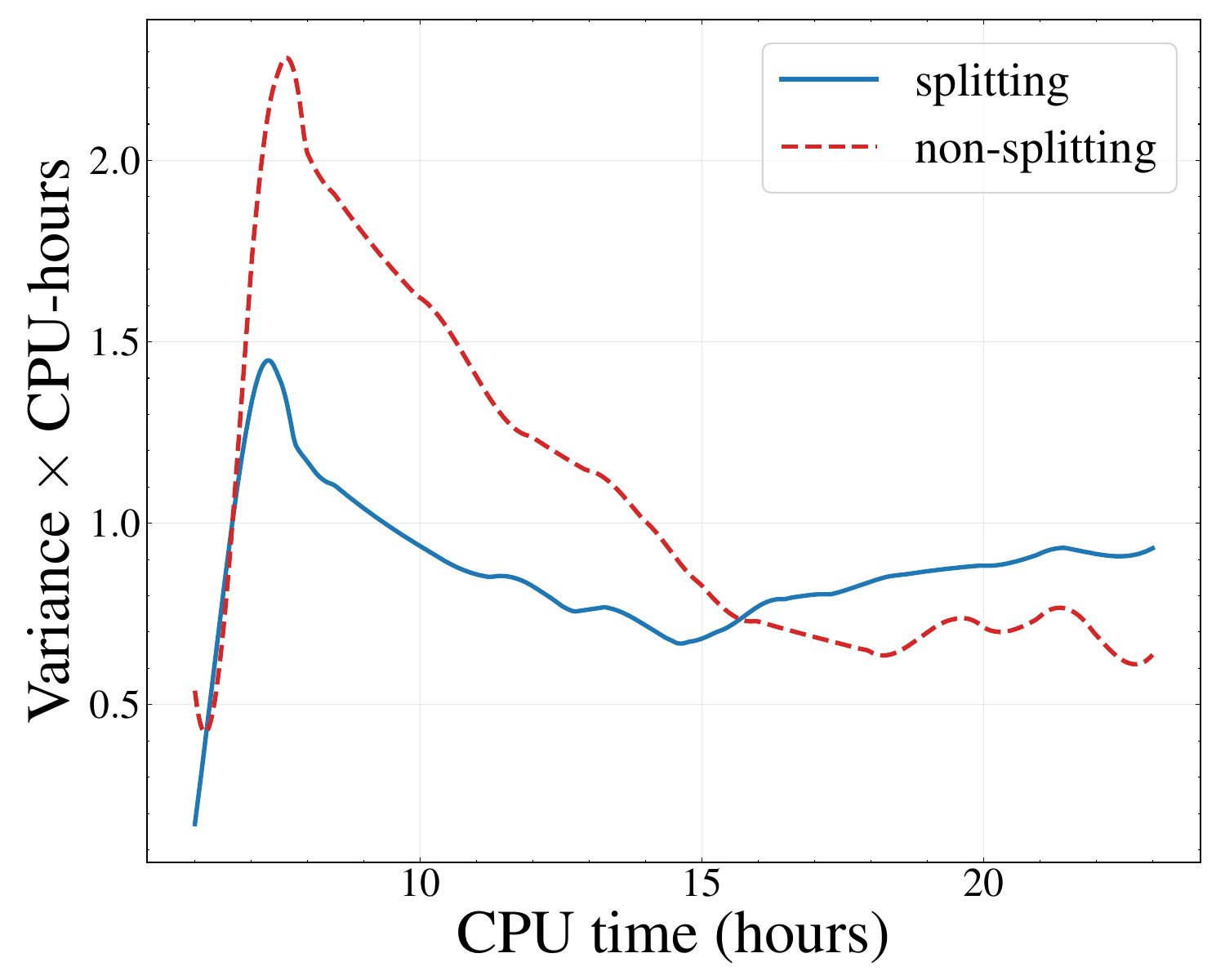}\hfill\includegraphics[width=0.19\textwidth]{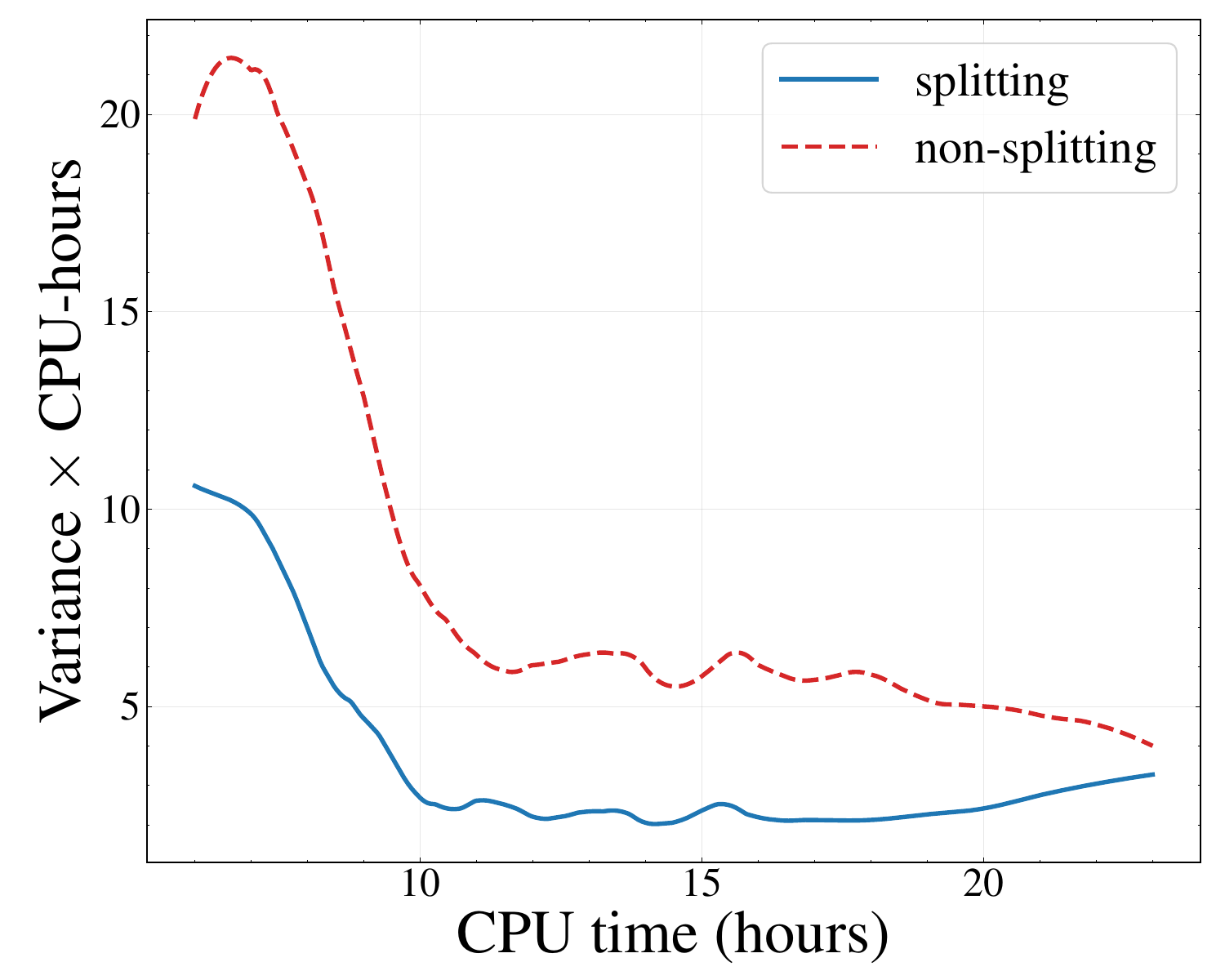}\hfill\includegraphics[width=0.19\textwidth]{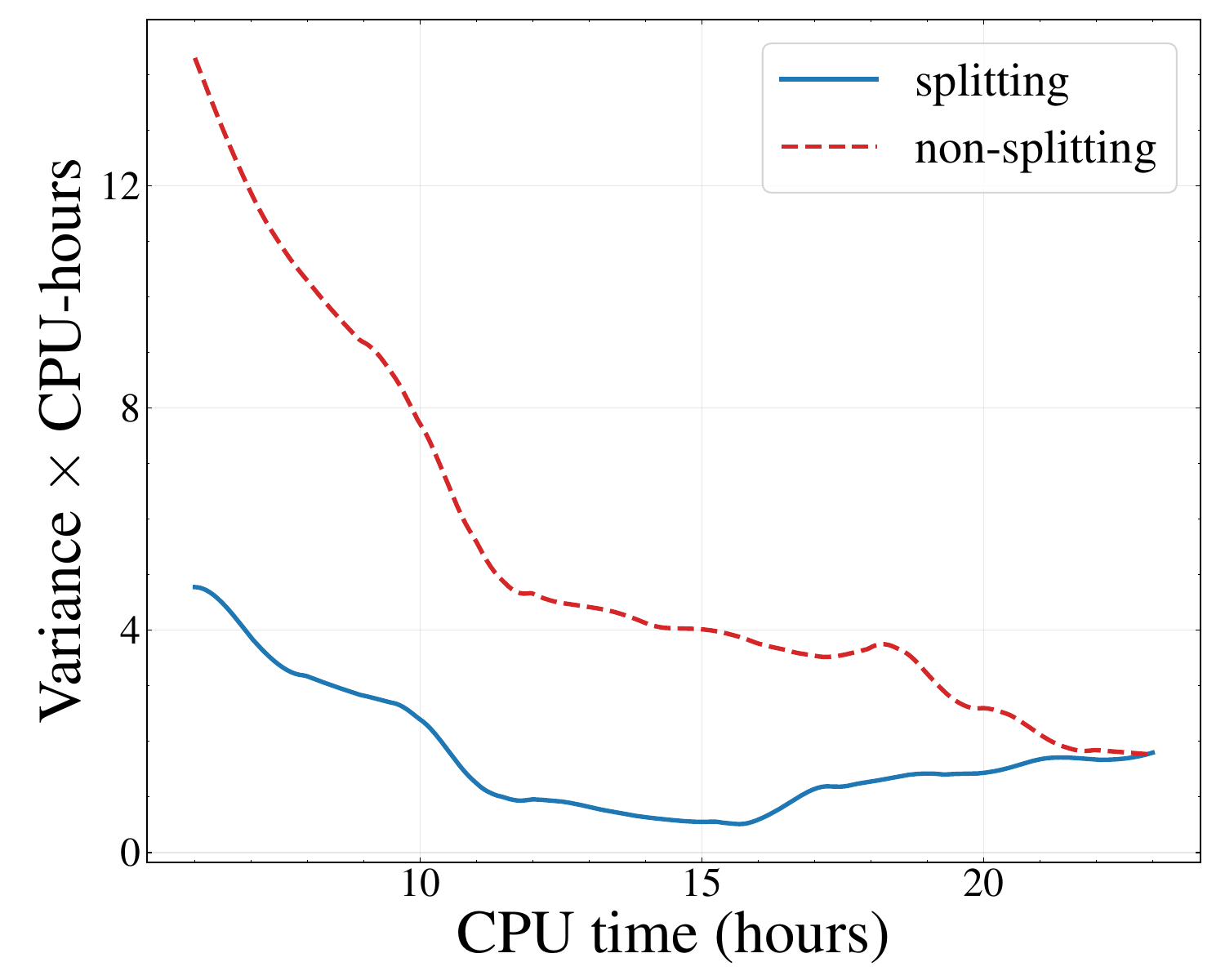}\\
    \makebox[0.19\textwidth]{\footnotesize (f) variance of matio}\hfill\makebox[0.19\textwidth]{\footnotesize (g) variance of openh264}\hfill\makebox[0.19\textwidth]{\footnotesize (h) variance of php}\hfill\makebox[0.19\textwidth]{\footnotesize (i) variance of poppler}\hfill\makebox[0.19\textwidth]{\footnotesize (j) variance of stb}
    \caption{Comparisons of variance across 10 benchmarks for fuzzer AFLFast.}
    \label{fig:comparison_aflfast}
\end{figure*}

\begin{figure*}[tp]
    \centering
    \includegraphics[width=0.19\textwidth]{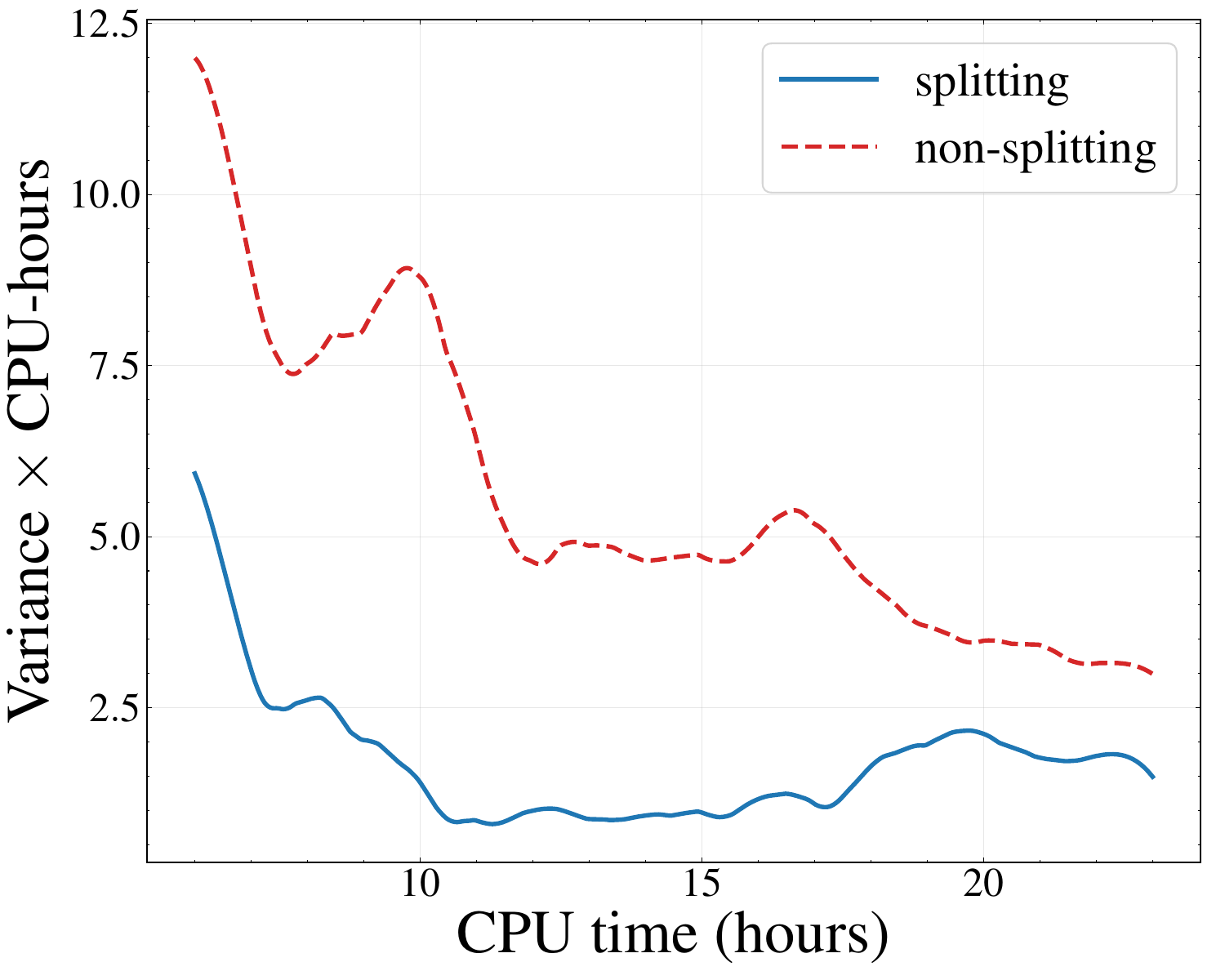}\hfill\includegraphics[width=0.19\textwidth]{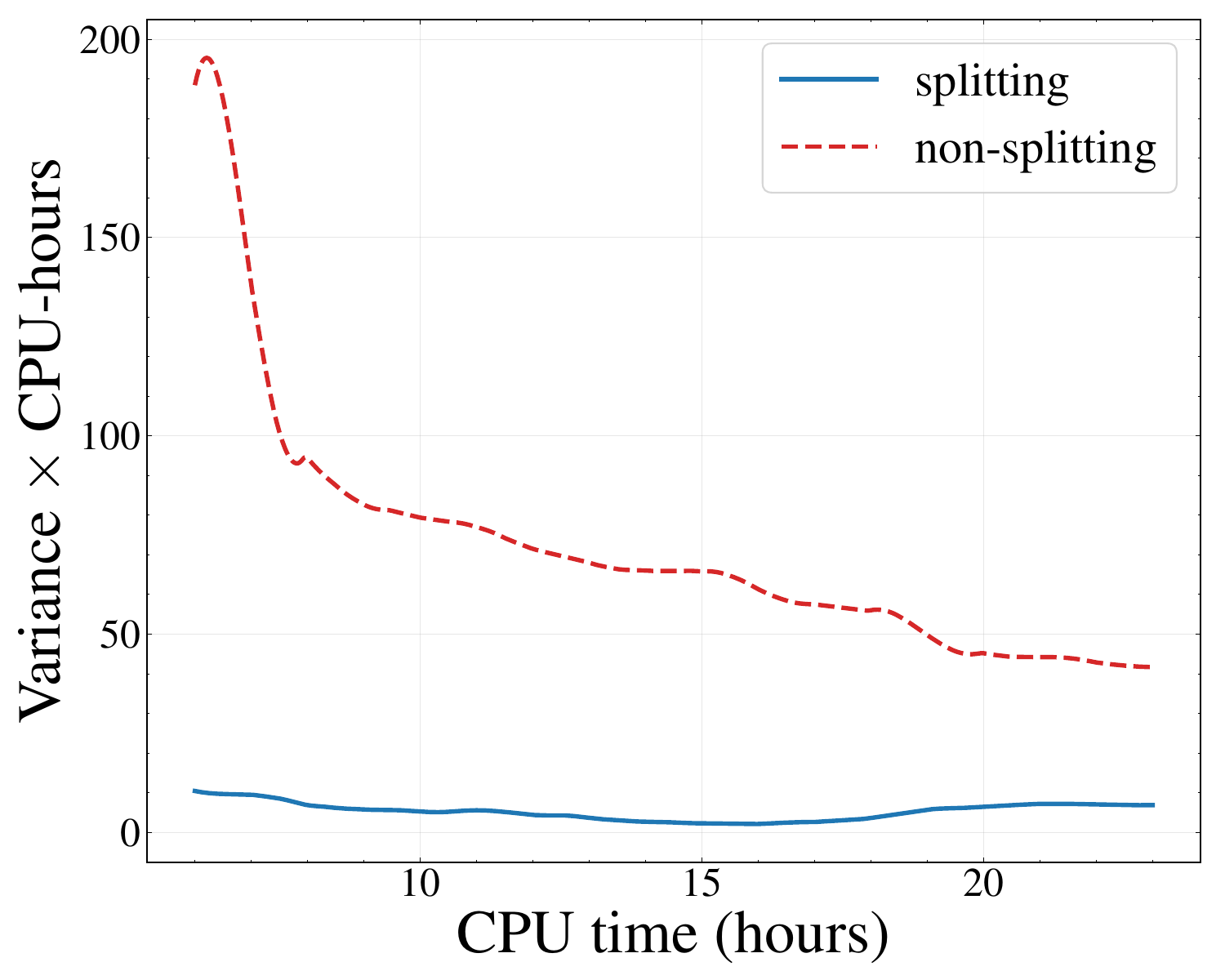}\hfill\includegraphics[width=0.19\textwidth]{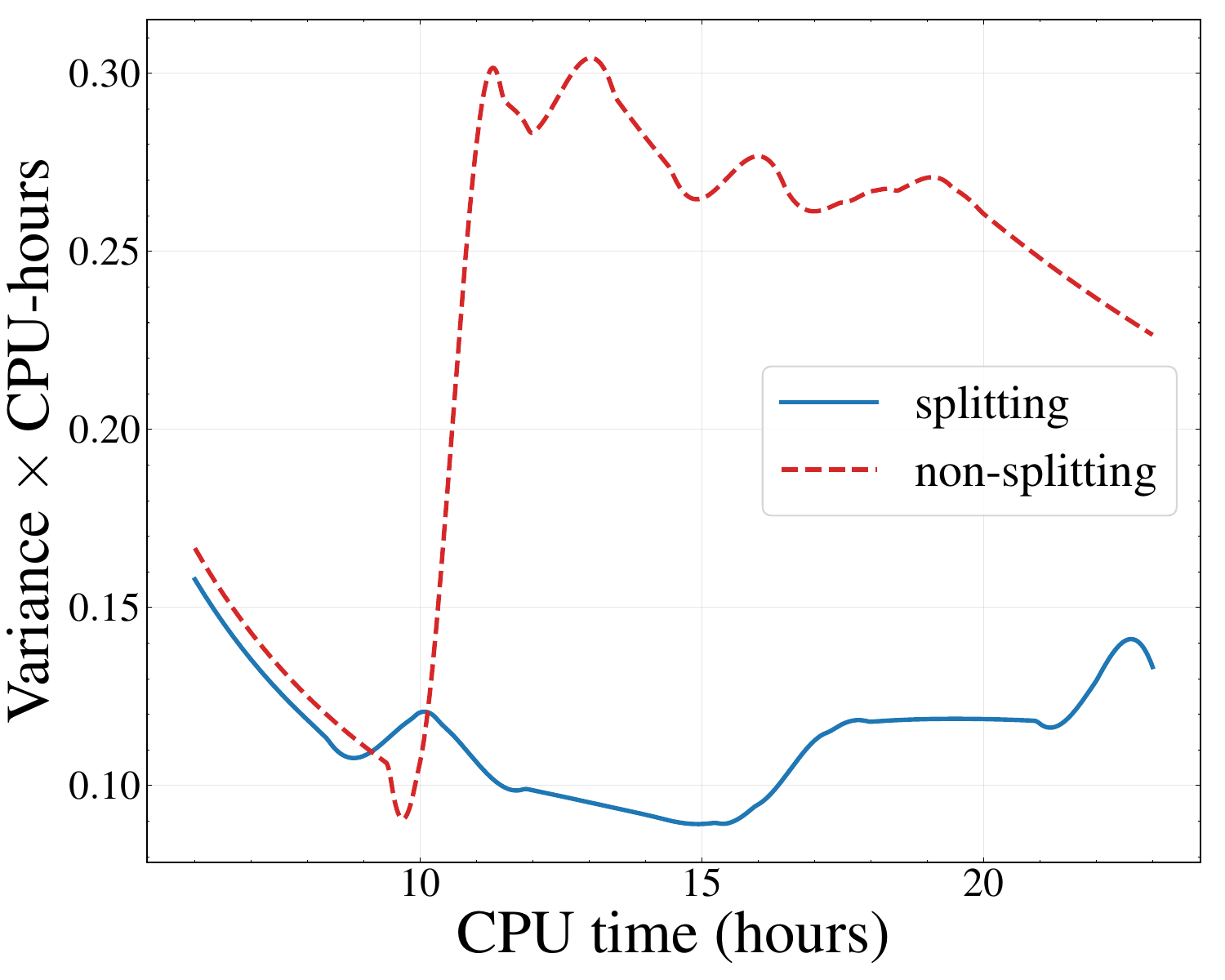}\hfill\includegraphics[width=0.19\textwidth]{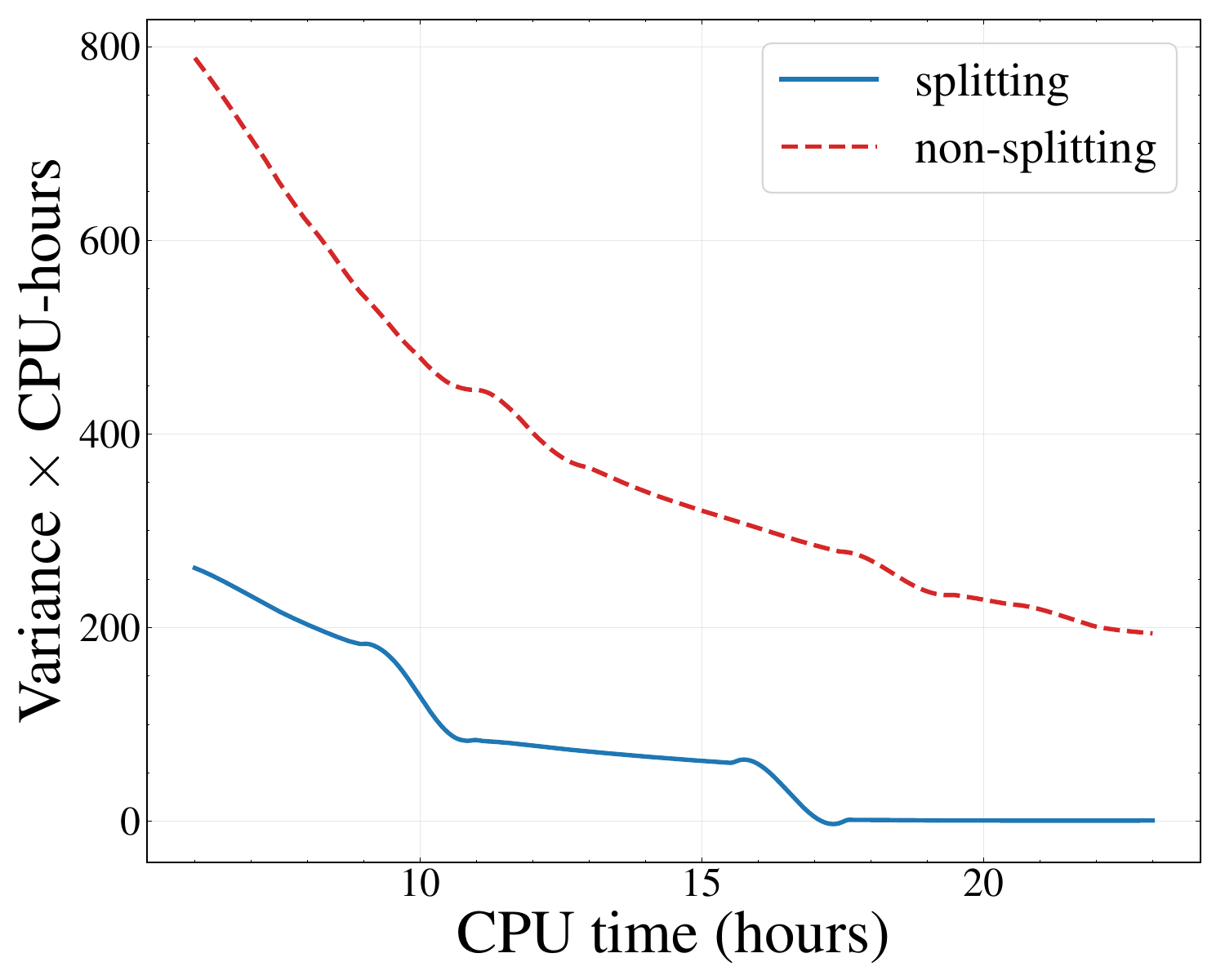}\hfill\includegraphics[width=0.19\textwidth]{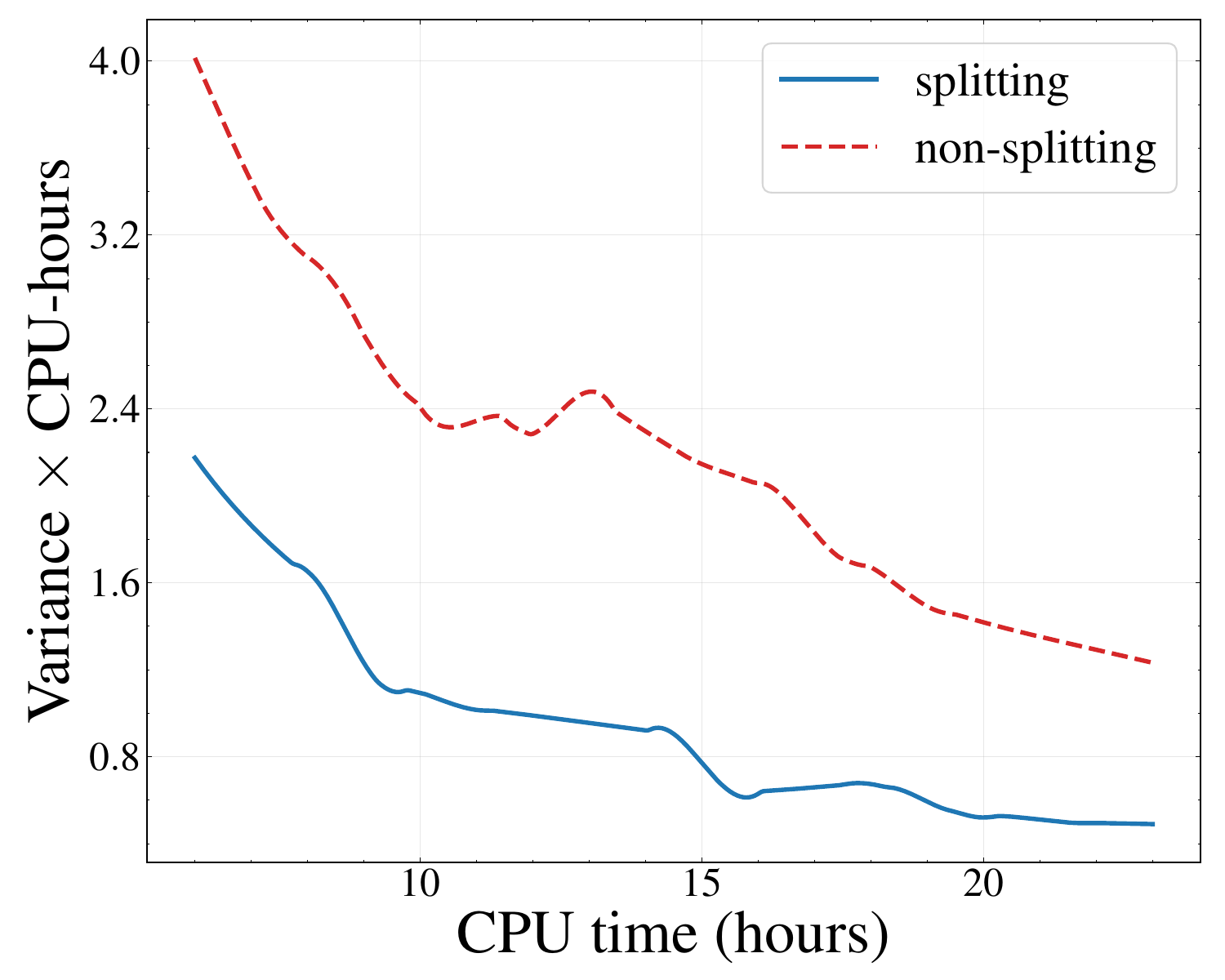}\\
    \makebox[0.19\textwidth]{\footnotesize (a) variance of arrow}\hfill\makebox[0.19\textwidth]{\footnotesize (b) variance of ffmpeg}\hfill\makebox[0.19\textwidth]{\footnotesize (c) variance of grok}\hfill\makebox[0.19\textwidth]{\footnotesize (d) variance of libhevc}\hfill\makebox[0.19\textwidth]{\footnotesize (e) variance of libhtp}\\[3pt]
    \includegraphics[width=0.19\textwidth]{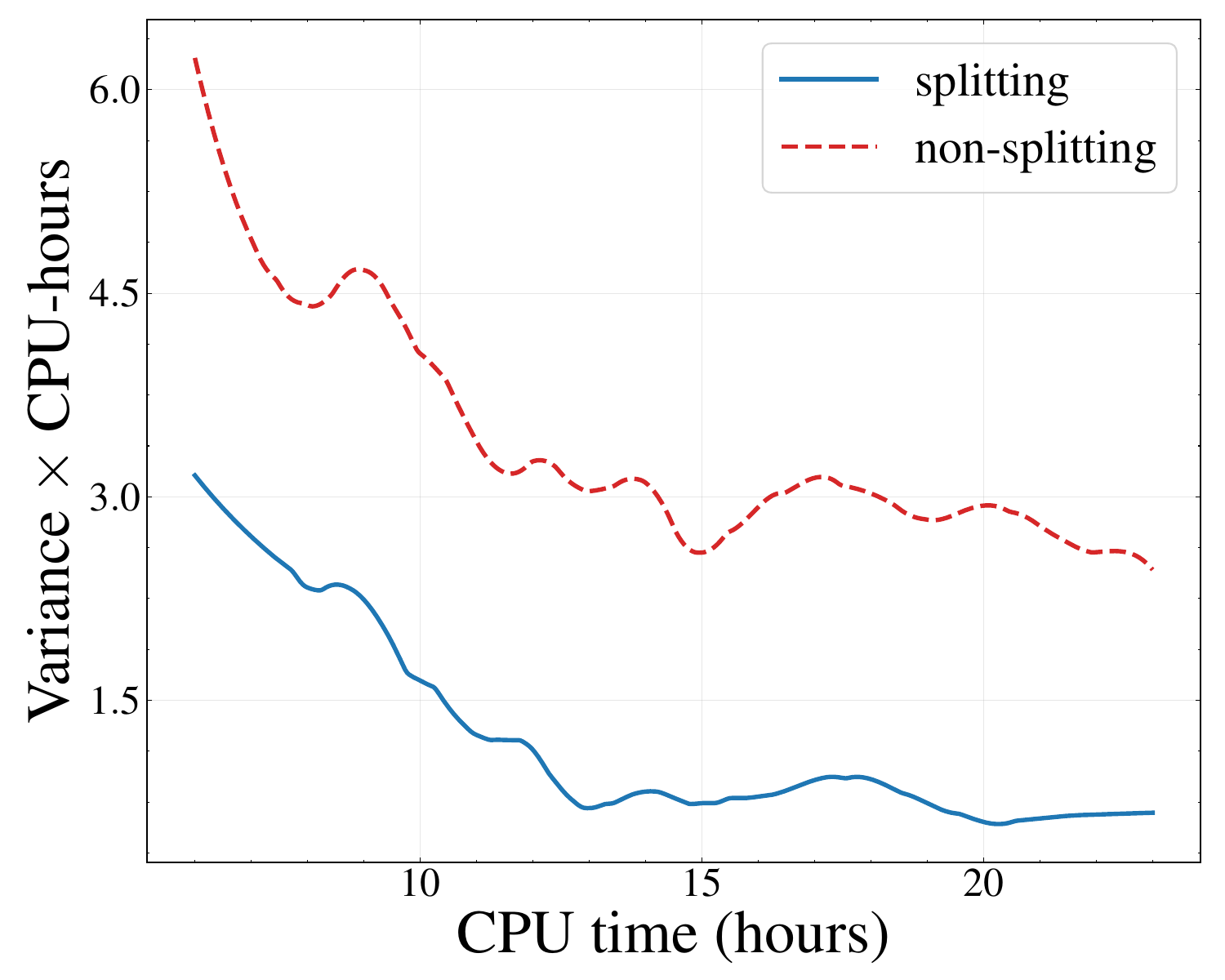}\hfill\includegraphics[width=0.19\textwidth]{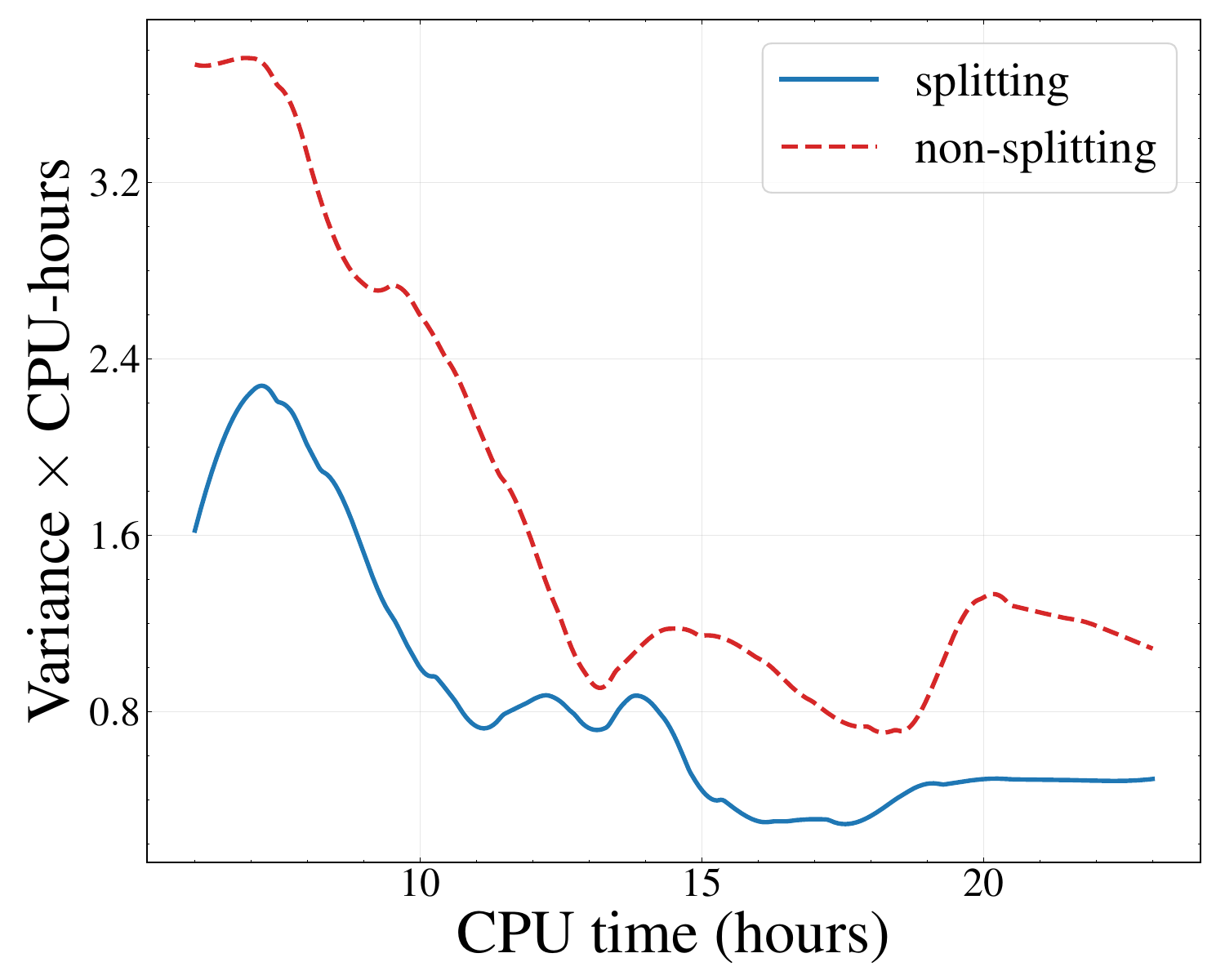}\hfill\includegraphics[width=0.19\textwidth]{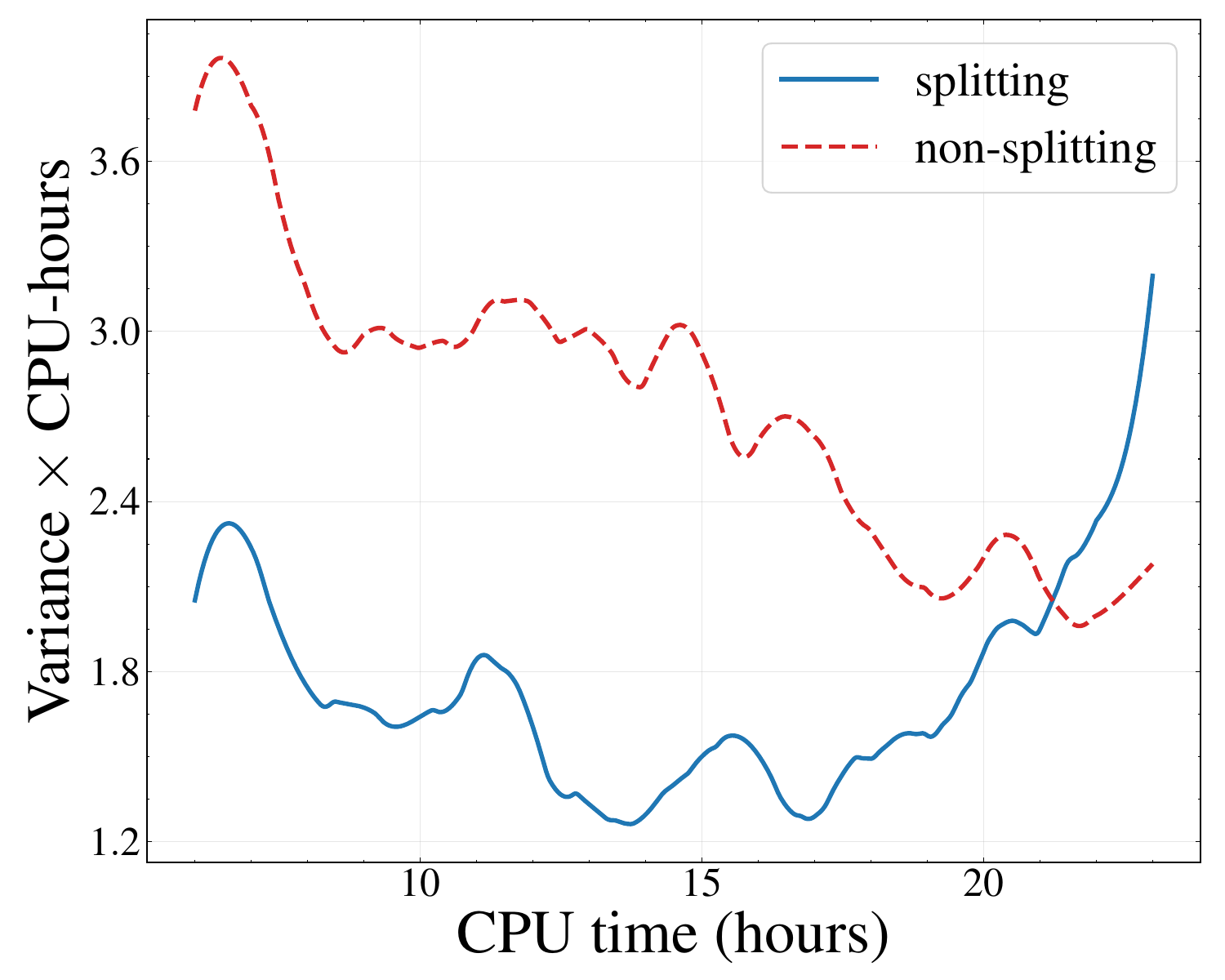}\hfill\includegraphics[width=0.19\textwidth]{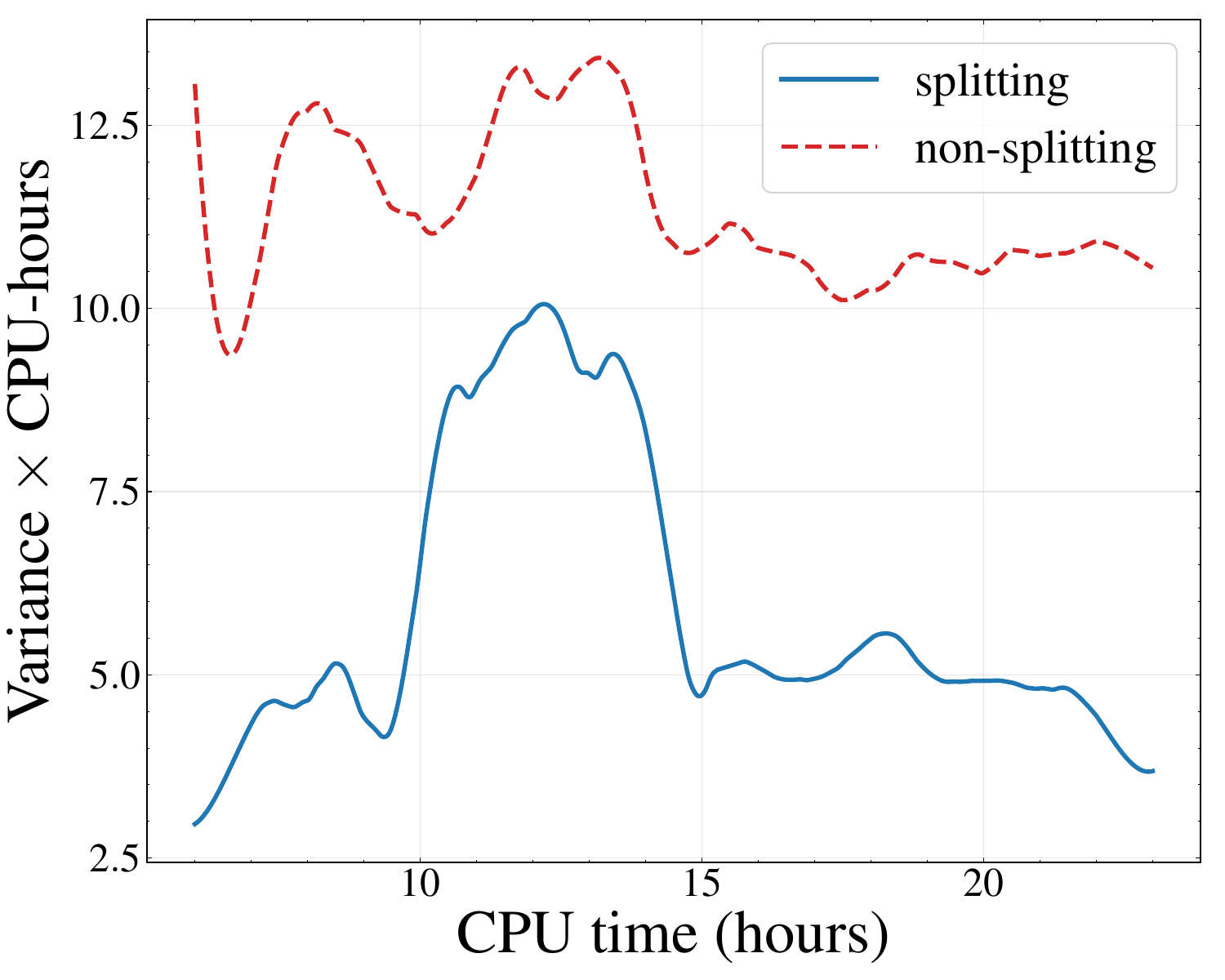}\hfill\includegraphics[width=0.19\textwidth]{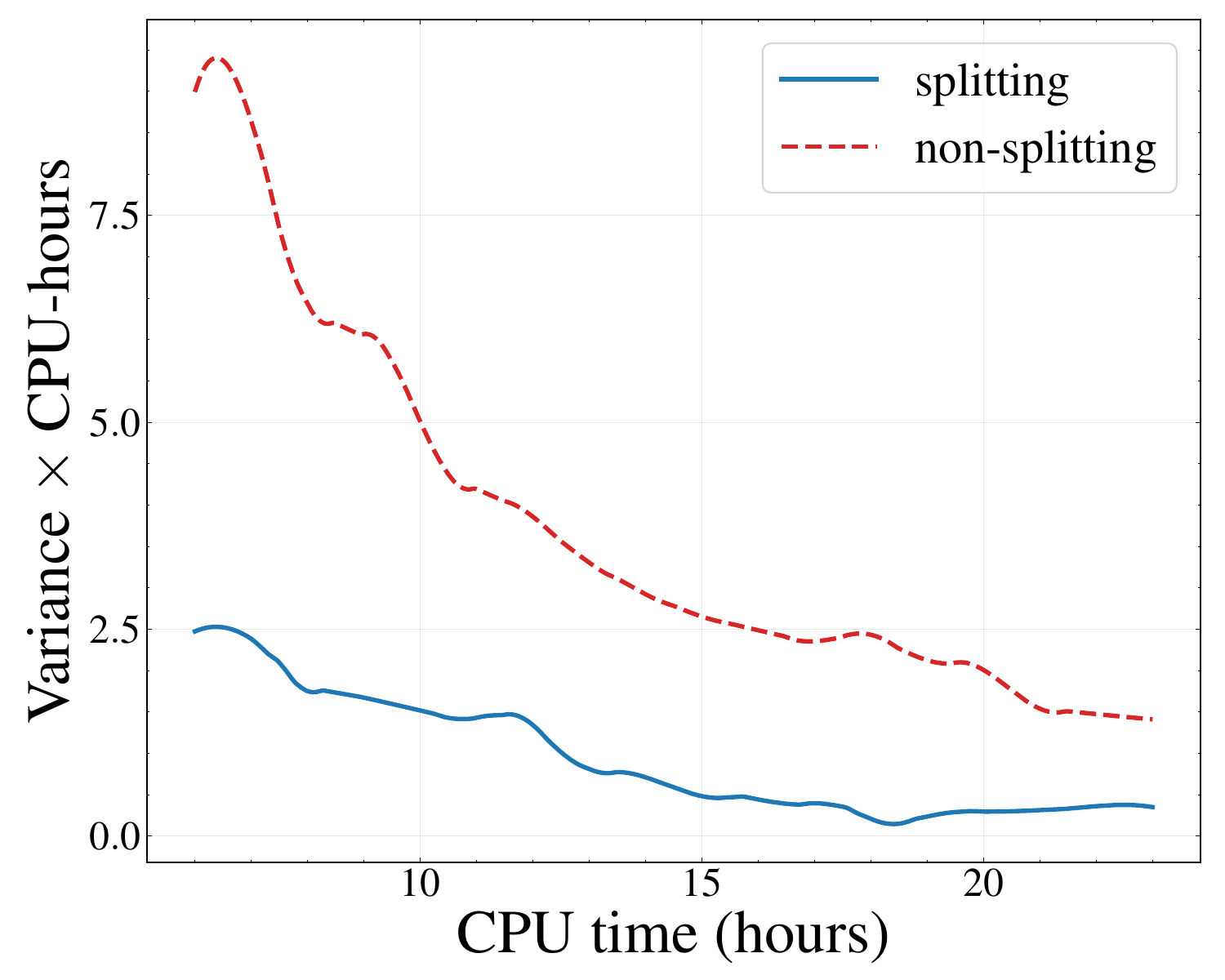}\\
    \makebox[0.19\textwidth]{\footnotesize (f) variance of matio}\hfill\makebox[0.19\textwidth]{\footnotesize (g) variance of openh264}\hfill\makebox[0.19\textwidth]{\footnotesize (h) variance of php}\hfill\makebox[0.19\textwidth]{\footnotesize (i) variance of poppler}\hfill\makebox[0.19\textwidth]{\footnotesize (j) variance of stb}
    \caption{Comparisons of variance across 10 benchmarks for fuzzer Honggfuzz.}
    \label{fig:comparison_honggfuzz}
\end{figure*}

\begin{figure*}[tp]
    \centering
    \includegraphics[width=0.19\textwidth]{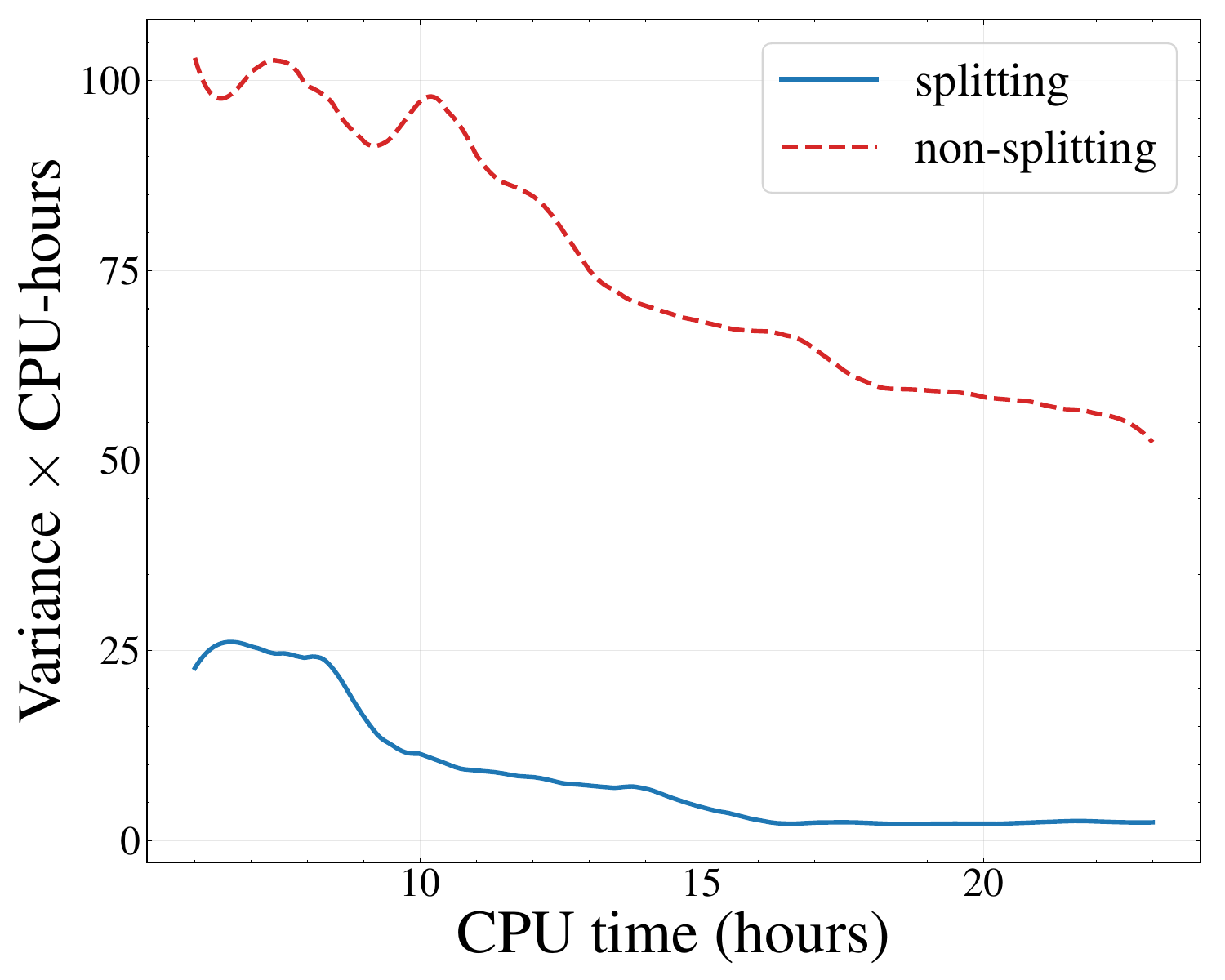}\hfill\includegraphics[width=0.19\textwidth]{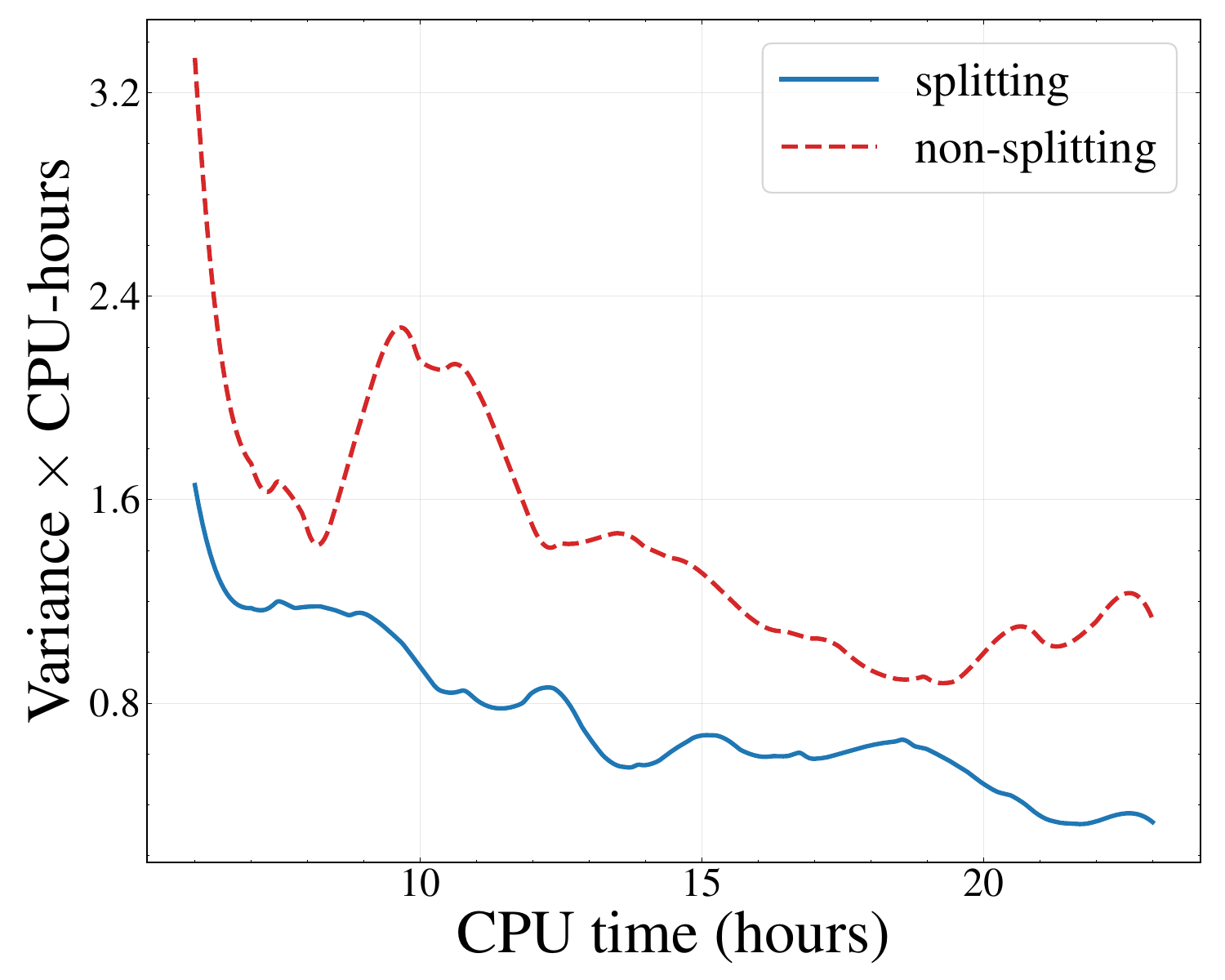}\hfill\includegraphics[width=0.19\textwidth]{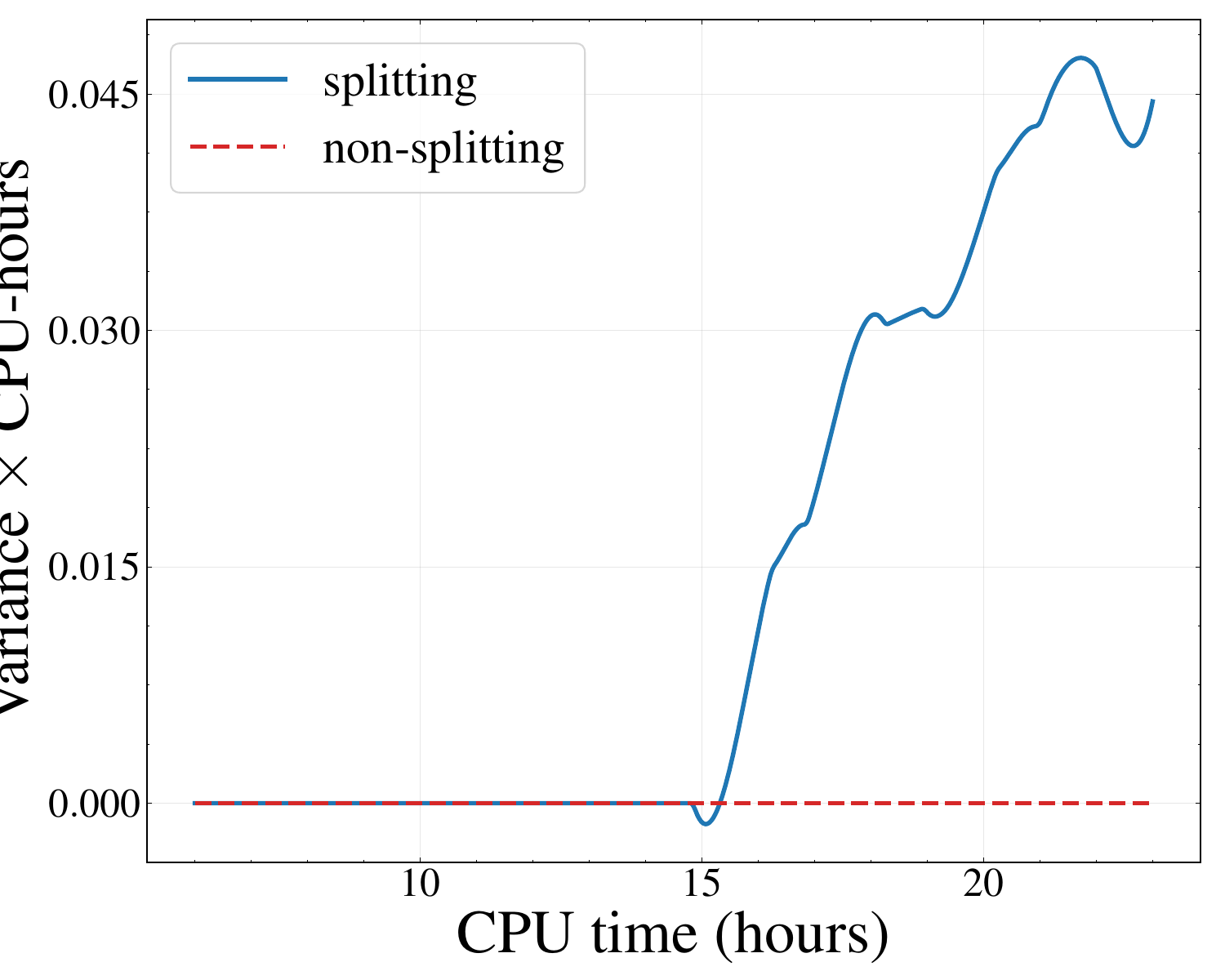}\hfill\includegraphics[width=0.19\textwidth]{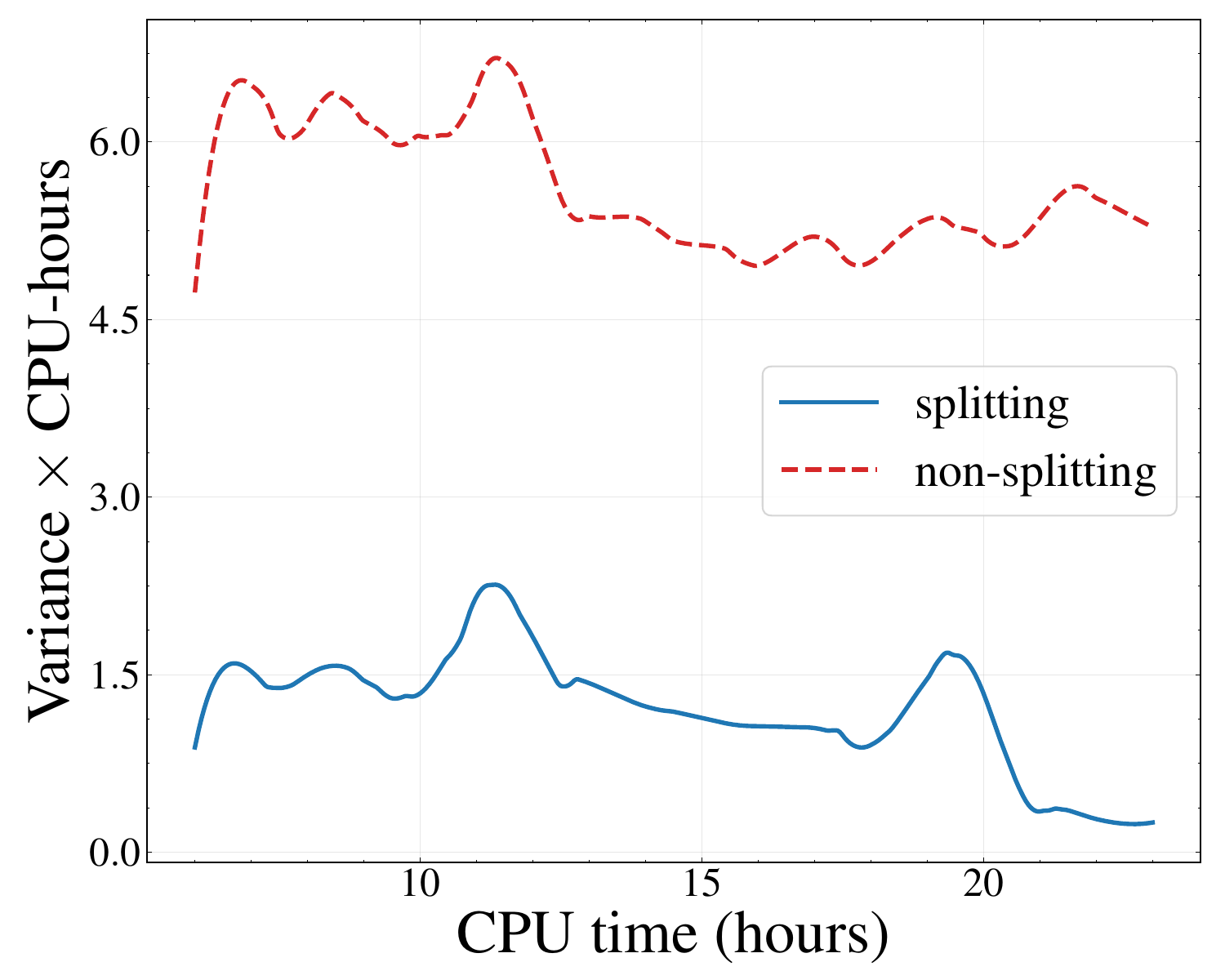}\hfill\includegraphics[width=0.19\textwidth]{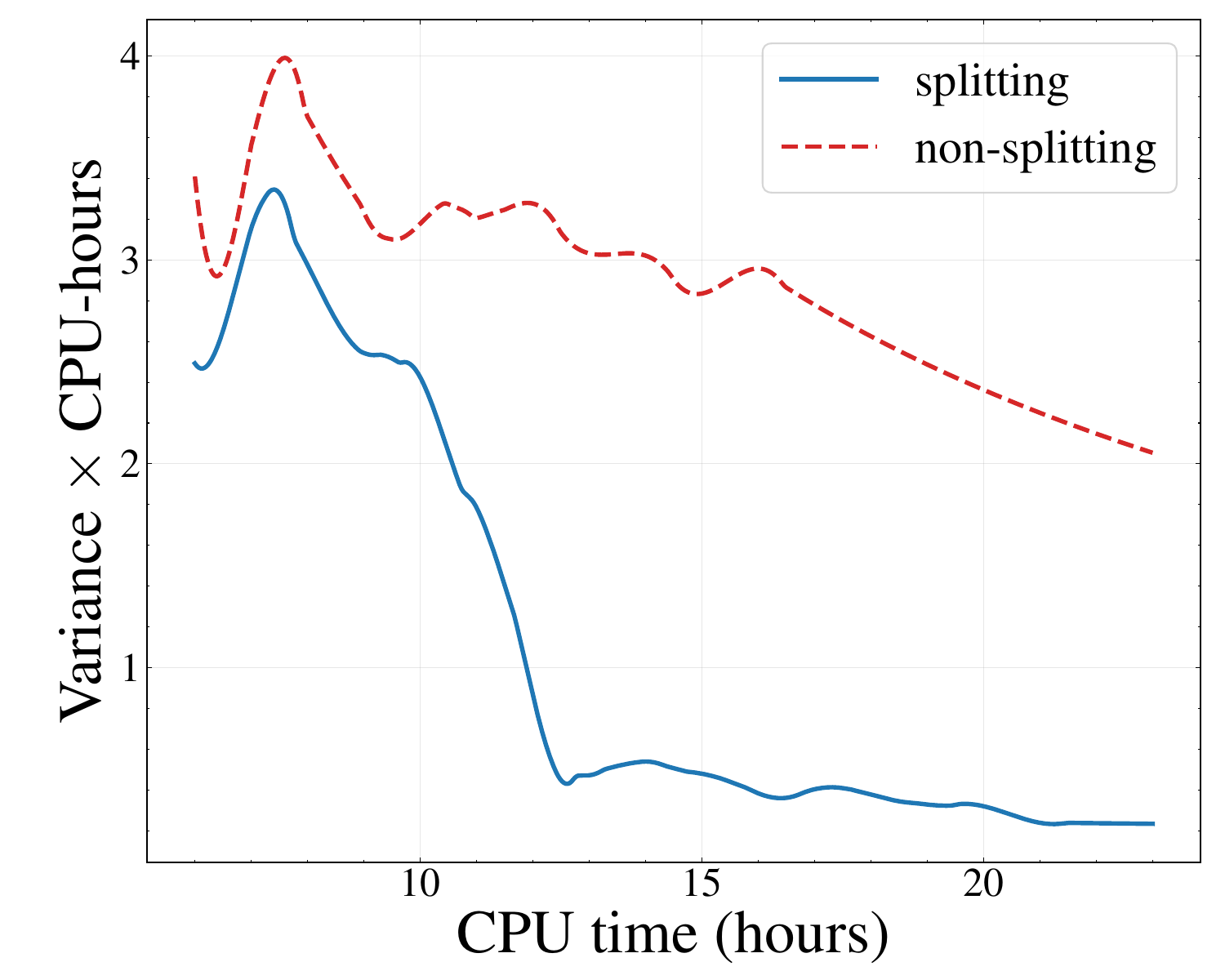}\\
    \makebox[0.19\textwidth]{\footnotesize (a) variance of arrow}\hfill\makebox[0.19\textwidth]{\footnotesize (b) variance of ffmpeg}\hfill\makebox[0.19\textwidth]{\footnotesize (c) variance of grok}\hfill\makebox[0.19\textwidth]{\footnotesize (d) variance of libhevc}\hfill\makebox[0.19\textwidth]{\footnotesize (e) variance of libhtp}\\[3pt]
    \includegraphics[width=0.19\textwidth]{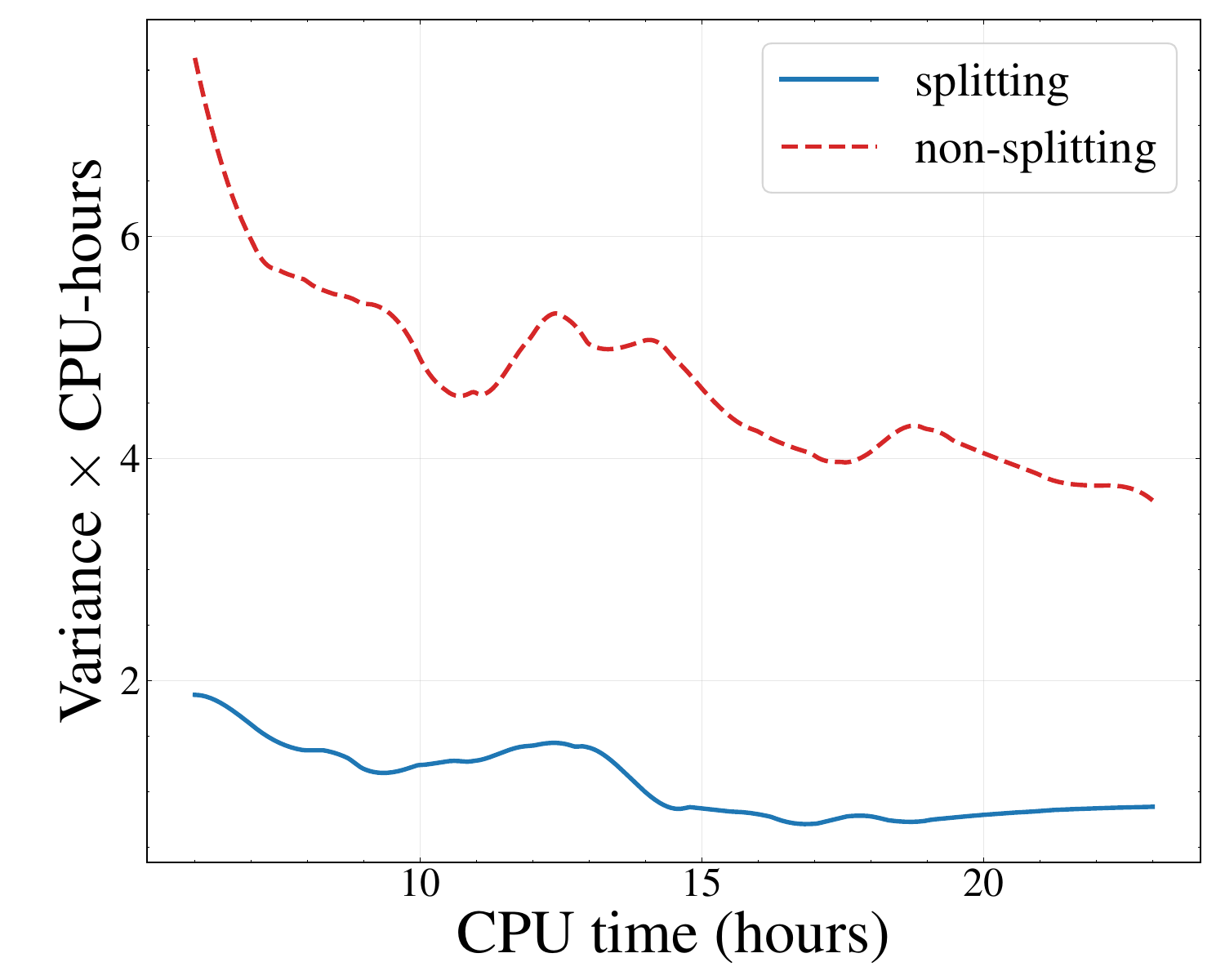}\hfill\includegraphics[width=0.19\textwidth]{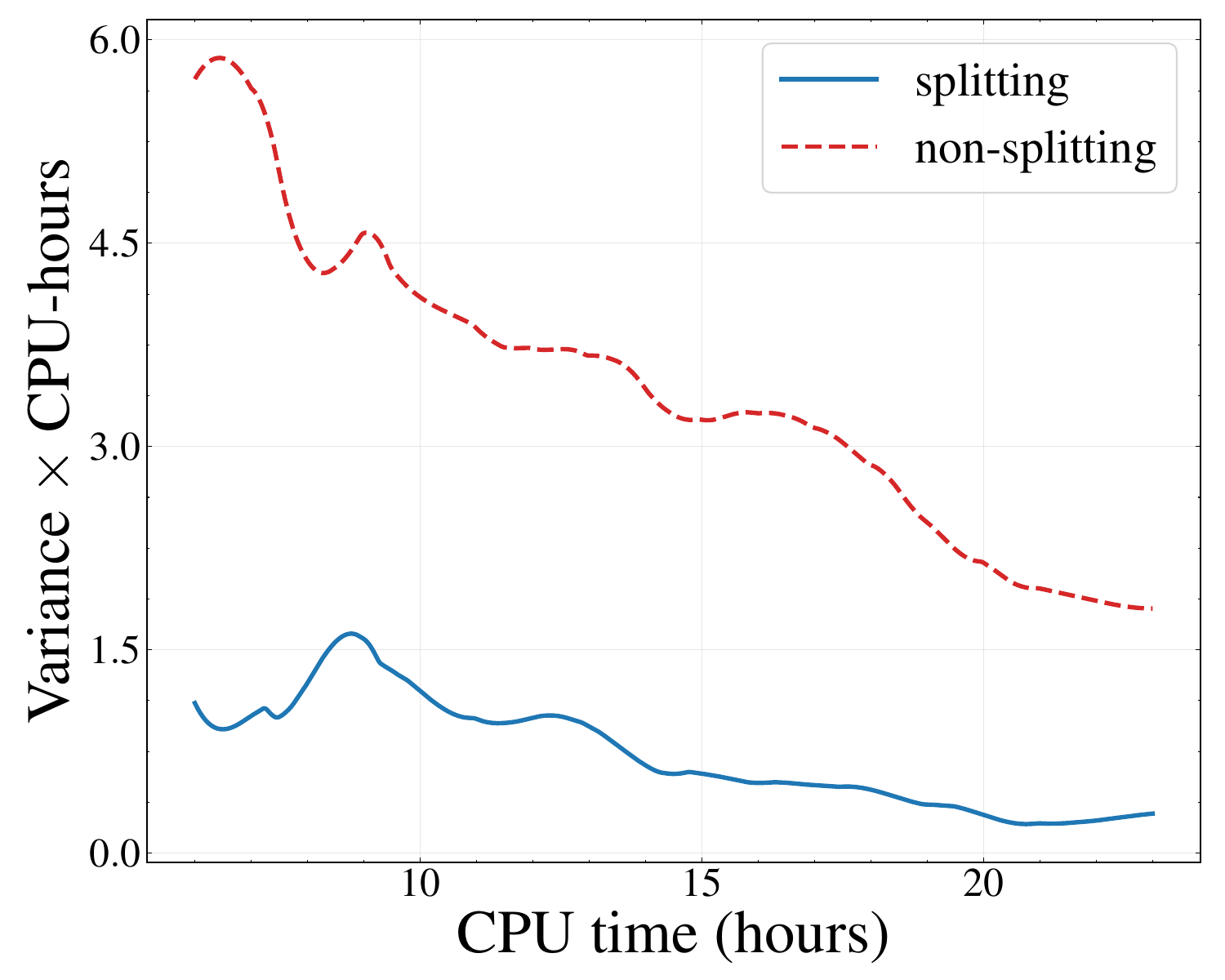}\hfill\includegraphics[width=0.19\textwidth]{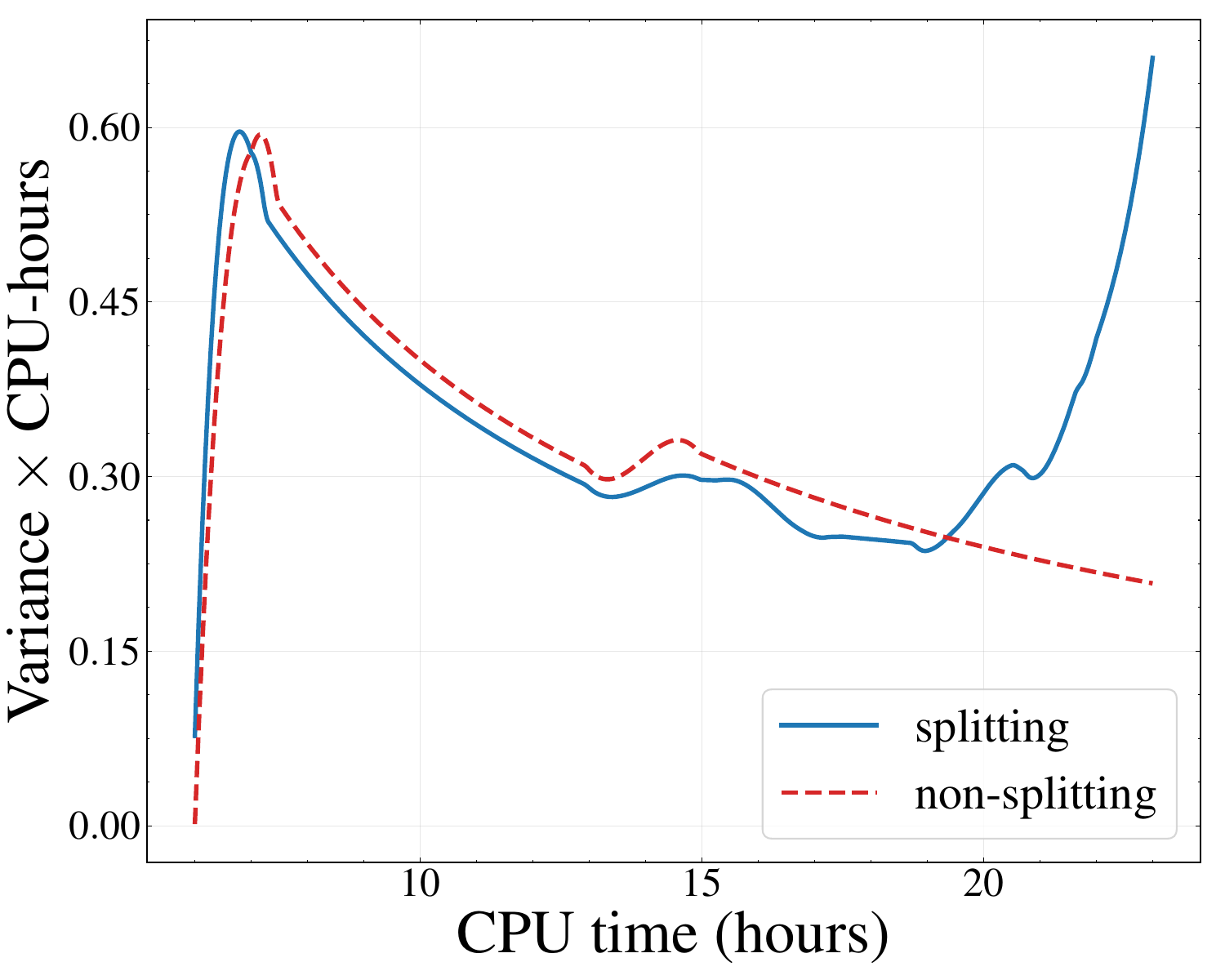}\hfill\includegraphics[width=0.19\textwidth]{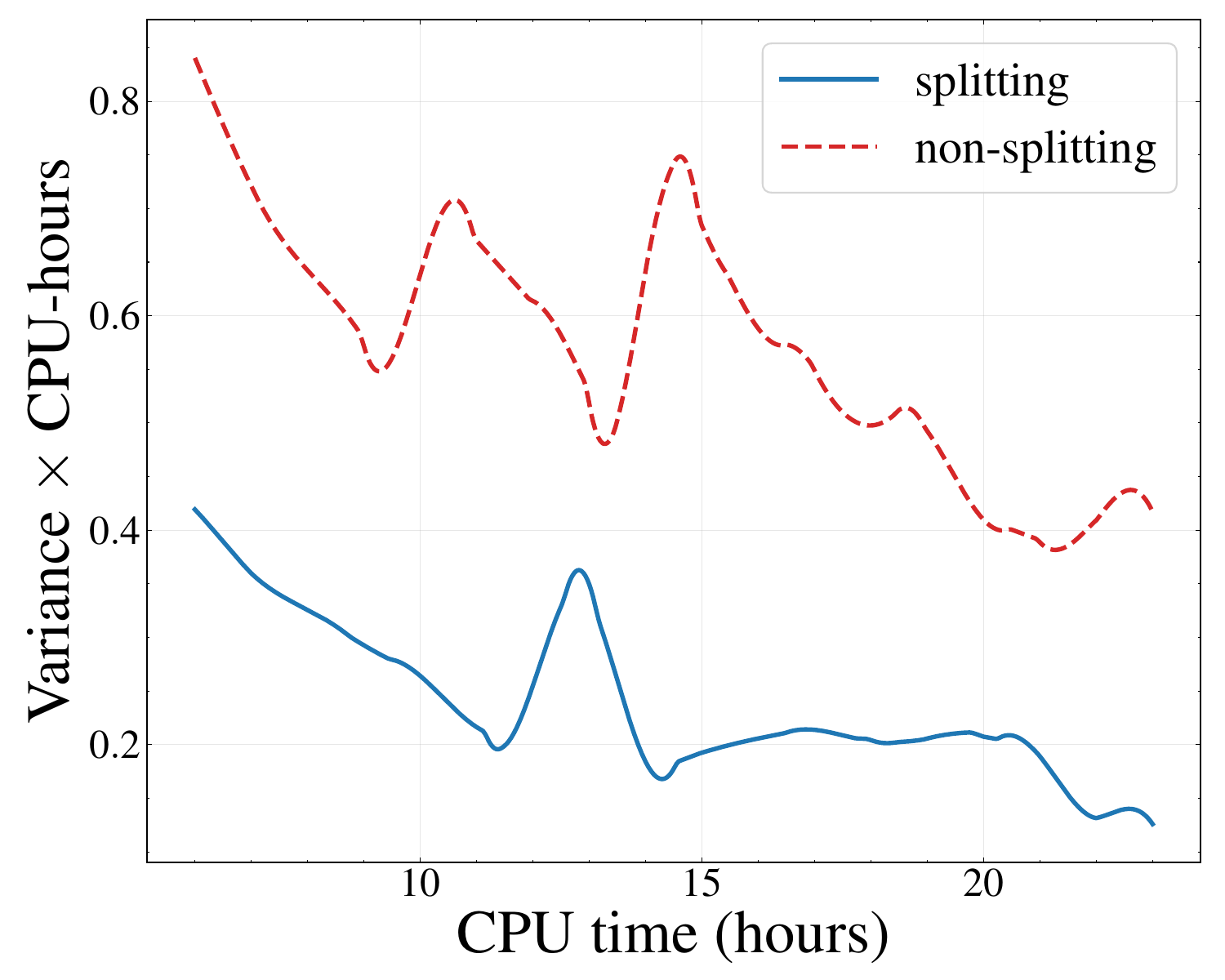}\hfill\includegraphics[width=0.19\textwidth]{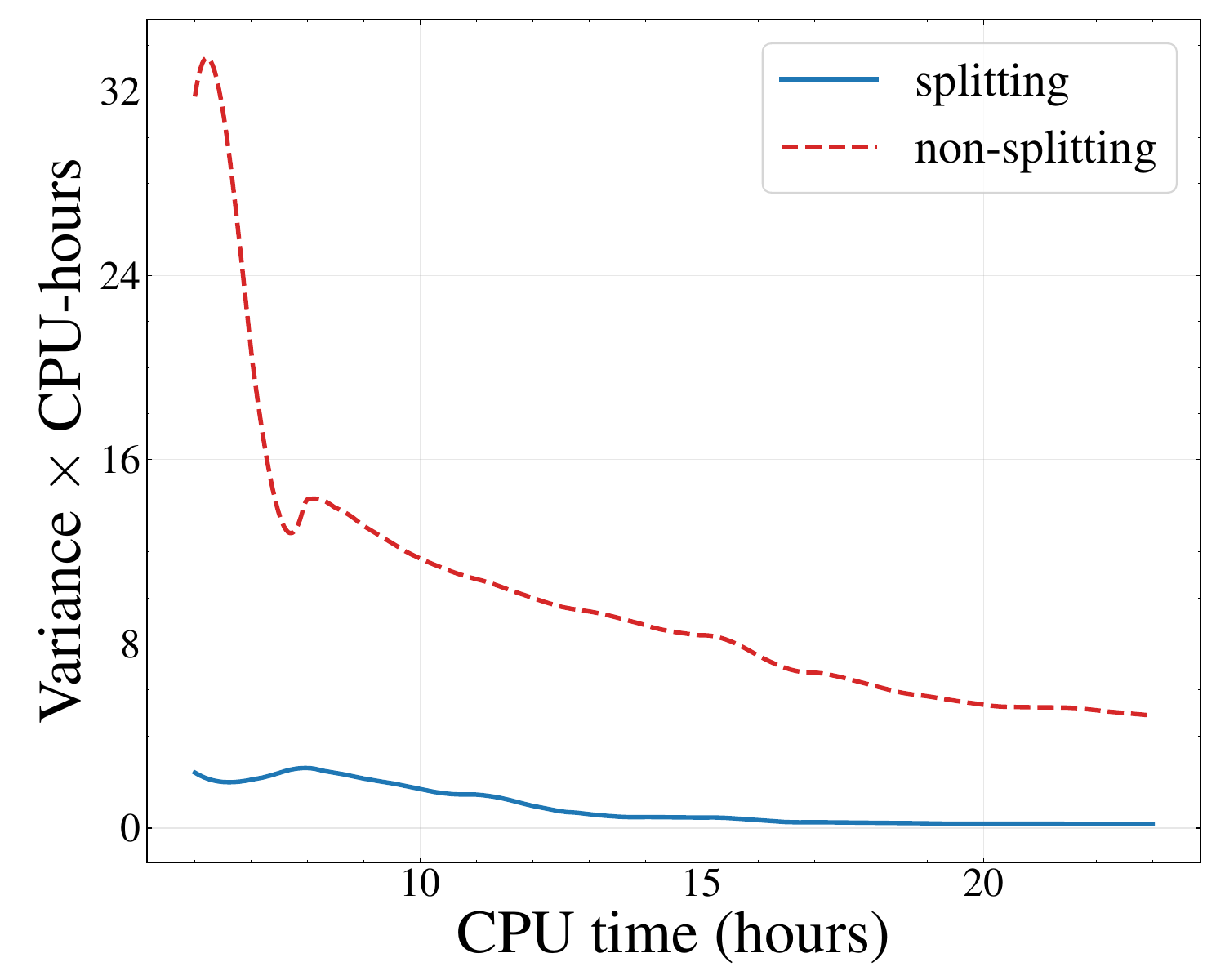}\\
    \makebox[0.19\textwidth]{\footnotesize (f) variance of matio}\hfill\makebox[0.19\textwidth]{\footnotesize (g) variance of openh264}\hfill\makebox[0.19\textwidth]{\footnotesize (h) variance of php}\hfill\makebox[0.19\textwidth]{\footnotesize (i) variance of poppler}\hfill\makebox[0.19\textwidth]{\footnotesize (j) variance of stb}
    \caption{Comparisons of variance across 10 benchmarks for fuzzer libFuzzer.}
    \label{fig:comparison_libfuzzer}
\end{figure*}

\begin{figure*}[tp]
    \centering
    \includegraphics[width=0.19\textwidth]{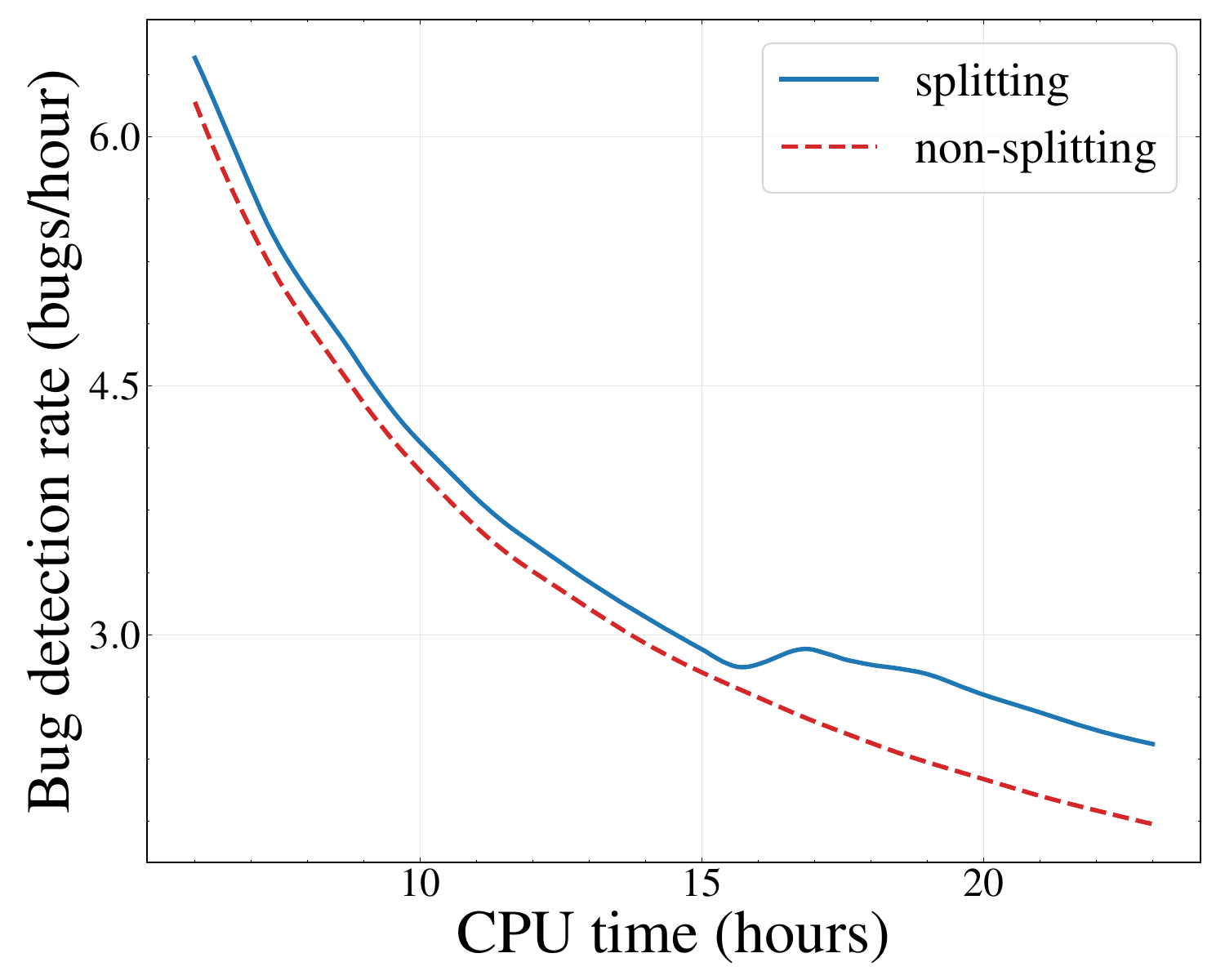}\hfill\includegraphics[width=0.19\textwidth]{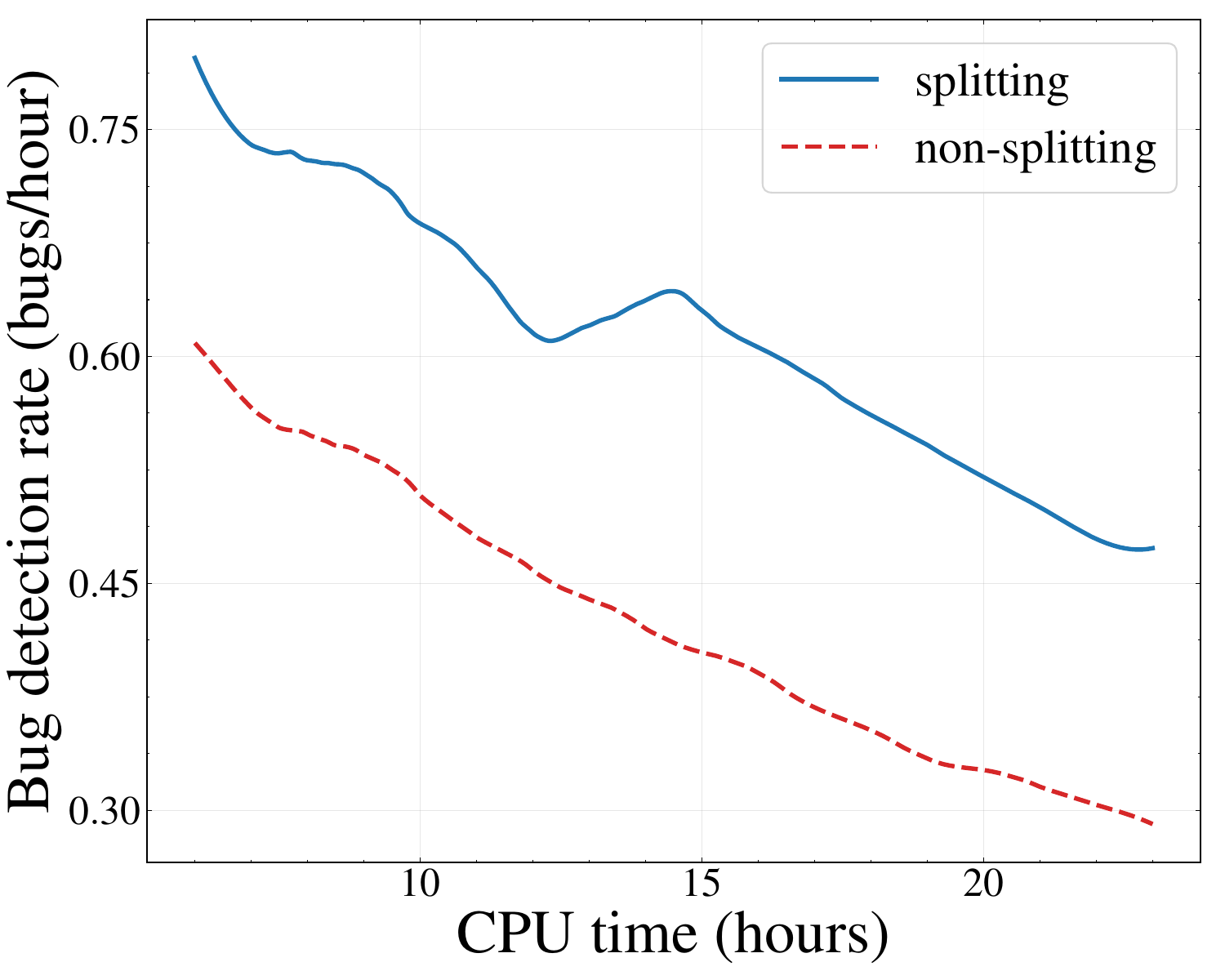}\hfill\includegraphics[width=0.19\textwidth]{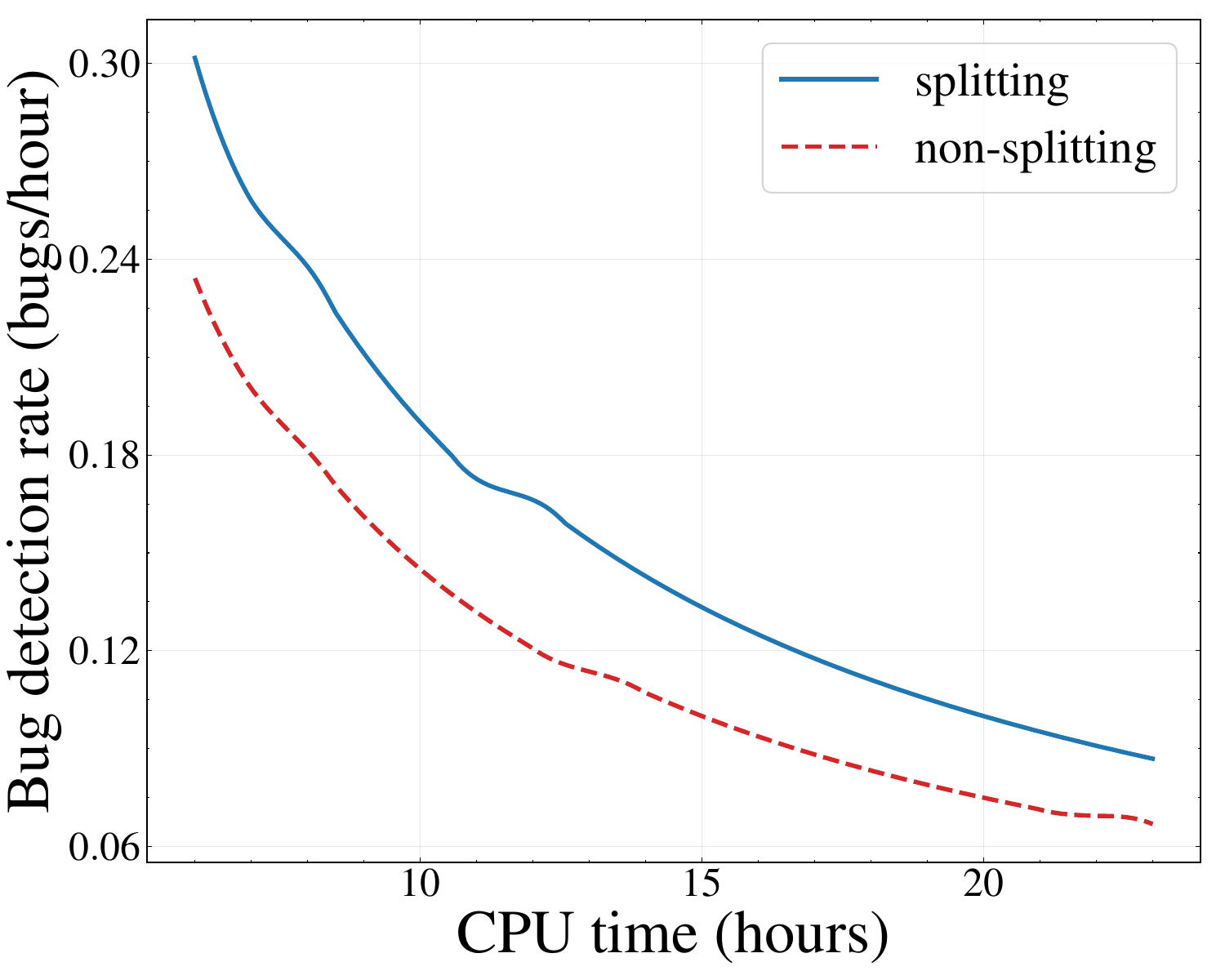}\hfill\includegraphics[width=0.19\textwidth]{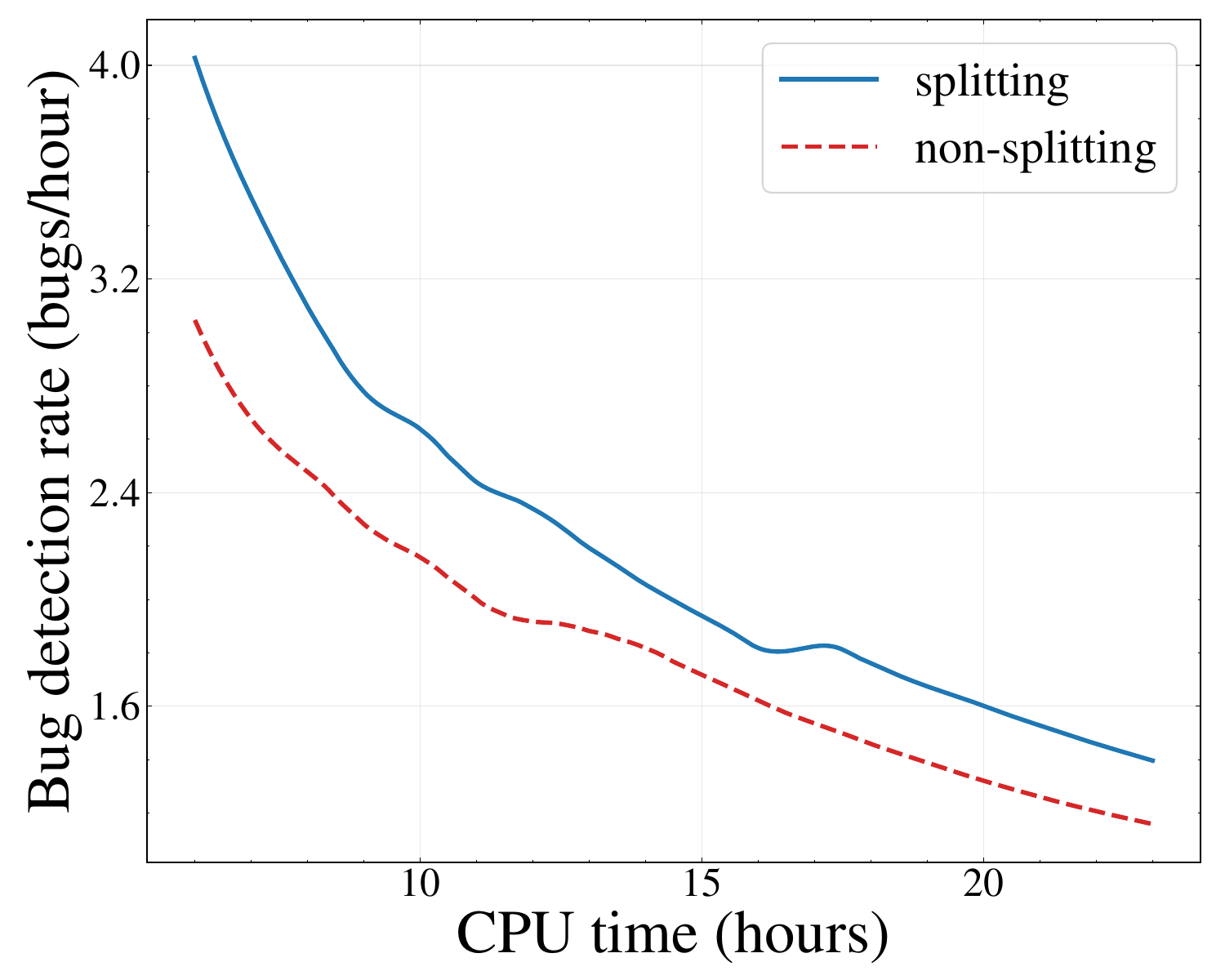}\hfill\includegraphics[width=0.19\textwidth]{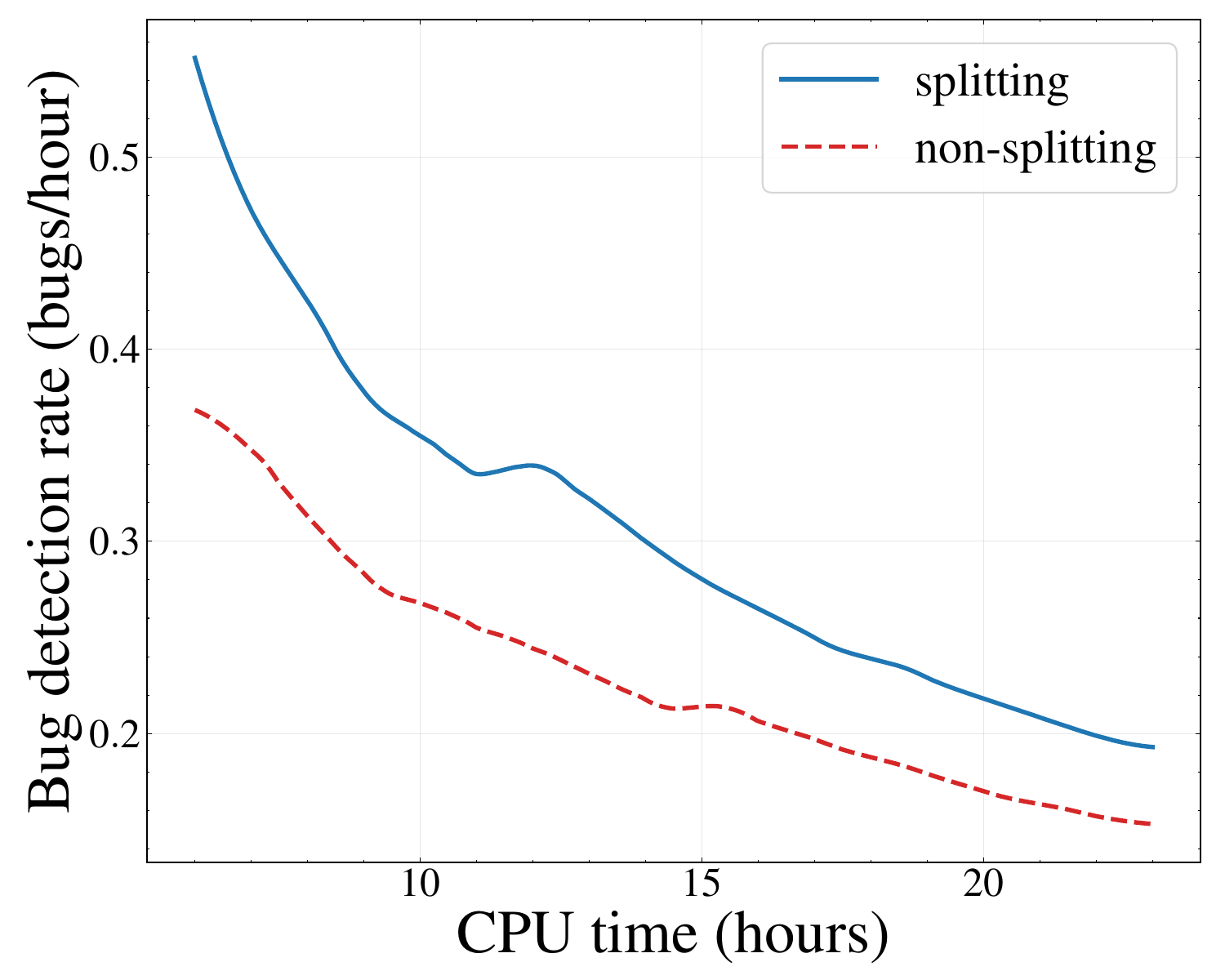}\\
    \makebox[0.19\textwidth]{\footnotesize (a) BDR of arrow}\hfill\makebox[0.19\textwidth]{\footnotesize (b) BDR of ffmpeg}\hfill\makebox[0.19\textwidth]{\footnotesize (c) BDR of grok}\hfill\makebox[0.19\textwidth]{\footnotesize (d) BDR of libhevc}\hfill\makebox[0.19\textwidth]{\footnotesize (e) BDR of libhtp}\\[3pt]
    \includegraphics[width=0.19\textwidth]{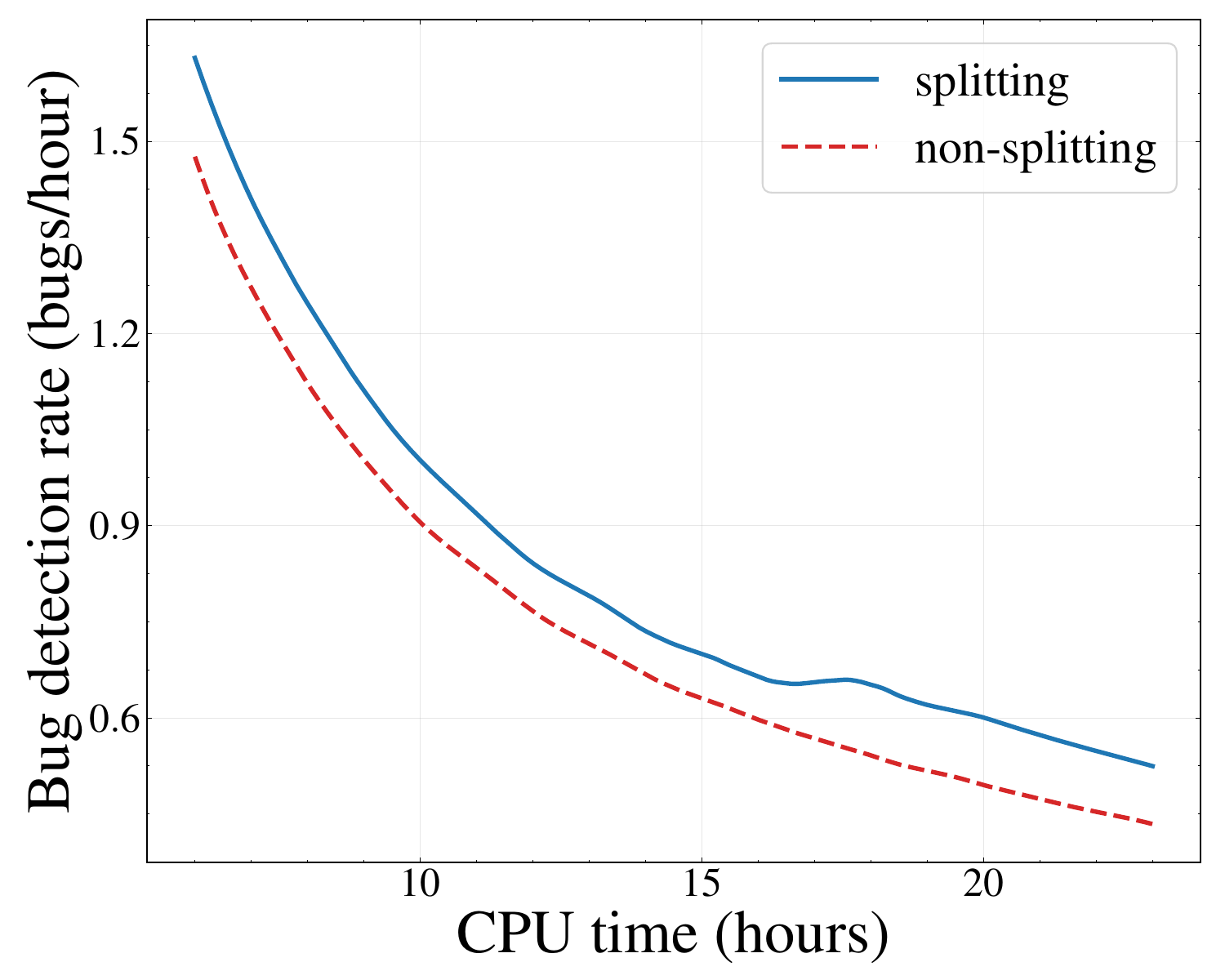}\hfill\includegraphics[width=0.19\textwidth]{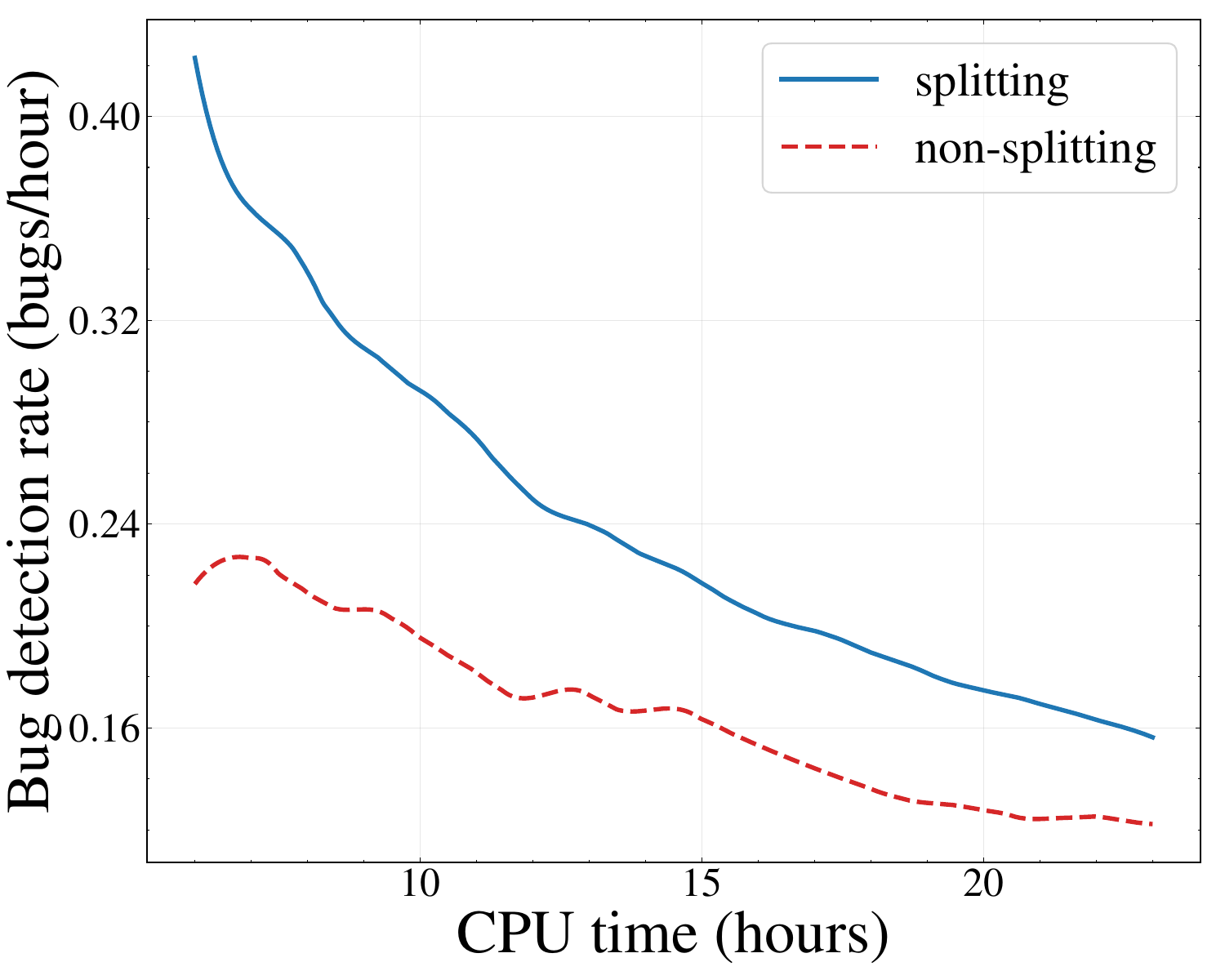}\hfill\includegraphics[width=0.19\textwidth]{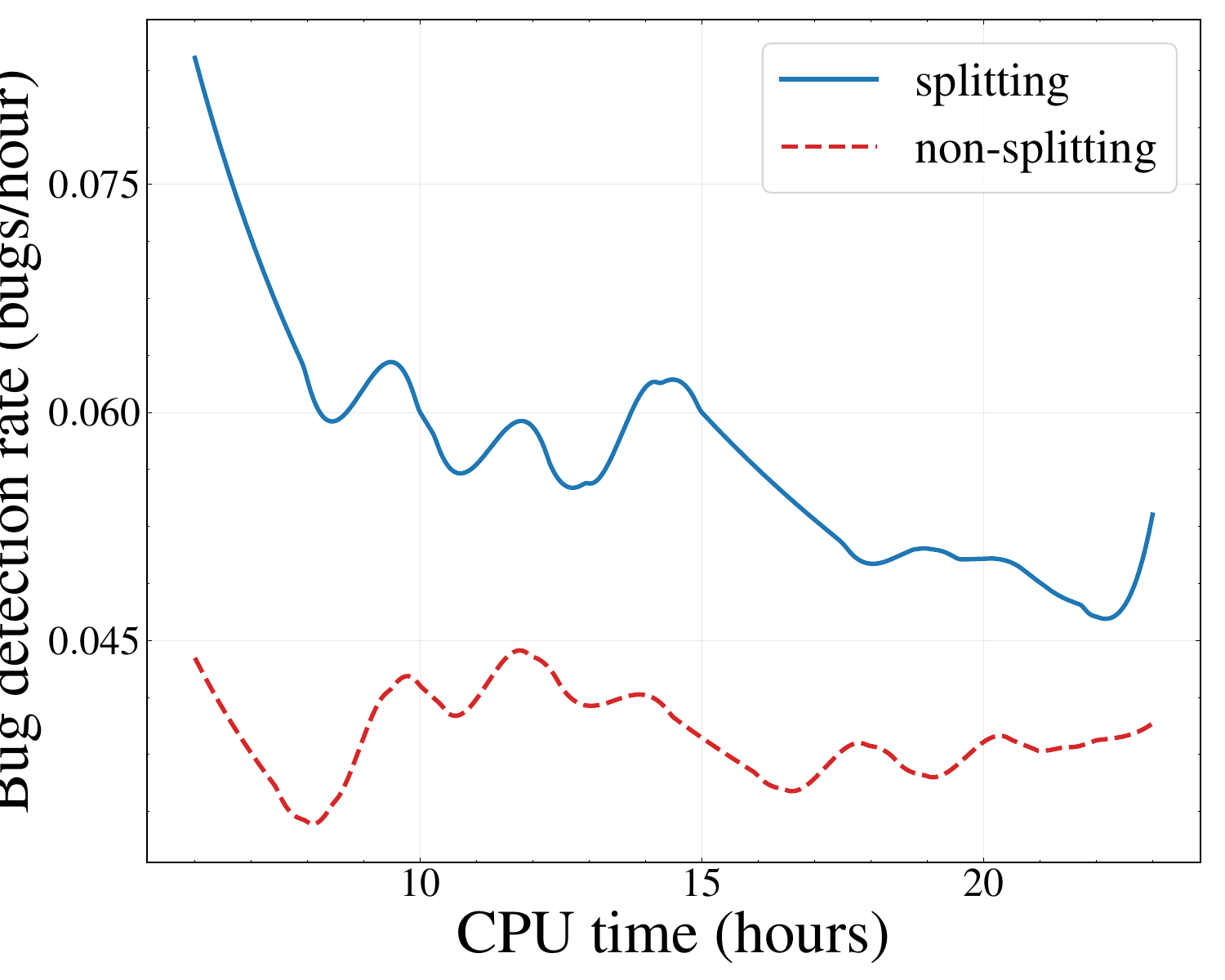}\hfill\includegraphics[width=0.19\textwidth]{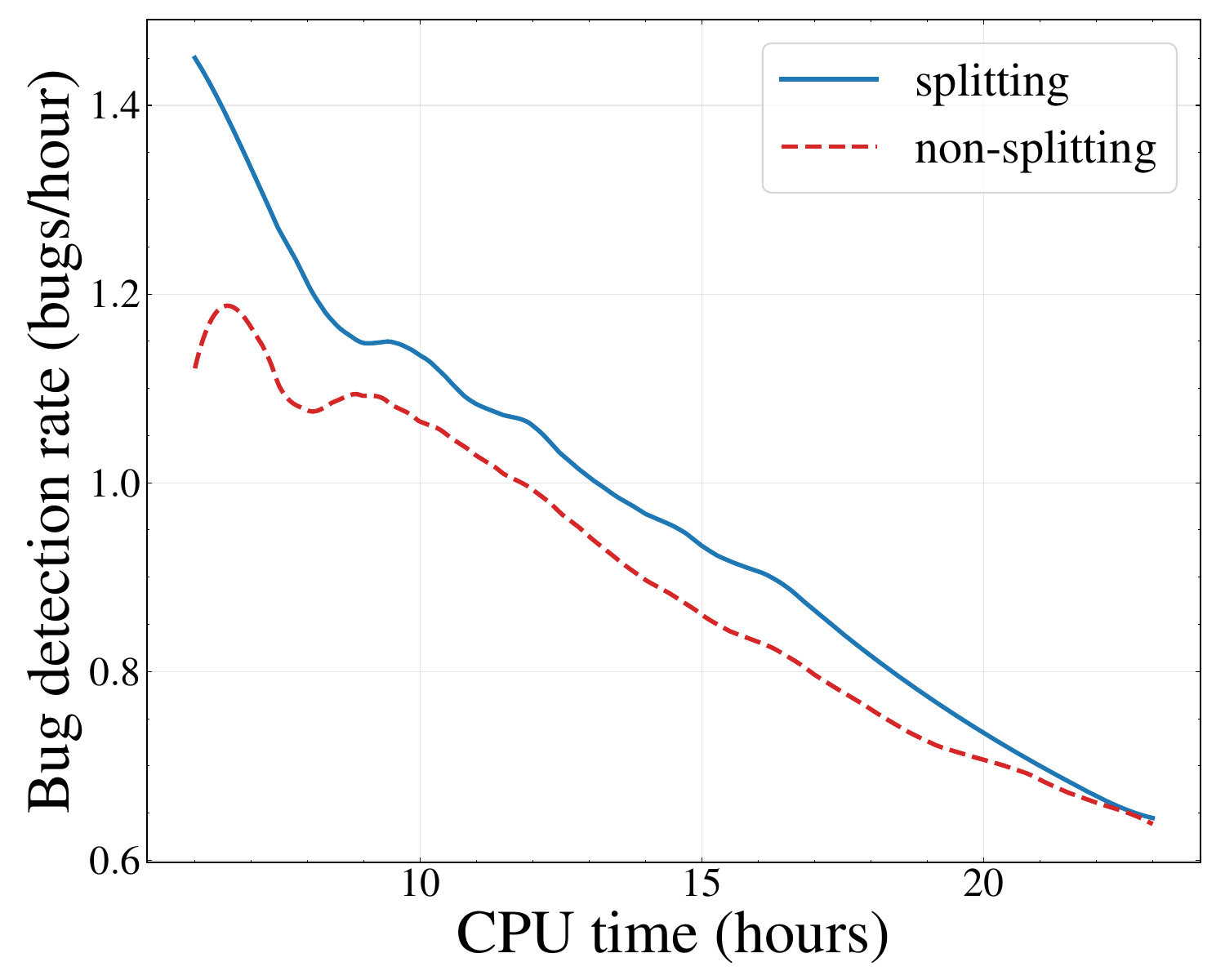}\hfill\includegraphics[width=0.19\textwidth]{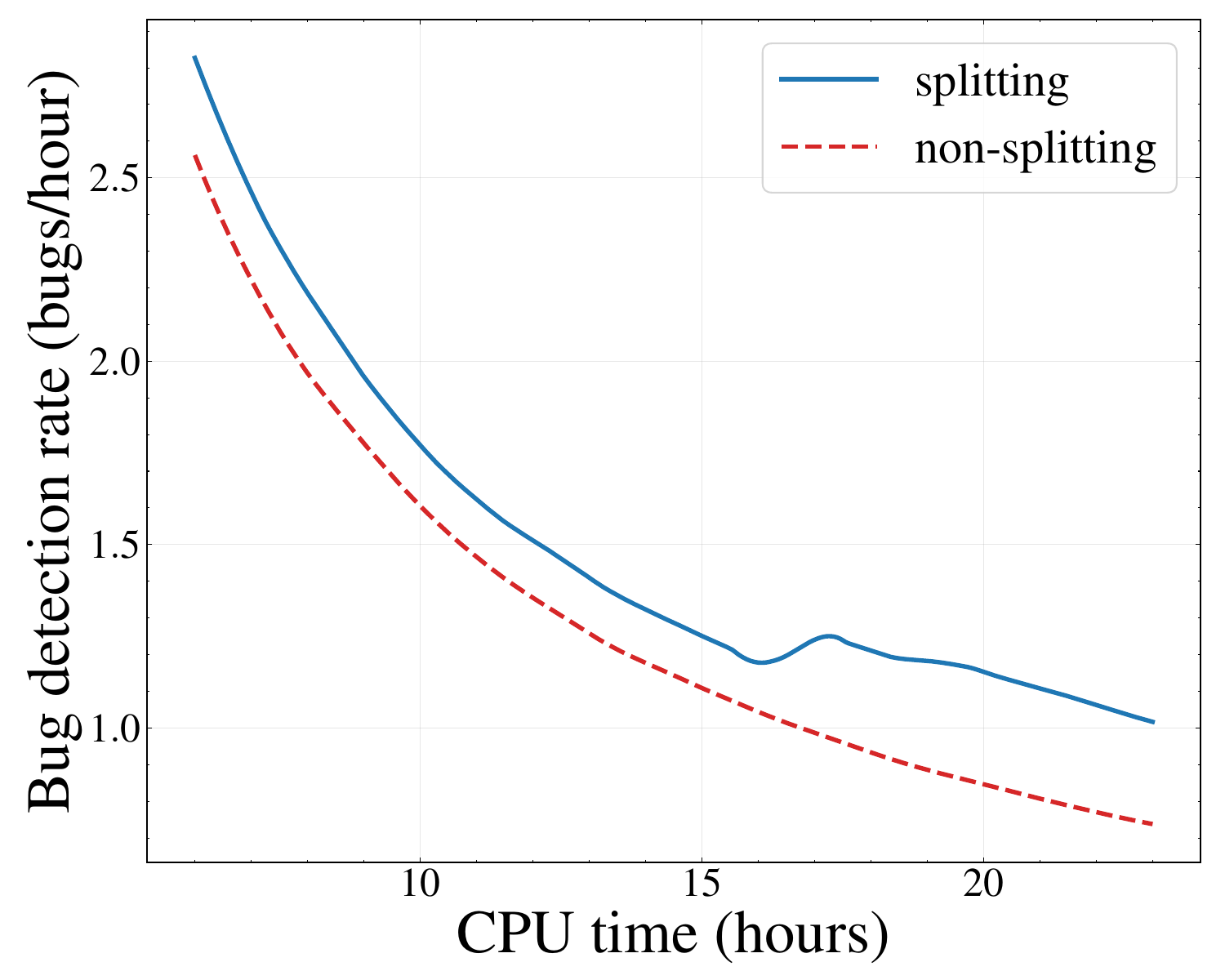}\\
    \makebox[0.19\textwidth]{\footnotesize (f) BDR of matio}\hfill\makebox[0.19\textwidth]{\footnotesize (g) BDR of openh264}\hfill\makebox[0.19\textwidth]{\footnotesize (h) BDR of php}\hfill\makebox[0.19\textwidth]{\footnotesize (i) BDR of poppler}\hfill\makebox[0.19\textwidth]{\footnotesize (j) BDR of stb}
    \caption{Comparisons of bug detection rate across 10 benchmarks for fuzzer MOpt.}
    \label{fig:bdr_mopt}
\end{figure*}

\begin{figure*}[tp]
    \centering
    \includegraphics[width=0.19\textwidth]{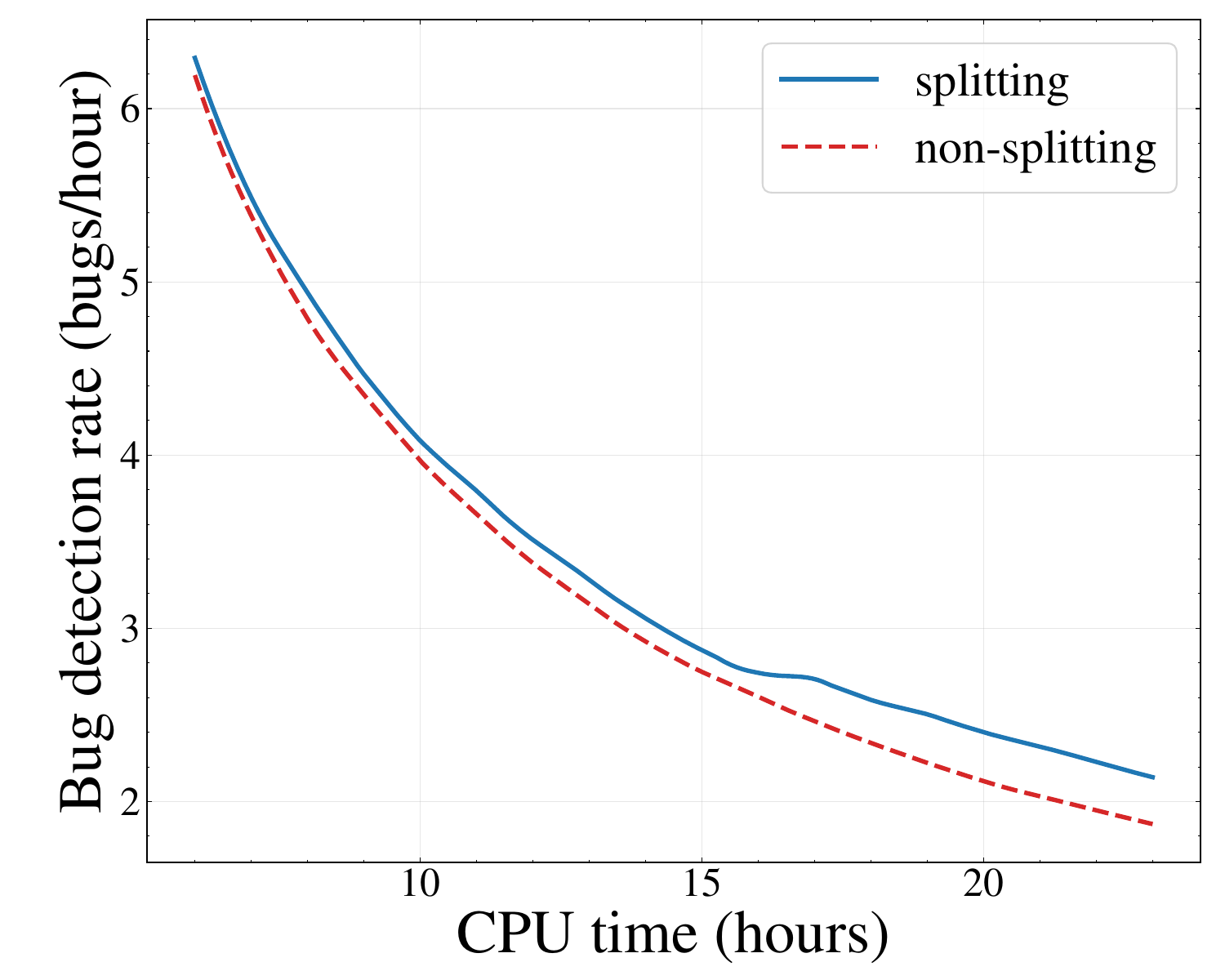}\hfill\includegraphics[width=0.19\textwidth]{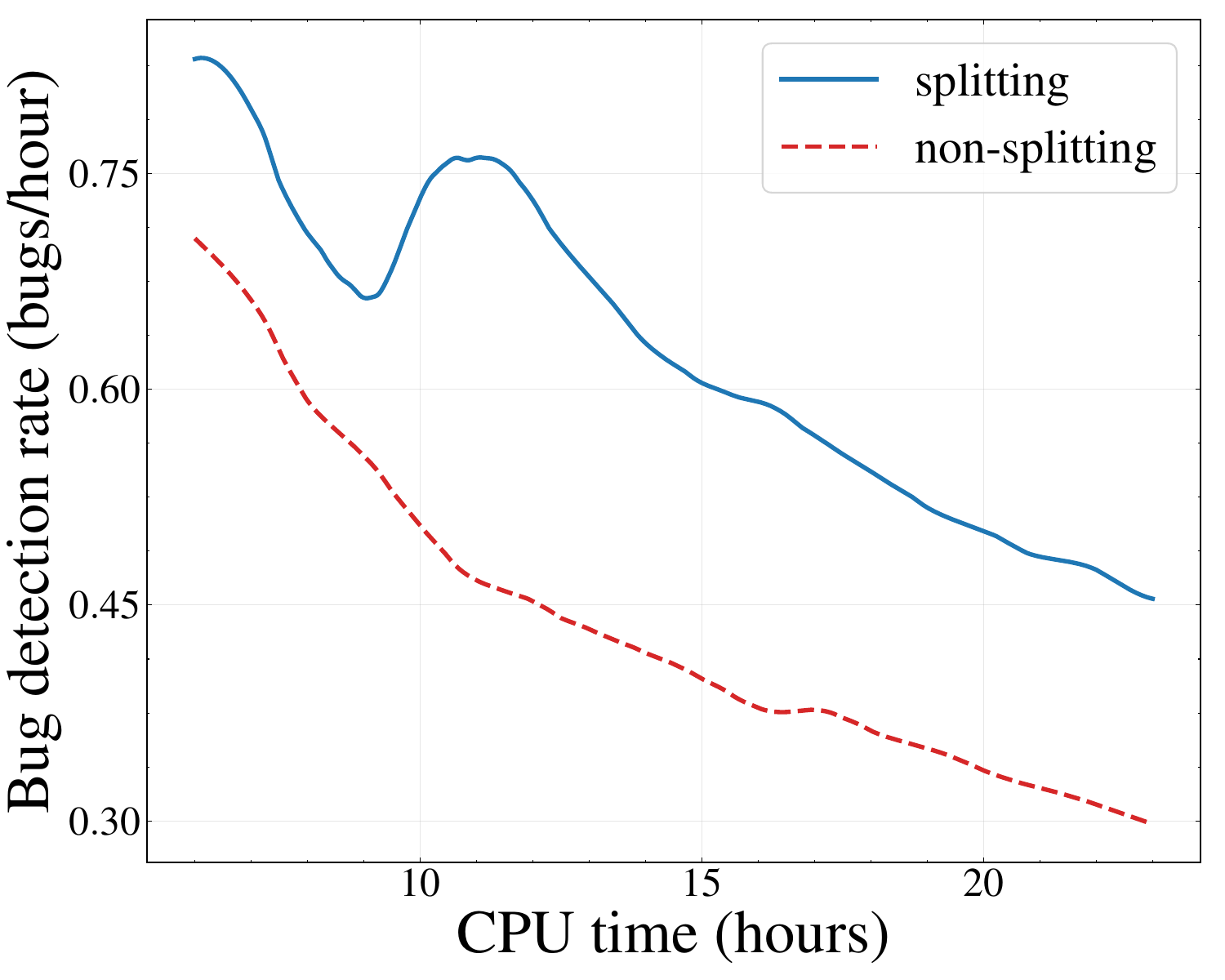}\hfill\includegraphics[width=0.19\textwidth]{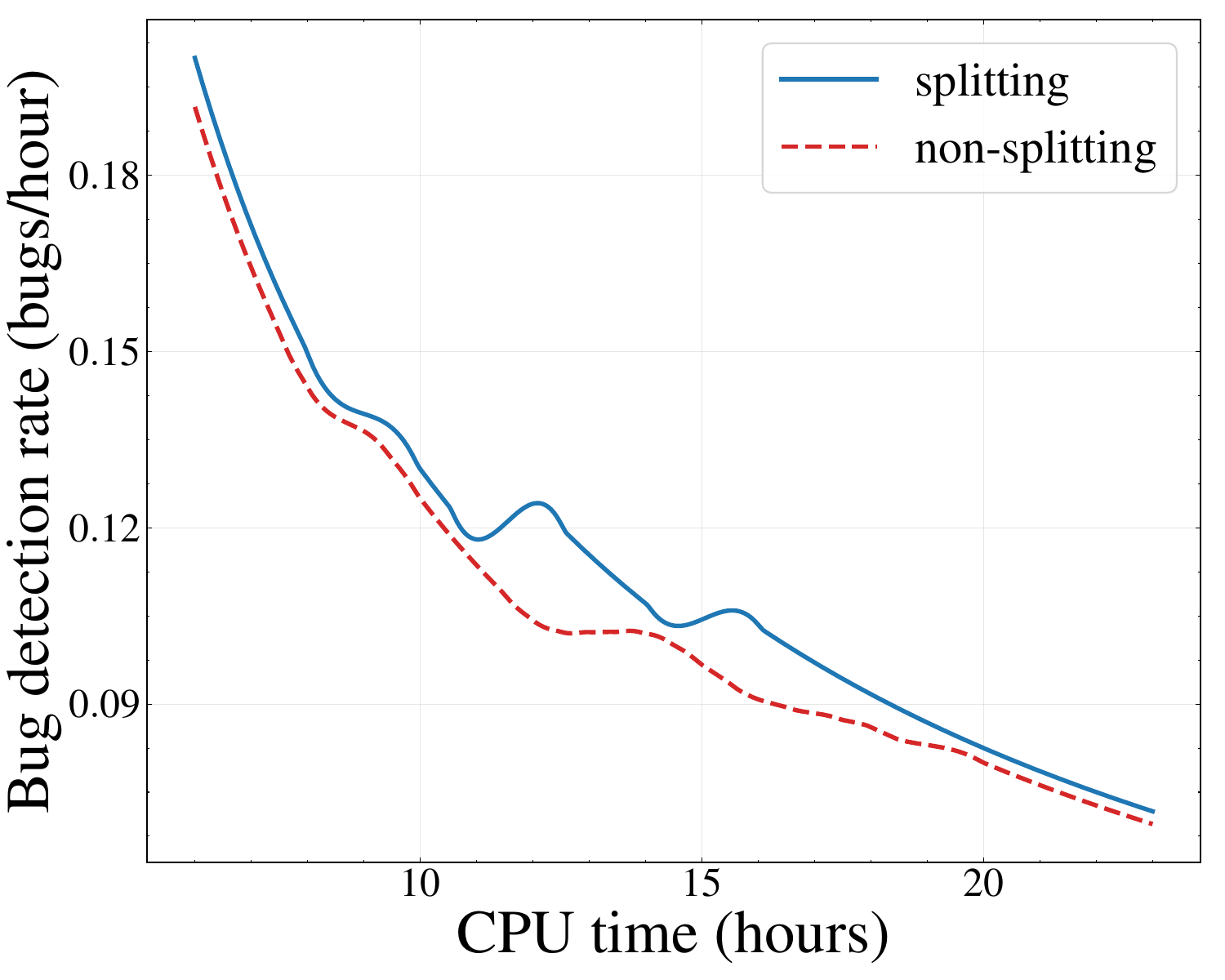}\hfill\includegraphics[width=0.19\textwidth]{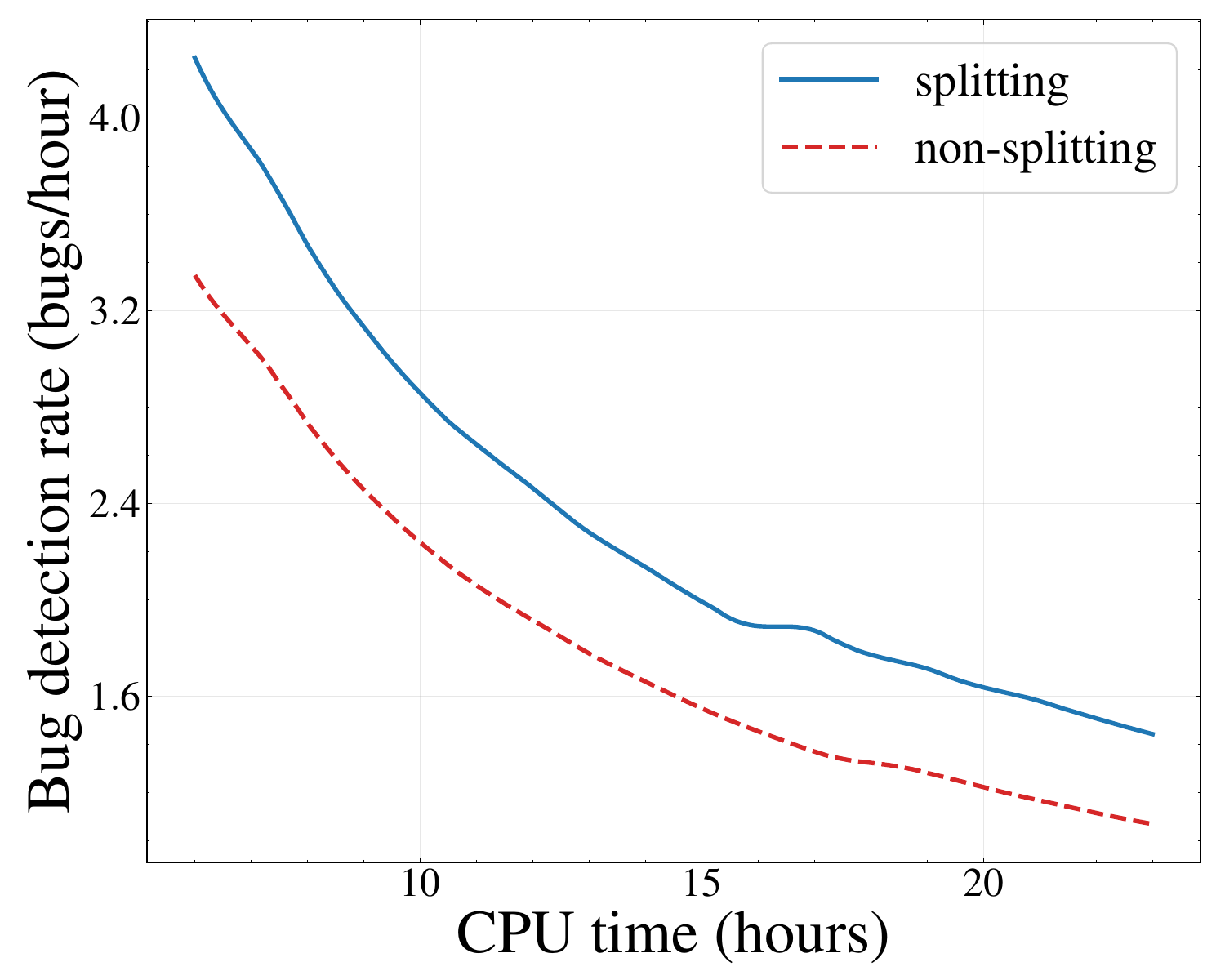}\hfill\includegraphics[width=0.19\textwidth]{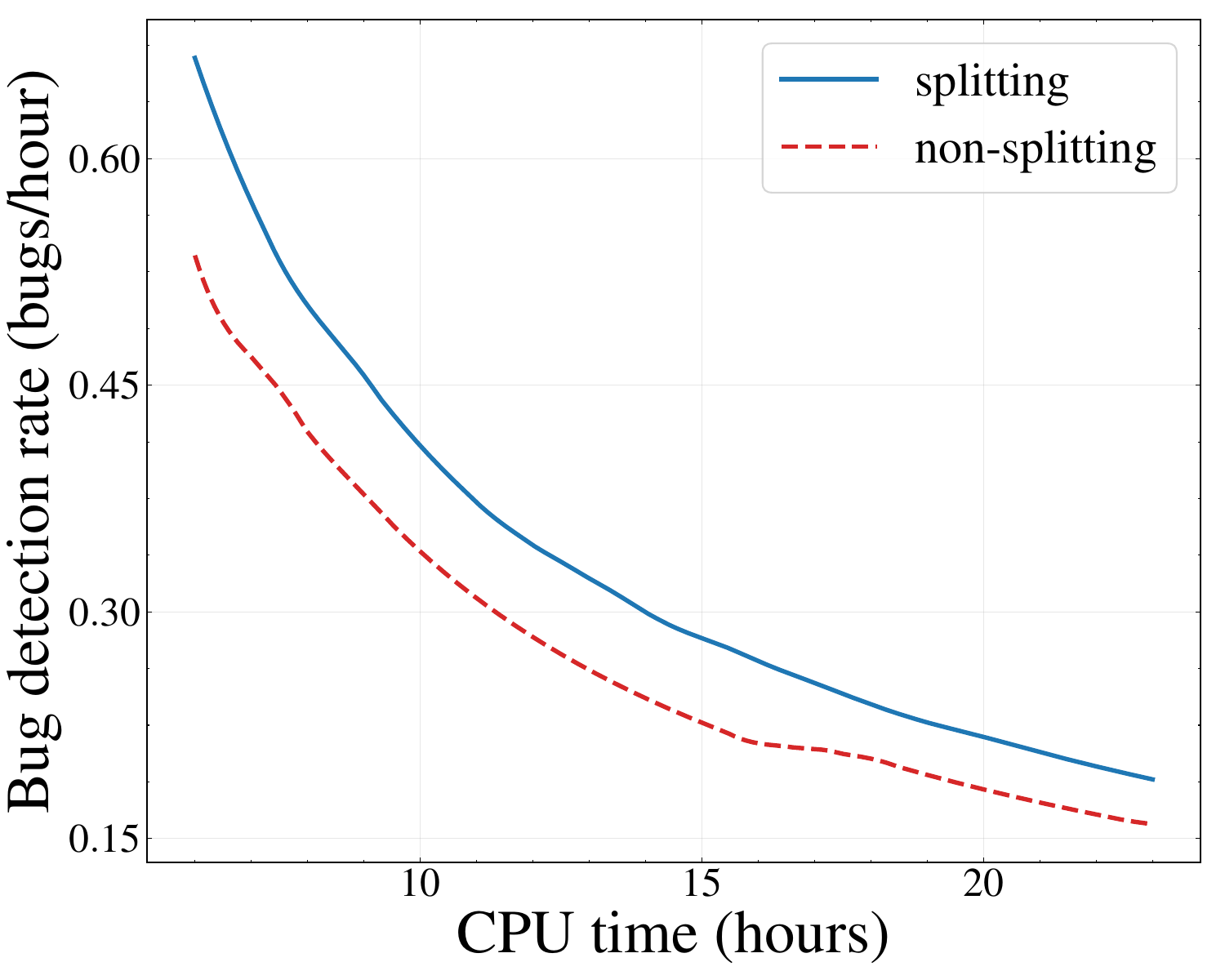}\\
    \makebox[0.19\textwidth]{\footnotesize (a) BDR of arrow}\hfill\makebox[0.19\textwidth]{\footnotesize (b) BDR of ffmpeg}\hfill\makebox[0.19\textwidth]{\footnotesize (c) BDR of grok}\hfill\makebox[0.19\textwidth]{\footnotesize (d) BDR of libhevc}\hfill\makebox[0.19\textwidth]{\footnotesize (e) BDR of libhtp}\\[3pt]
    \includegraphics[width=0.19\textwidth]{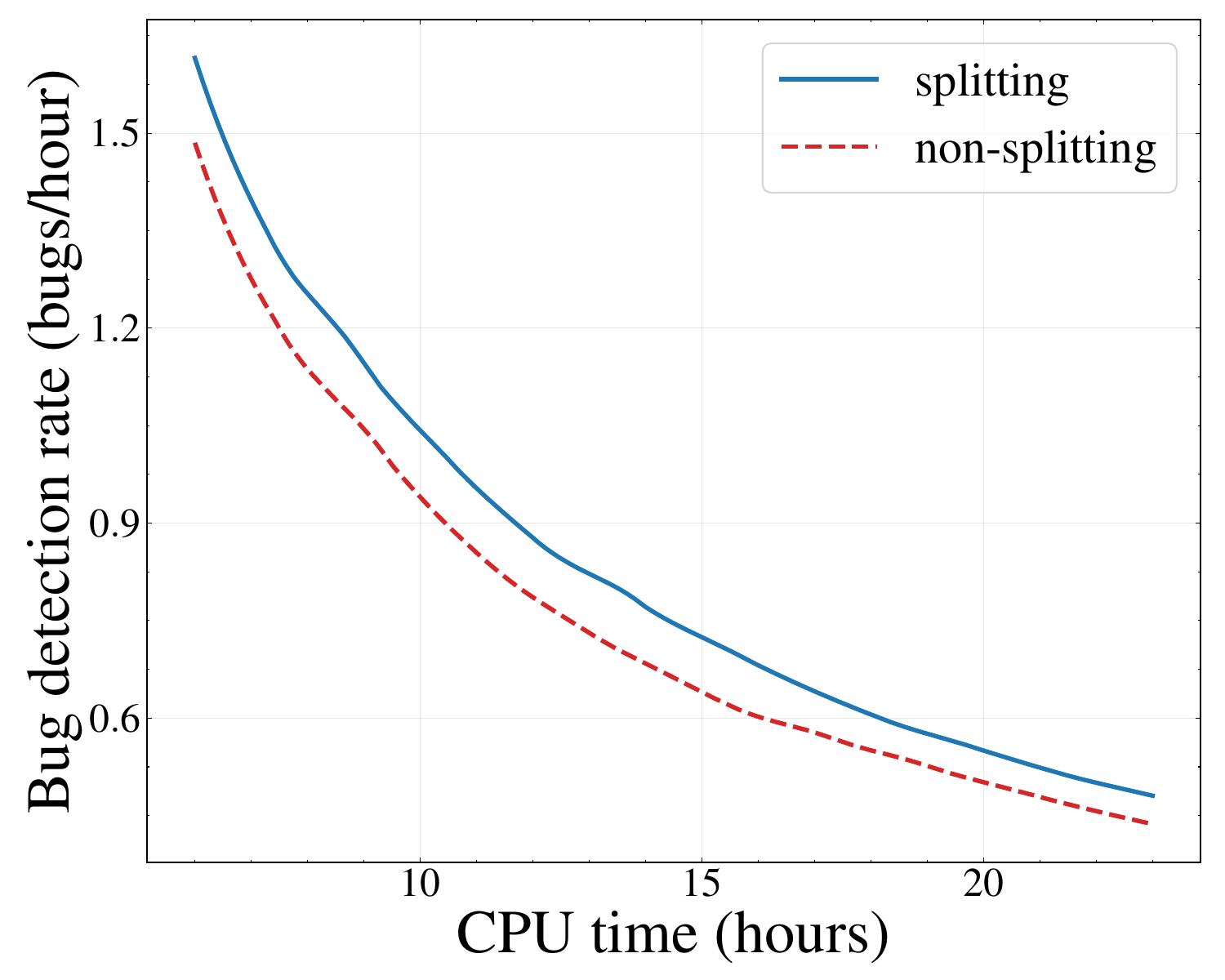}\hfill\includegraphics[width=0.19\textwidth]{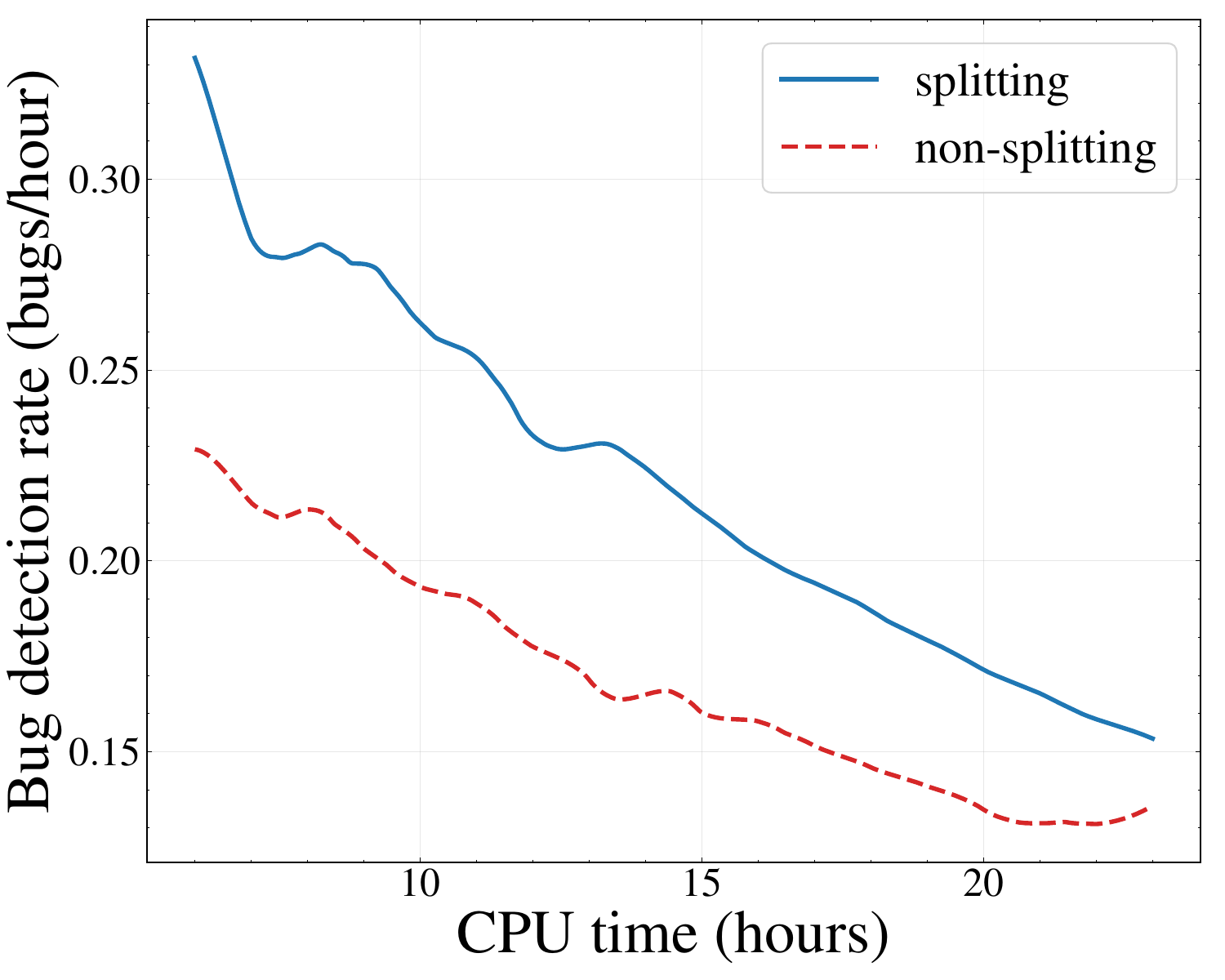}\hfill\includegraphics[width=0.19\textwidth]{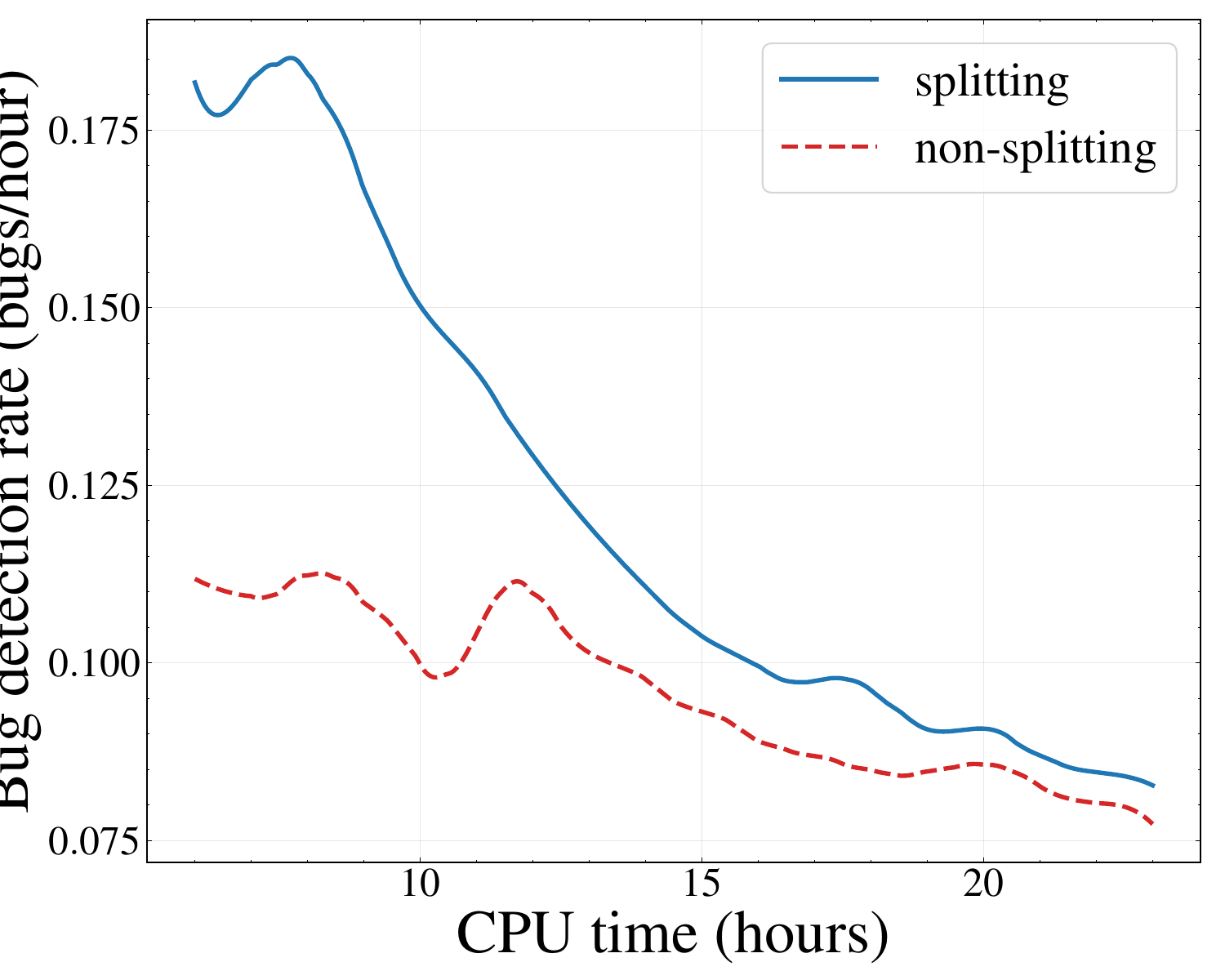}\hfill\includegraphics[width=0.19\textwidth]{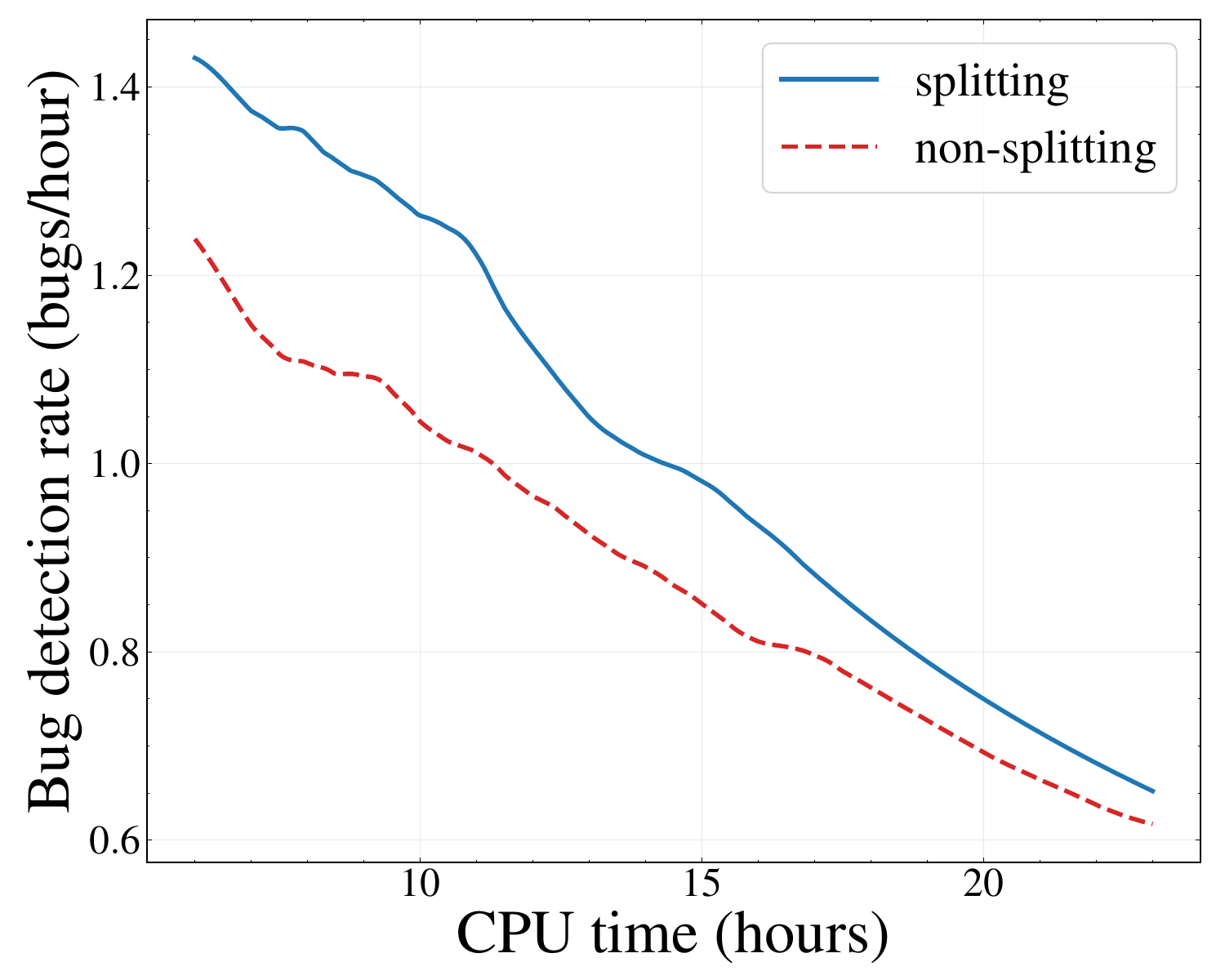}\hfill\includegraphics[width=0.19\textwidth]{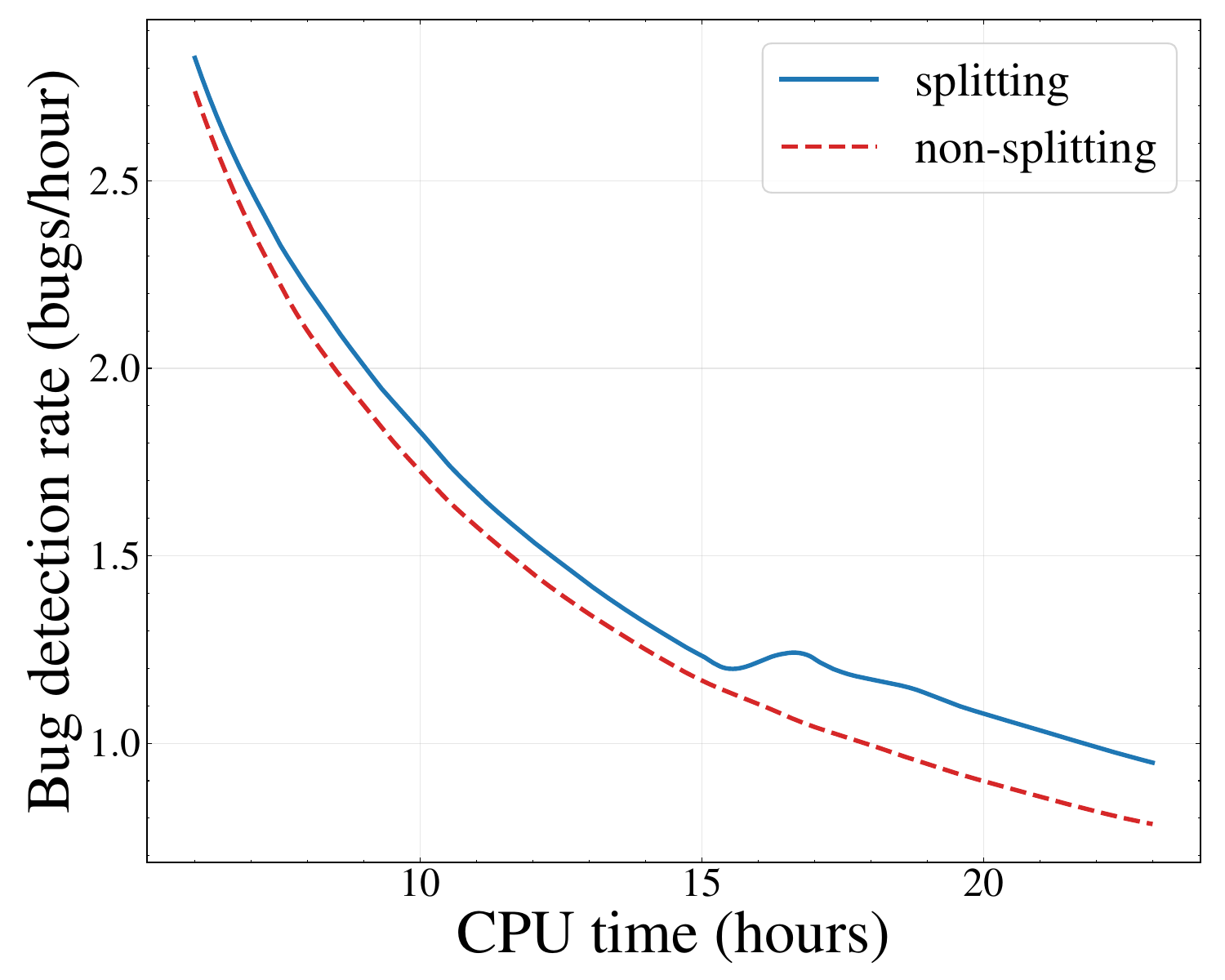}\\
    \makebox[0.19\textwidth]{\footnotesize (f) BDR of matio}\hfill\makebox[0.19\textwidth]{\footnotesize (g) BDR of openh264}\hfill\makebox[0.19\textwidth]{\footnotesize (h) BDR of php}\hfill\makebox[0.19\textwidth]{\footnotesize (i) BDR of poppler}\hfill\makebox[0.19\textwidth]{\footnotesize (j) BDR of stb}
    \caption{Comparisons of bug detection rate across 10 benchmarks for fuzzer AFL.}
    \label{fig:bdr_afl}
\end{figure*}

\begin{figure*}[tp]
    \centering
    \includegraphics[width=0.19\textwidth]{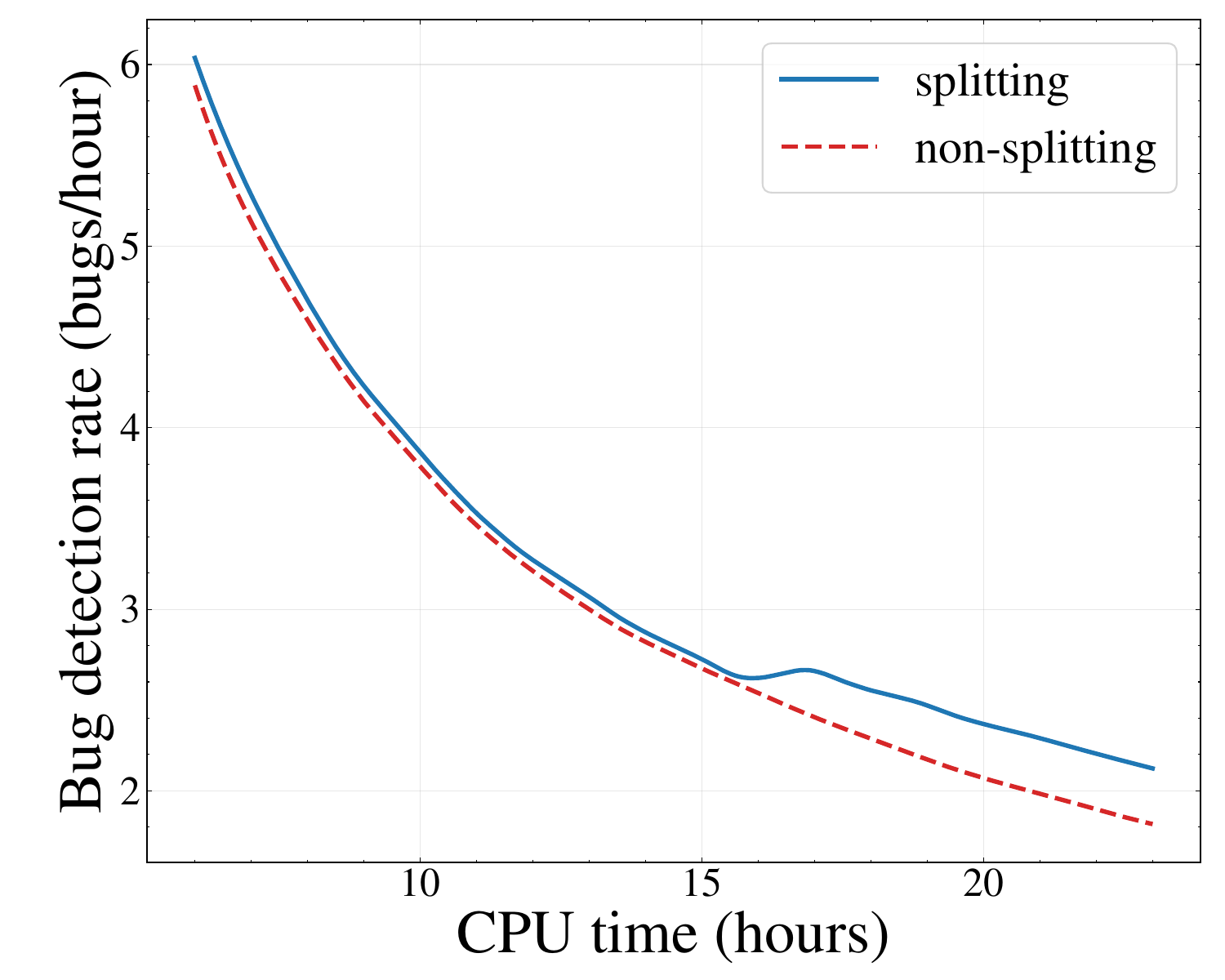}\hfill\includegraphics[width=0.19\textwidth]{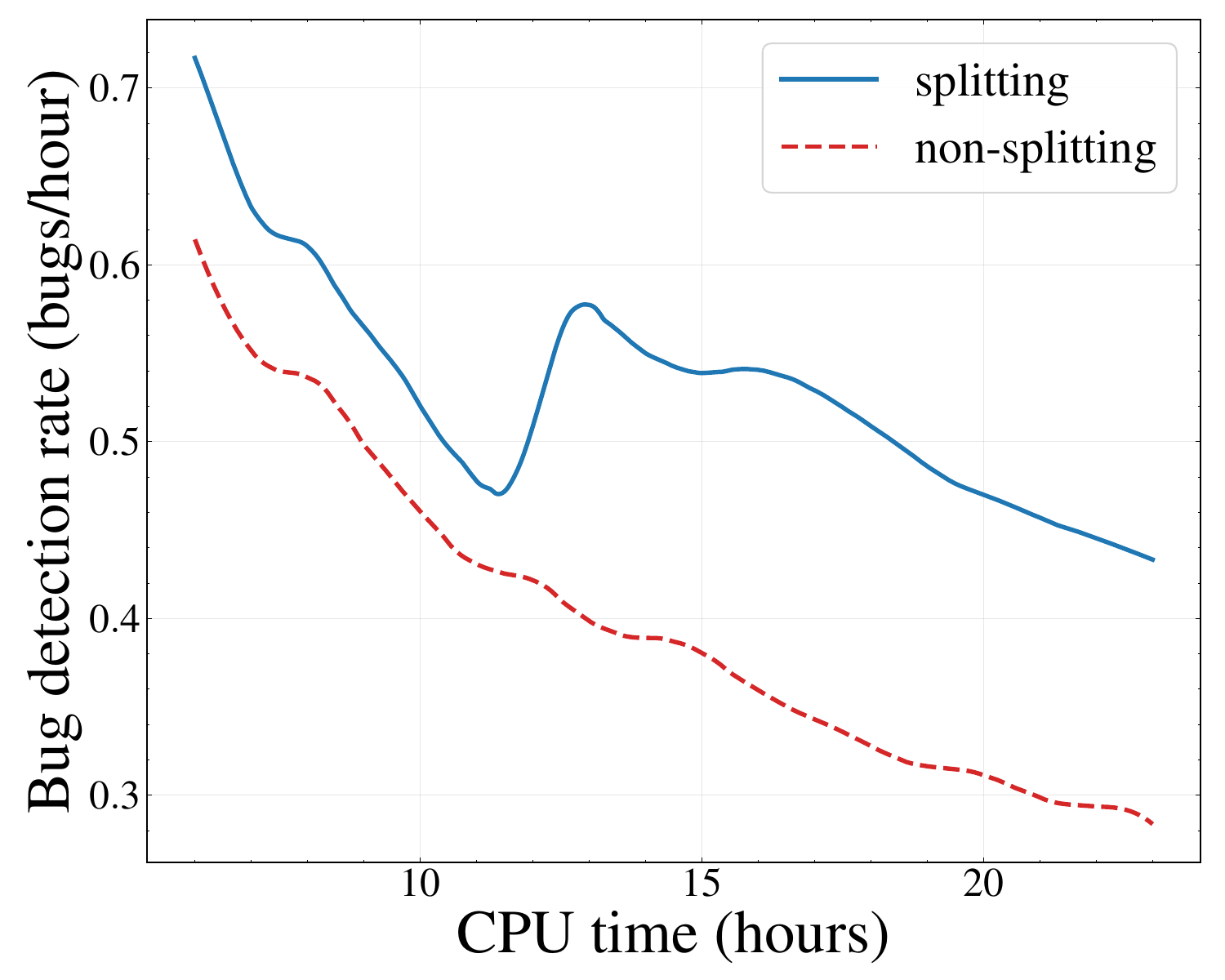}\hfill\includegraphics[width=0.19\textwidth]{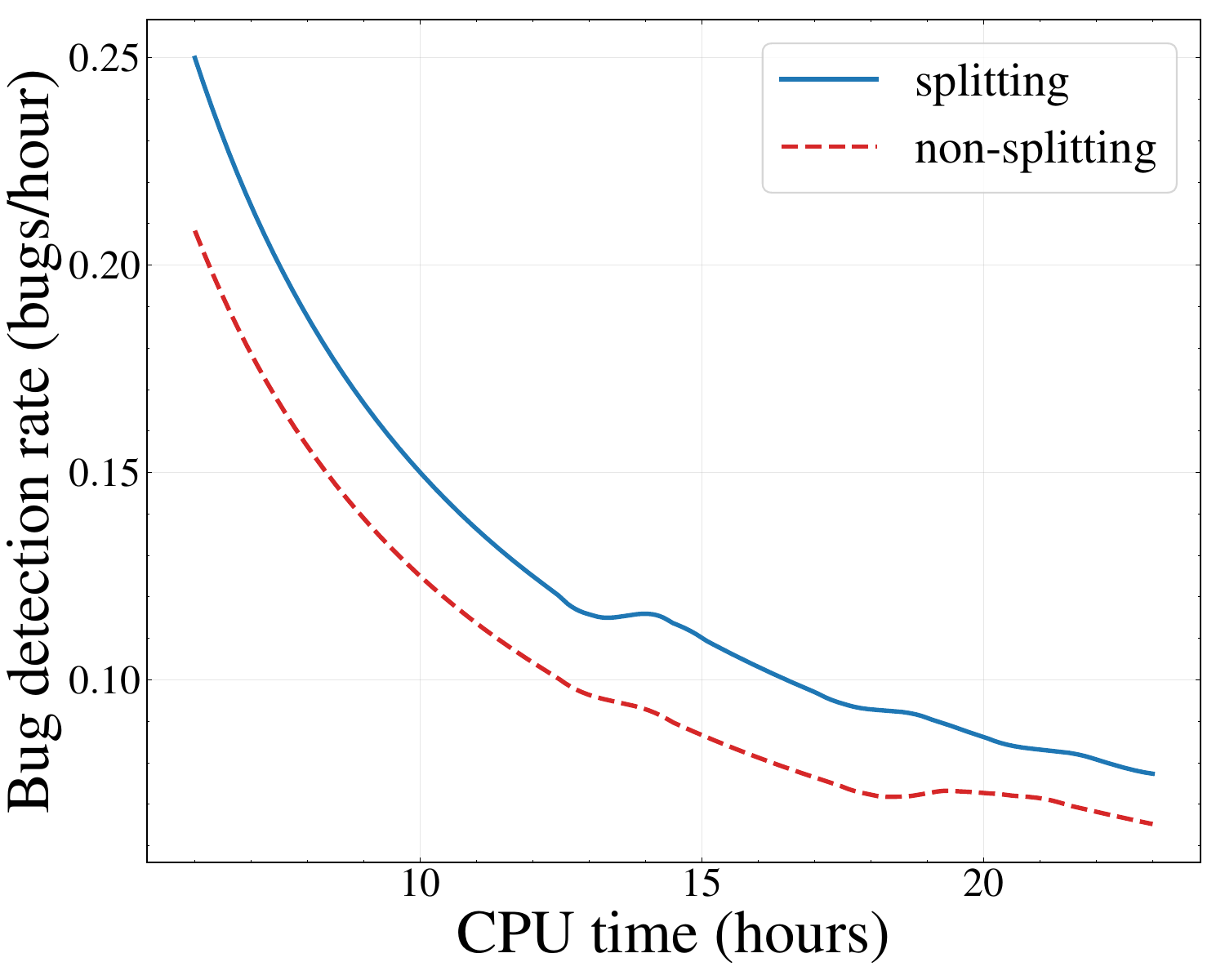}\hfill\includegraphics[width=0.19\textwidth]{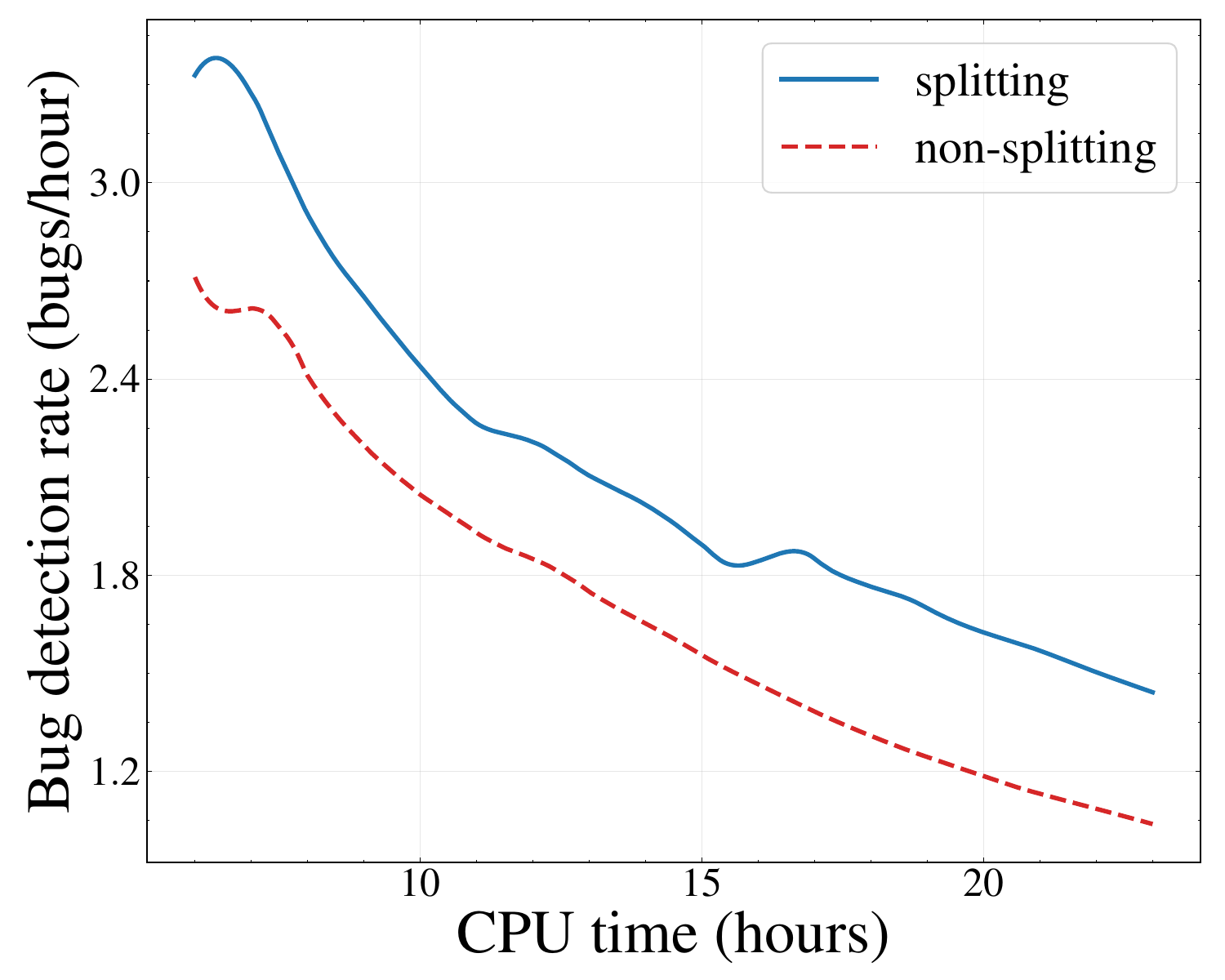}\hfill\includegraphics[width=0.19\textwidth]{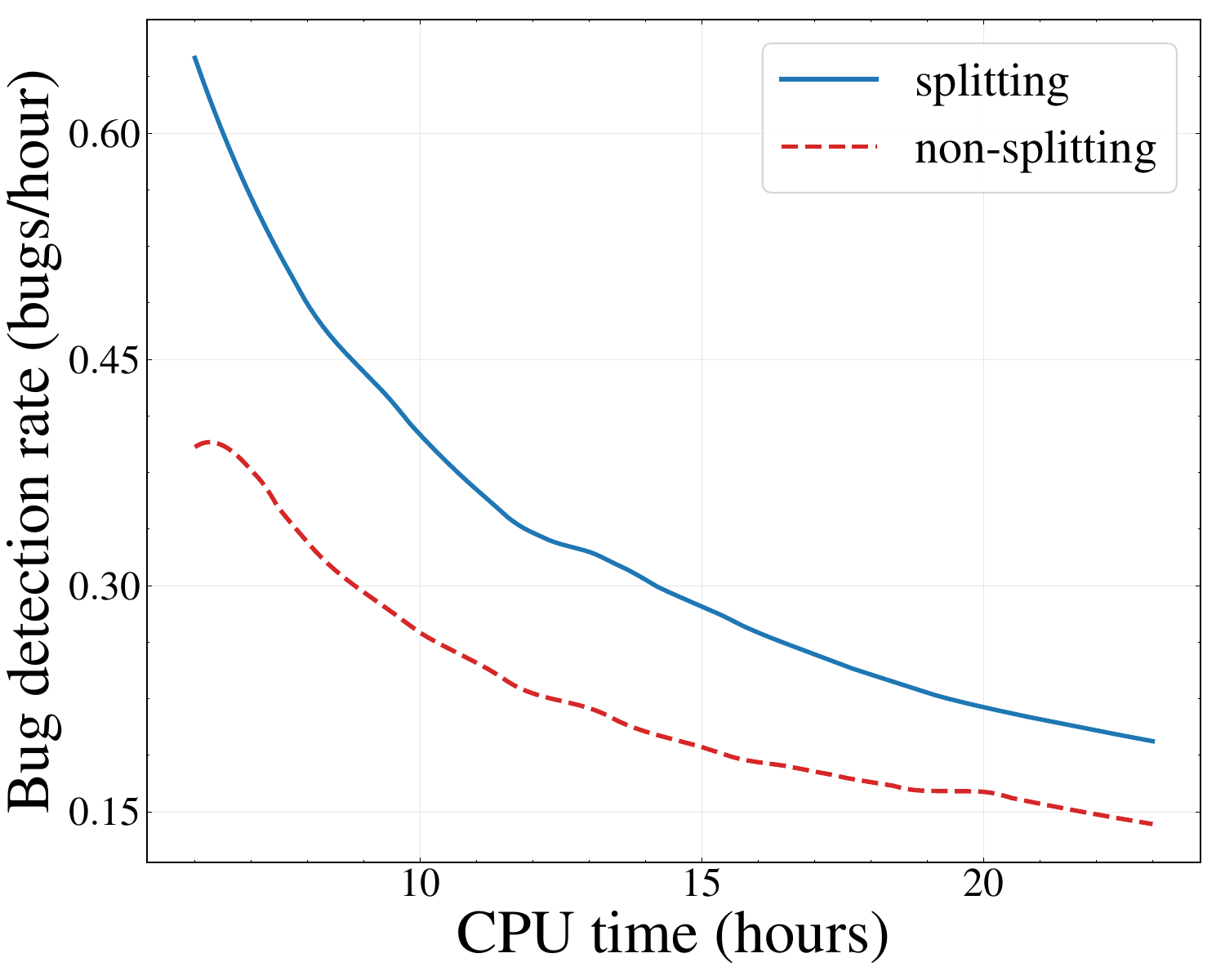}\\
    \makebox[0.19\textwidth]{\footnotesize (a) BDR of arrow}\hfill\makebox[0.19\textwidth]{\footnotesize (b) BDR of ffmpeg}\hfill\makebox[0.19\textwidth]{\footnotesize (c) BDR of grok}\hfill\makebox[0.19\textwidth]{\footnotesize (d) BDR of libhevc}\hfill\makebox[0.19\textwidth]{\footnotesize (e) BDR of libhtp}\\[3pt]
    \includegraphics[width=0.19\textwidth]{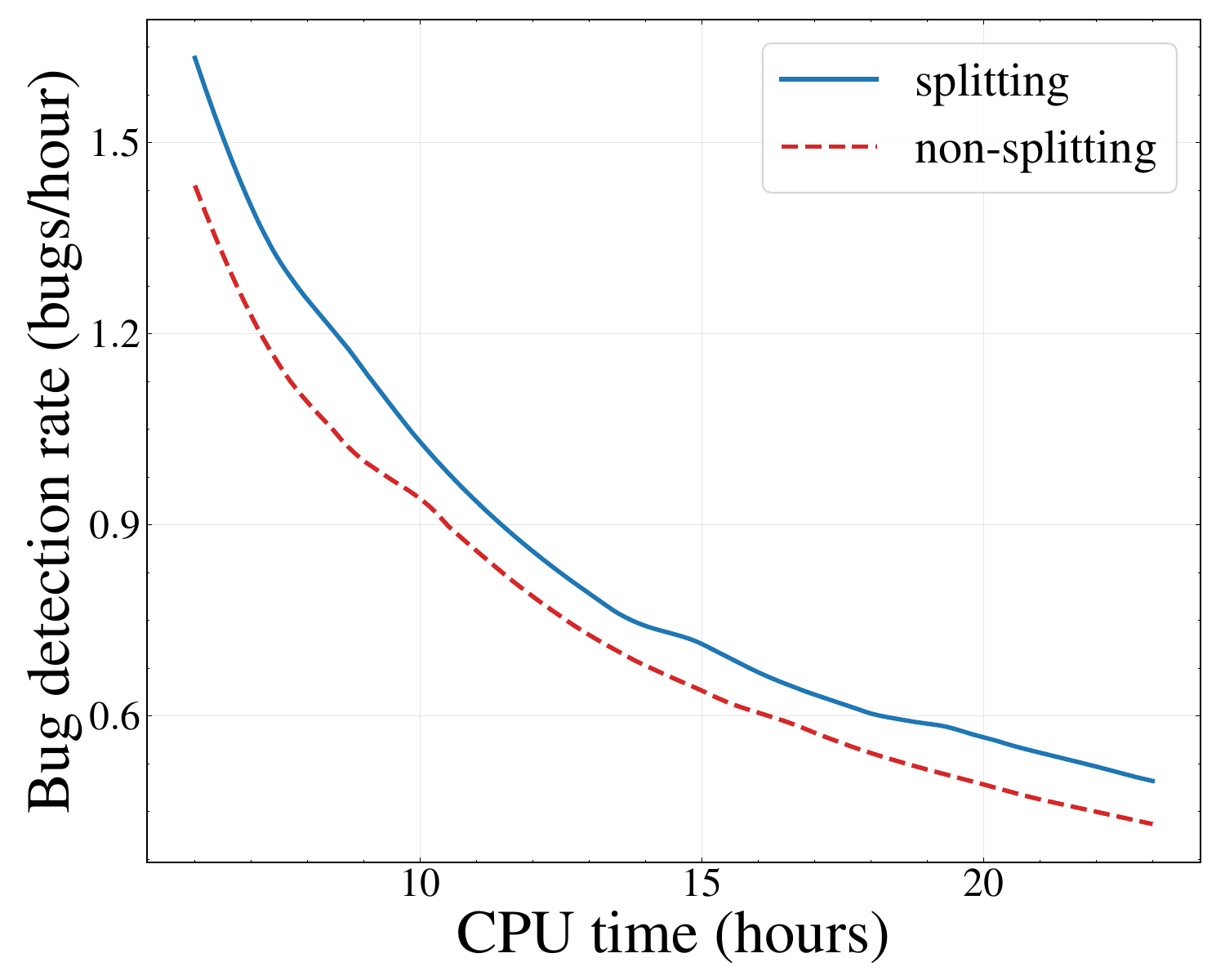}\hfill\includegraphics[width=0.19\textwidth]{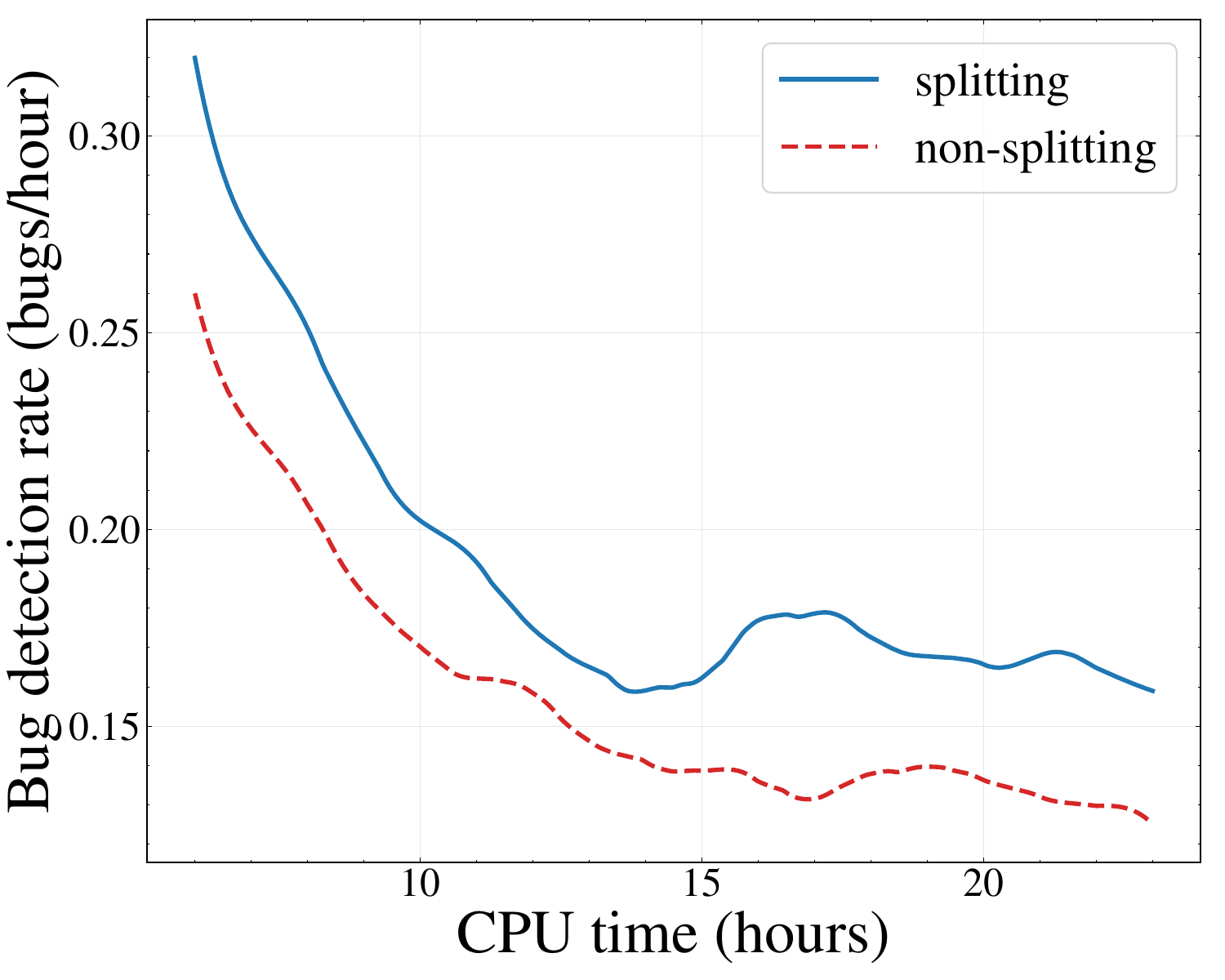}\hfill\includegraphics[width=0.19\textwidth]{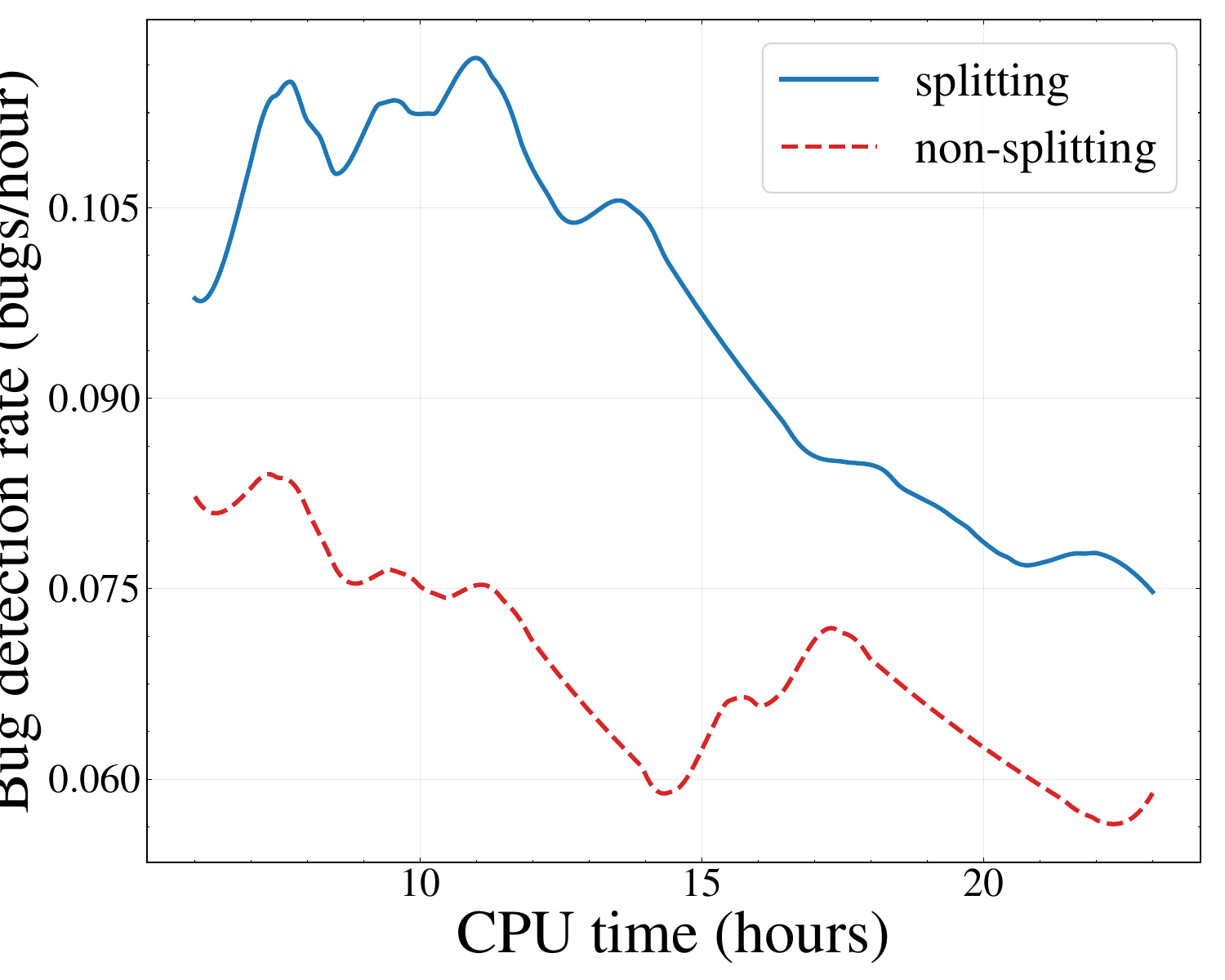}\hfill\includegraphics[width=0.19\textwidth]{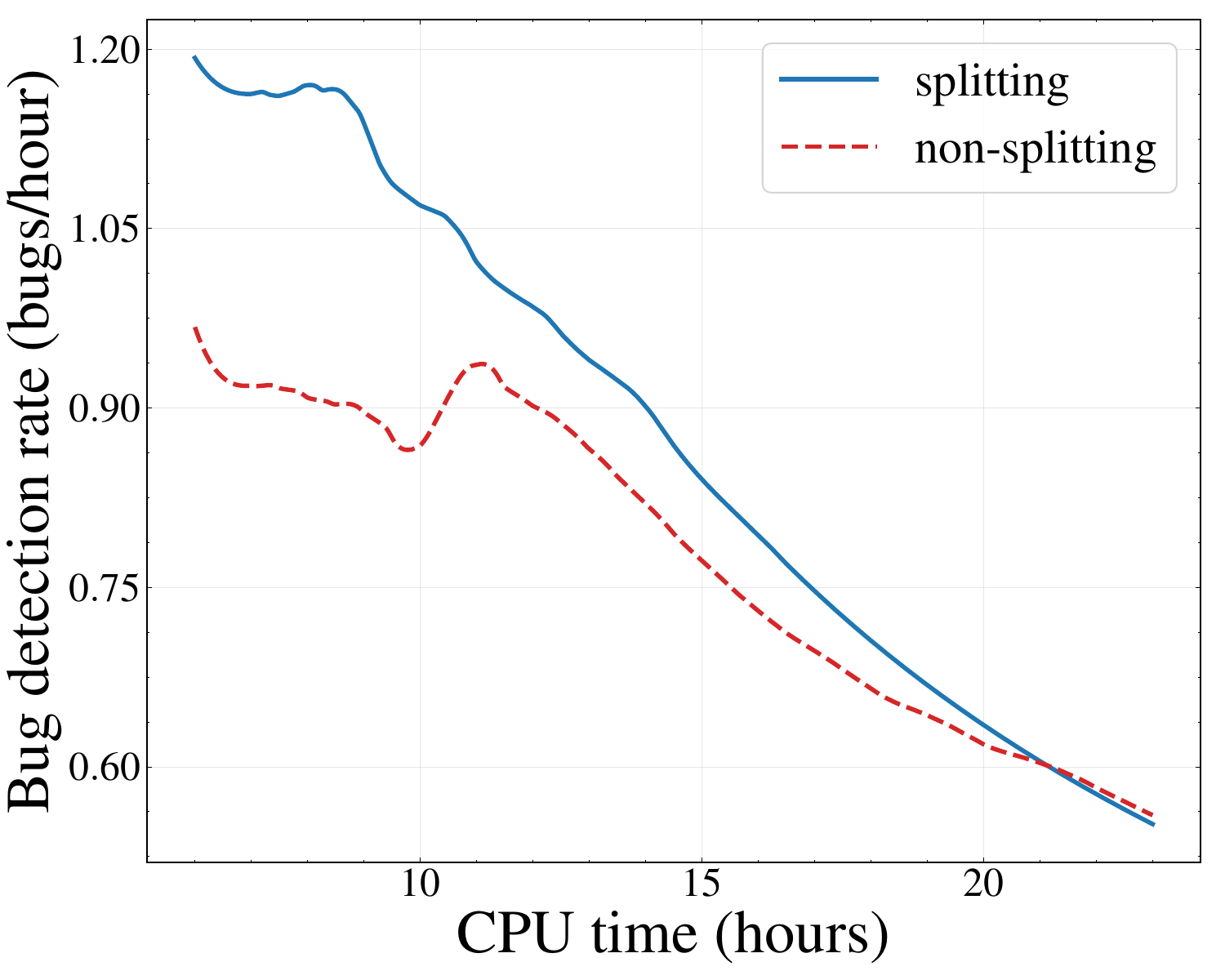}\hfill\includegraphics[width=0.19\textwidth]{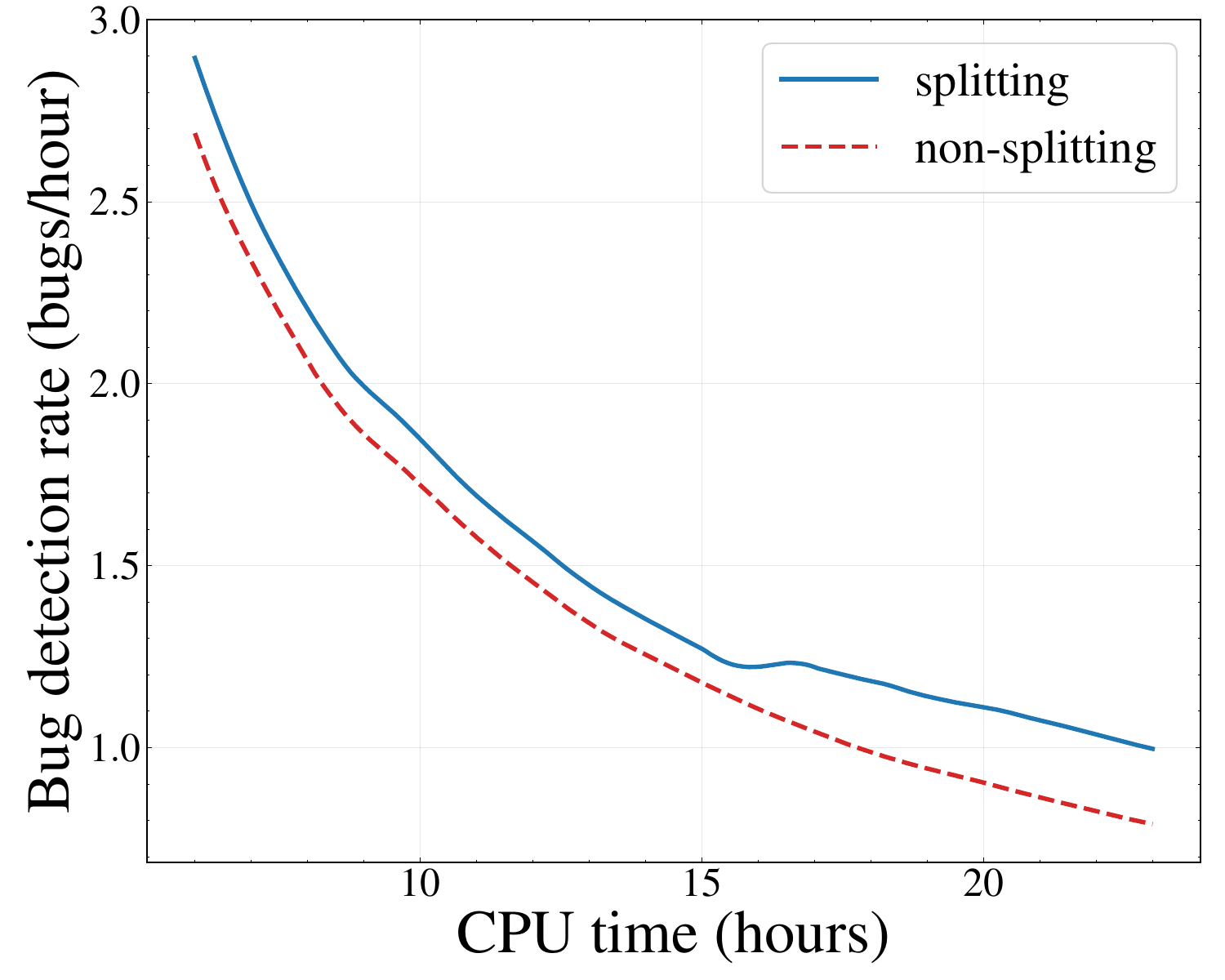}\\
    \makebox[0.19\textwidth]{\footnotesize (f) BDR of matio}\hfill\makebox[0.19\textwidth]{\footnotesize (g) BDR of openh264}\hfill\makebox[0.19\textwidth]{\footnotesize (h) BDR of php}\hfill\makebox[0.19\textwidth]{\footnotesize (i) BDR of poppler}\hfill\makebox[0.19\textwidth]{\footnotesize (j) BDR of stb}
    \caption{Comparisons of bug detection rate across 10 benchmarks for fuzzer AFLSmart.}
    \label{fig:bdr_aflsmart}
\end{figure*}

\begin{figure*}[tp]
    \centering
    \includegraphics[width=0.19\textwidth]{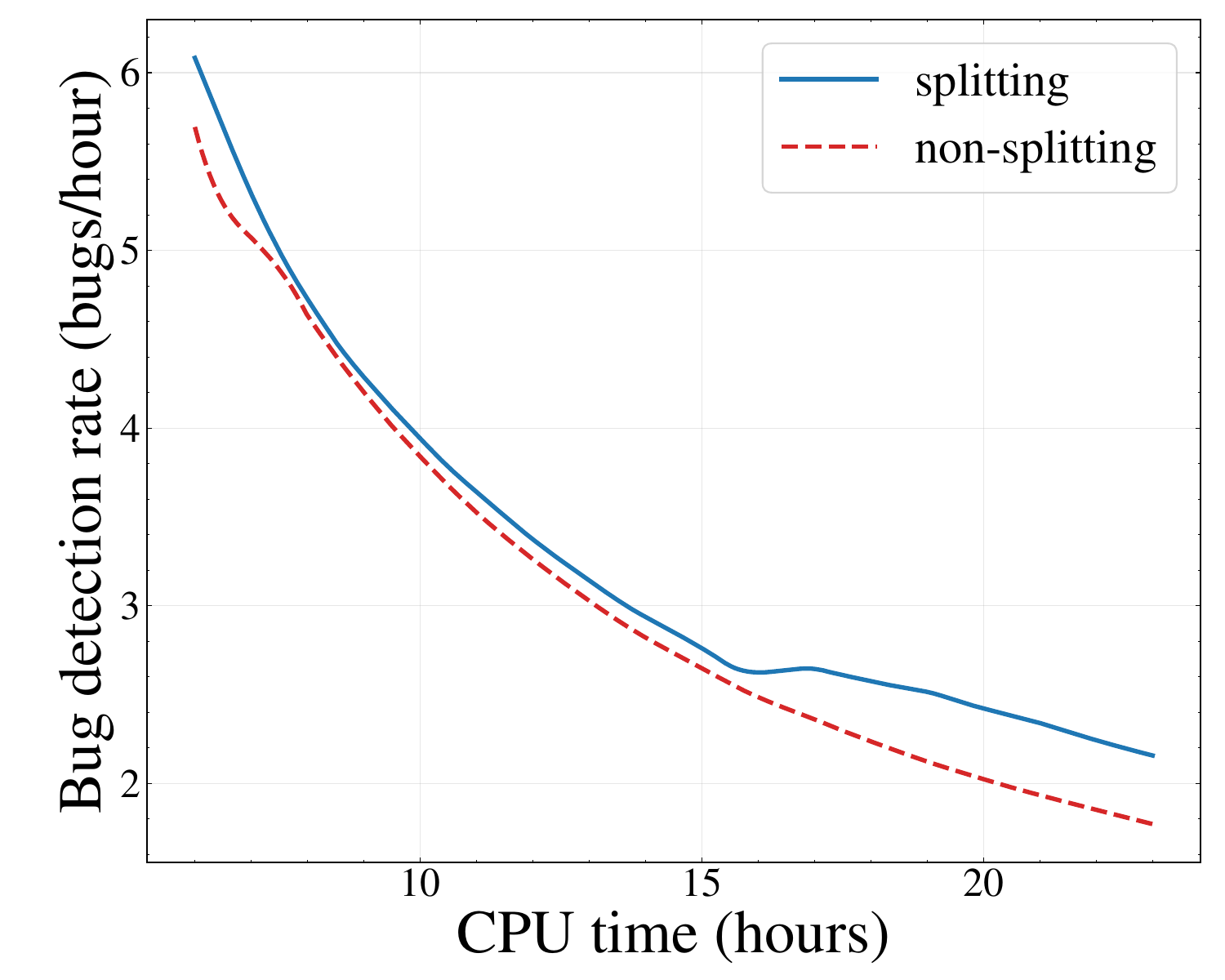}\hfill\includegraphics[width=0.19\textwidth]{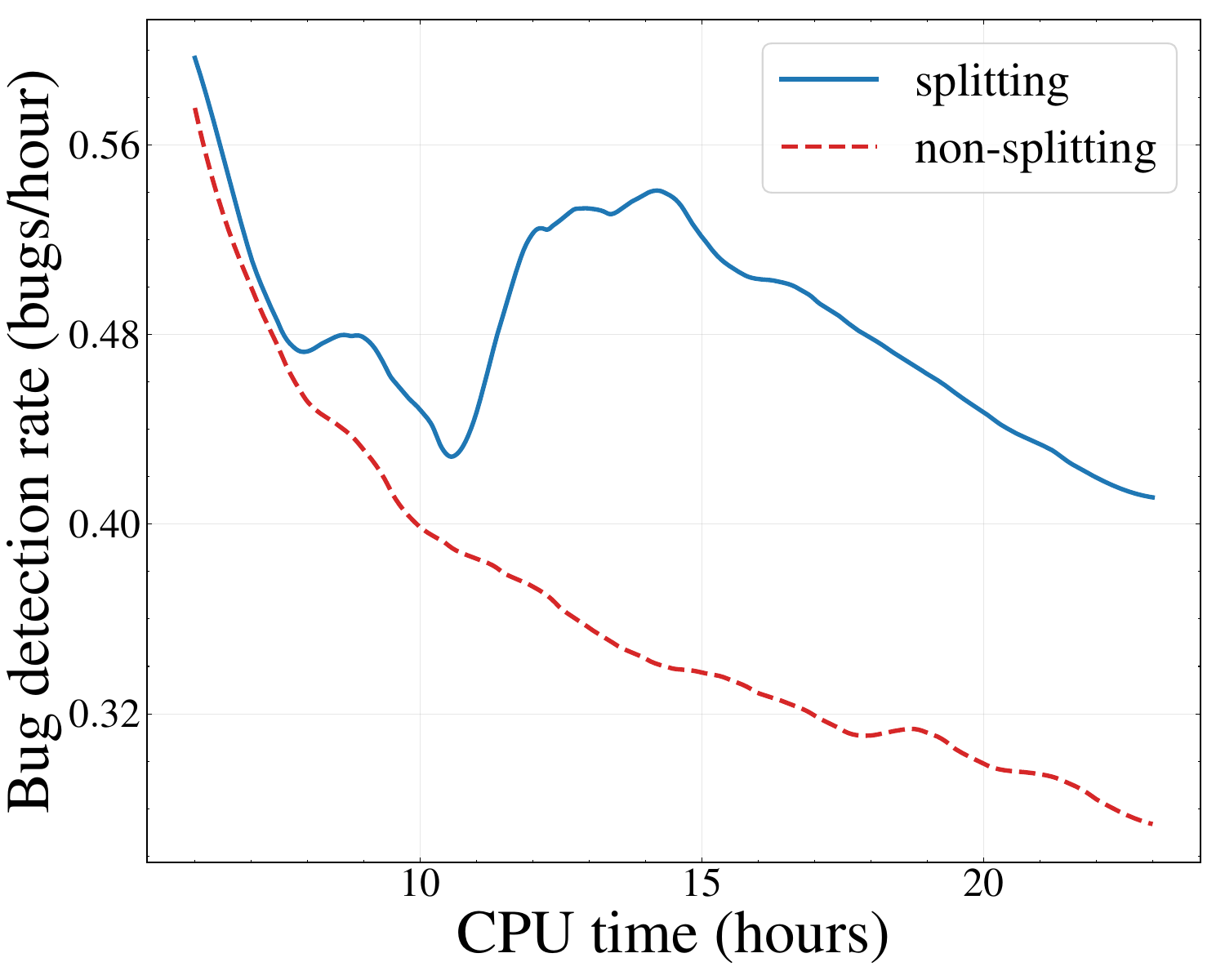}\hfill\includegraphics[width=0.19\textwidth]{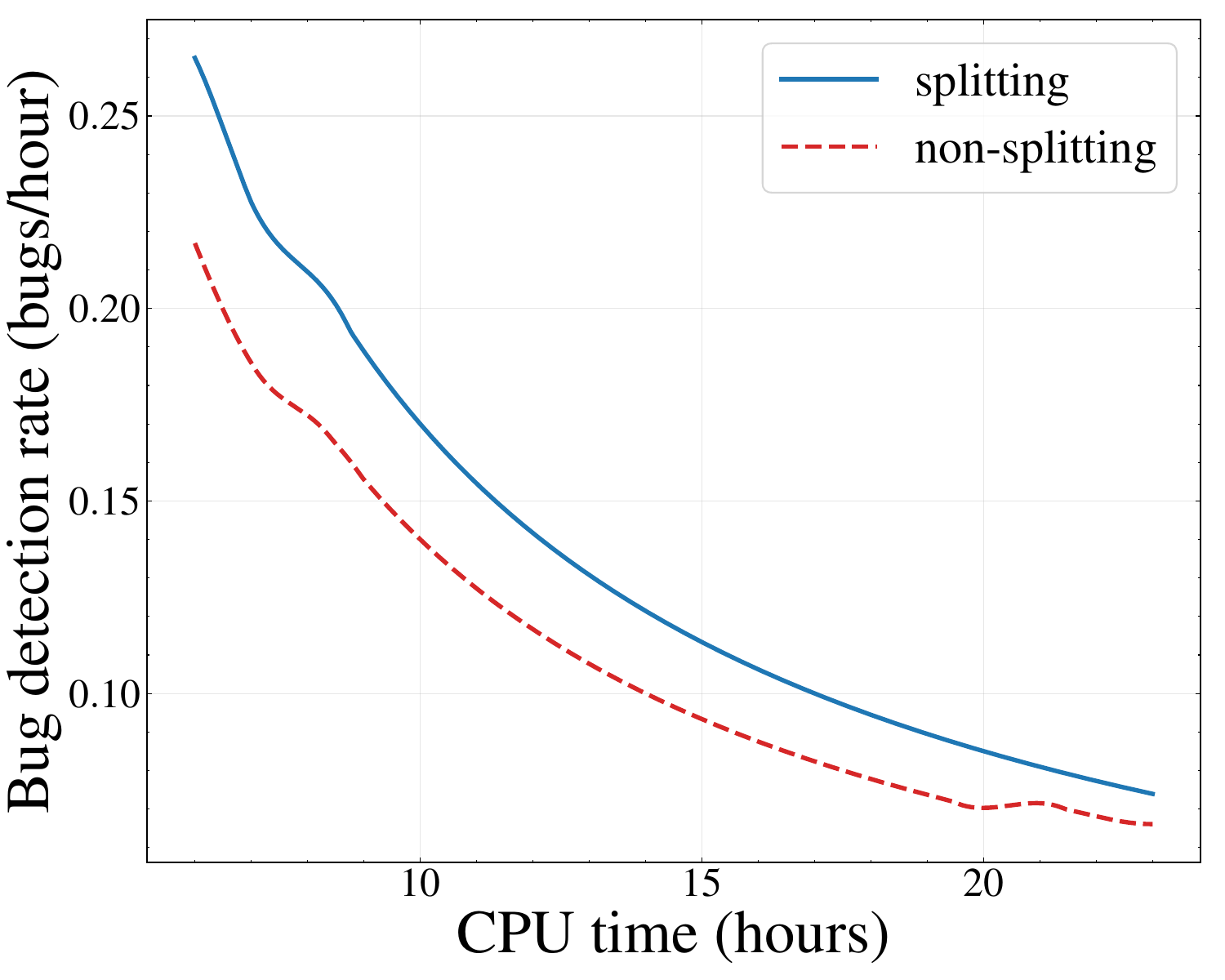}\hfill\includegraphics[width=0.19\textwidth]{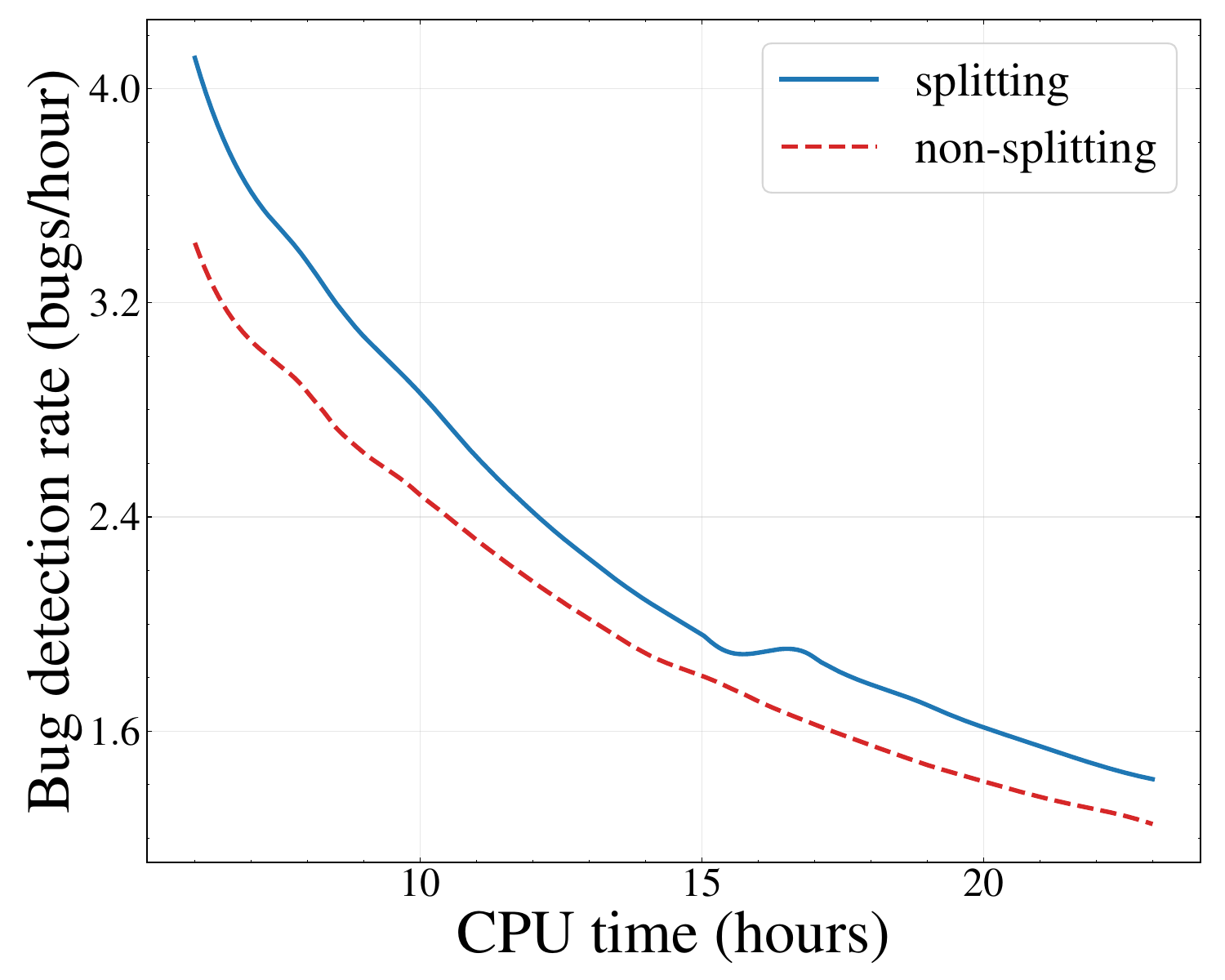}\hfill\includegraphics[width=0.19\textwidth]{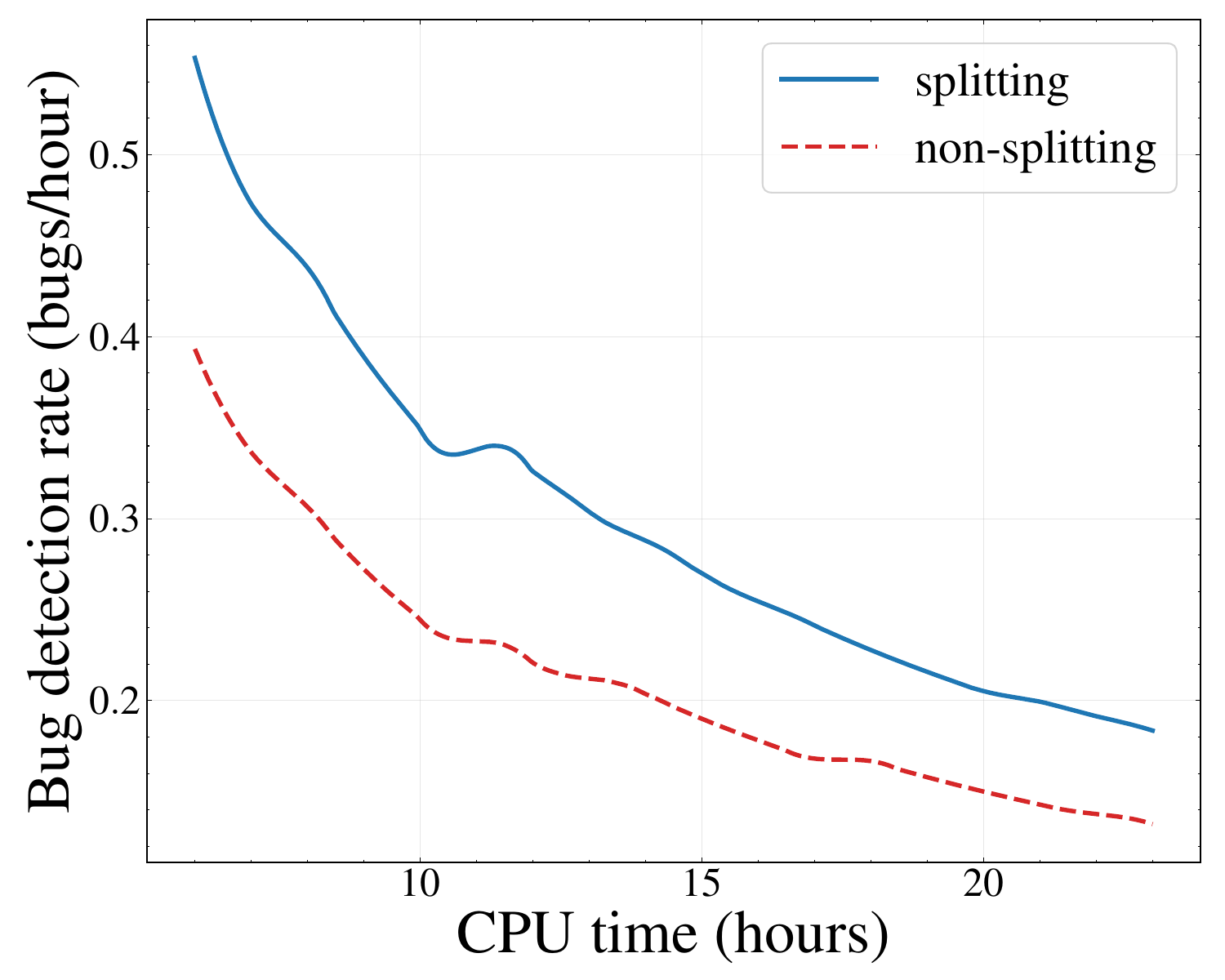}\\
    \makebox[0.19\textwidth]{\footnotesize (a) BDR of arrow}\hfill\makebox[0.19\textwidth]{\footnotesize (b) BDR of ffmpeg}\hfill\makebox[0.19\textwidth]{\footnotesize (c) BDR of grok}\hfill\makebox[0.19\textwidth]{\footnotesize (d) BDR of libhevc}\hfill\makebox[0.19\textwidth]{\footnotesize (e) BDR of libhtp}\\[3pt]
    \includegraphics[width=0.19\textwidth]{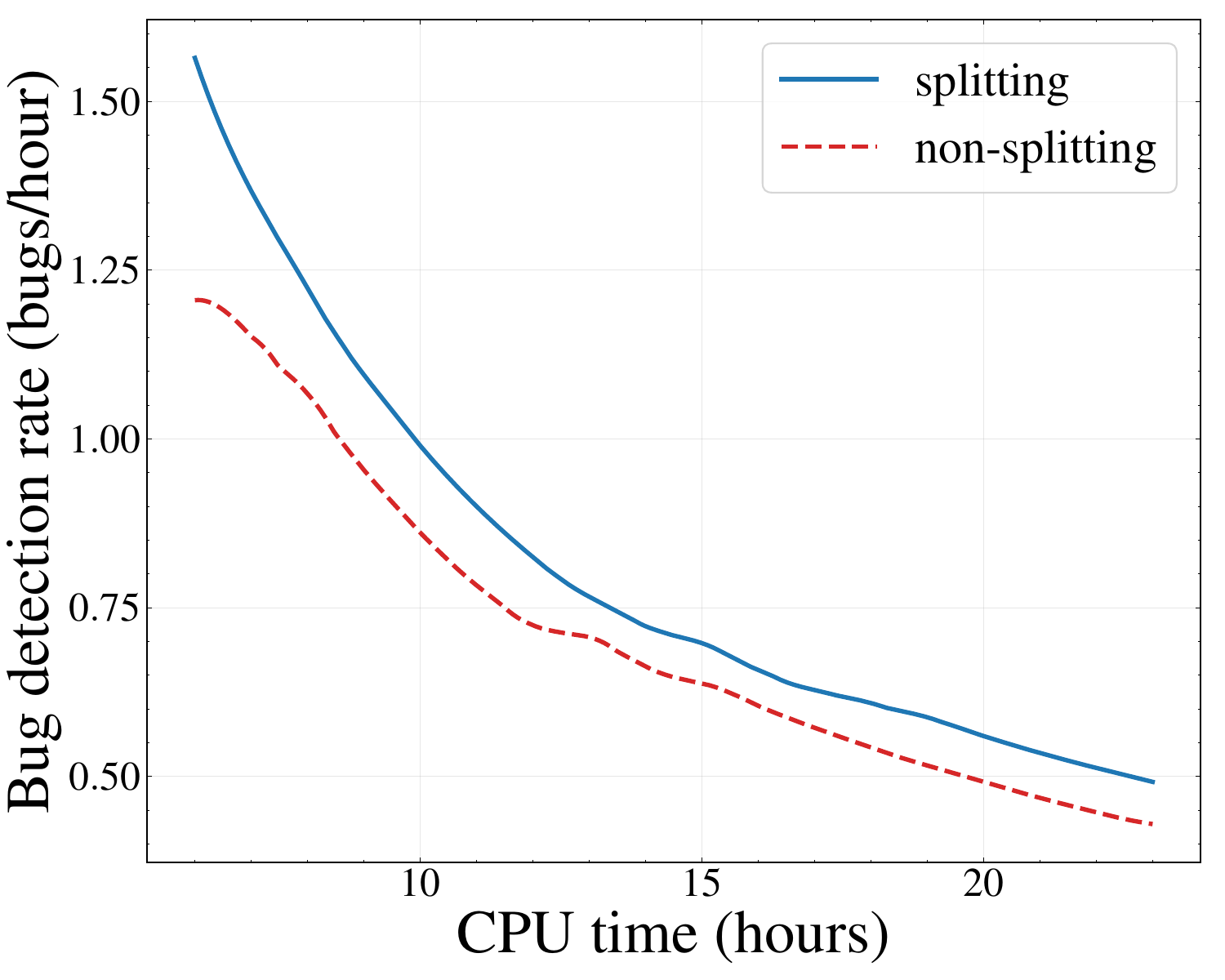}\hfill\includegraphics[width=0.19\textwidth]{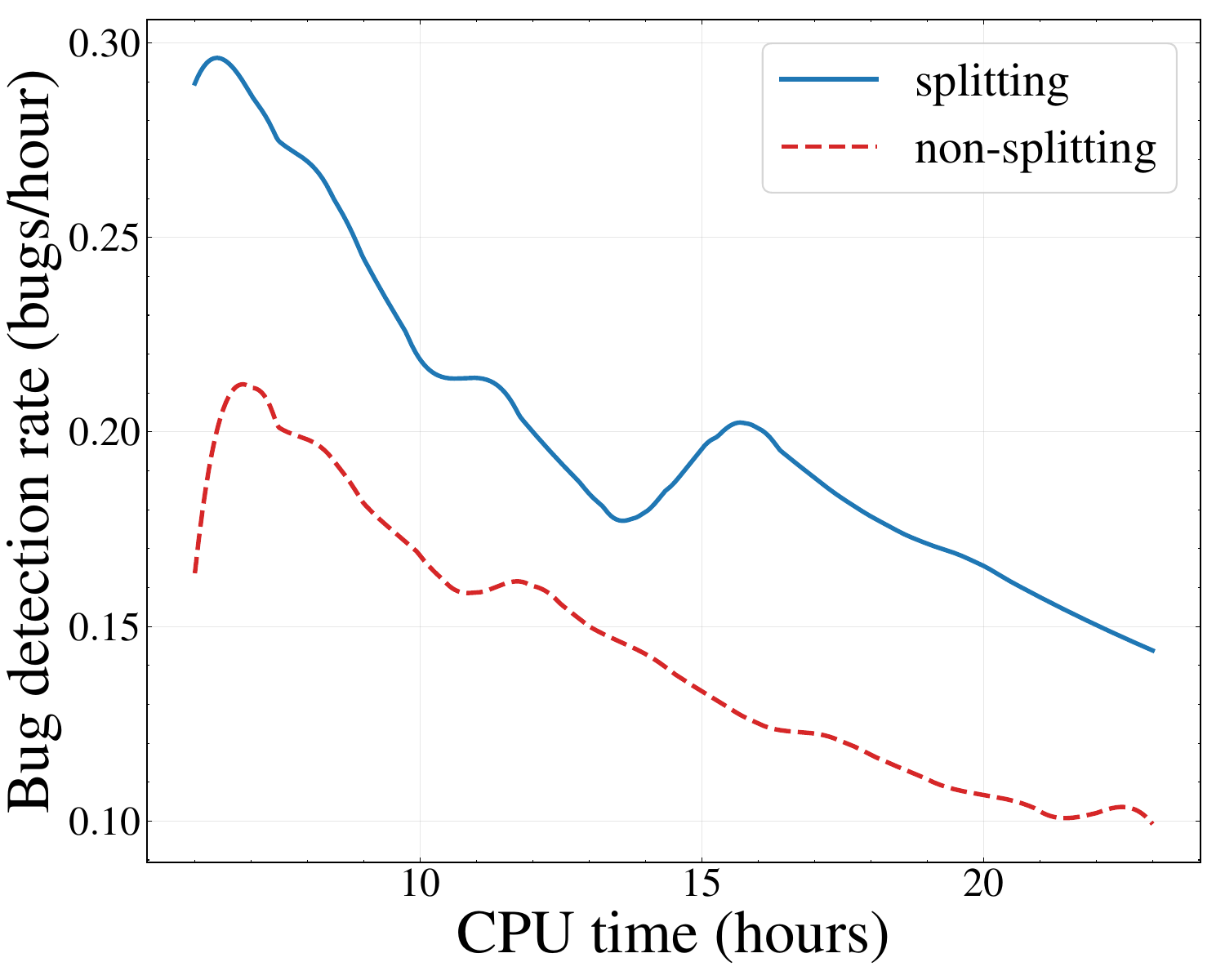}\hfill\includegraphics[width=0.19\textwidth]{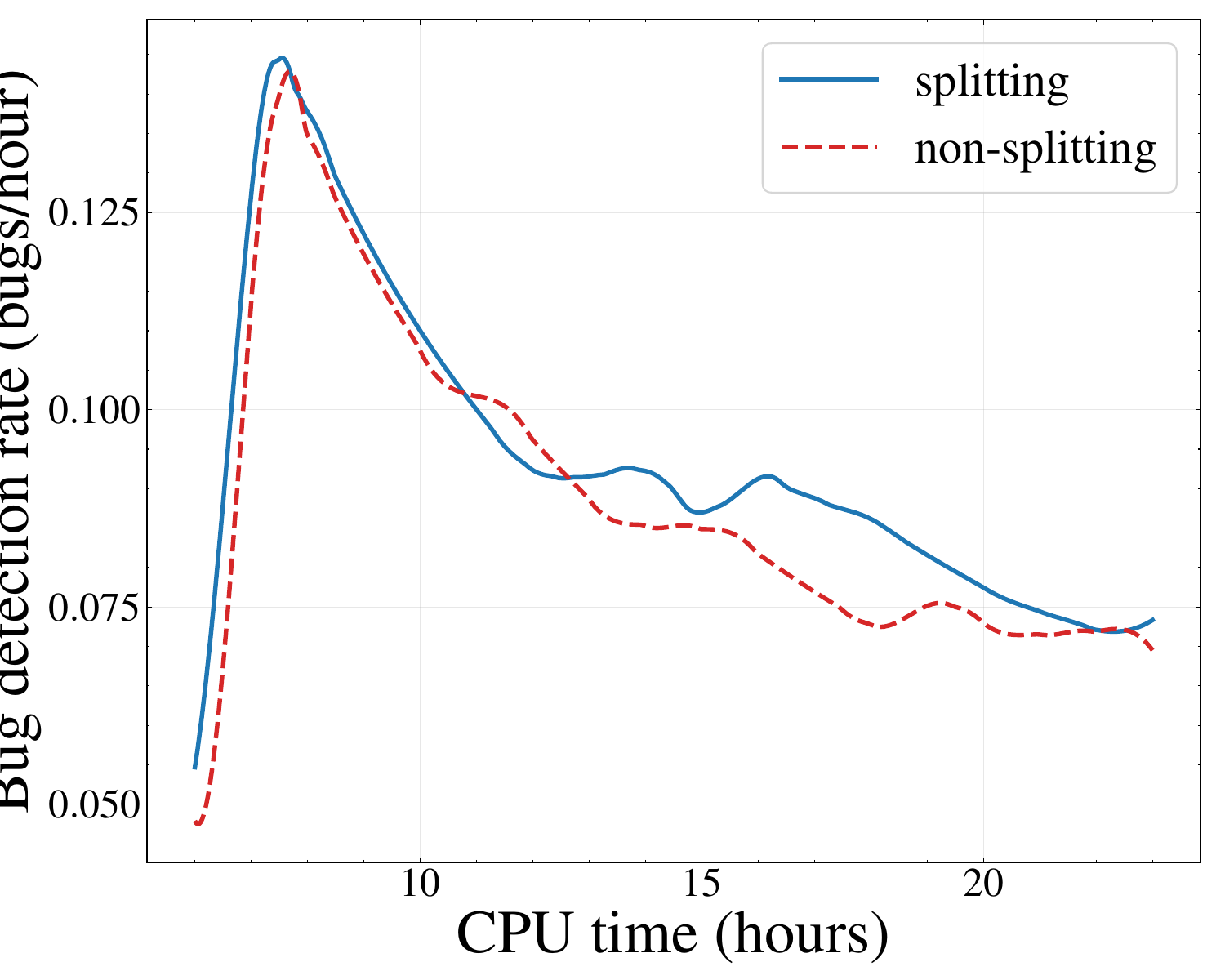}\hfill\includegraphics[width=0.19\textwidth]{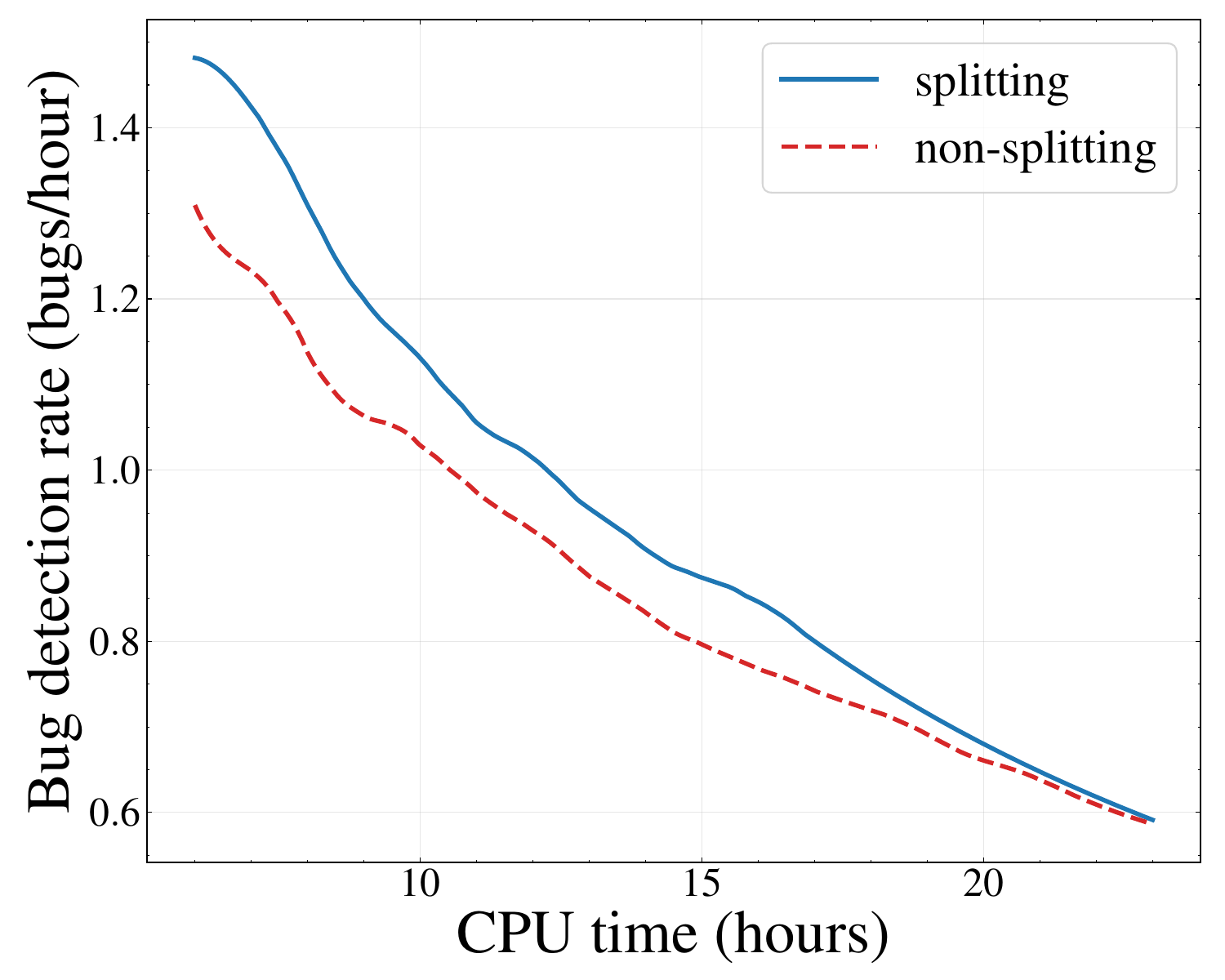}\hfill\includegraphics[width=0.19\textwidth]{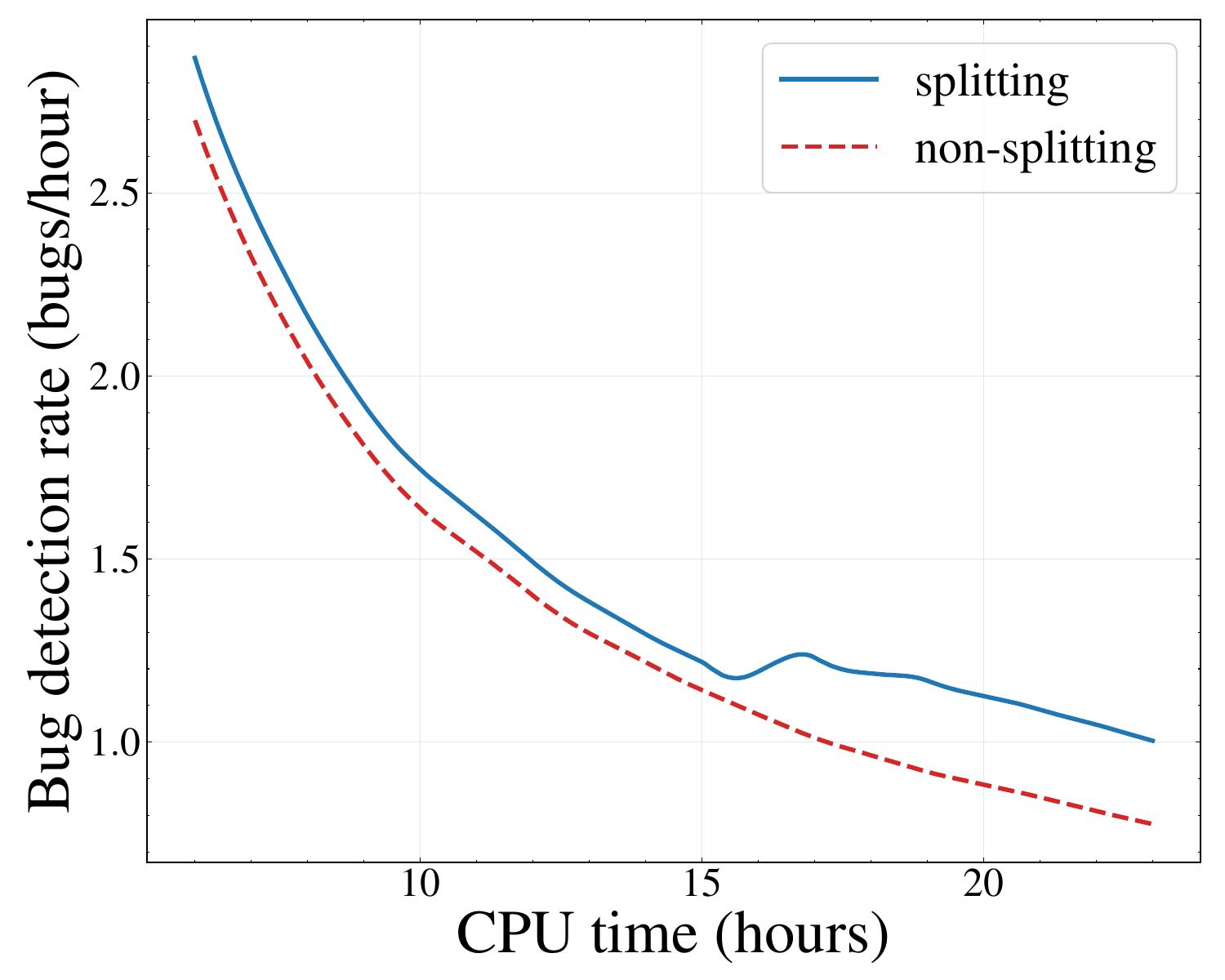}\\
    \makebox[0.19\textwidth]{\footnotesize (f) BDR of matio}\hfill\makebox[0.19\textwidth]{\footnotesize (g) BDR of openh264}\hfill\makebox[0.19\textwidth]{\footnotesize (h) BDR of php}\hfill\makebox[0.19\textwidth]{\footnotesize (i) BDR of poppler}\hfill\makebox[0.19\textwidth]{\footnotesize (j) BDR of stb}
    \caption{Comparisons of bug detection rate across 10 benchmarks for fuzzer AFLFast.}
    \label{fig:bdr_aflfast}
\end{figure*}

\begin{figure*}[tp]
    \centering
    \includegraphics[width=0.19\textwidth]{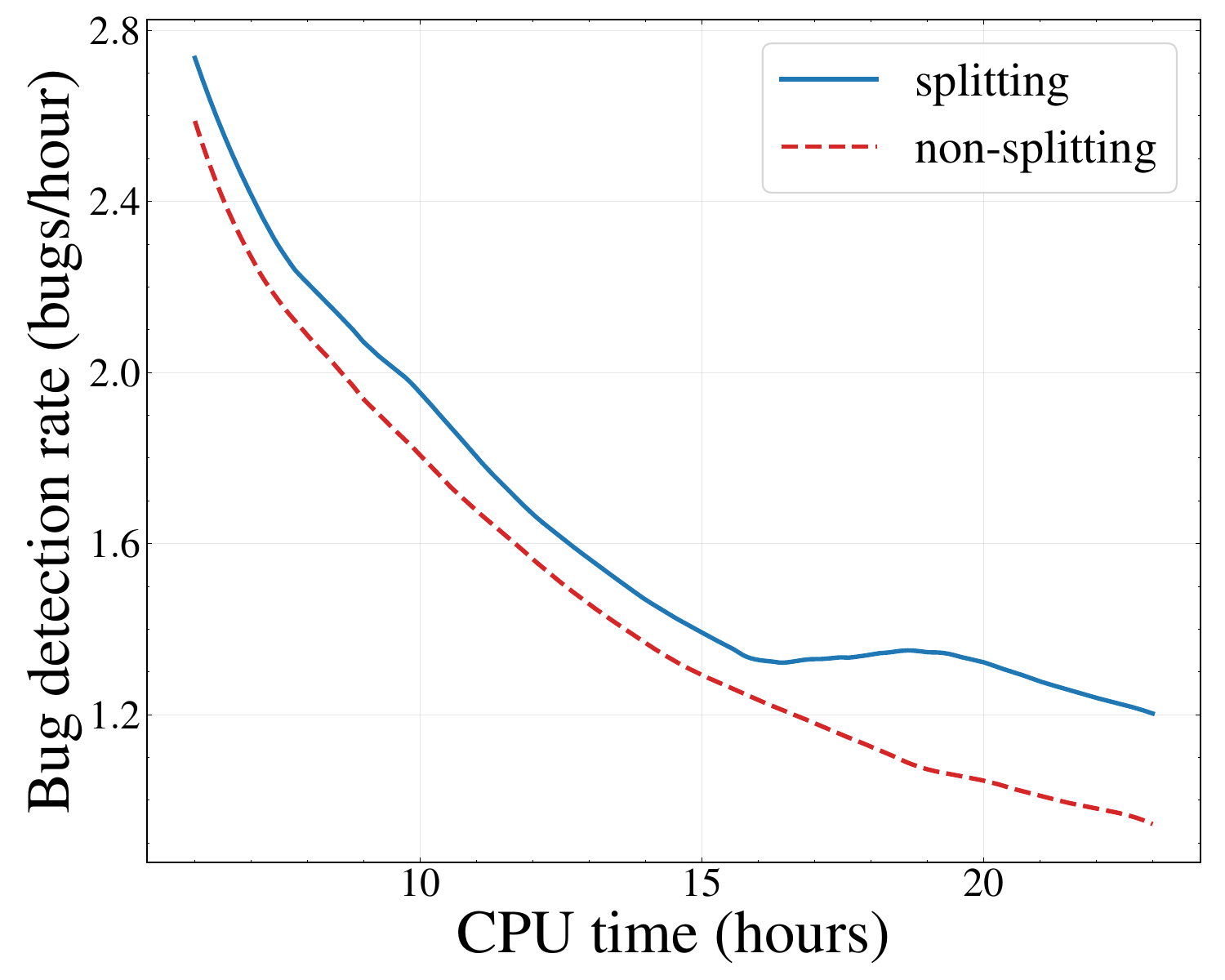}\hfill\includegraphics[width=0.19\textwidth]{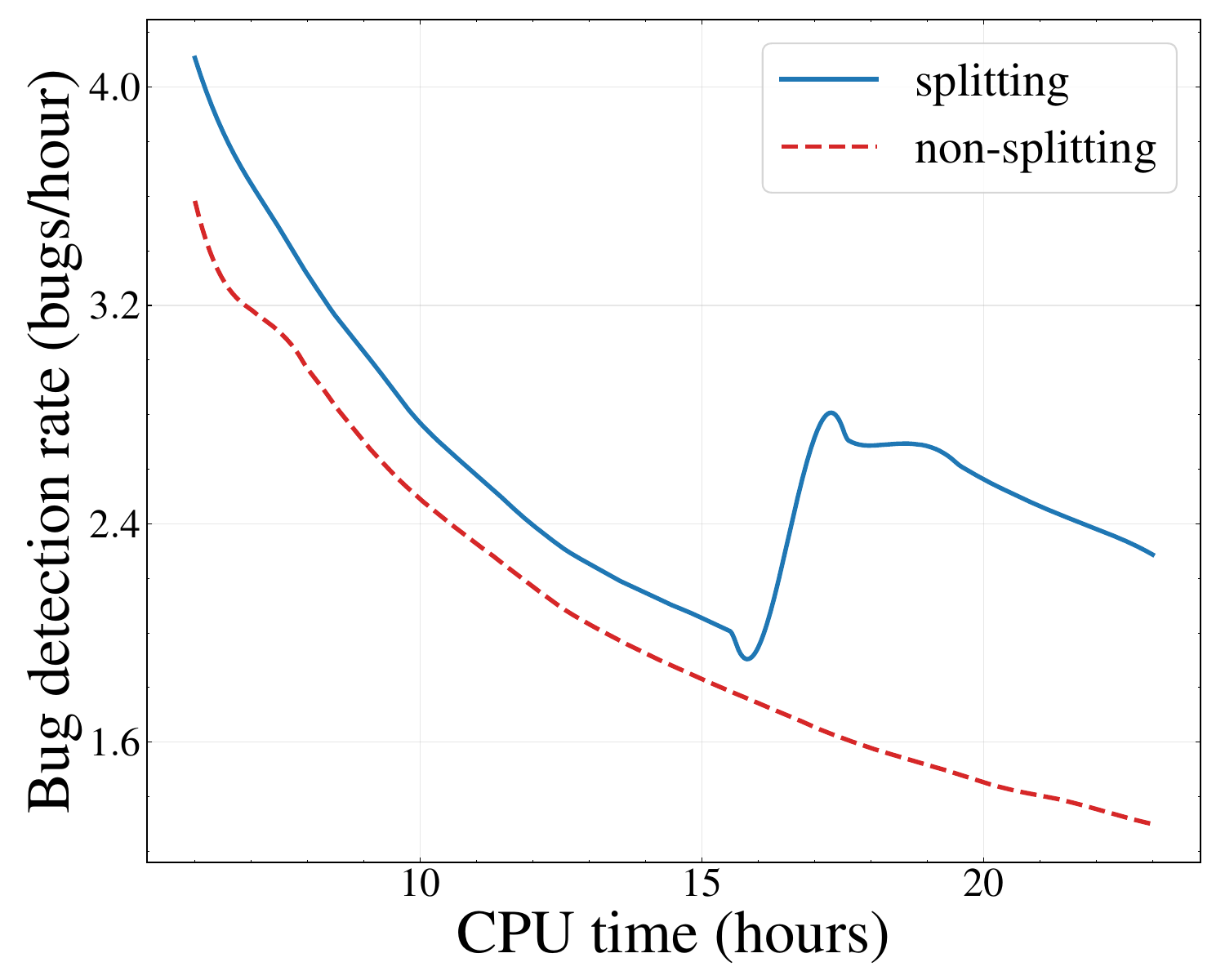}\hfill\includegraphics[width=0.19\textwidth]{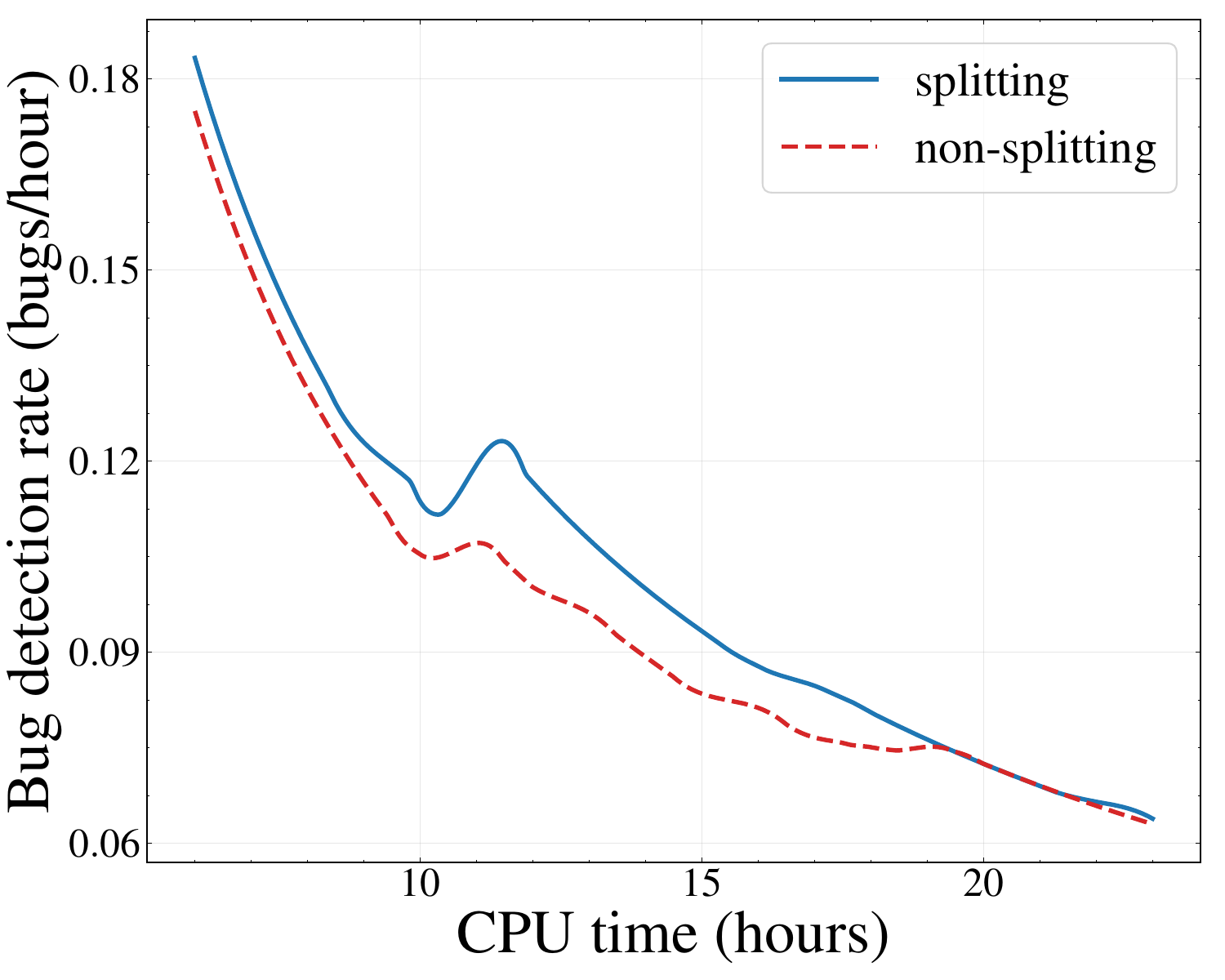}\hfill\includegraphics[width=0.19\textwidth]{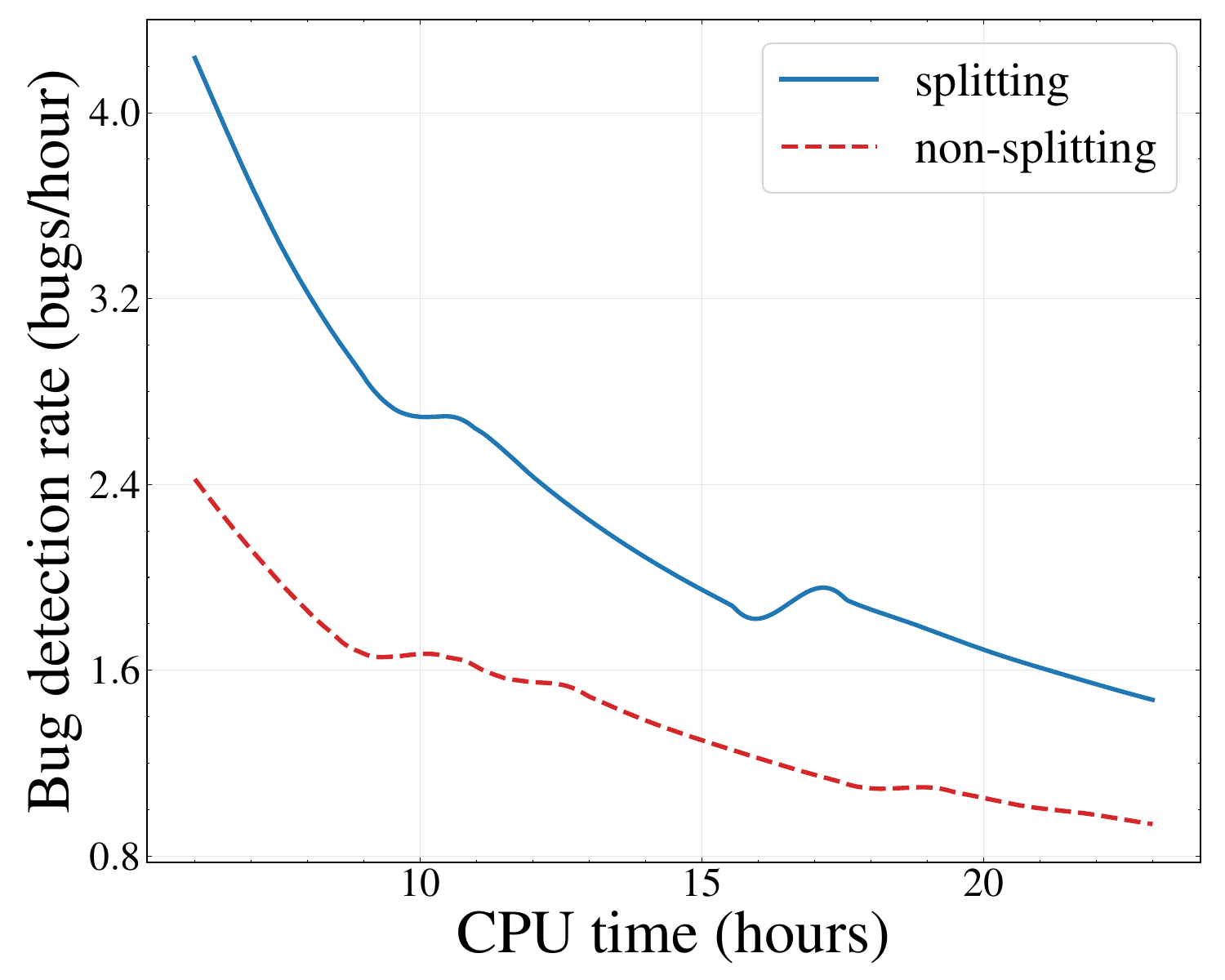}\hfill\includegraphics[width=0.19\textwidth]{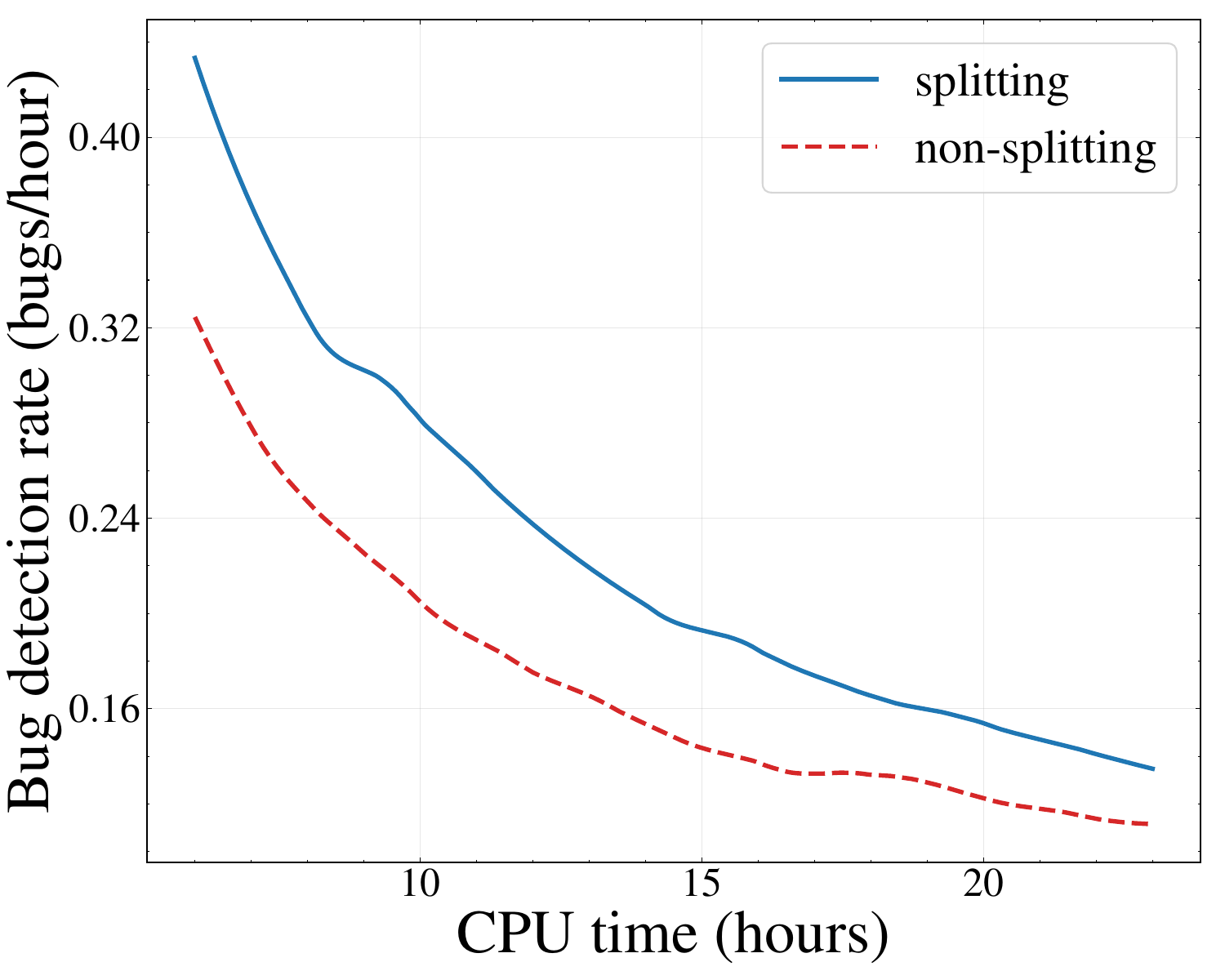}\\
    \makebox[0.19\textwidth]{\footnotesize (a) BDR of arrow}\hfill\makebox[0.19\textwidth]{\footnotesize (b) BDR of ffmpeg}\hfill\makebox[0.19\textwidth]{\footnotesize (c) BDR of grok}\hfill\makebox[0.19\textwidth]{\footnotesize (d) BDR of libhevc}\hfill\makebox[0.19\textwidth]{\footnotesize (e) BDR of libhtp}\\[3pt]
    \includegraphics[width=0.19\textwidth]{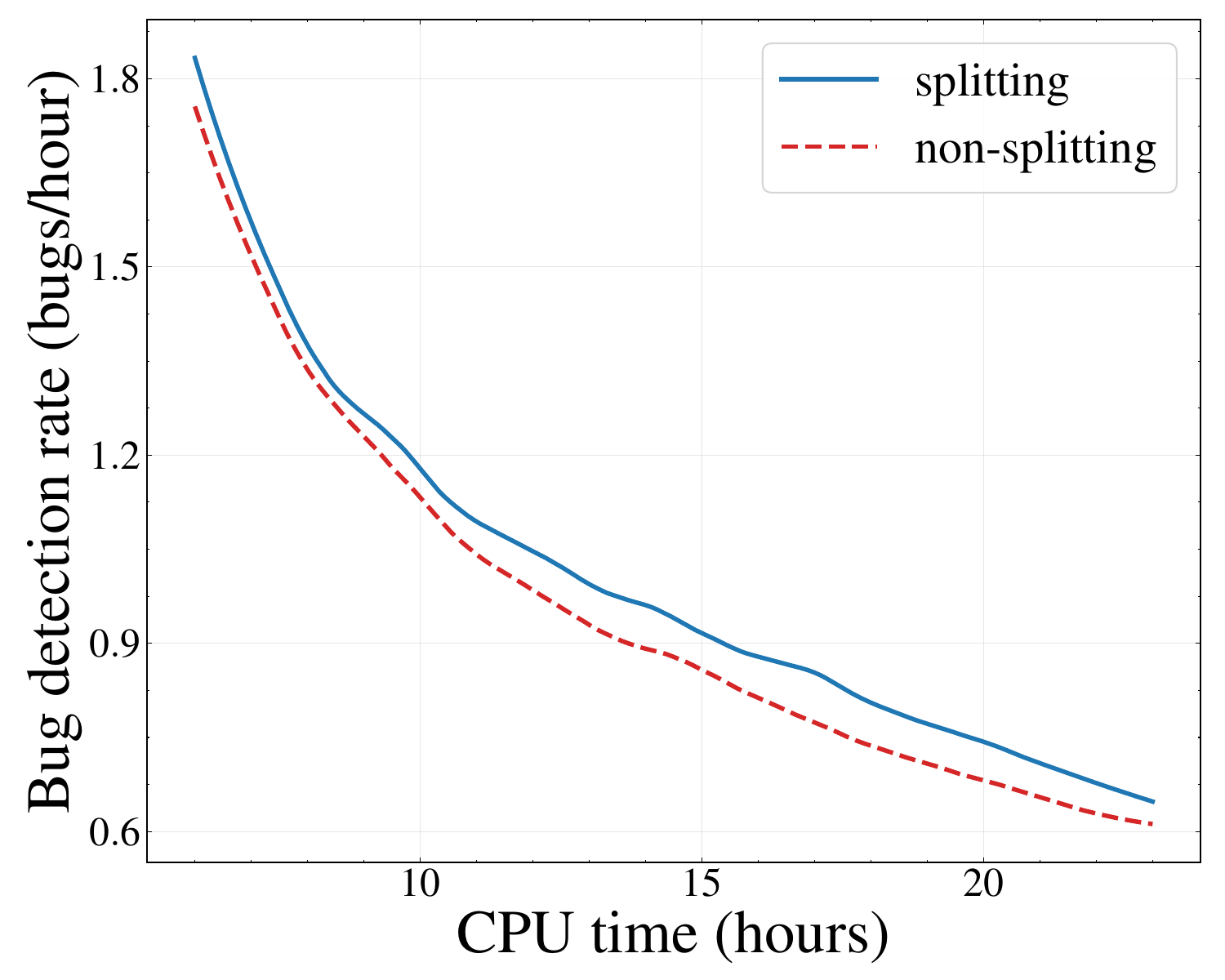}\hfill\includegraphics[width=0.19\textwidth]{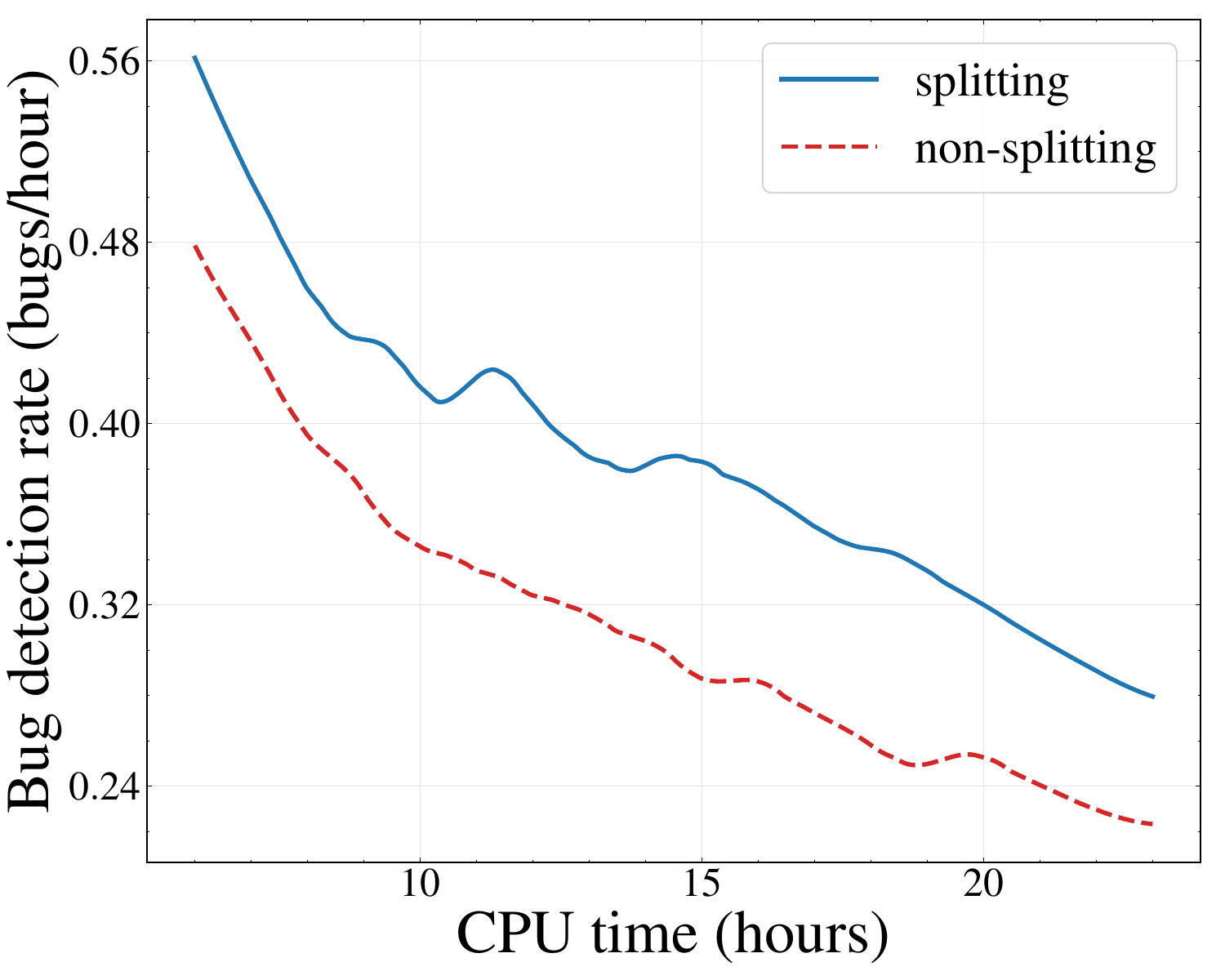}\hfill\includegraphics[width=0.19\textwidth]{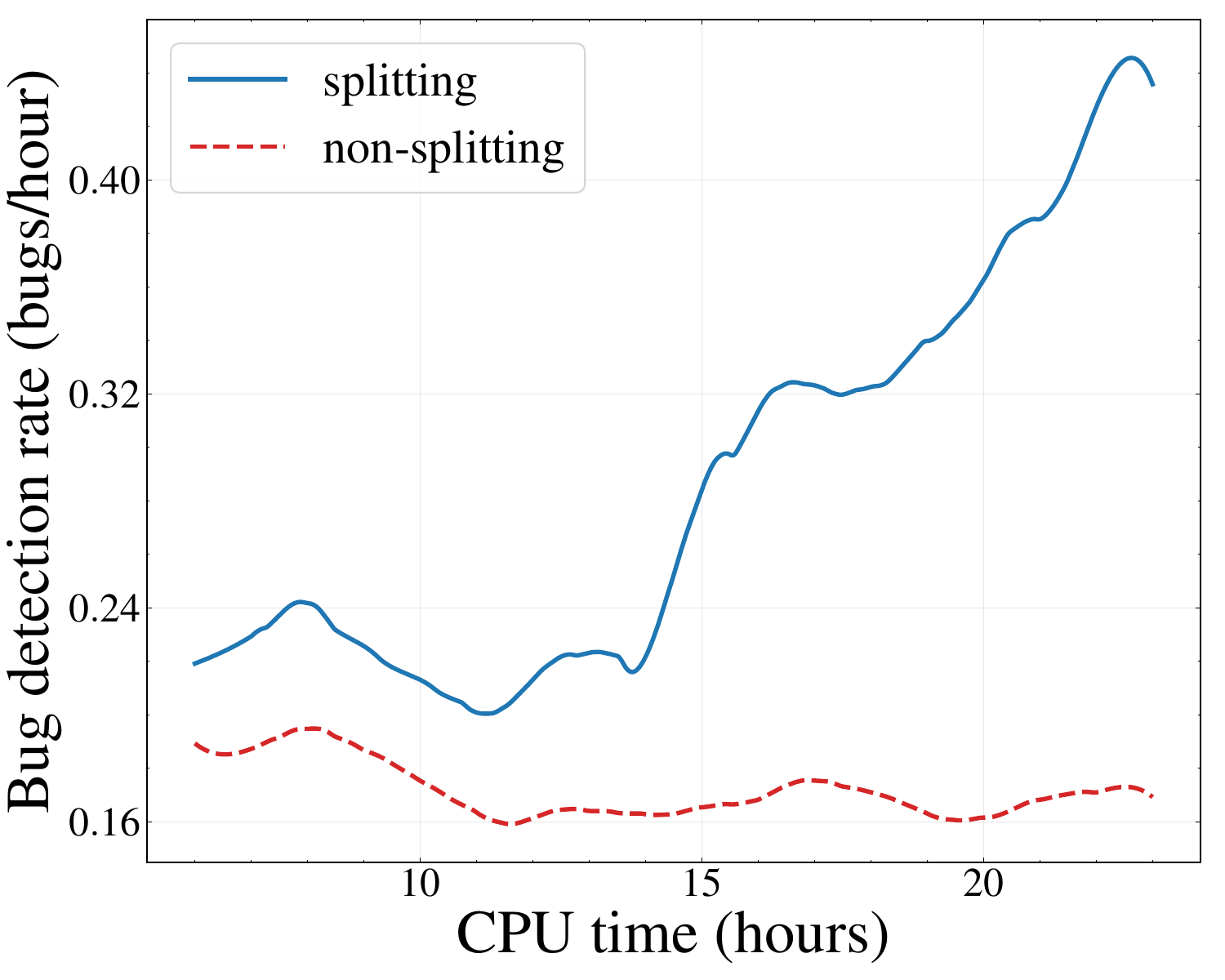}\hfill\includegraphics[width=0.19\textwidth]{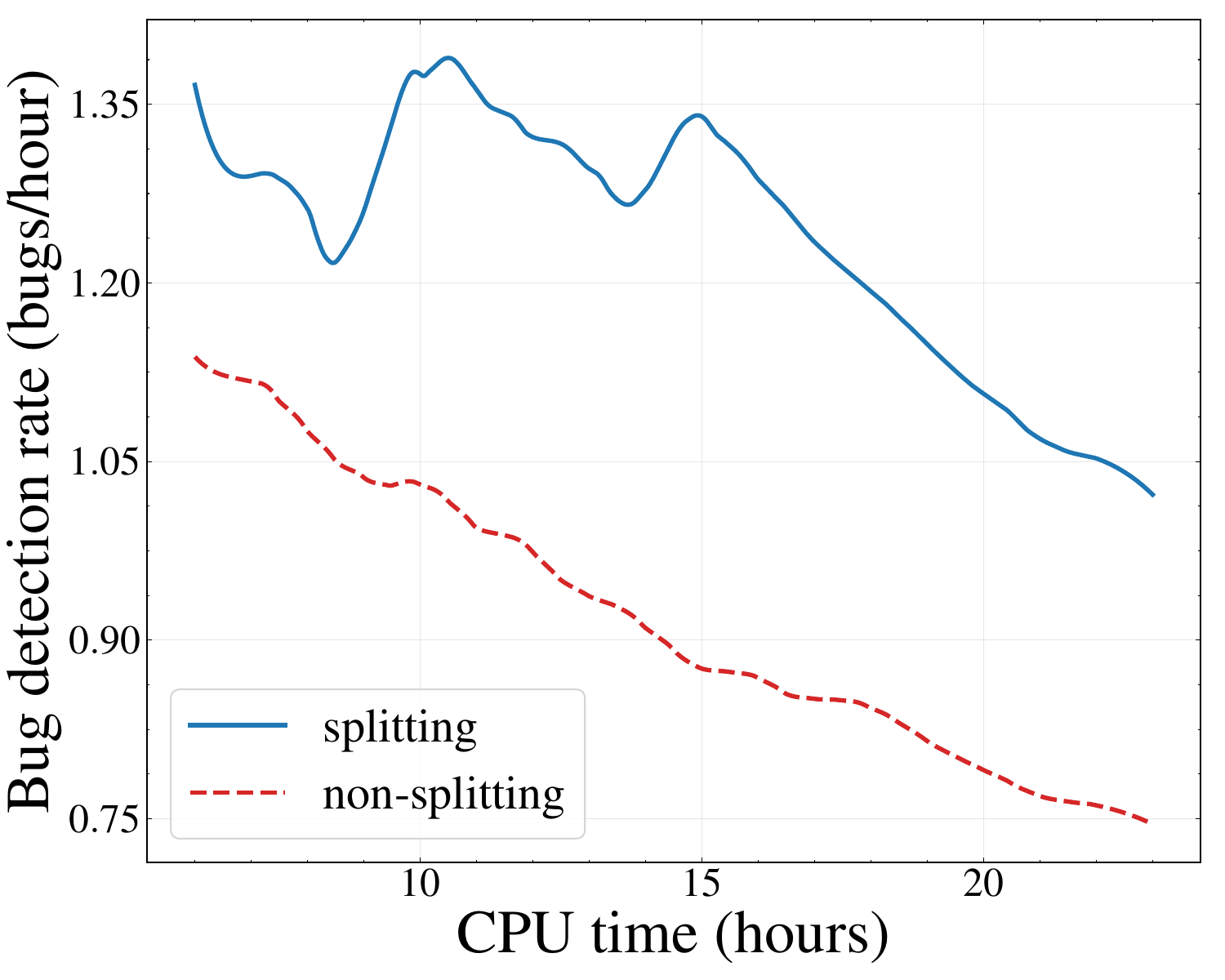}\hfill\includegraphics[width=0.19\textwidth]{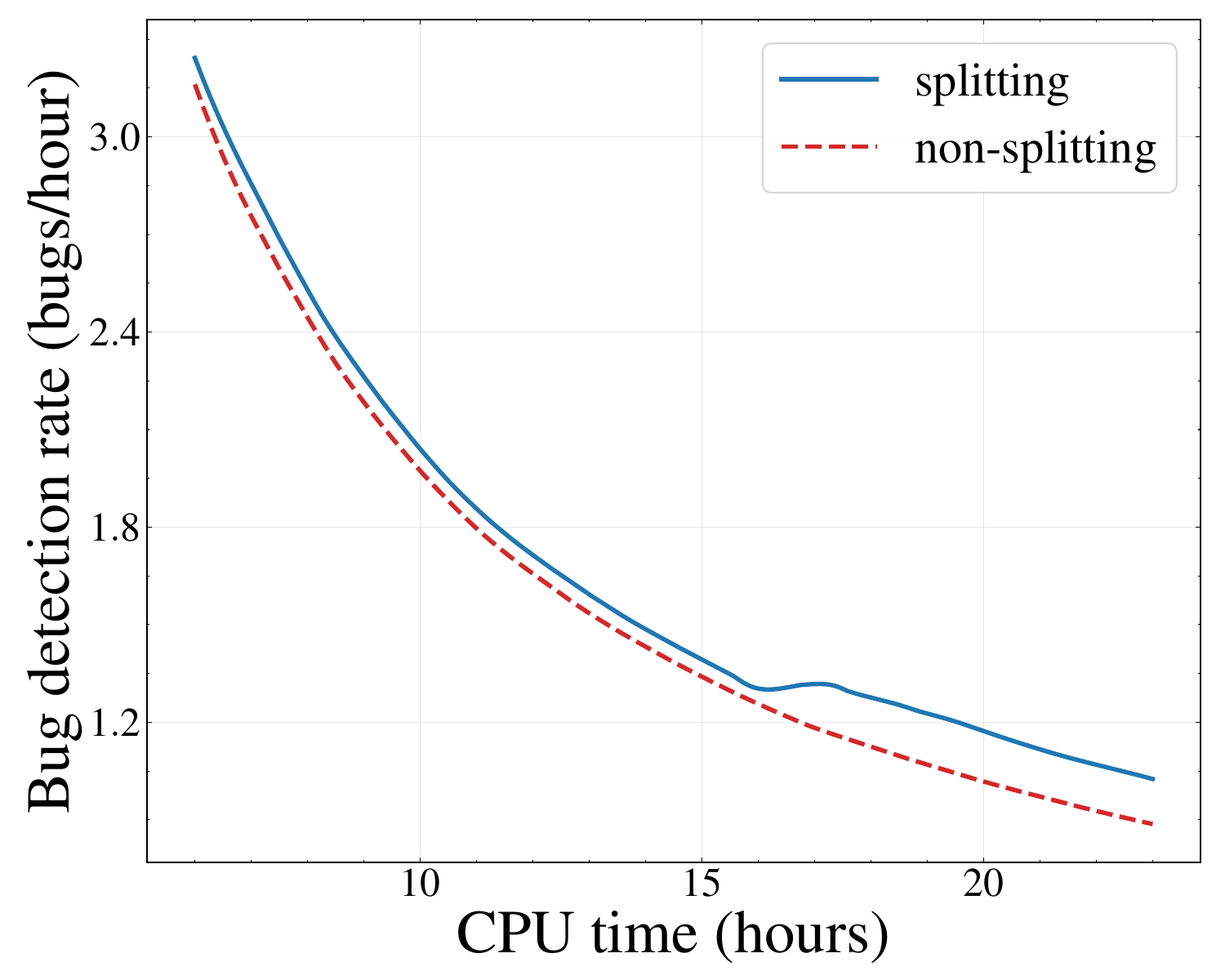}\\
    \makebox[0.19\textwidth]{\footnotesize (f) BDR of matio}\hfill\makebox[0.19\textwidth]{\footnotesize (g) BDR of openh264}\hfill\makebox[0.19\textwidth]{\footnotesize (h) BDR of php}\hfill\makebox[0.19\textwidth]{\footnotesize (i) BDR of poppler}\hfill\makebox[0.19\textwidth]{\footnotesize (j) BDR of stb}
    \caption{Comparisons of bug detection rate across 10 benchmarks for fuzzer Honggfuzz.}
    \label{fig:bdr_honggfuzz}
\end{figure*}

\begin{figure*}[tp]
    \centering
    \includegraphics[width=0.19\textwidth]{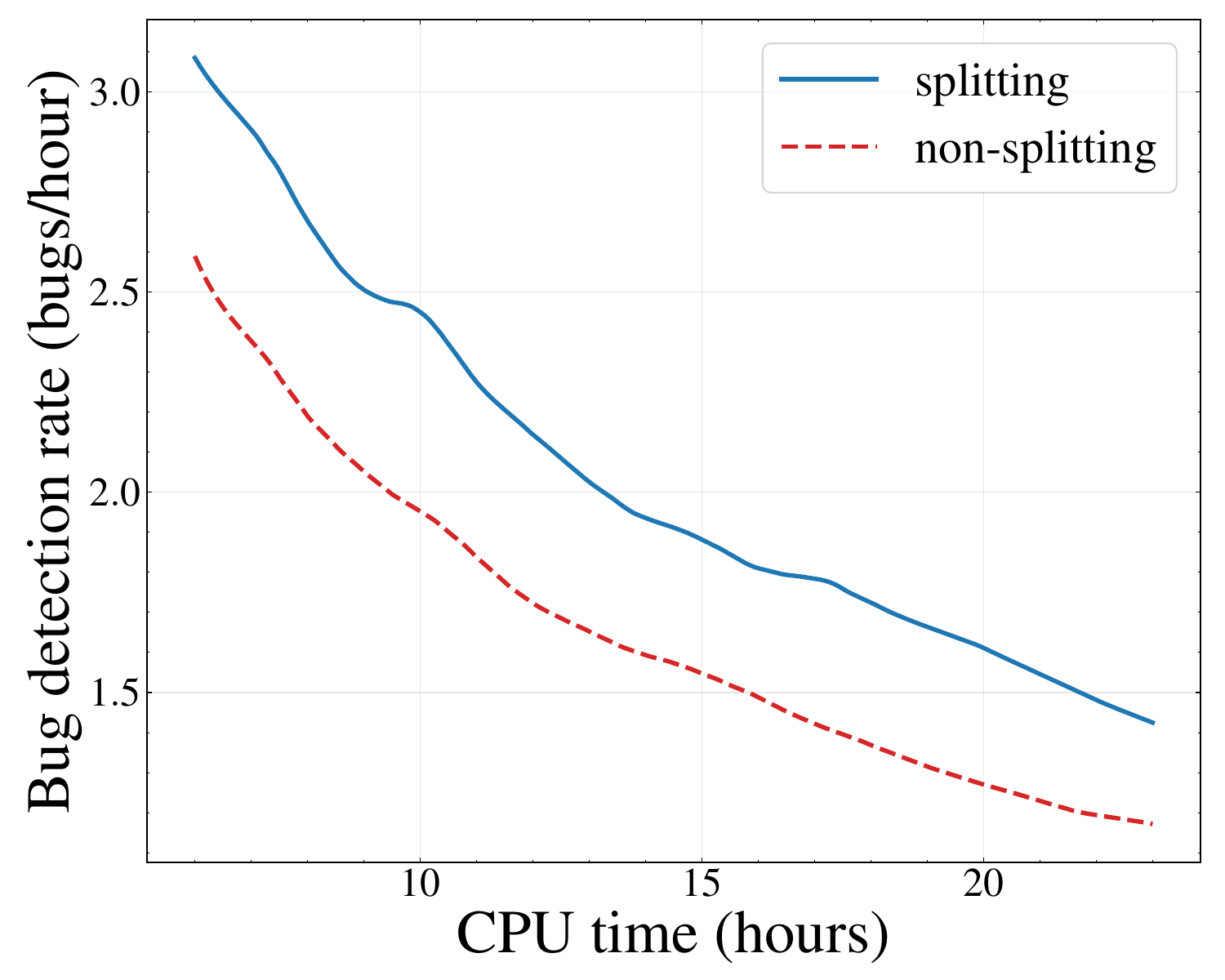}\hfill\includegraphics[width=0.19\textwidth]{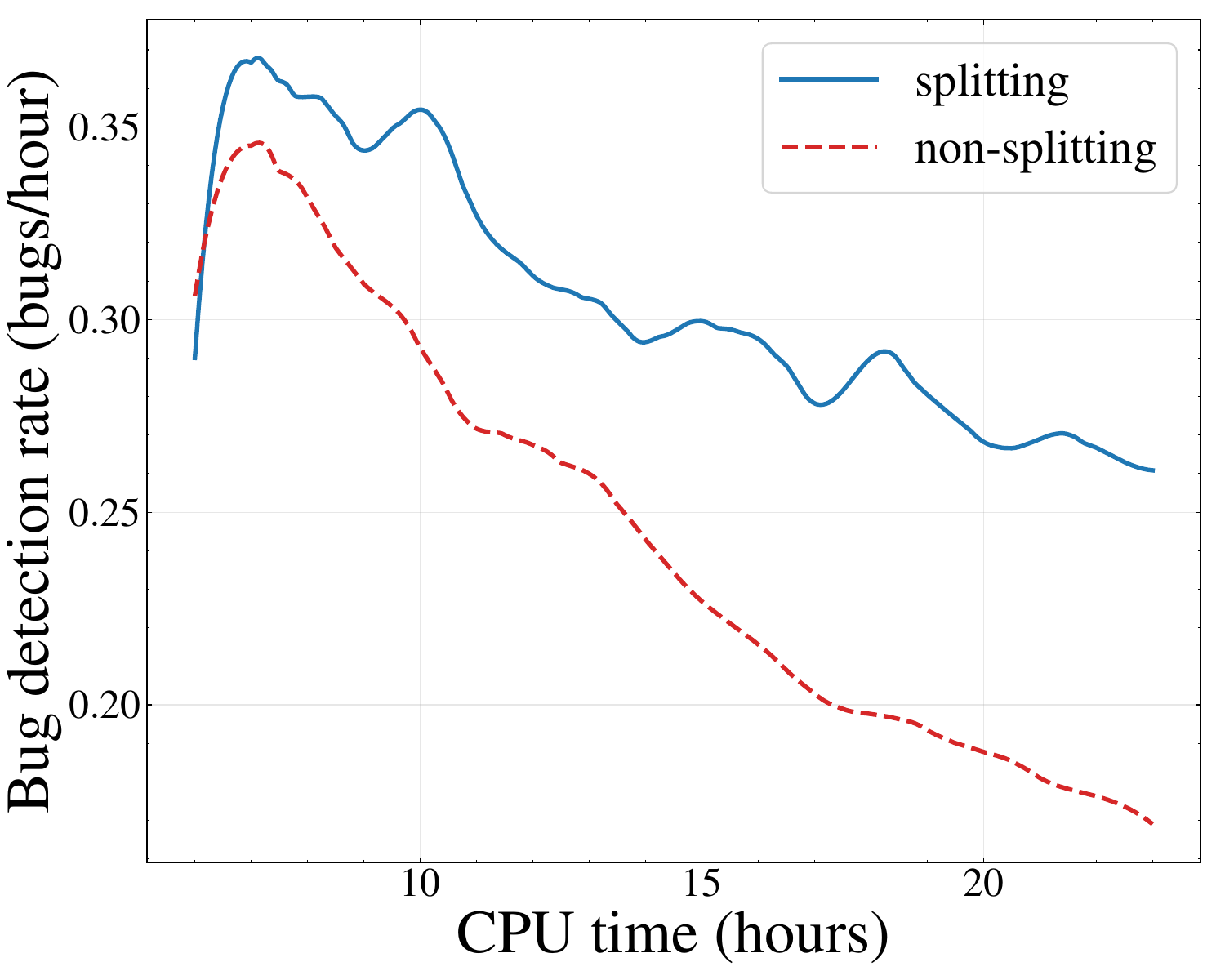}\hfill\includegraphics[width=0.19\textwidth]{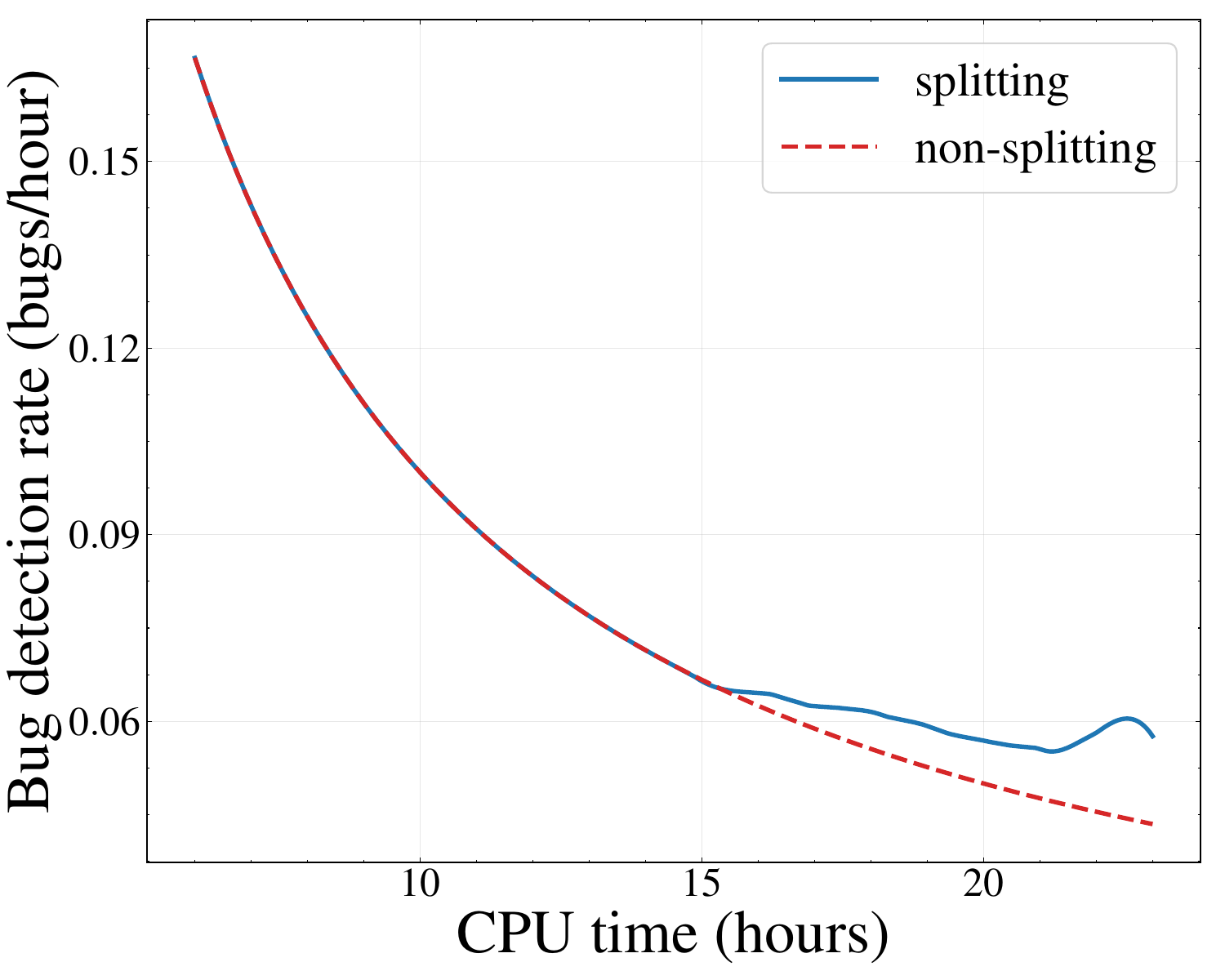}\hfill\includegraphics[width=0.19\textwidth]{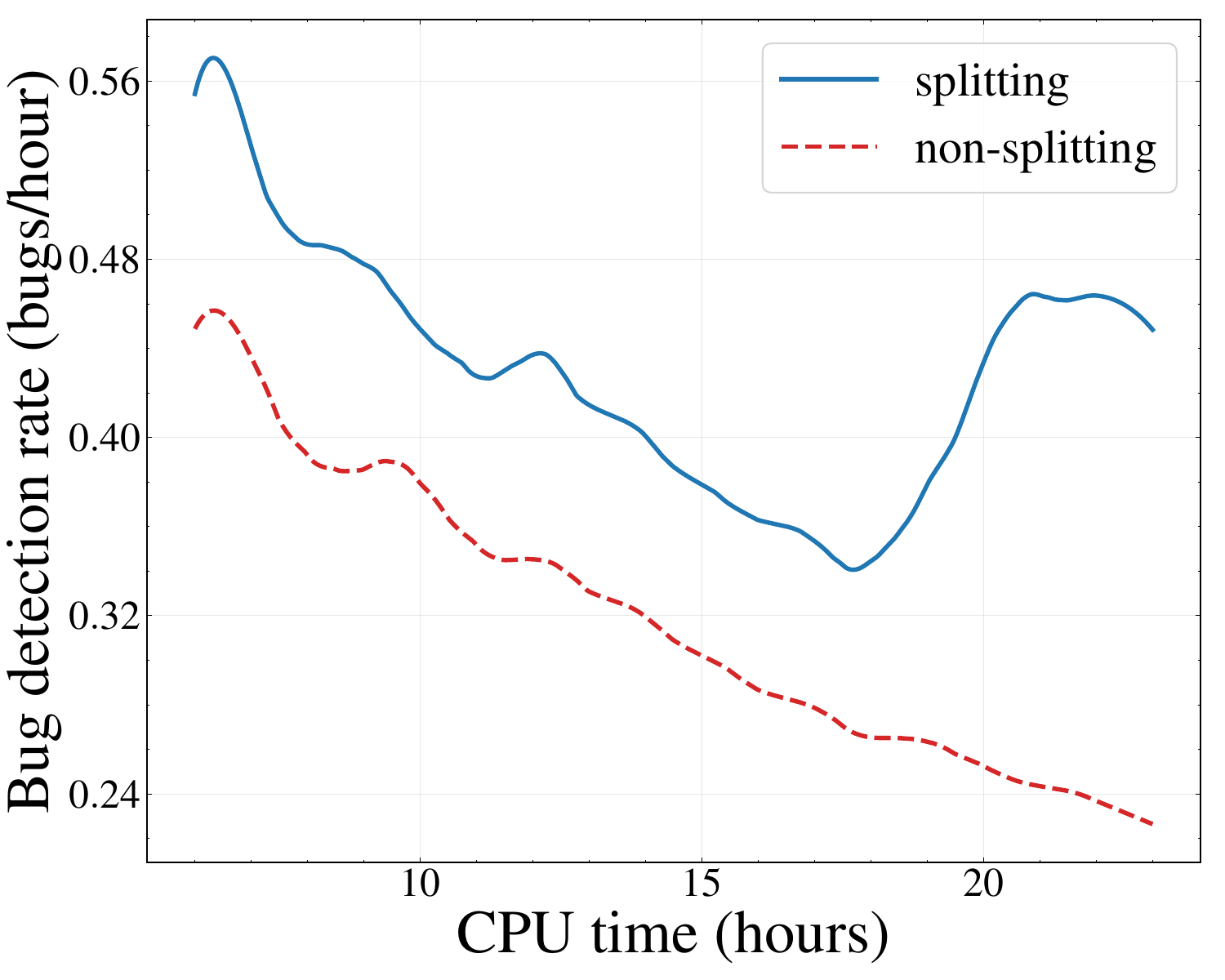}\hfill\includegraphics[width=0.19\textwidth]{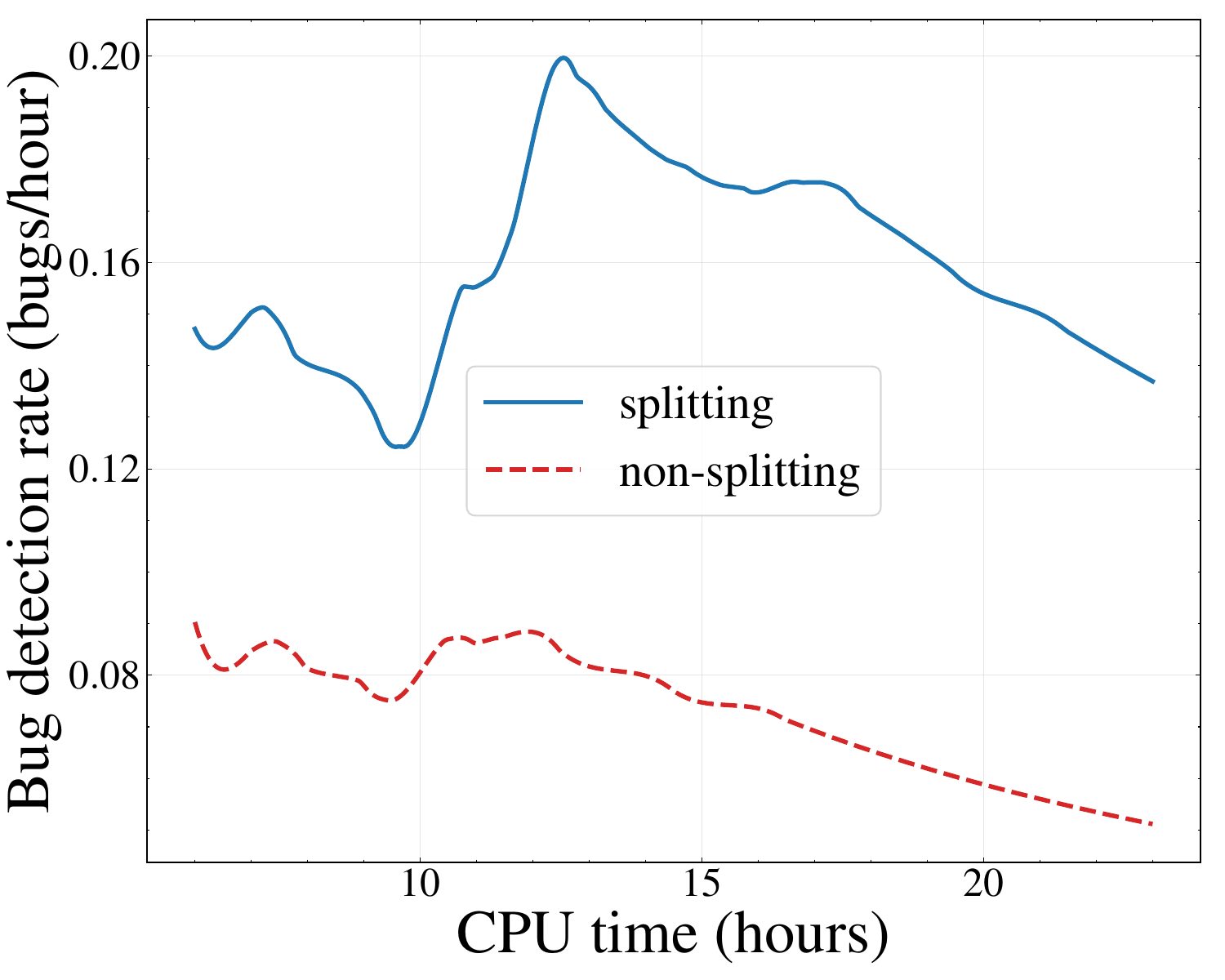}\\
    \makebox[0.19\textwidth]{\footnotesize (a) BDR of arrow}\hfill\makebox[0.19\textwidth]{\footnotesize (b) BDR of ffmpeg}\hfill\makebox[0.19\textwidth]{\footnotesize (c) BDR of grok}\hfill\makebox[0.19\textwidth]{\footnotesize (d) BDR of libhevc}\hfill\makebox[0.19\textwidth]{\footnotesize (e) BDR of libhtp}\\[3pt]
    \includegraphics[width=0.19\textwidth]{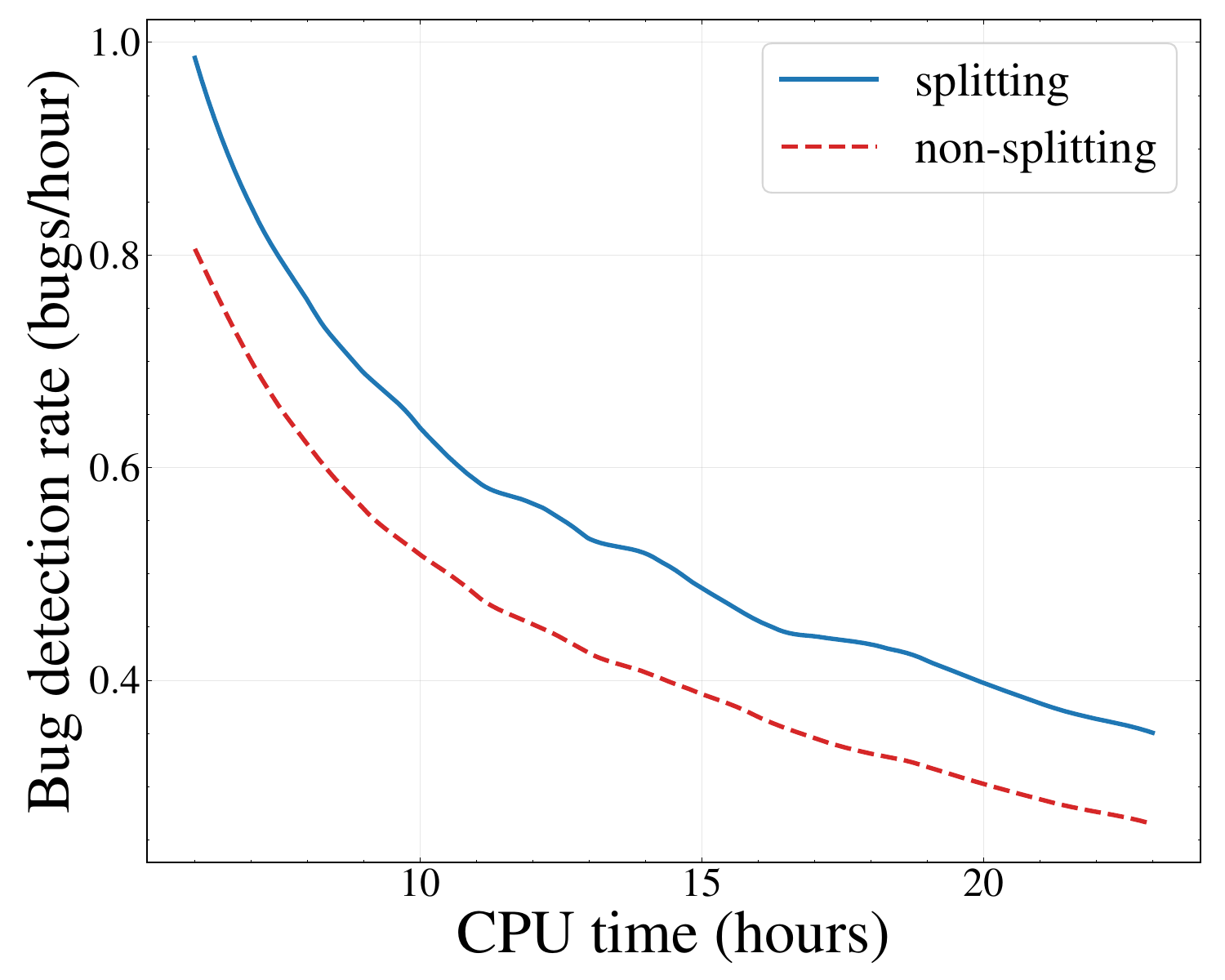}\hfill\includegraphics[width=0.19\textwidth]{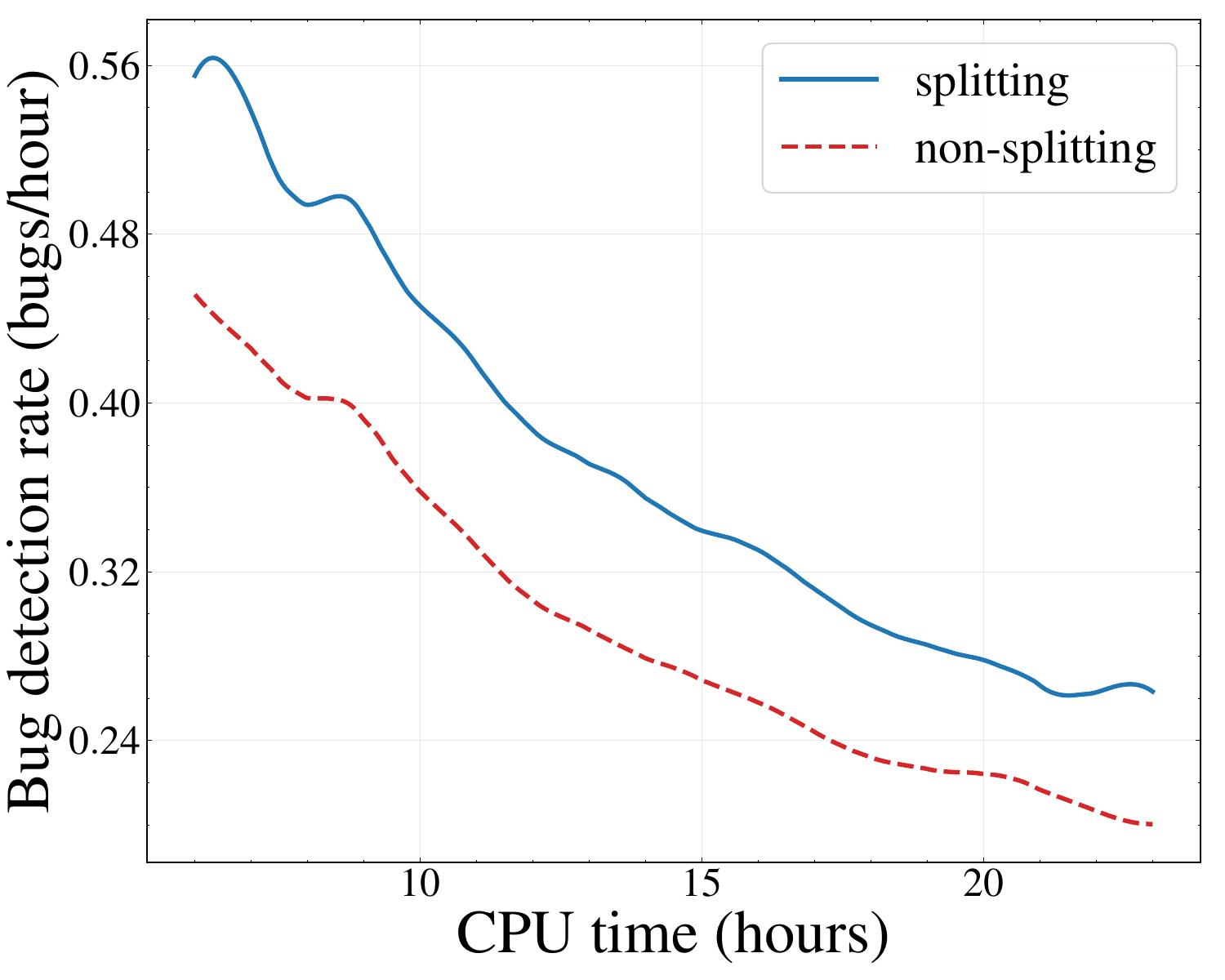}\hfill\includegraphics[width=0.19\textwidth]{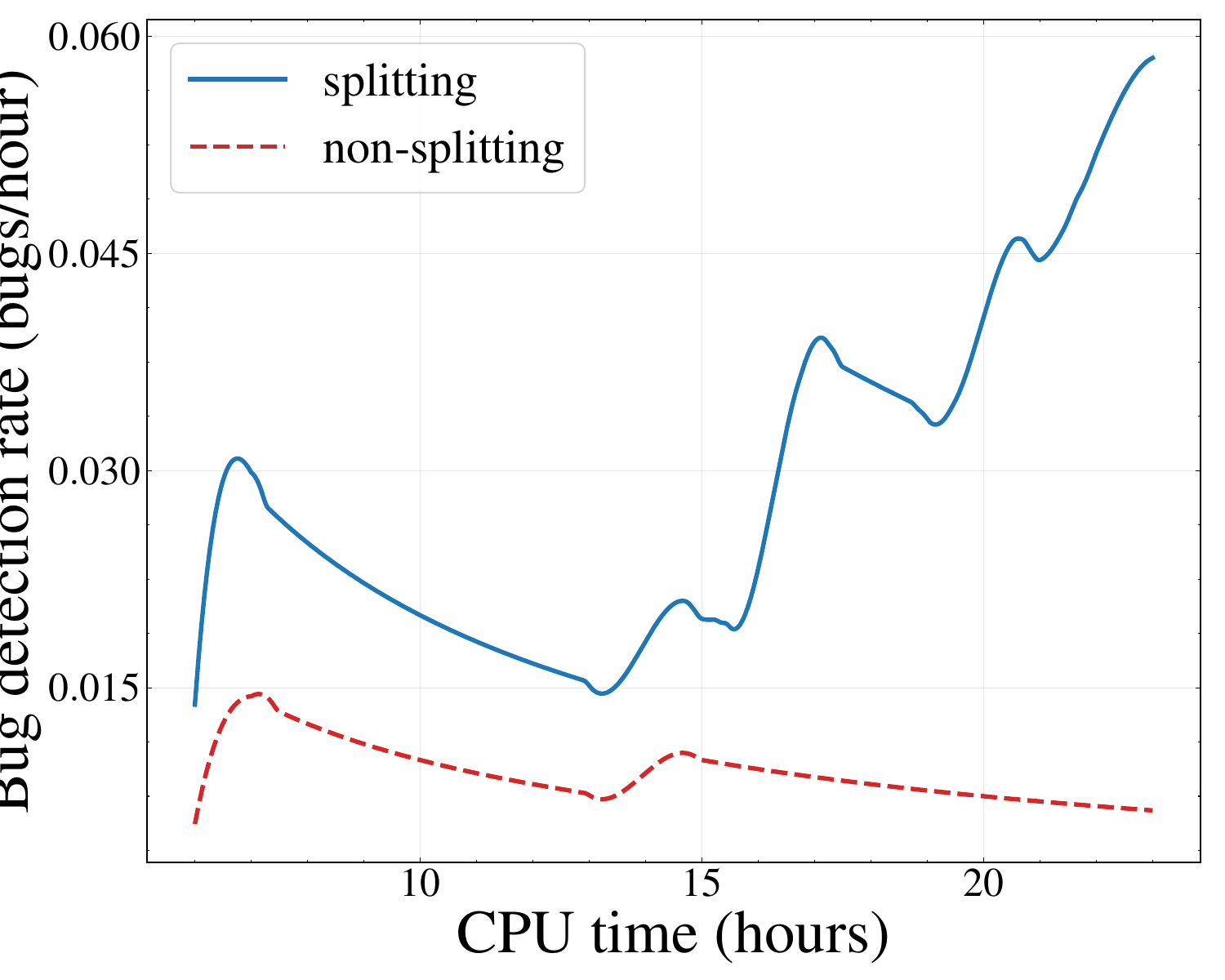}\hfill\includegraphics[width=0.19\textwidth]{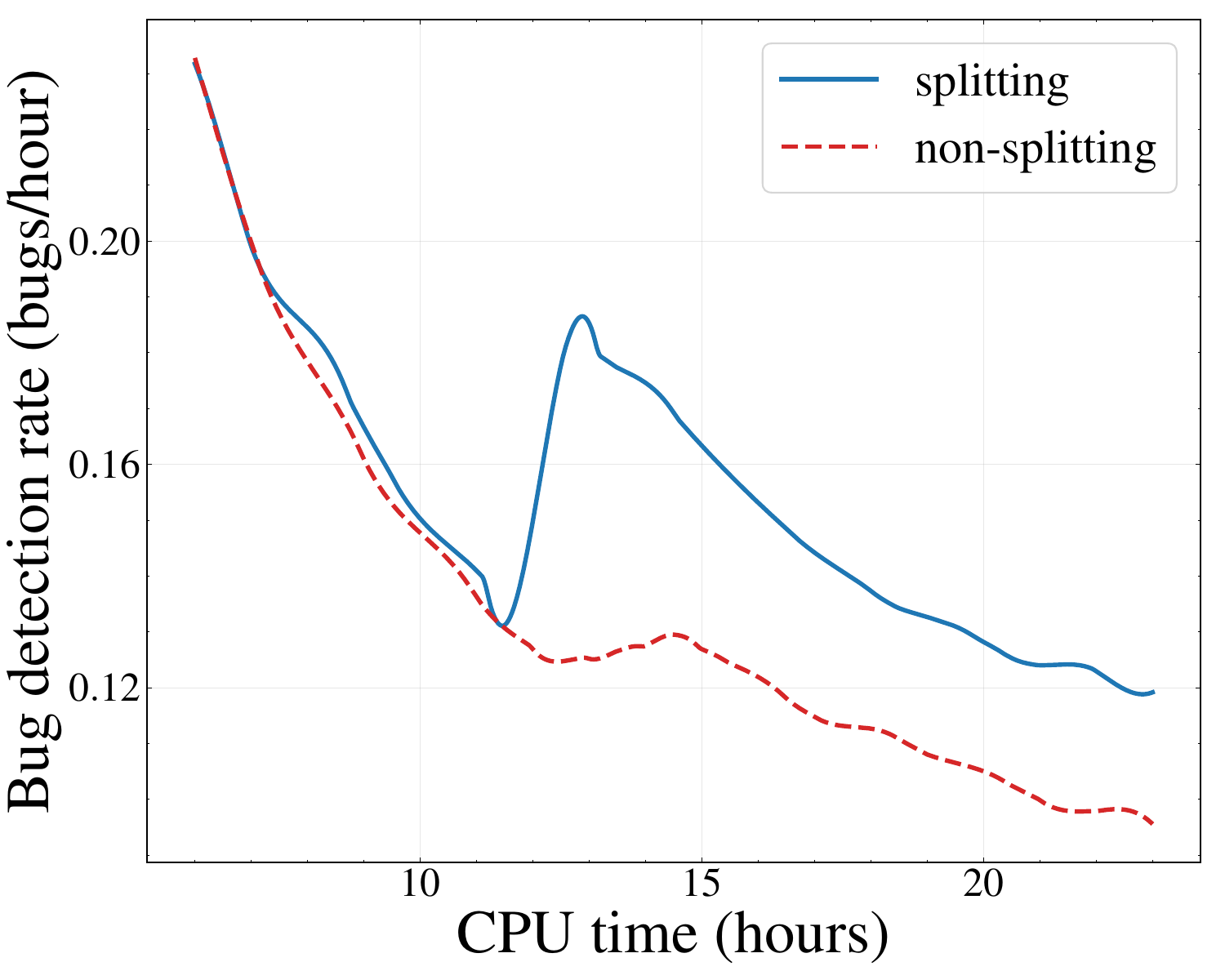}\hfill\includegraphics[width=0.19\textwidth]{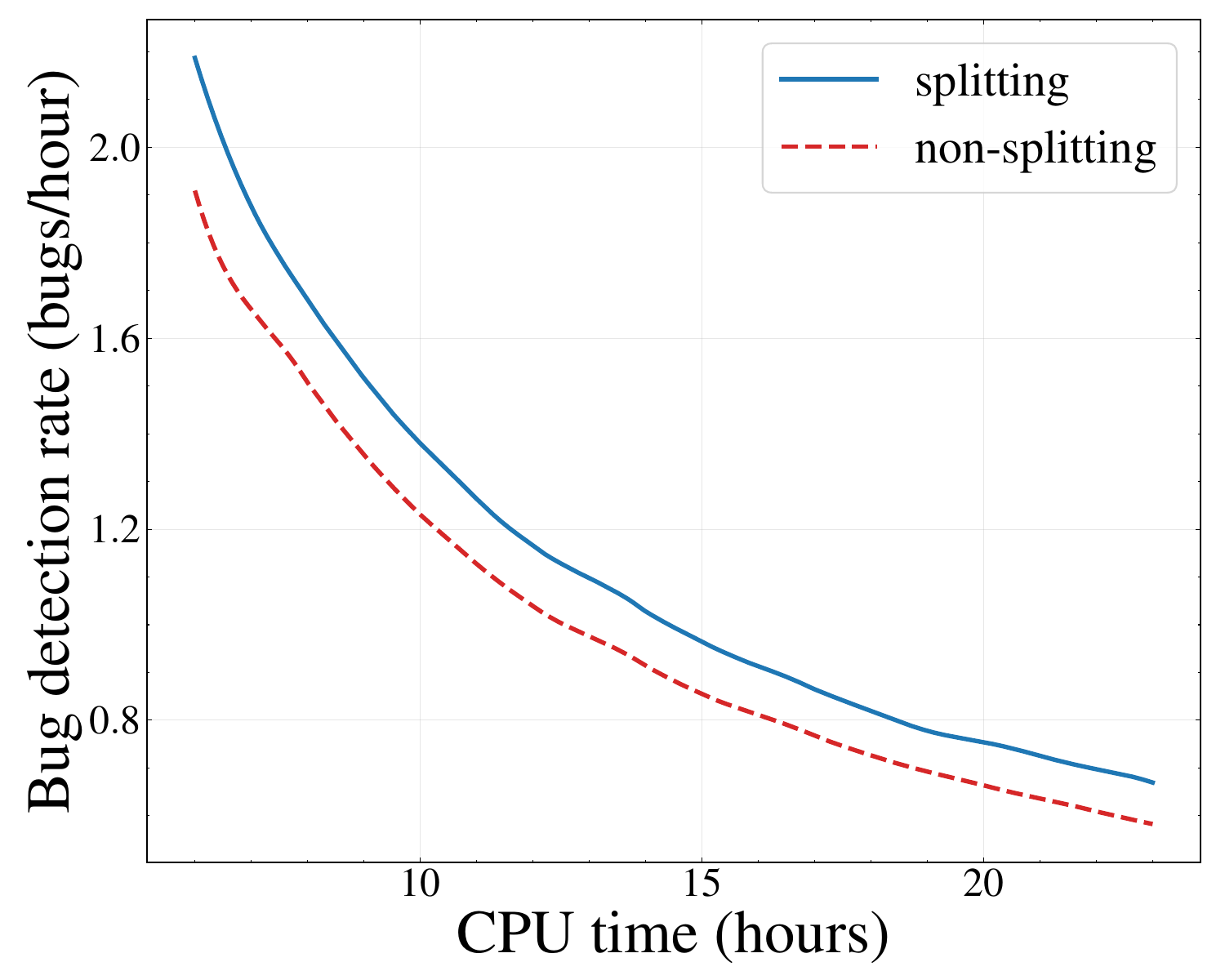}\\
    \makebox[0.19\textwidth]{\footnotesize (f) BDR of matio}\hfill\makebox[0.19\textwidth]{\footnotesize (g) BDR of openh264}\hfill\makebox[0.19\textwidth]{\footnotesize (h) BDR of php}\hfill\makebox[0.19\textwidth]{\footnotesize (i) BDR of poppler}\hfill\makebox[0.19\textwidth]{\footnotesize (j) BDR of stb}
    \caption{Comparisons of bug detection rate across 10 benchmarks for fuzzer libFuzzer.}
    \label{fig:bdr_libfuzzer}
\end{figure*}

\begin{figure}[htb]
    \centering
    \includegraphics[width=0.235\textwidth]{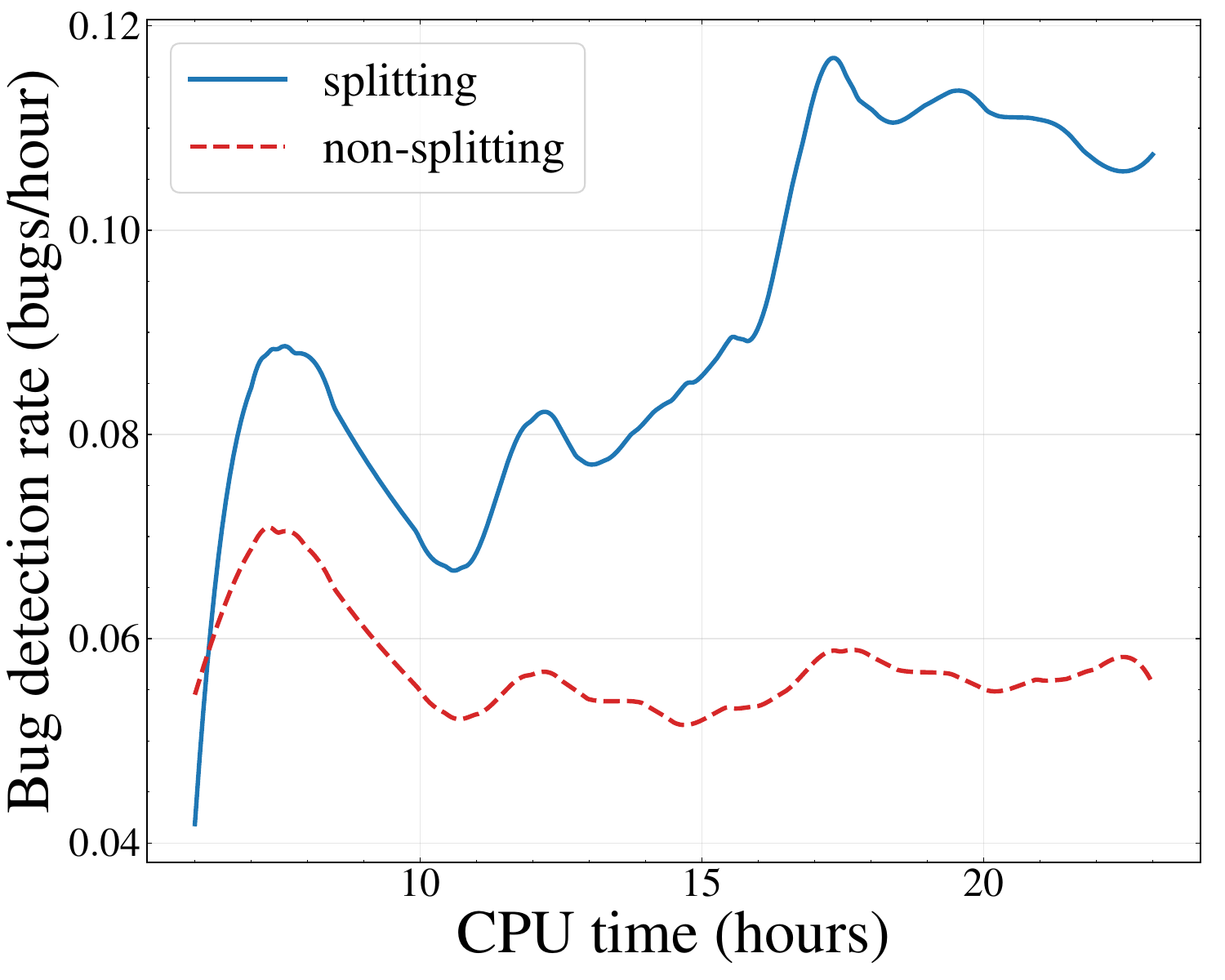}\hfill\includegraphics[width=0.235\textwidth]{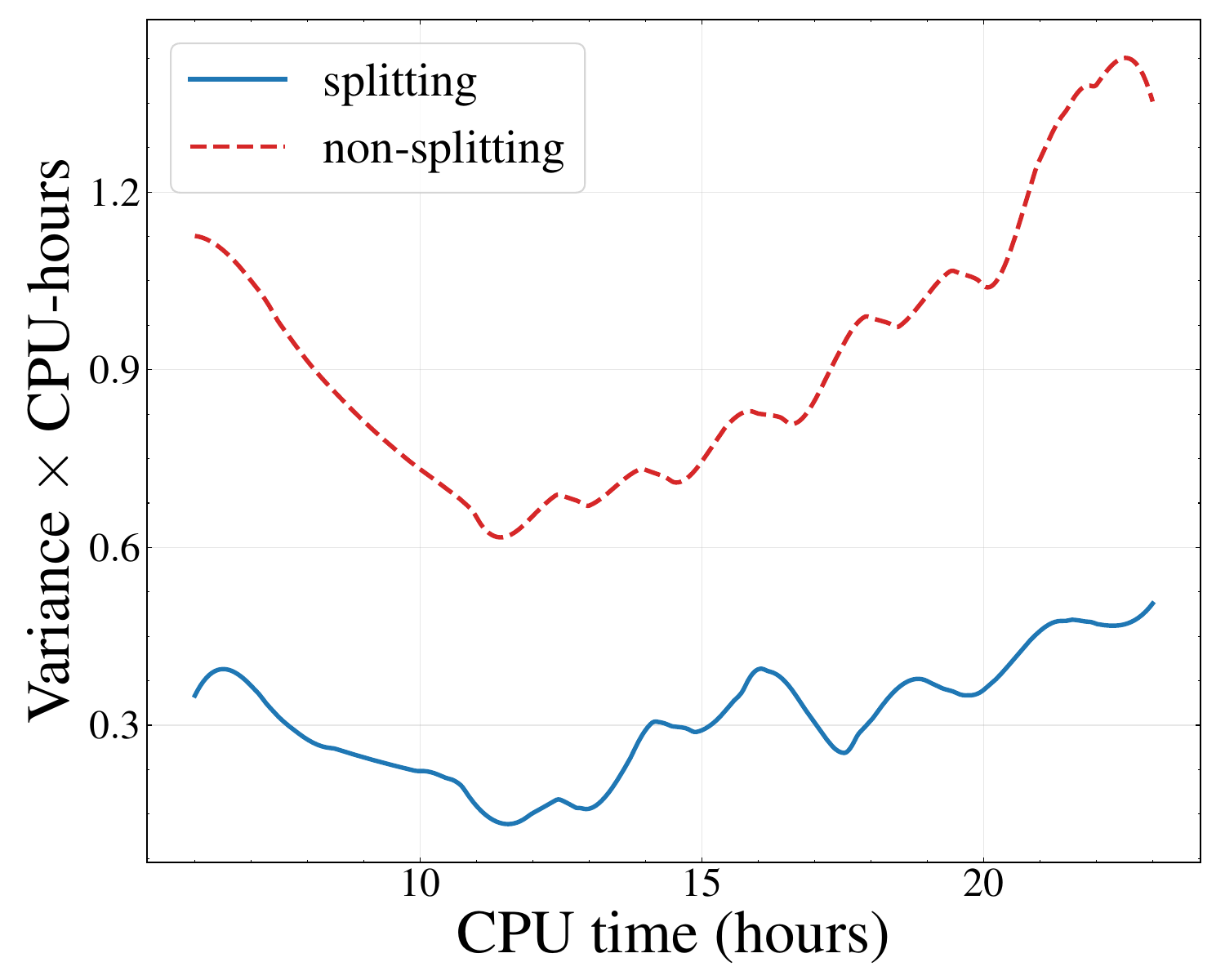}\\
    \makebox[0.235\textwidth]{\small (a) BDR of php}\hfill\makebox[0.235\textwidth]{\small (b) variance of php}\\[3pt]
    \includegraphics[width=0.235\textwidth]{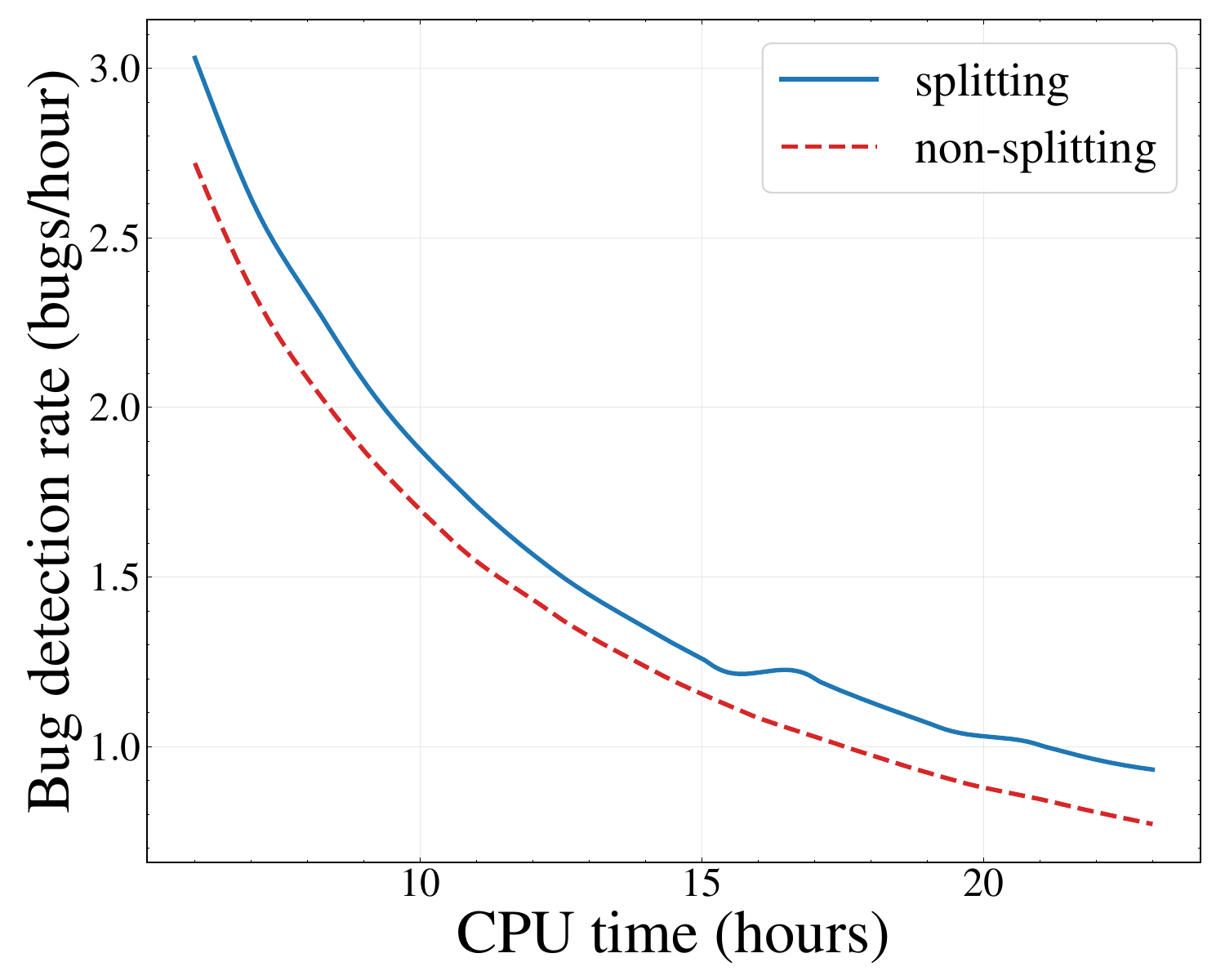}\hfill\includegraphics[width=0.235\textwidth]{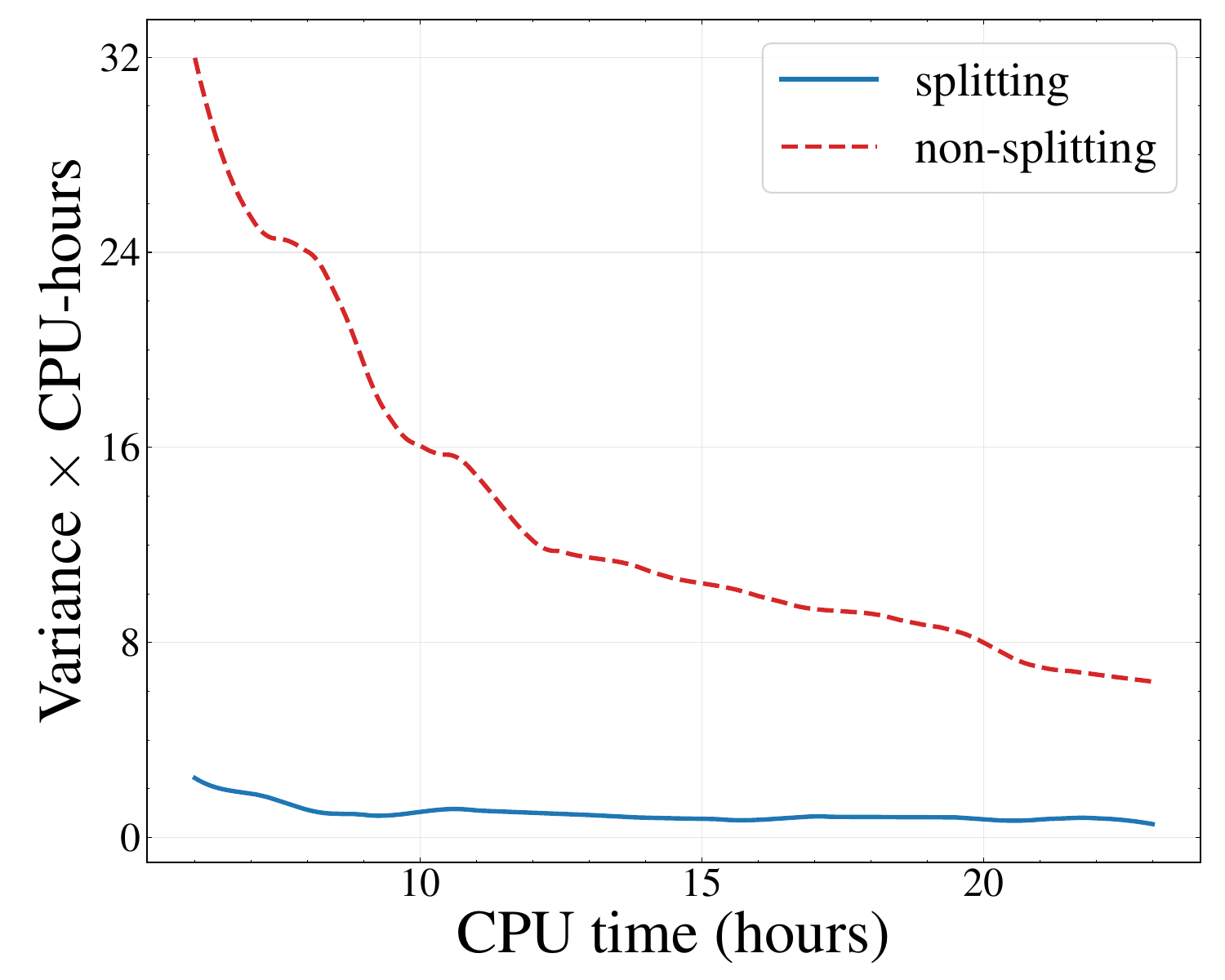}\\
    \makebox[0.235\textwidth]{\small (c) BDR of stb}\hfill\makebox[0.235\textwidth]{\small (d) variance of stb}
    \caption{Ablation study for fuzzer aflplusplus on php and stb (companion to Figure~\ref{fig:ablation_aflplusplus_1}).}
    \label{fig:ablation_aflplusplus_full}
\end{figure}

\begin{figure}[htb]
    \centering
    \includegraphics[width=0.235\textwidth]{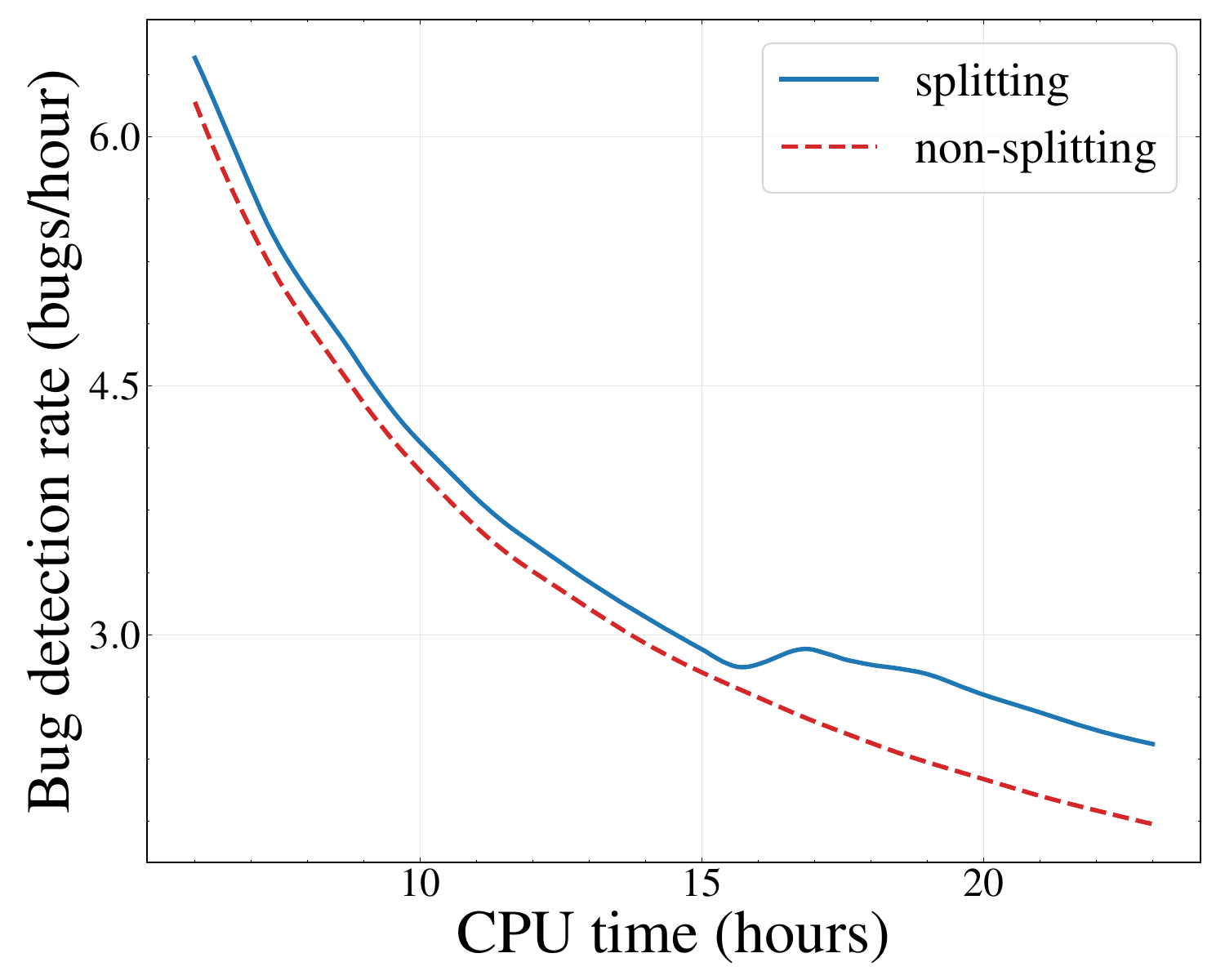}\hfill\includegraphics[width=0.235\textwidth]{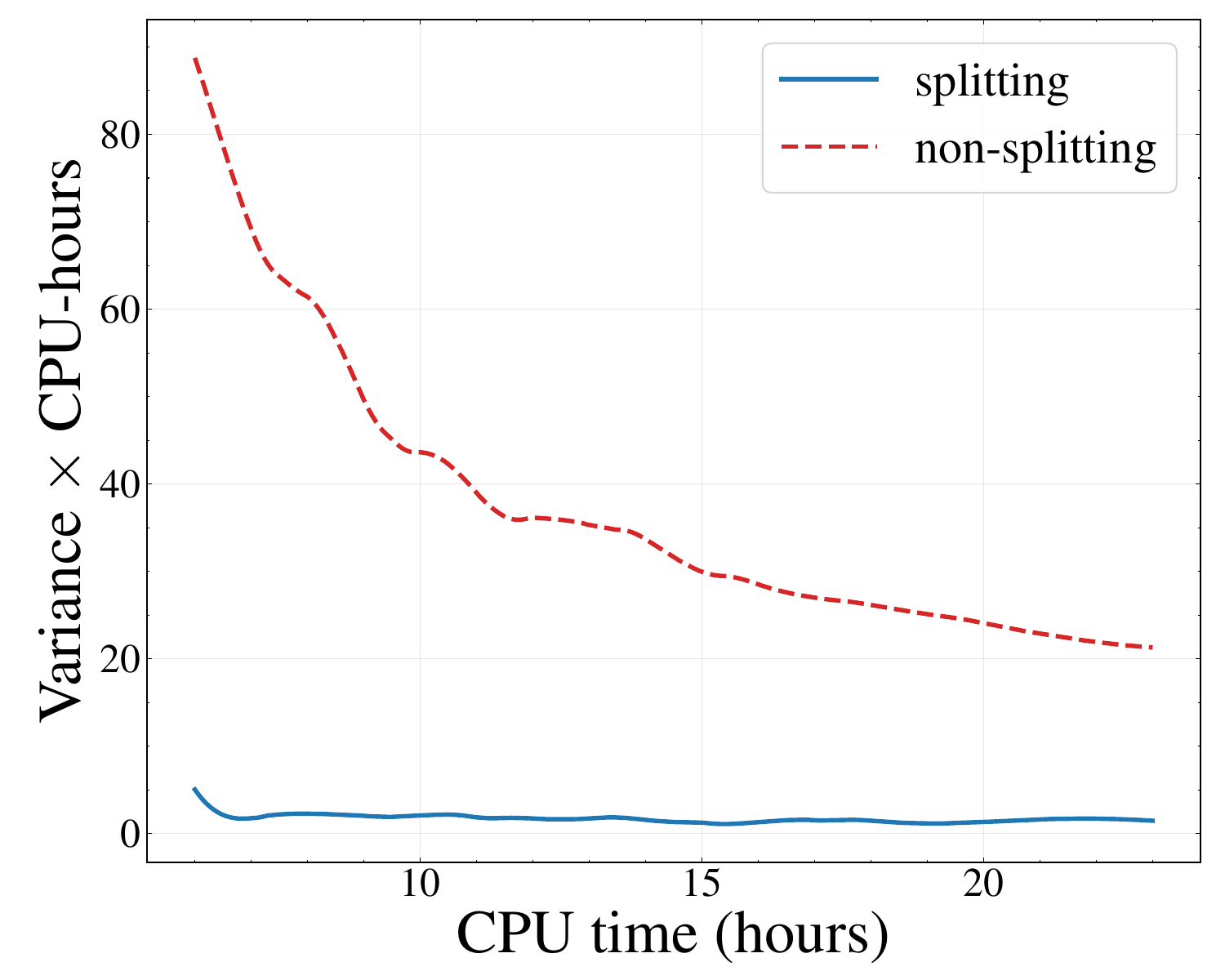}\\
    \makebox[0.235\textwidth]{\small (a) BDR of arrow}\hfill\makebox[0.235\textwidth]{\small (b) variance of arrow}\\[3pt]
    \includegraphics[width=0.235\textwidth]{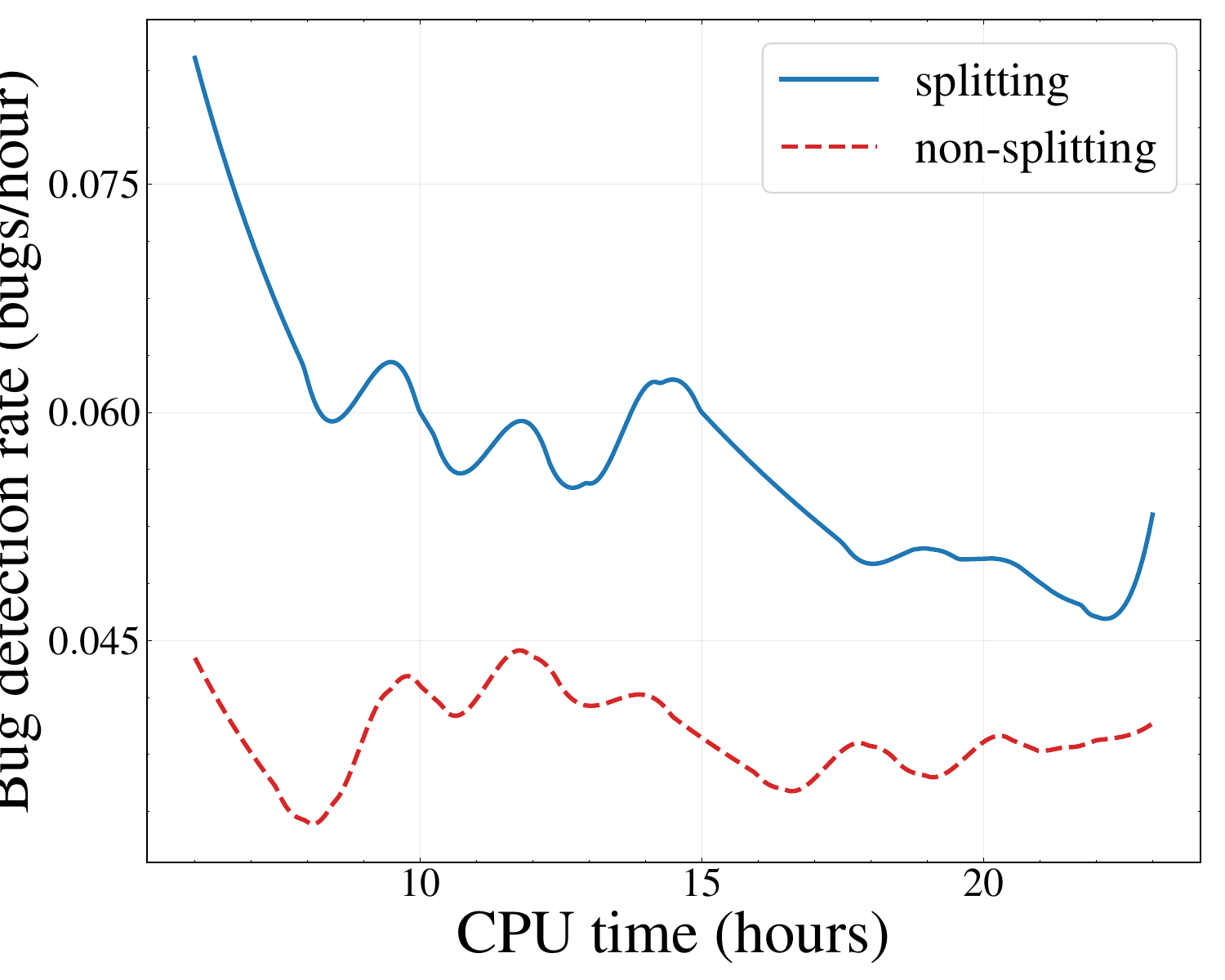}\hfill\includegraphics[width=0.235\textwidth]{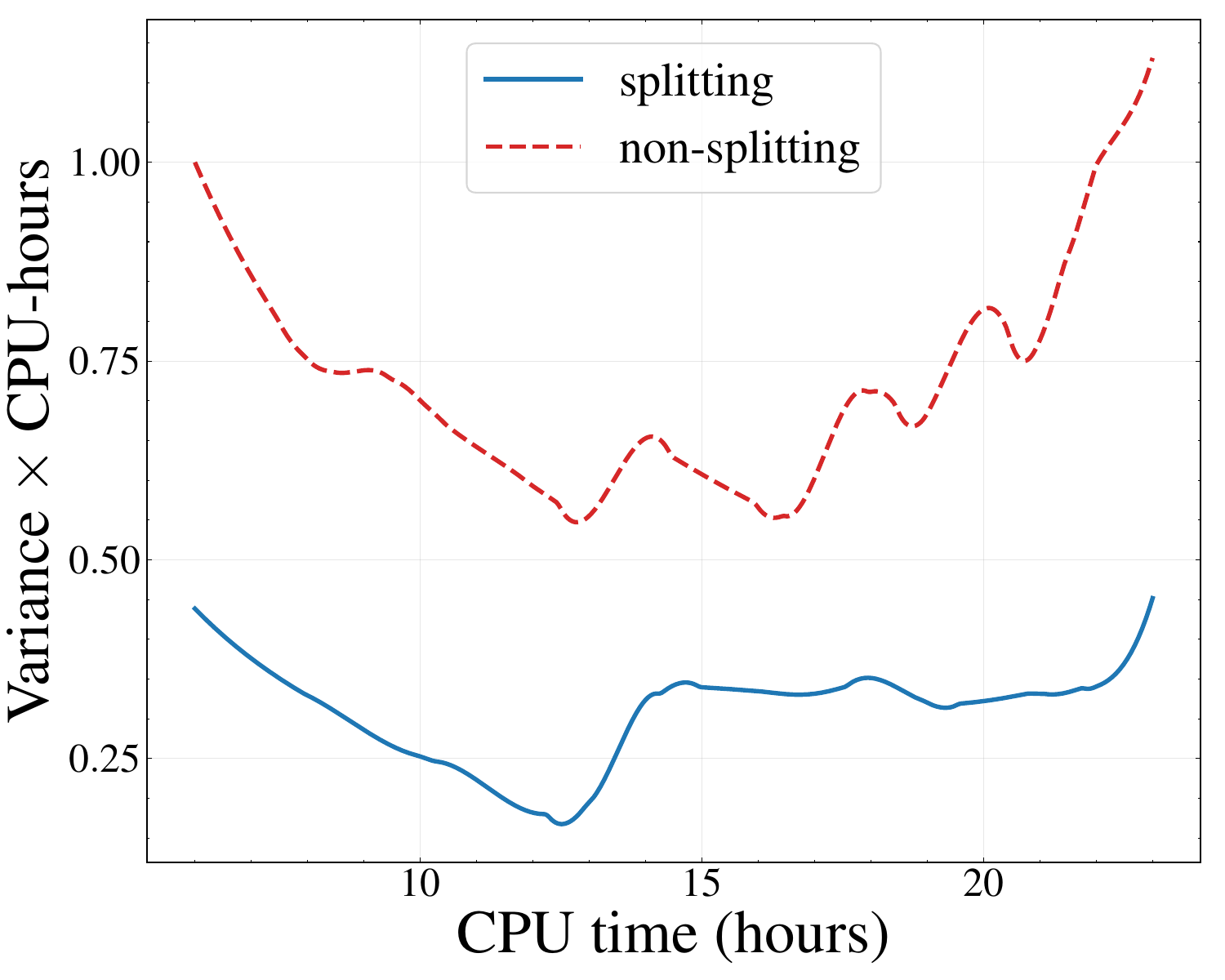}\\
    \makebox[0.235\textwidth]{\small (c) BDR of php}\hfill\makebox[0.235\textwidth]{\small (d) variance of php}\\[3pt]
    \includegraphics[width=0.235\textwidth]{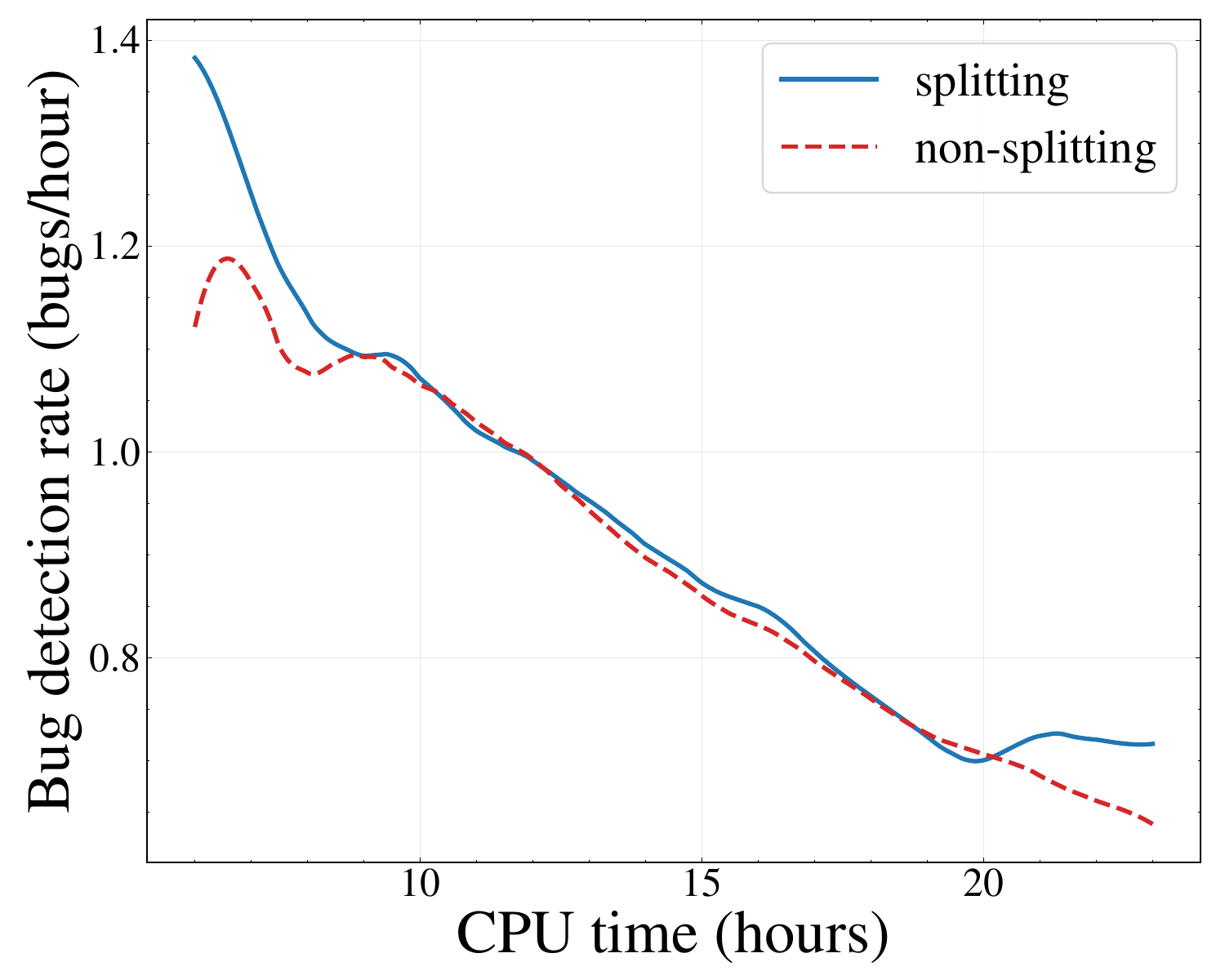}\hfill\includegraphics[width=0.235\textwidth]{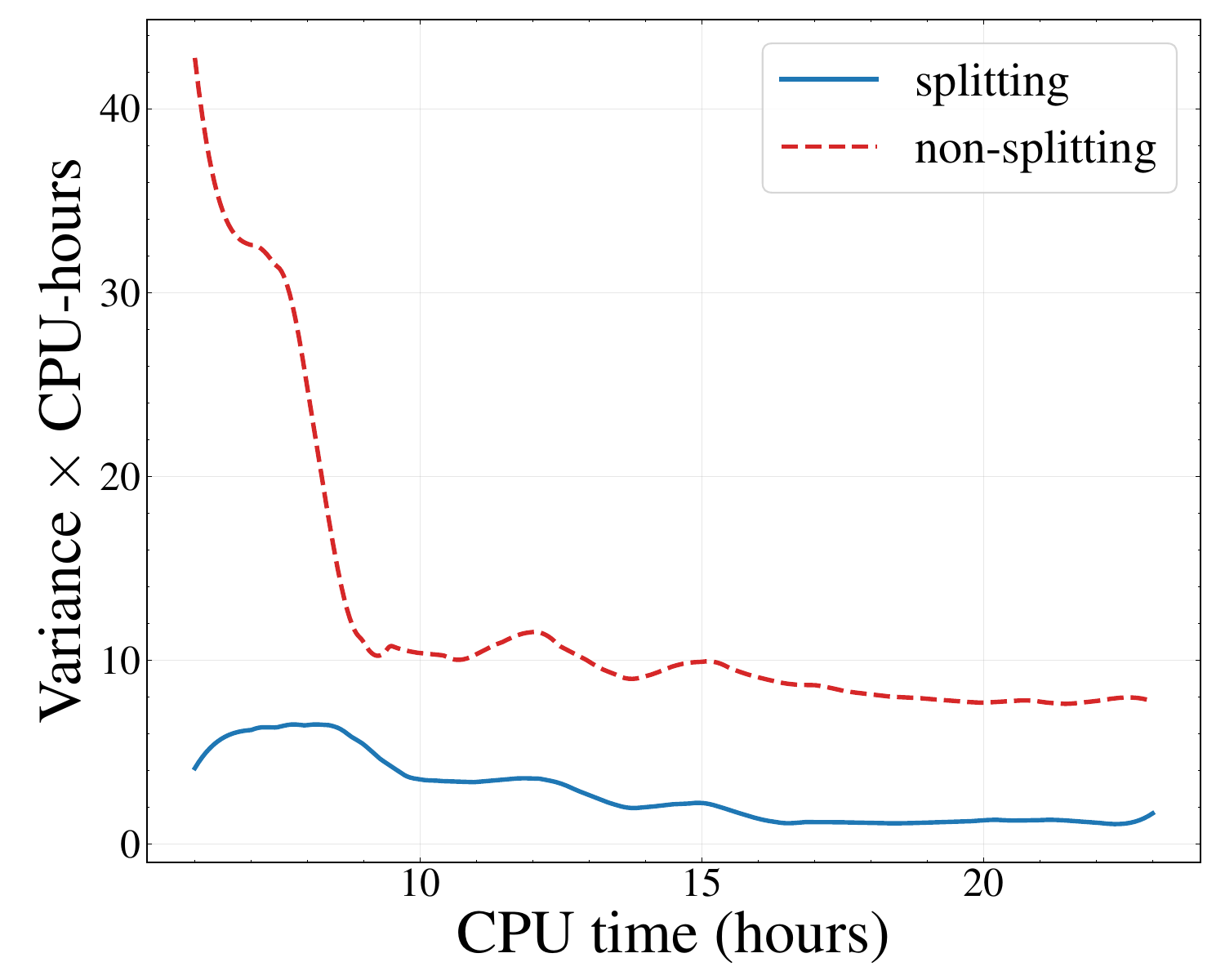}\\
    \makebox[0.235\textwidth]{\small (e) BDR of poppler}\hfill\makebox[0.235\textwidth]{\small (f) variance of poppler}\\[3pt]
    \includegraphics[width=0.235\textwidth]{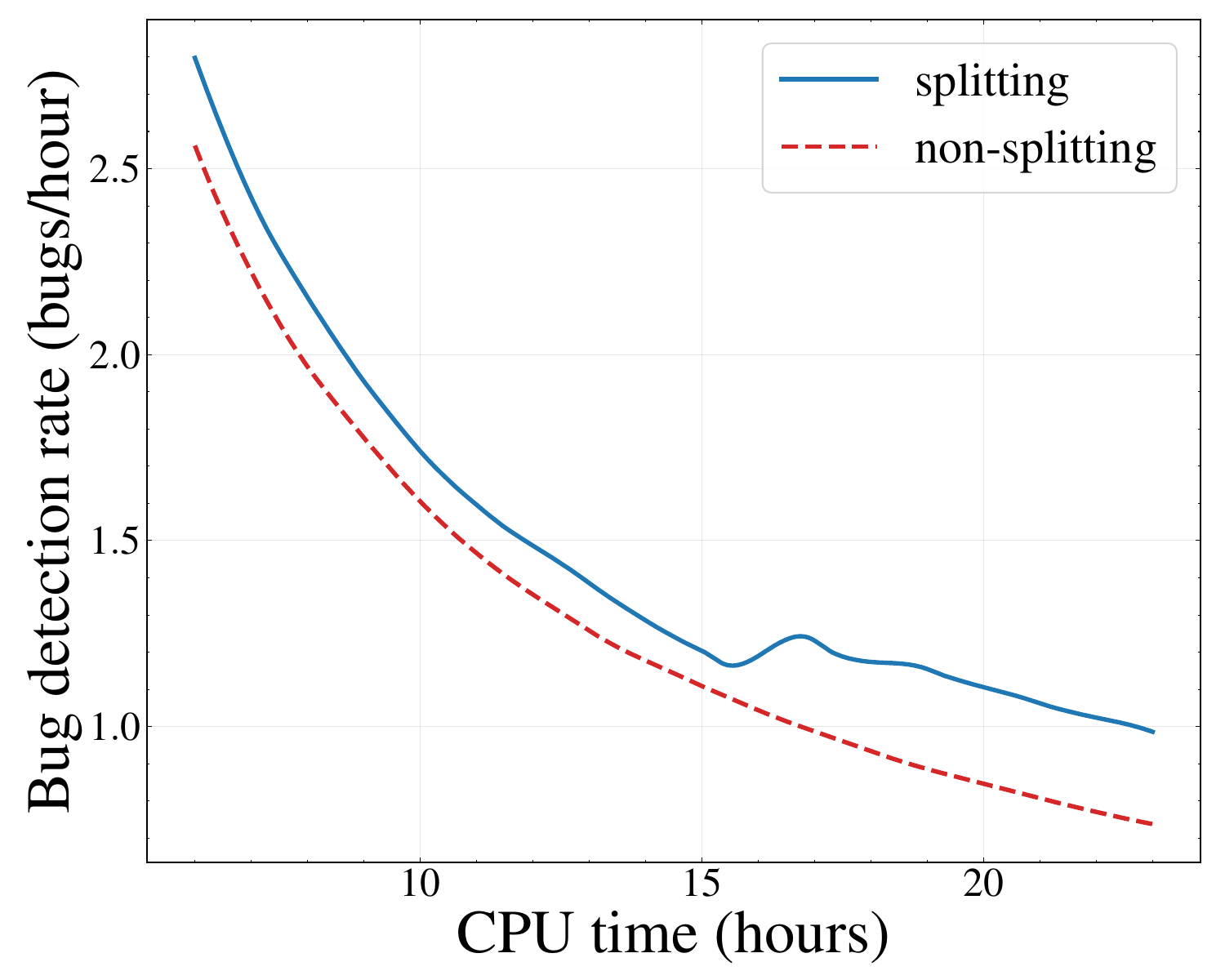}\hfill\includegraphics[width=0.235\textwidth]{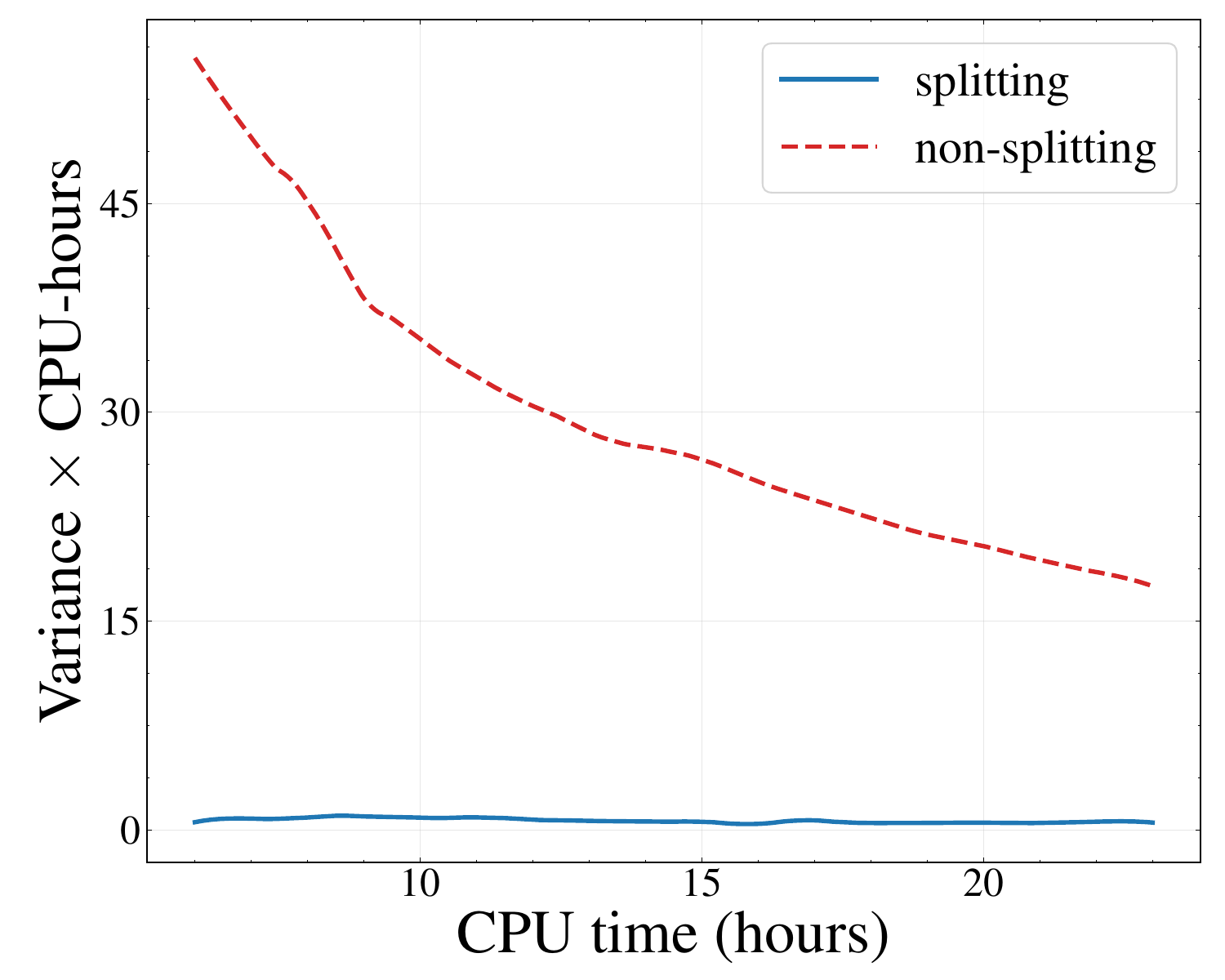}\\
    \makebox[0.235\textwidth]{\small (g) BDR of stb}\hfill\makebox[0.235\textwidth]{\small (h) variance of stb}
    \caption{Ablation study for fuzzer MOpt on arrow, php, poppler and stb.}
    \label{fig:ablation_mopt_1}
\end{figure}

\begin{figure}[htb]
    \centering
    \includegraphics[width=0.235\textwidth]{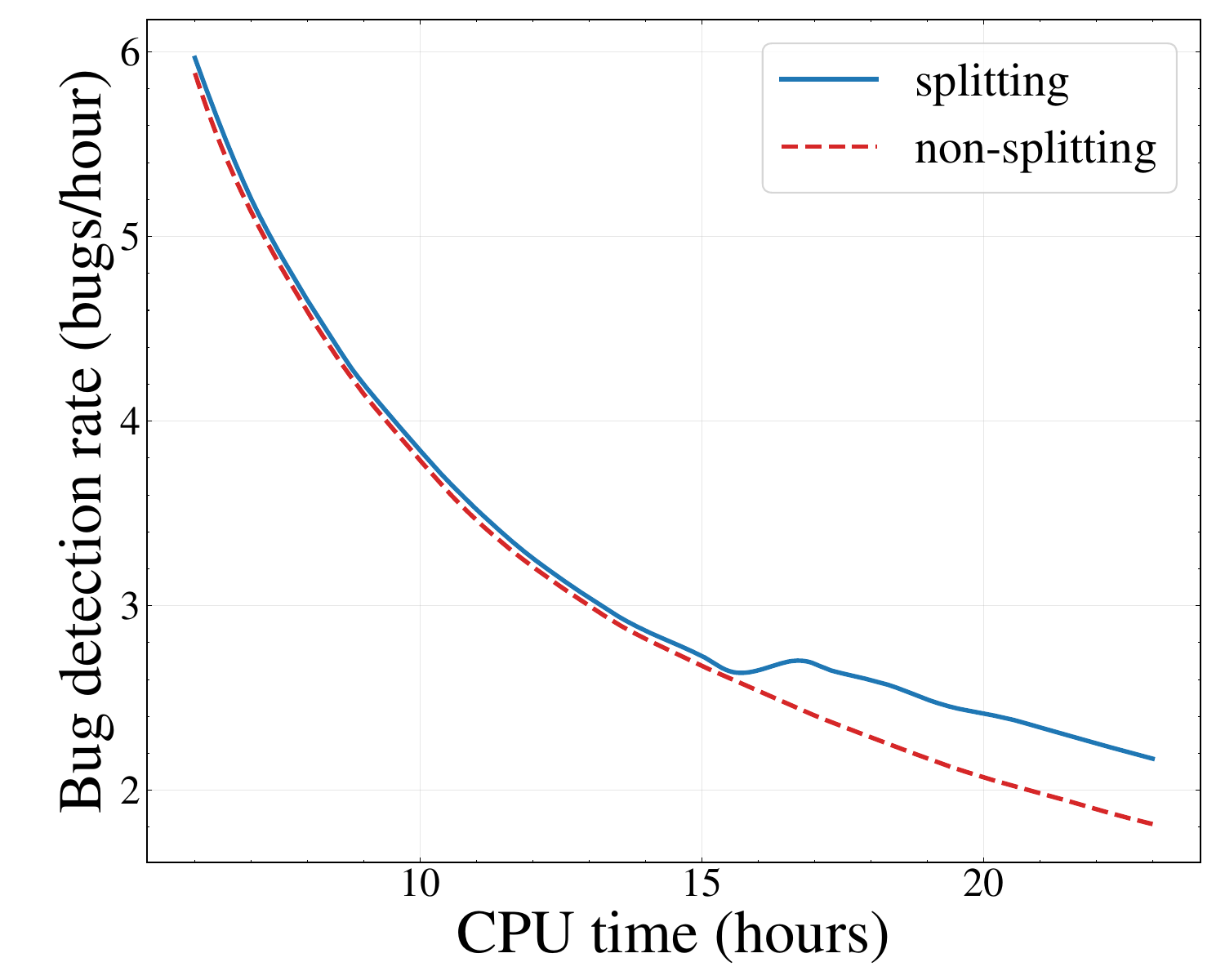}\hfill\includegraphics[width=0.235\textwidth]{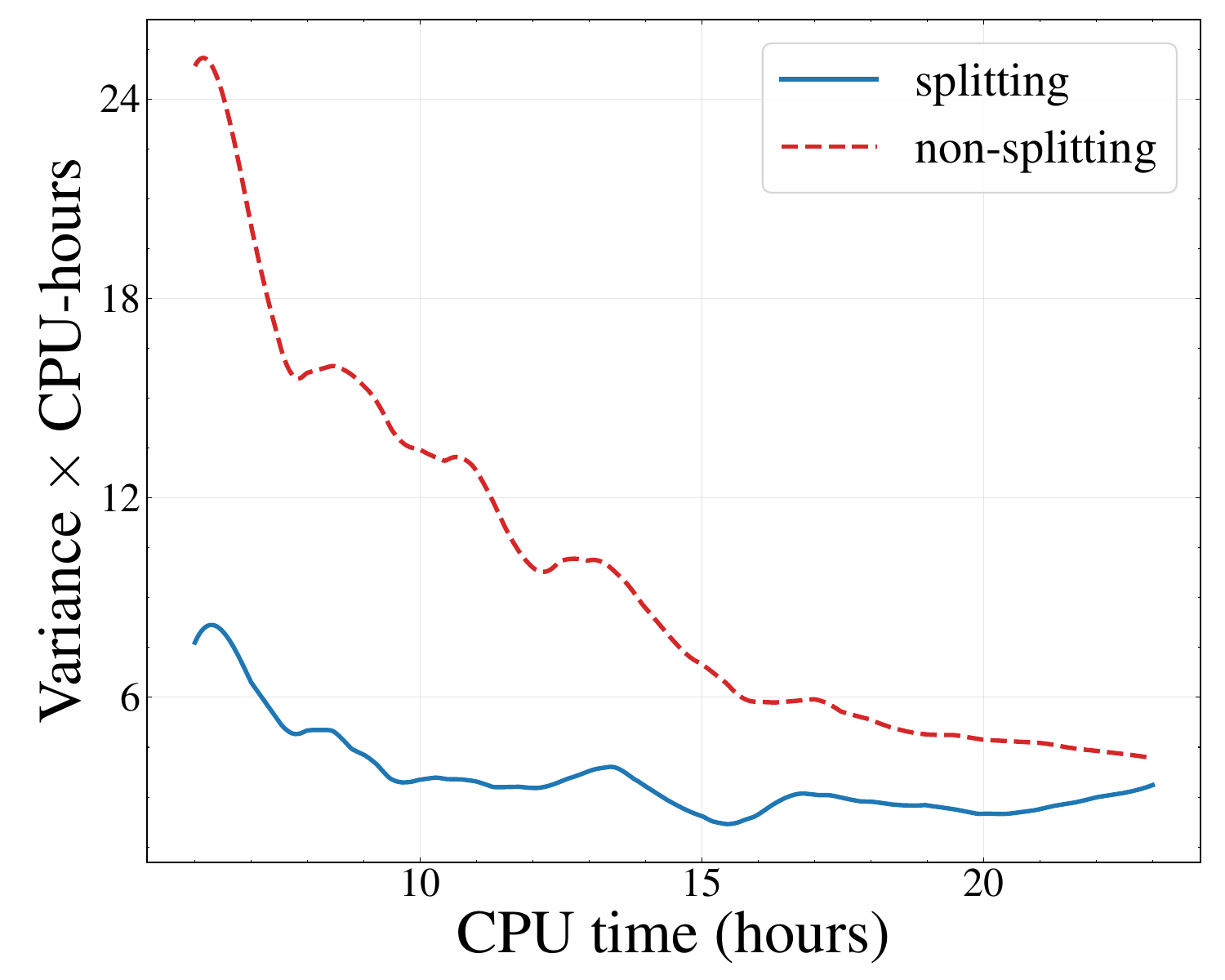}\\
    \makebox[0.235\textwidth]{\small (a) BDR of arrow}\hfill\makebox[0.235\textwidth]{\small (b) variance of arrow}\\[3pt]
    \includegraphics[width=0.235\textwidth]{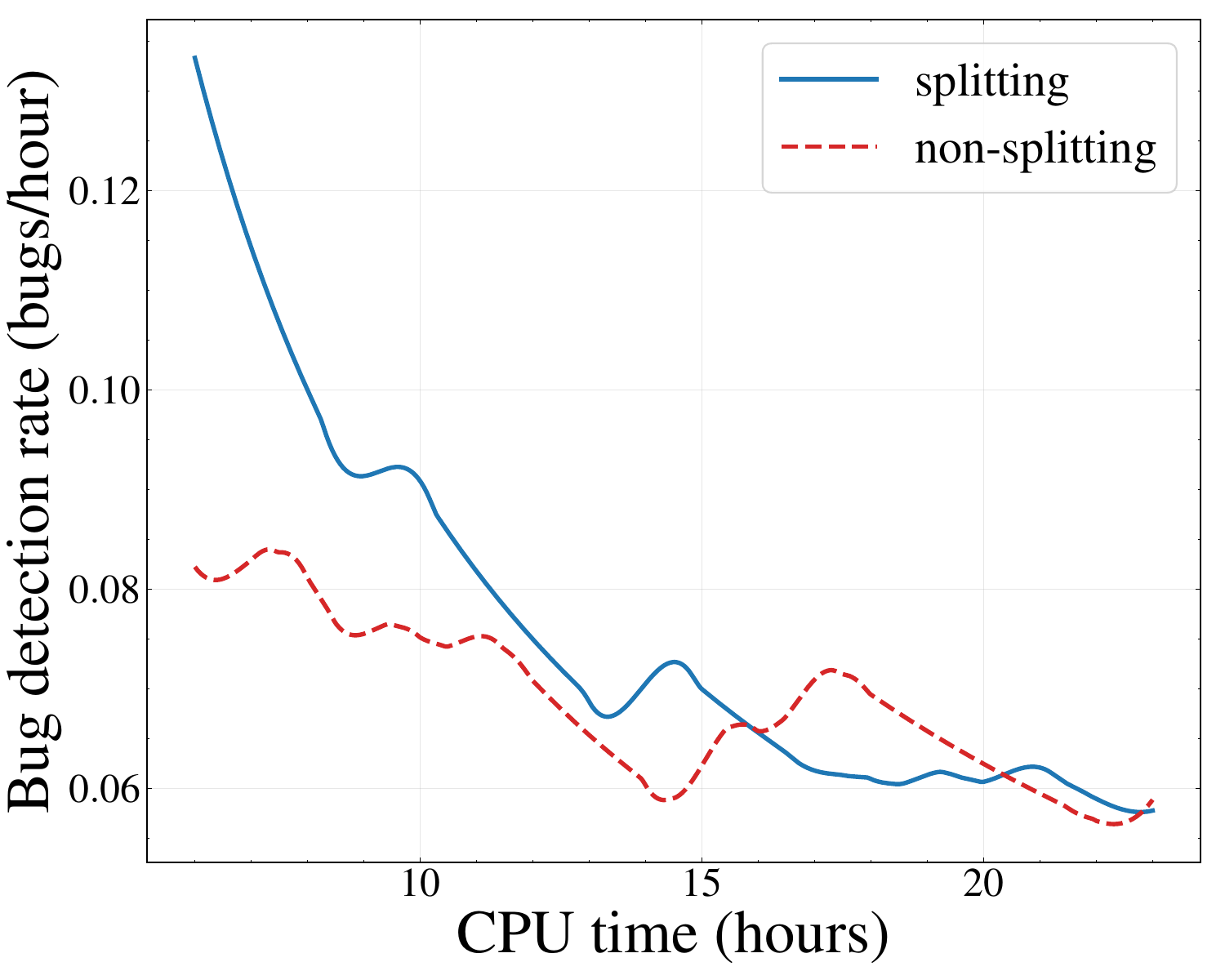}\hfill\includegraphics[width=0.235\textwidth]{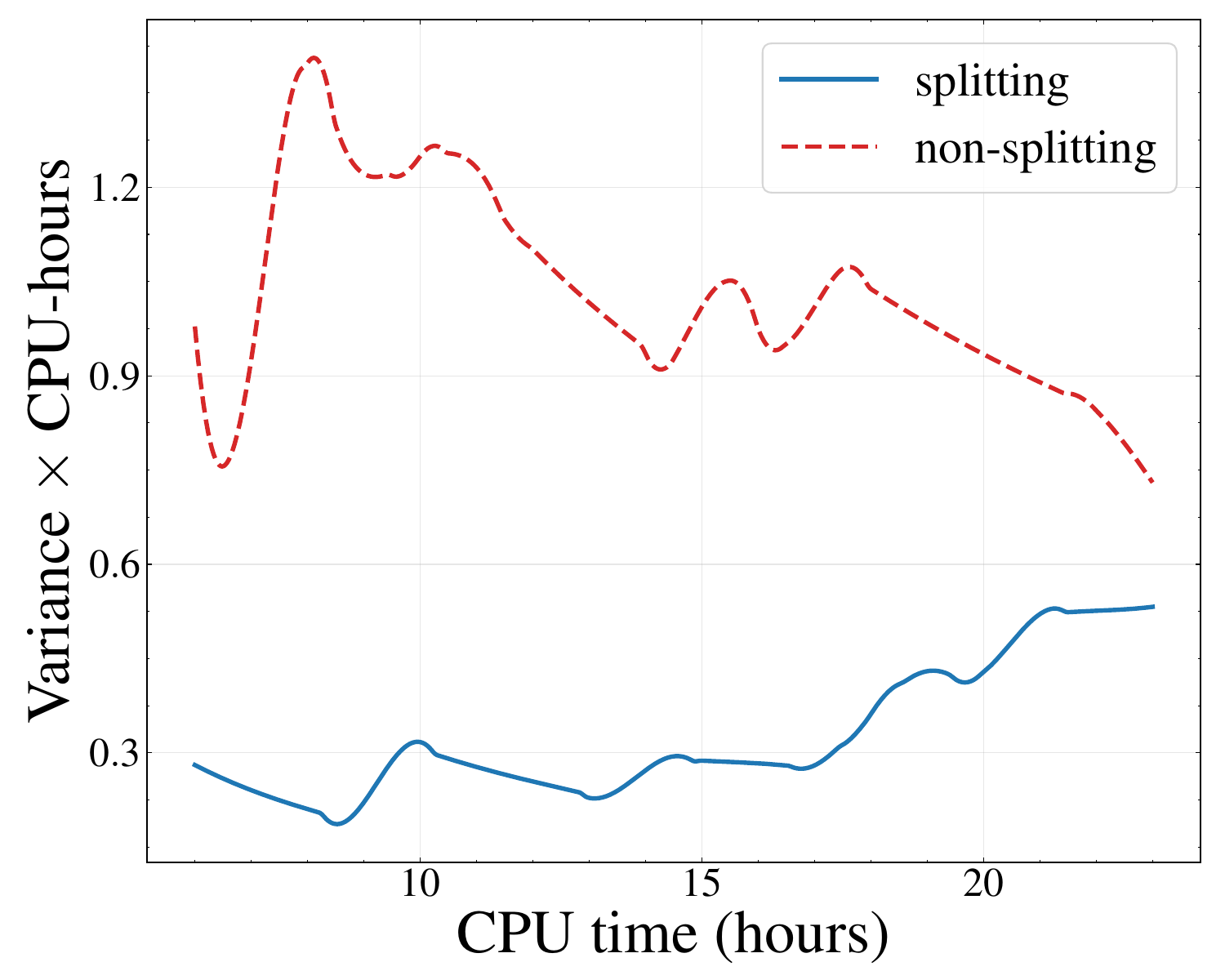}\\
    \makebox[0.235\textwidth]{\small (c) BDR of php}\hfill\makebox[0.235\textwidth]{\small (d) variance of php}\\[3pt]
    \includegraphics[width=0.235\textwidth]{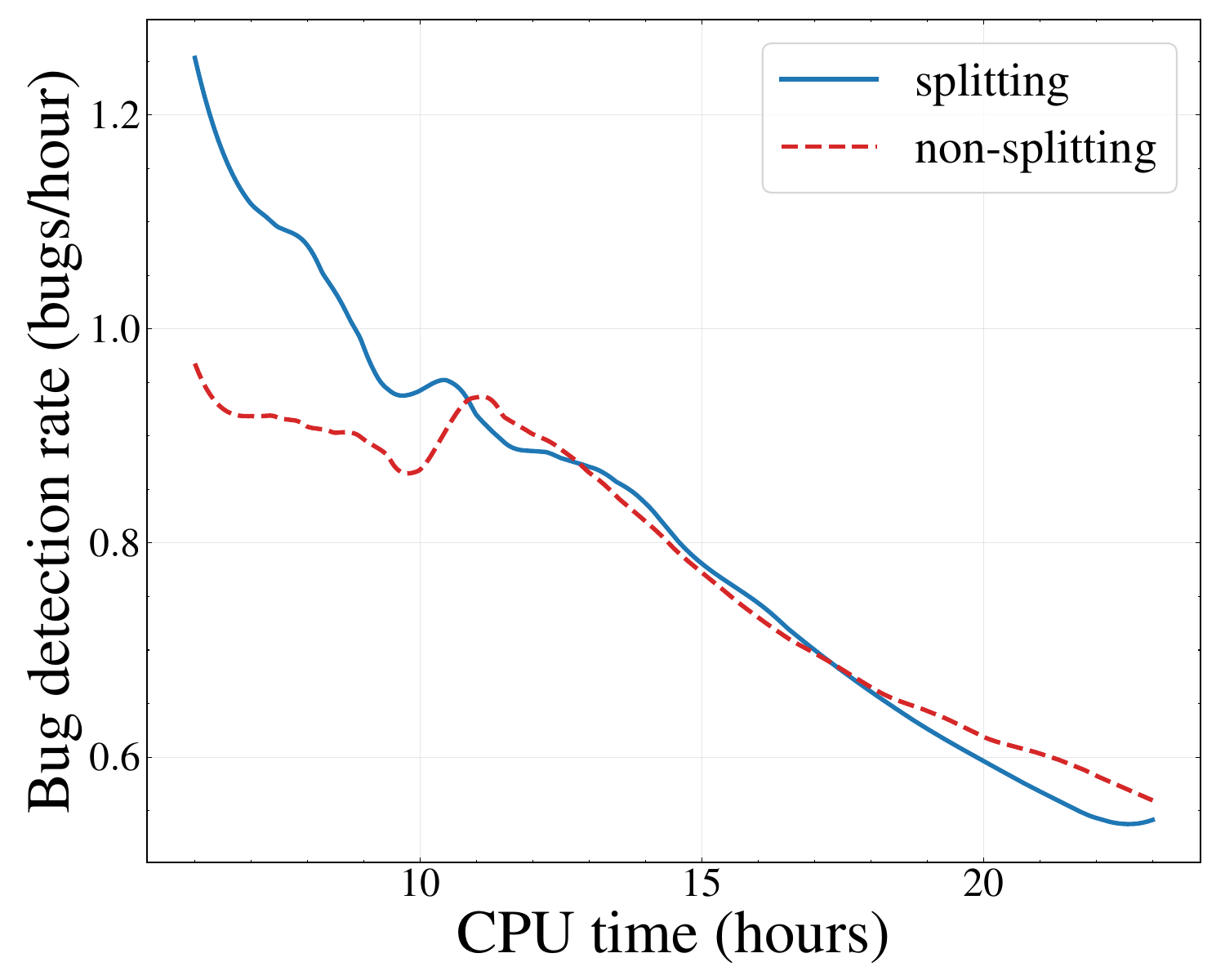}\hfill\includegraphics[width=0.235\textwidth]{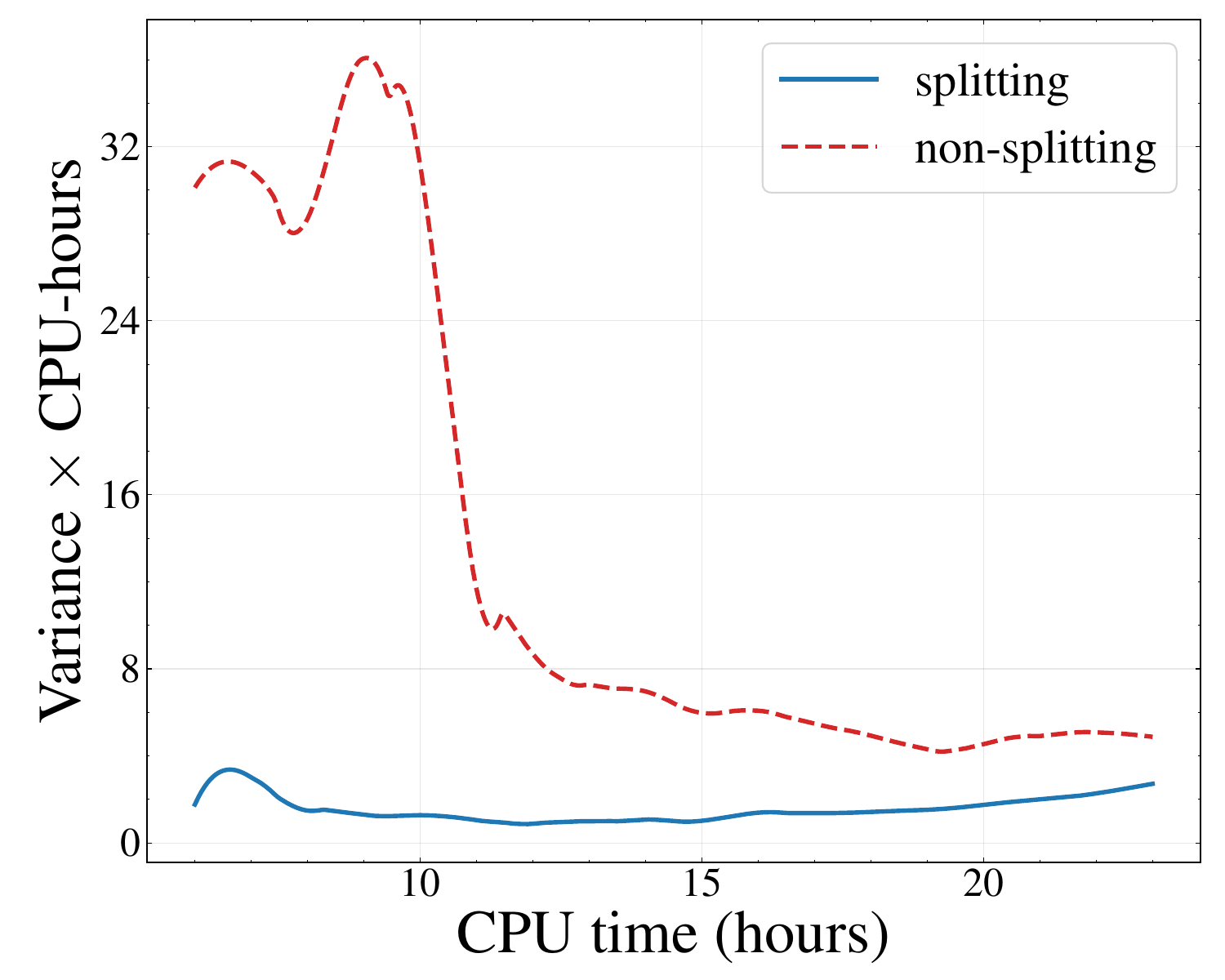}\\
    \makebox[0.235\textwidth]{\small (e) BDR of poppler}\hfill\makebox[0.235\textwidth]{\small (f) variance of poppler}\\[3pt]
    \includegraphics[width=0.235\textwidth]{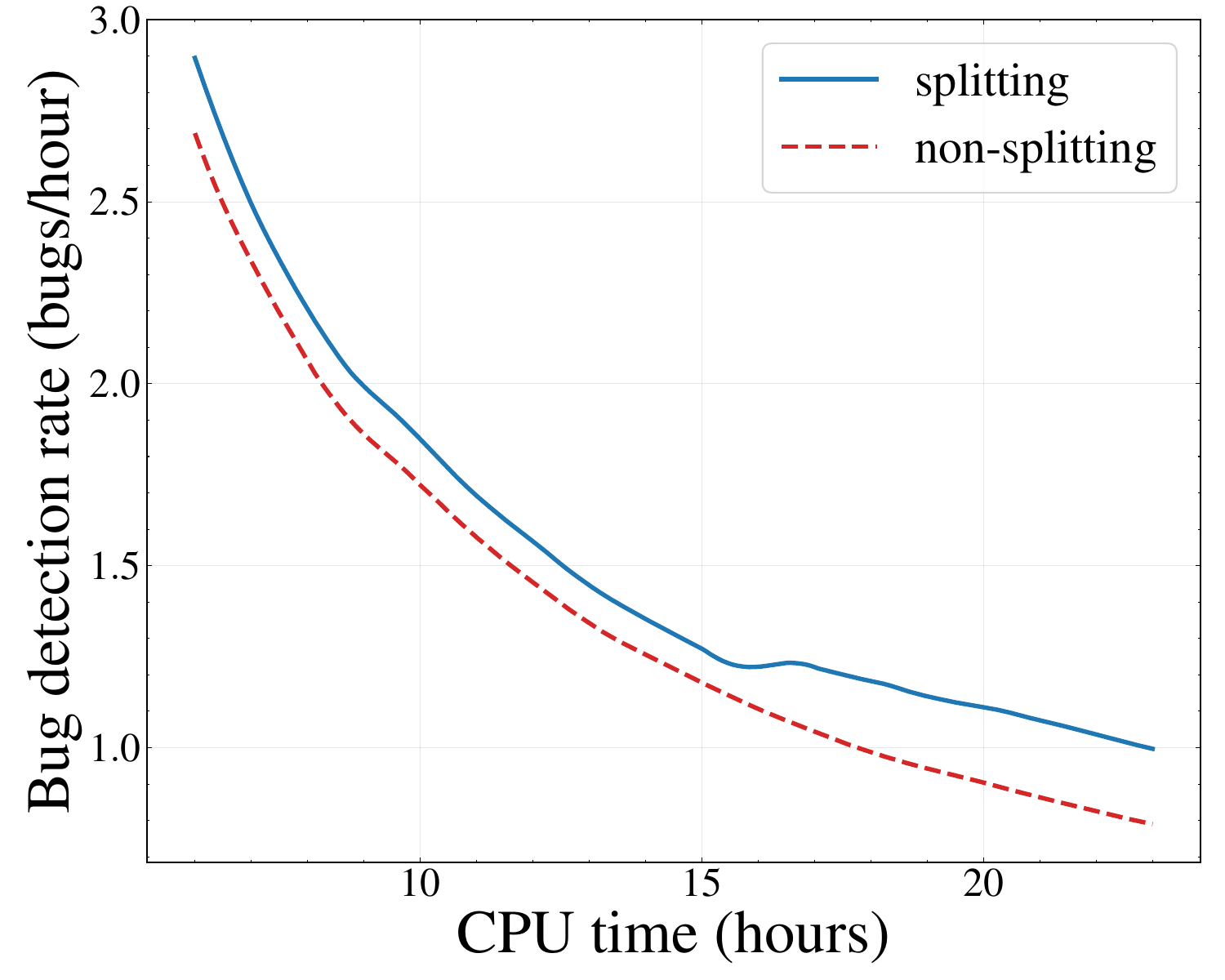}\hfill\includegraphics[width=0.235\textwidth]{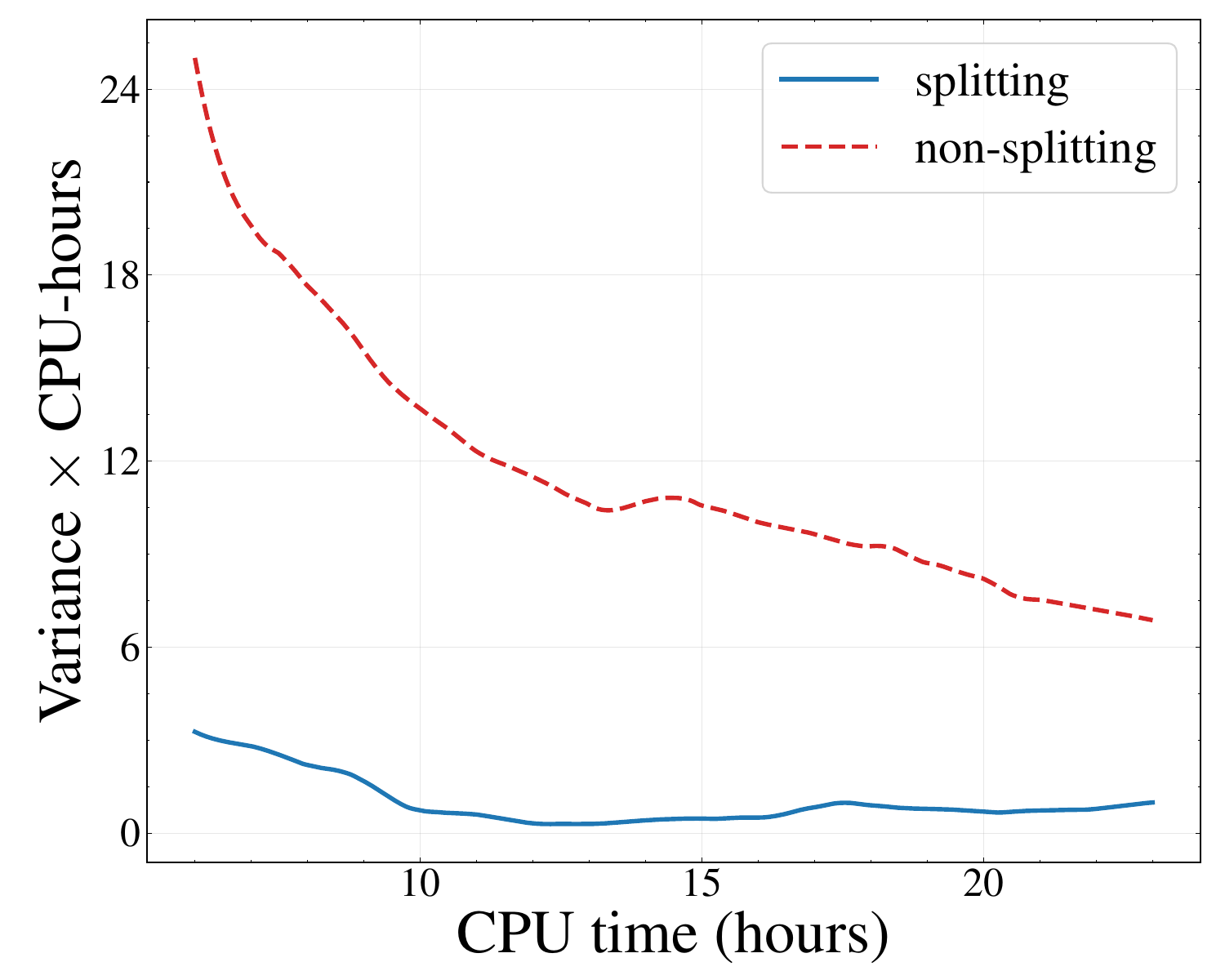}\\
    \makebox[0.235\textwidth]{\small (g) BDR of stb}\hfill\makebox[0.235\textwidth]{\small (h) variance of stb}
    \caption{Ablation study for fuzzer AFLSmart on arrow, php, poppler and stb.}
    \label{fig:ablation_aflsmart_1}
\end{figure}

\fi

%% file: section/proof.tex
\section{Discussions on the Assumptions}\label{app:discussion}

Although our analysis is built on an ergodicity assumption, this condition plays an important practical role: it rules out poorly performing fuzzers whose mutation processes become trapped in restricted regions of the input space, and our empirical results align with the theoretical predictions even when it is not enforced. Such assumptions are standard in the Markov chain literature; relaxing them is interesting future work.

Finally, multiple seeds are accommodated naturally: seeds executed in parallel are equivalent to multiple independent trials, already covered by the framework; seeds cycled over time can be modeled by augmenting the transition matrix with deterministic transitions to new seeds followed by restarts under the original dynamics, a regeneration structure.
\iffullversion
\section{Proof}\label{app:proof}

\subsection{Proof of Theorem~\ref{thm:splitting_var}}\label{app:pf-thm-splitting}
\begin{proof}[Proof of Theorem~\ref{thm:splitting_var}]
Conditional on the split times, and under the independent-replica idealization the theorem assumes,
\begin{align*}
\operatorname{Var}(\hat{\mu}_T)\le\frac{1}{T^{2}}\sum_t \frac{c_t}{k_t}.
\end{align*}
With $k_t\in\{1,2\}$ and $C_{\mathrm{in}},C_{\mathrm{out}}$ the in- and out-region autocovariance masses, we arrive at
\begin{align*}
\mathcal{E}_{\mathrm{split}}\le\frac{T+|S|}{T^{2}}\Bigl(C_{\mathrm{out}}+\tfrac12 C_{\mathrm{in}}\Bigr).
\end{align*}
Because the chain is reversible and positive, every lag-autocovariance of $g=\mathbf 1_B-\pi(f)$ is non-negative, so the same bound with $k_t\equiv1$ holds with equality:
\begin{align*}
\mathcal{E}_{\mathrm{MC}}=\frac1T\bigl(C_{\mathrm{out}}+C_{\mathrm{in}}\bigr),
\end{align*}
which supplies the matching lower bound. Dividing, with $r=C_{\mathrm{in}}/C_{\mathrm{out}}$ and $p=|S|/T$,
\begin{align*}
\frac{\mathcal{E}_{\mathrm{split}}}{\mathcal{E}_{\mathrm{MC}}}
\le\frac{T+|S|}{T}\cdot\frac{C_{\mathrm{out}}+\tfrac12 C_{\mathrm{in}}}{C_{\mathrm{out}}+C_{\mathrm{in}}}
=(1+p)\frac{2+r}{2(1+r)},
\end{align*}
which is below $1$ exactly when $p<r/(2+r)$. On the two-state surrogate $r=(1-p)/p$, so the ratio becomes $(1+p)^2/2$ and
\begin{align*}
\frac{(1+p)^2}{2}<1\iff p<\sqrt2-1 .
\end{align*}
\end{proof}

\subsection{Proof of Proposition~\ref{prop:tree}}\label{app:pf-prop-tree}
\begin{proof}[Proof of Proposition~\ref{prop:tree}]
The tree is fixed, so the variance is taken at that tree throughout. Same-lineage pairs contribute the within-chain autocovariance weights $c_t$, averaged over the $|L_t|$ active leaves. For $\ell\neq j$ with $r=r_{\ell j}$, conditional independence after $r$ gives
\begin{align*}
\operatorname{Cov}(Y_{\ell,t},Y_{j,u})=\operatorname{Cov}\bigl(P^{\,t-r}g(X_r),P^{\,u-r}g(X_r)\bigr).
\end{align*}
Bounding a product by the product of root-mean-square sizes and re-weighting from the fork law to $\pi$ at a cost of at most $\kappa$, we obtain
\begin{align*}
\bigl|\operatorname{Cov}(Y_{\ell,t},Y_{j,u})\bigr|
&\le\kappa\lVert P^{\,t-r}g\rVert_\pi\lVert P^{\,u-r}g\rVert_\pi\\
&\le\kappa\,\rho^{\,(t-r)+(u-r)}\operatorname{Var}_\pi(g),
\end{align*}
which gives the stated bound.
\end{proof}

\subsection{Proof of Theorem~\ref{bug_rate_splitting}}\label{app:pf-thm-bugrate}
\begin{proof}[Proof of Theorem \ref{bug_rate_splitting}]
Fix any realized split tree. At time $t$ there are $k_t\ge1$ live branches, indexed $\ell=1,\dots,k_t$, and the designated continuation path occupies branch $\ell_0(t)$; relabel so that $\ell_0(t)=1$. The two rates are
\[
\widehat P_{\mathcal T}=\frac1T\sum_{t=1}^{T}\frac1{k_t}\sum_{\ell=1}^{k_t}Y_{\ell,t},
\qquad
\widehat P_{0}=\frac1T\sum_{t=1}^{T}Y_{1,t}.
\]
Subtracting termwise and using $\frac1{k_t}\sum_{\ell=1}^{k_t}Y_{\ell,t}-Y_{1,t}=\frac1{k_t}\sum_{\ell=2}^{k_t}(Y_{\ell,t}-Y_{1,t})$ gives
\[
\widehat P_{\mathcal T}-\widehat P_{0}
=\frac1T\sum_{t=1}^{T}\frac1{k_t}\sum_{\ell=2}^{k_t}\bigl(Y_{\ell,t}-Y_{1,t}\bigr),
\]
which is the first identity; when $k_t=1$ the inner sum is empty. For the raw counts, the tree records every branch and the designated path only branch $\ell_0(t)$, so
\[
N_{\mathcal T}-N_0=\sum_{t=1}^{T}\sum_{\ell\neq\ell_0(t)}Y_{\ell,t}\ \ge\ 0,
\]
each summand being an indicator.
\end{proof}

\subsection{Proof of Proposition~\ref{prop:sizebias}}\label{app:pf-prop-sizebias}
\begin{proof}[Proof of Proposition~\ref{prop:sizebias}]
By definition $\mu_m(dx)=m(x)\mu(dx)/\mathbb{E}_\mu[m]$, so
\[
\begin{aligned}
\mathbb{E}_{\mu_m}[V]&=\frac{\mathbb{E}_\mu[mV]}{\mathbb{E}_\mu[m]},\quad\text{hence}\\
\mathbb{E}_{\mu_m}[V]-\mathbb{E}_{\mu}[V]&=\frac{\mathbb{E}_\mu[mV]-\mathbb{E}_\mu[m]\,\mathbb{E}_\mu[V]}{\mathbb{E}_\mu[m]},
\end{aligned}
\]
and the numerator is $\operatorname{Cov}_\mu(m,V)$ by definition of covariance. Since $\mathbb{E}_\mu[m]>0$, the difference is strictly positive exactly when $\operatorname{Cov}_\mu(m,V)>0$.
\end{proof}

\subsection{Proof of Theorem~\ref{thm:main_body}}\label{app:pf-thm-main}
\begin{proof}[Proof of Theorem~\ref{thm:main_body}]
\emph{Part 1.} The roots are independent, so $S^2_{M,R}(T)$ is the unbiased sample variance of the i.i.d.\ variables $Z_{m,R}(T)$. Expanding $\operatorname{Var}(S^2)$ in the second and fourth central moments, we obtain
\begin{align*}
\operatorname{Var}\bigl(T\,S^2_{M,R}(T)\bigr)=\frac1M\Bigl[\mu_{4,T}-\tfrac{M-3}{M-1}(v^\star_T)^2\Bigr].
\end{align*}
Since $\mathbb E[T\,S^2_{M,R}(T)]=v^\star_T$, Chebyshev's inequality gives, with probability at least $1-\alpha$,
\begin{align*}
\bigl|T\,S^2_{M,R}(T)-v^\star_T\bigr|\le E_M(\alpha).
\end{align*}

\emph{Part 2.} Expanding the variance of a time average, we arrive at
\begin{align*}
T\,v_R(T)&=\Gamma(0)+2\sum_{h=1}^{T-1}\bigl(1-\tfrac hT\bigr)\Gamma(h),\\
\sigma^2&=\Gamma(0)+2\sum_{h\ge1}\Gamma(h).
\end{align*}
Subtracting, we derive
\begin{align*}
B_R(T)&=\Bigl|2\sum_{h=1}^{T-1}\tfrac hT\,\Gamma(h)+2\sum_{h\ge T}\Gamma(h)\Bigr|\\
&\le\frac{2}{T}\sum_{h\ge1}h\,C\rho^{\,h}+2\sum_{h\ge T}C\rho^{\,h}
=\frac{2C\rho}{T(1-\rho)^{2}}+\frac{2C\rho^{T}}{1-\rho}.
\end{align*}

\emph{Part 3.} Under either start, $T\operatorname{Var}(Z_{1,R}(T))=\tfrac1T\sum_{s,t}\operatorname{Cov}(Y_s,Y_t)$. Hence
\begin{align*}
b_{\mathrm{start}}(T)&\le\frac1T\sum_{s,t}\bigl\lvert\operatorname{Cov}_{\pi_0}(Y_s,Y_t)-\Gamma(t-s)\bigr\rvert\\
&\le\frac1T\Bigl[\sum_{s}C_0\rho^{\,s-1}+2\sum_{s<t}C_0\rho^{\,s-1}\rho^{\,t-s}\Bigr]\\
&=\frac{C_0}{T(1-\rho)}+\frac{2C_0\rho}{T(1-\rho)^{2}}.
\end{align*}
\emph{Part 4.} Combining Parts 1 to 3 by the triangle inequality, we obtain $A_R(T)\le E_M(\alpha)+b_{\mathrm{start}}(T)+U_{B,R}(T)$. Applying Part 1 at level $\alpha/|\mathcal R|$ and a union bound completes the proof.
\end{proof}

\subsection{The exact starting-state term of Theorem~\ref{thm:main_body}}\label{app:pf-start}
\begin{proof}[Proof of the exact starting-state term in Theorem~\ref{thm:main_body}, part~3]
For the two-state chain the $t$-step marginal from start mass $p_0$ is $m_t=\pi_1+D\lambda^{t-1}$ with $D=p_0-\pi_1$. Hence
\begin{align*}
\operatorname{Var}(f(X_t))&=m_t(1-m_t),\\
\operatorname{Cov}(f_s,f_t)&=m_s\bigl(\pi_1+(1-\pi_1)\lambda^{t-s}\bigr)-m_sm_t,\qquad s<t.
\end{align*}
Subtracting the stationary case $D=0$ and summing over $1\le s\le t\le T$, the $\lambda^{t-s}$ cross-terms cancel and we obtain
\begin{align*}
\operatorname{Var}_{\pi_0}\bigl(\textstyle\sum f\bigr)-\operatorname{Var}_\pi\bigl(\textstyle\sum f\bigr)=D(1-2\pi_1)L_T-D^2Q_T,
\end{align*}
where $L_T=\sum_{t}(2t-1)\lambda^{t-1}$ and $Q_T=\sum_t\lambda^{2(t-1)}+2\sum_{s<t}\lambda^{(t-1)+(s-1)}=\bigl(\sum_{i=0}^{T-1}\lambda^i\bigr)^2$. Dividing by $T$ gives $b_{\mathrm{start}}(T)$.

Write $g(D)=D(1-2\pi_1)L_T-D^2Q_T$, so that $b_{\mathrm{start}}(T)=|g(D)|/T$. Since $g$ is a downward parabola, the maximum of $|g|$ over the range $D\in[-\pi_1,1-\pi_1]$ of any start is attained at
\begin{align*}
D_\star=\frac{(1-2\pi_1)L_T}{2Q_T}\quad\text{or at an endpoint,}
\end{align*}
so $\max_{p_0}b_{\mathrm{start}}(T)$ is available in closed form and dominates $b_{\mathrm{start}}(T)$ for every start.
\end{proof}

\subsection{Proof of Lemma~\ref{lem:gap-transfer}}\label{app:pf-lem-gap}
\begin{proof}[Proof of Lemma~\ref{lem:gap-transfer}]
Write $g=f-\pi(f)$. By definition,
\begin{align*}
\operatorname{Cov}_\pi\bigl(f(X_0),f(X_t)\bigr)=\langle g,P^{t}g\rangle_\pi .
\end{align*}
For a reversible chain this quantity contracts by at least $1-\gamma$ at each step, so
\begin{align*}
\bigl|\langle g,P^{t}g\rangle_\pi\bigr|\le(1-\gamma)^{t}\lVert g\rVert_\pi^{2}=(1-\gamma)^{t}\operatorname{Var}_\pi(f).
\end{align*}
Summing the geometric series gives the stated bound on $\sigma^2$.
\end{proof}

\subsection{Robustness of \texorpdfstring{$U_{A,R}$}{U\_A} to the plug-in gap and to drift}\label{app:gap-robust}
Write $U(g)$ for the computable form of Theorem~\ref{thm:main_body} behind the reported numbers, in which the initialization term is written in its worst-case form through the spectral gap, $U(g)=\tfrac{1}{\sqrt M}+C_1\tfrac{(1-g)(2-g)}{g^2}+\tfrac{1}{T}\tfrac{1}{g^2}+\tfrac{(1-g)^T}{g}$, with $C_1=\tfrac14+\tfrac{1}{16\pi(1)^2}>0$, $g=\gamma^X_R\in(0,1]$, $M,T$ fixed.

\noindent\emph{Monotonicity.} Termwise, $\tfrac{d}{dg}\tfrac{(1-g)(2-g)}{g^2}=\tfrac{3g-4}{g^3}<0$, $\tfrac{d}{dg}\tfrac{1}{g^2}=-\tfrac{2}{g^3}<0$, and $\tfrac{d}{dg}\tfrac{(1-g)^T}{g}=-\tfrac{(1-g)^{T-1}((T-1)g+1)}{g^2}\le0$ on $(0,1]$. Hence $U'(g)<0$: $U$ is strictly decreasing. Therefore $\hat\gamma\le\gamma\Rightarrow U(\hat\gamma)\ge U(\gamma)\ge A_R(T)$ on the $1-\delta$ event, i.e.\ under-estimating the gap is conservative.

\noindent\emph{Lipschitz sensitivity.} On $[g_0,1]$, $|U_2'|\le 4C_1/g_0^3$, $|U_3'|\le 2/(Tg_0^3)$, and (using $\sup_{n\ge0}n x^n\le \tfrac{1}{e\ln(1/x)}\le\tfrac{1}{eg}$ for $x=1-g$) $|U_4'|\le (1+e^{-1})/g_0^2$, so $|U(\hat g)-U(g)|\le L(g_0)|\hat g-g|$ with $L(g_0)=\tfrac{4C_1+2/T}{g_0^3}+\tfrac{1+e^{-1}}{g_0^2}=\Theta(g_0^{-3})$. The bias $b(T,g)=U(g)-M^{-1/2}$ has the same envelope, so the closed-form $M(\gamma)=\log(2/\delta)/(c(\varepsilon-b)^2)$ obeys $|M(\hat\gamma)-M(\gamma)|\le \tfrac{2\log(2/\delta)}{c\,m_0^3}L(g_0)|\hat\gamma-\gamma|$ for $m_0=\varepsilon-\max(b(T,\hat\gamma),b(T,\gamma))>0$; and $\hat\gamma\le\gamma\Rightarrow M(\hat\gamma)\ge M(\gamma)$ (under-estimation only raises the prescribed $M$).

\noindent\emph{Validity when the output is not a two-state chain.} If $\hat\gamma\le\gamma_{\mathrm{input}}$, that is, the two-state fit under-states the input chain's gap, then by Lemma~\ref{lem:gap-transfer} $|\mathrm{Cov}(Y_0,Y_t)|\le\mathrm{Var}_\pi(f)(1-\gamma_{\mathrm{input}})^t\le\mathrm{Var}_\pi(f)(1-\hat\gamma)^t$ with $\mathrm{Var}_\pi(f)=\pi(1)\pi(0)\le\tfrac14$, so the autocovariance-tail terms hold with $\gamma:=\hat\gamma$ and $U(\hat\gamma)$ is a valid conservative bound, with no assumption that the bug indicator is itself a Markov chain.

\noindent\emph{What is estimated under drift.} Under piecewise homogeneity, $\sigma_w^2=\mathrm{Var}_{\pi_w}(f)\phi(\gamma_w)$, $\phi(\gamma)=\tfrac{2-\gamma}{\gamma}$, $|\phi'|=2/\gamma^2$. For $\bar\sigma^2=\sum_w a_w\sigma_w^2$ and any $w^\star$, $|\bar\sigma^2-\sigma_{w^\star}^2|\le\max_w|\sigma_w^2-\sigma_{w^\star}^2|\le \tfrac{2\bar V}{g_0^2}|\Delta\gamma|+\phi(g_0)V_\Delta$ with $\bar V=\max_w\mathrm{Var}_{\pi_w}(f)\le\tfrac14$ and $V_\Delta=\max_{w,w'}|\mathrm{Var}_{\pi_w}(f)-\mathrm{Var}_{\pi_{w'}}(f)|$. Applying Theorem~\ref{thm:main_body} within each window and a union bound over the $n_{\mathrm{win}}$ windows gives $A_R^{(w)}\le U(\gamma_w)$ simultaneously w.p.\ $1-\delta$, hence aggregate deviation $\le\max_w U(\gamma_w)=U(g_0)$ by monotonicity. (Shorter windows reduce the drift bias but inflate $U$ through the $1/W$ and $(1-\gamma)^W$ horizon terms, a bias/variance trade-off in $W$.)

\subsection{Bug Detection Rate of Splitting}

\noindent Both identities are pathwise: they hold for the realized tree, with no assumption on how the fork times were chosen. Setting $k_t\in\{1,2\}$ recovers the two-way fork, for which the first identity reads $\frac1{2T}\sum_{t\in A}(Y_{2,t}-Y_{1,t})$ with $A=\{t:k_t=2\}$.

\fi

\smallskip
\noindent Two consequences. First, the bound depends on the tree only through the ancestry times $r_{\ell j}$, so it covers any shape: $K$ stages, branching factor $B$, up to $B^{K}$ concurrent leaves, and descendants that persist to the horizon. Second, $\kappa$ makes the effect of the starting point explicit: a fork occurs at a bug-trigger state, whose law is not $\pi$, and on the two-state surrogate $\kappa=1/\pi(B)$. Setting $K=1$, $B=2$ with a single fork and $\kappa=1$ recovers Theorem~\ref{thm:splitting_var} and hence the closed-form rule $p<\sqrt2-1$.